\documentclass[11pt, a4paper, fleqn]{report}

\usepackage{verbatim}
\usepackage{amsfonts}
\usepackage{amsmath}
\usepackage{amssymb}
\usepackage{amsthm}
\usepackage{latexsym}
\usepackage{pstricks}
\usepackage{makeidx}
\usepackage{hyperref}
\usepackage{wasysym}
\usepackage{slashed}
\usepackage[dvips]{graphicx}
\usepackage[left=4.0cm, top=3.9cm, right=4.0cm, bottom=3.9cm]{geometry}
\usepackage[none]{hyphenat}
\usepackage{cite}
\usepackage{caption}
\usepackage{subcaption}
\usepackage[cal=dutchcal,calscaled=.88]{mathalfa}
\usepackage{xcolor}
\definecolor{Green}{rgb}{0,0.4,0}

\DeclareMathAlphabet{\mathpzc}{OT1}{pzc}{m}{it}

\newcommand{\ud}{\mathrm{d}}
\newcommand{\diag}{\mathrm{diag}}

\newcommand{\Tr}{\,\mathrm{Tr}}
\renewcommand{\Re}{\re}
\renewcommand{\Im}{\im}
\newcommand{\re}{\mathrm{Re}}
\newcommand{\im}{\mathrm{Im}}

\newcommand{\vol}{\mathrm{vol}}

\newcommand{\sgn}{\mathrm{sgn}}

\newcommand{\smallo}{\scriptstyle\mathcal{O}\displaystyle}
\newcommand{\Li}{\mathrm{Li}}

\newcommand{\mathscr}{\mathcal}

\let\oldsqrt\sqrt
\def\sqrt{\mathpalette\DHLhksqrt}
\def\DHLhksqrt#1#2{%
\setbox0=\hbox{$#1\oldsqrt{#2\,}$}\dimen0=\ht0
\advance\dimen0-0.2\ht0
\setbox2=\hbox{\vrule height\ht0 depth -\dimen0}%
{\box0\lower0.4pt\box2}}

\numberwithin{equation}{chapter}

\theoremstyle{definition}
\newtheorem{dfnt}{Definition}[chapter]

\theoremstyle{plain}
\newtheorem{lemma}{Lemma}[chapter]
\newtheorem{thrm}{Theorem}[chapter]
\newtheorem{prps}{Proposition}[chapter]
\newtheorem{cor}{Corollary}[chapter]

\theoremstyle{remark}
\newtheorem{rmk}{Remark}[chapter]
\newtheorem{exm}{Example}[chapter]
\newtheorem*{rmk*}{Remark}

{

\parindent 0pt

\usepackage{fouriernc}

\title{Partition function methods for the quartic scalar quantum field theory on the Moyal plane}
\author{vorgelegt von\\Jins de Jong\\Delfzijl}
\date{\today}

\begin{document}
\sloppy
\begin{titlepage}
\centering
{\scshape\LARGE Mathematik \par}
\vspace{4cm}
{\huge\bfseries Partition function methods for the quartic scalar quantum field theory on the Moyal plane\par}
\vspace{4cm}

Inaugural-Dissertation zur Erlangung des Doktorgrades der Naturwissenschaften im Fachbereich Mathematik und Informatik der Mathematisch-Naturwissenschaftlichen Fakult\"at der Westf\"alischen Wilhelms-Universit\"at M\"unster 

\vspace{2cm}
{vorgelegt von\par}
{\Large\itshape Jins de Jong\par}
{aus \itshape Delfzijl\par}

  \vfill

{\large - 2018 -\par}
\end{titlepage}

\newpage
\begin{titlepage}
\begin{flushleft}
{\par}
\vspace*{4cm}
{\par}
\vfill

\begin{tabular}{ll}
Dekan:&Prof. Dr. Xiaoyi Jiang\\
Erster Gutachter:&Prof. Dr. Raimar Wulkenhaar\\
Zweiter Gutachter:&Prof. Dr. Harald Grosse\\
Tag der m\"undlichen Pr\"ufung:\phantom{Ruimte}&10.10.2018\\
Tag der Promotion:&10.10.2018\\
\end{tabular}	
\end{flushleft}
\end{titlepage}

\newpage
\tableofcontents

\chapter*{Summary}
\addcontentsline{toc}{chapter}{Summary}

Since the first definition of the Wightman axioms a search for a four dimensional nontrivial quantum field theory has been going on. In recent years much progress has been booked, especially on the Grosse-Wulkenhaar model. To understand this particular model and the Wightman axioms better, the Grosse-Wulkenhaar model is studied using partition function methods in this work.\\

In general, the partition function is useful only for very simple models or for approximations. However, the properties of the Grosse-Wulkenhaar model suggest that an exact treatment of the partition functions in this model may be possible. One indication of this is the Kontsevich model, which can be solved completely. The differences between these partition functions are the main cause of the technical challenges here.\\

In this work partition function methods for the Grosse-Wulkenhaar model without source term have been discussed, although several aspects can be translated directly to this model with a source term or other single-matrix models quantum field theories with varying kinematics. After the diagonalisation of these matrices the partition functions is factorised using the asymptotic volume of diagonal subpolytopes of symmetric stochastic matrices. A consequence of this is that the free energy density before renormalisation can be determined nonperturbatively for weak coupling. This means that it never has to be assumed that the coupling is zero. Additionally, some modifications for strong coupling are discussed.

\newpage
\chapter*{Zusammenfassung}
\addcontentsline{toc}{chapter}{Zusammenfassung}

Seit der ersten Definition der Wightman-Axiome wird ein Beispiel einer nicht-trivialen Quantenfeldtheorie in vier Dimensionen gesucht. Erhebliche Fortschritte in diesem Bereich sind w\"ahrend der letzten Jahre erreicht worden, speziell bez\"uglich des Grosse-Wulkenhaar-Models. Damit sowohl dieses Model, als auch die Axiome besser verstanden werden, ist in dieser Arbeit das Grosse-Wulkenhaar-Model durch die direkte Berechnung der Zustandsumme untersucht worden.\\

Generell ist die Zustandsumme nur n\"utzlich f\"ur besonders einfache Theorien oder N\"aherungsmethoden. Es sind, allerdings, die Eigenschaften des Grosse-Wulkenhaar-Models, die eine exakte Behandlung der Zustandsumme m\"oglich erscheinen lassen. Ein Indiz daf\"ur ist die L\"osung des Kontsevich-Models. Die Unterschiede zwischen den Zustandsummen beider Modelle verursachen erhebliche technische Schwierigkeiten bei der Berechnung.\\

In dieser Arbeit werden Zustandsummemethoden f\"ur das Grosse-Wulkenhaar-Model ohne Quellenterm besprochen. Viele Aspekte k\"onnen jedoch ohne Probleme f\"ur das Model mit Quellenterm oder quantenfeldtheoretische Matrizenmodelle mit variierender Dynamik \"ubertragen werden. Nach Diagonaliserung der Matrizen wird die Zustandsumme mit Hilfe des asymptotischen Volumens des diagonalen Subpolytops von symmetrischen stochastischen Matrizen faktorisiert. Eine Konsequenz dessen ist, dass die Freie-Energiedichte vor Renormierung ohne St\"orungstheorie f\"ur schwache Kopplung bestimmt werden kann. An keiner Stelle muss angenommen werden, dass die Kopplung verschwindet. Zus\"atzlich werden Varianten dieser Methode f\"ur starke Kopplung besprochen.


\chapter{Introduction\label{sec:Intro}}

\section{Quantum field theory}
The development of quantum mechanics in the first decades of the twentieth century has also made the limitations of the theory clear. It does not accurately describe quantum physics in relativistic circumstances. Quantum field theory is the attempt to unify quantum mechanics and special relativity. Forced by the technical challenges of this attempt and inspired by its success in quantum mechanics, a perturbative approach to quantum field theory was embraced. Although not without complications, this approach has been extremely successful. It has explained experimental results up to unprecedented precision with very few exceptions.

\subsection{Path integral formulation}
The path integral formulation by Feynman is one of the most popular approaches to perturbative quantum field theory (\textsc{qft}). It provides the physicist with an interpretation of the abstract concepts involved and works for large classes of models. Without mathematical rigour a minimal introduction of this formulation, emphasising the reasoning and concepts, can be given. This is done in the Euclidean formulation, where the `time' argument used is the Wick rotated physical time.\\
A model, typically called a theory in the jargon and in the rest of this thesis, is in this context a set of fields $\{\varphi\}$ together with an action functional $S$ that assigns to every field an energy content. The theory is designed to describe some physical system in time and space and the set of fields is the set of all possible states this system may adopt. For every state $\varphi$ the action functional provides an energy content $S[\varphi]$, which is then translated to a (generalised) probability $p=e^{-S[\varphi]}$. This means that in this formulation a theory is a set of states with a probabilty distribution.\\

To every field in the theory a particle is associated. This is an interpretational step, which translates between physical experiment and theoretical model. A consequence is that the set of fields also describes all possible particles states. These should contain the quantum analogons of at least position and velocity.\\

An experiment on such a system typically then takes the following form: ``Given an initial condition $X$, what is the probability that system will evolve into final condition $Y$?'' This question will be formulated as the event $\langle Y,X\rangle$. The prediction of this experiment is then the expectation value of this event
\begin{equation*}
\mathbb{E}[\langle Y,X\rangle]=\frac{1}{\sum_{\varphi}\,e^{-S[\varphi]}}\sum_{\varphi}\langle Y,X\rangle\,e^{-S[\varphi]}\quad.
\end{equation*}
A source $J$ is a technical tool to simplify this process. In practice, all interesting events can be generated by acting with derivatives with respect to the source field on the partition function. The partition function is the sum of states
\begin{equation*}
Z[J]=\sum_{\varphi}e^{-S[\varphi]+J\varphi}
\end{equation*}
and its logarithm is called the free energy. For any event a suitable combination of derivatives with respect to $J$ can be found. This shows that the entire theory can be derived from the free energy. The free energy contains all physical information about a theory.\\

It has not been made clear what the fields mathematically are. This is on purpose and is also the flexibility of this formulation. There are many possible mathematical objects that can serve as quantum field, provided that it satisfies all physical requirements. As long as every step in the process from defining a theory to producing the experimental predictions is well-defined, there is complete freedom to choose what to work with.

\subsection{Locality}
One of these physical requirements is locality. The physical reasoning behind this is that any interaction requires the involved parts must be nearby in both time and space. To touch someone it is necessary to be in the same room at the same time. Another option is interaction through a medium, which also satisfies this localisation condition. To write  a letter one must sit at a desk to write the letter, bring it to the letter box and the receiver must read it at a certain moment at the kitchen table.\\

A typical example of a local theory in the Euclidean path integral formulation is the $\varphi^{4}$-theory, given by the actional functional
\begin{align}
&S[\varphi]=\frac{1}{\hbar}\int \ud^{d} x\, \mathcal{L}_{\varphi}(x)\qquad\text{and}\nonumber\\
&\mathcal{L}_{\varphi}(x) = \frac{1}{2}\sum_{\mu=1}^{d} \partial_{\mu}\varphi(x)\,\partial_{\mu}\varphi(x)+\frac{m^{2}}{2}\varphi^{2}(x)+g\varphi^{4}(x)\label{e:phi4}
\end{align}
and to leave the set of fields unspecified, assuming that a set exists fulfilling all requirements. The parameter $d$ is the dimension of spacetime and in standard physical situations $d=4$. It is furthermore standard to use so-called natural units, $\hbar=c=1$. The quantum of action and the speed of light are set equal to one, so that all actions and speeds are given in multiples/fractions of these.\\

The step from quantum mechanics to quantum field theory means that special relativity must be incorporated into the framework. Special relativity follows from Einstein's insight that the observed physics may not depend on the inertial frame chosen to conduct the experiment in. This means that a Lorentz transformation, which encodes a change of inertial frame, should leave the observed physics invariant. The standard way to satisfy this demand is to use local Lagrangian densities $\mathcal{L}$, such as the one in (\ref{e:phi4}). These depend only on the fields at one specific point. The action functional is obtained from this by integrating over al possible points. In this way there is no geometry that may be deformed by a Lorentz transformation.

\subsection{Problems of perturbative QFT}
So far, this has been an optimistic story about a physical theory. However, the physical practice is not as polished. In all common theories the steps from the definition of the theory, formally writing down the partition function, to the actual experimental predictions cannot be made without many significant assumptions, which often lack a mathematical justification. An example of such an assumption is the size of the coupling constant in $4$-dimensional theories. In perturbation theory the size of the coupling constant is used as an expansion parameter and assumed to be small. However, the number of contributing terms, the number of Feynman diagrams, increases much faster than the expansion parameter decreases. The only way to make sense of the partition function is to set the coupling to zero, removing the interaction from the theory. This would suggest that the only existing theory is the theory without interaction, also called the free or trivial theory.\\

Also the set of states is often left implicit, so that the partition function is not well-defined. As will be explained in the next paragraph, this in itself is not an insurmountable problem, but it implies that the partition function cannot be used to study the theory. From the partition function one cannot show in these cases that the theory satisfies the physical requirements of a quantum field theory or demonstrate its characteristics, such as locality.\\
Nevertheless, all the perturbative quantum field theories that have been so successful in describing the physics of subatomic particles at relativistic energies are commonly introduced by their partition function. And almost every measurable consequence drawn from this formulation has been verified experimentally. It seems that the station where this method could be called a ``lucky shot'' is long passed.

\section{Axiomatic quantum field theory}
Soon after the first quantum mechanical theories a mathematical framework was written down that captured the entire theory. From this mathematical side the consistency and general structure are must better understood. This beautiful and very pleasing situation has influenced the development of \textsc{qft} strongly. Besides the physical perturbative \textsc{qft} a mathematical discipline emerged, focusing on consistency and demonstrating the physical requirements of a quantum field theory with mathematical rigour. The first step is to give a precise formulation of the physical requirements. These are then called axioms and they define mathematically what is meant by a quantum field theory. These axioms should mimic the physical requirements in the traditional sense in some way.\\

The second step is to produce examples of this definition and analyse them. Ideally, an example is found that leads to experimentally verifiable results. In this case, both branches of \textsc{qft} would be reunited. It should be stressed that the precise formulation of the axioms is subordinate to the experimental results. Depending on the preferred mathematical methods to use, either the Wightman axioms or the Haag-Kastler axioms~\cite{haag1} can be taken as the definition of a \textsc{qft}. A different set of axioms could be used as well, as long as the mimic the physical requirements and they are form a mathematical definition of a \textsc{qft}.

\subsection{Wightman axioms}
The Wightman axioms~\cite{wightman1,streater,wightman2} are the most common axioms to define a \textsc{qft} on Minkowski space. They describe a separable Hilbert space of pure states and fields as operator-valued tempered distributions that transform under some representation of the Lorentz group. The fields must satisfy either commutation or anticommutation relations. Furthermore, the Hilbert space must include a unique state, the vacuum, which is invariant under the unitary transformations. Using this technical machinery some reconstruction theorems can be proved. The Wightman reconstruction theorem shows that the Hilbert space and the quantum fields can be reconstructed from the complete set of $n$-point functions under certain conditions. It brings the axioms therefore much closer to the observable results.\\

Another reconstruction theorem is the Osterwalder-Schrader theorem~\cite{osterwalder1,osterwalder2}. They are formulated in Schwinger functions, which are analytic continuations of the $n$-point functions to a Euclidean space. The Schwinger functions are given by
\begin{equation*}
\langle\varphi(x_{1}),\ldots,\varphi(x_{n})\rangle=\frac{1}{Z[0]}\sum_{\varphi}\varphi(x_{1}),\ldots,\varphi(x_{n})e^{-S[\varphi]}\quad.
\end{equation*}
They are the moments of the probabilty distribution mentioned before. If these functions are Euclidean covariant, symmetric under permutation and satisfy property called reflection positivity, then they lead via this reconstruction theorem to a \textsc{qft} in the sense of the Wightman axioms. This is the mathematical equivalent of the Wick rotation common in perturbative \textsc{qft}.\\

Working on a Euclidean space, instead of on Minkowski space, has the advantage that the number of applicable mathematical techniques to treat the partition function rigourously is much larger. The price paid for this is that computed quantities lose their straightforward physical interpretation. In most case we will refer to the Schwinger functions as $n$-point functions as well, since we will always be working in the Euclidean context.\\

The lesson taught by these reconstruction theorems is that a \textsc{qft} may be defined from the $n$-point functions, provided that all required conditions on these can be proved for them. There are no requirements on the action or the set of fields, so that the formulation's flexibility can be maintained.\\

Both the Wightman axioms and the Osterwalder-Schrader theorem can be generalised to lower dimensions. In this way it also becomes meaningful to discuss two and three-dimensional \textsc{qft}'s. Several examples of nontrivial quantum field theories have been found in lower dimensions, but in four dimensions only the free theory is known. In spaces of dimension higher than four, no quantum field theories are known.

\section{Matrix model quantum field theories}
Partition functions over Hermitean matrices are a common tool for gravitational theories. A typical action would consist of a quadratic term for the edges and a cubic interaction term for the vertices. Mathematically, such models correspond usually to a $\tau$-function of the KdV-hierarchy. The logarithm of the partition function is the free energy, which generates all connected graphs. The dual of such a graph is a triangulation of an oriented surface. This shows that such models generate triangulations of two-dimensional spaces. The continuum limit of such a model would be a model of gravity. It was shown by Kontsevich~\cite{kontsevich} that these models are equal to the intersection theory on the compactified moduli space of complex curves. This was done using nonperturbative techniques for the partition function.

\subsection{The Moyal product\label{sec:moyalp}}
The Hermitean matrices appear in another way in \textsc{qft}. They are the coefficients of an eigenfunction decomposition of functions under the Moyal product
\begin{equation}
\left(a \star b\right)(x)=\int\frac{\ud^{2d}k\,\ud^{2d}y}{(2\pi)^{2d}}a(x+\frac{1}{2}\Theta k)b(x+y)e^{ik\cdot y}\quad,\label{e:m1}
\end{equation}
where $x,y$ and $k$ are vectors and $\Theta$ a real antisymmetric $2d\times 2d$-matrix. It is an associative, but noncommutative product.\\

An orthonormal basis $\{a_{m,n}|m,n\in\mathbb{N}_{0}\}$ with inner product
\begin{equation*}
\langle a,b \rangle= \int \ud x\, \big(a\star b\big)(x)=\int \ud x\, a(x)b(x)
\end{equation*}
exists satisfying
\begin{equation*}
a_{m,n}\star a_{p,q}=\delta_{n,p}a_{m,q}\qquad\text{and}\qquad\overline{a_{m,n}}=a_{n,m}\quad.
\end{equation*}
and
\begin{equation*}
\int \ud x a_{m,n}(x)=\delta_{m,n}\quad.
\end{equation*}
These identities show how the spacetime integration over the Lagrangian density becomes the trace over the matrices of field coefficients with respect to the orthonormal basis. A remaining issue is the matrix formulation of the Laplacian. It is not straightforward how to deal with this and the method chosen will severly influence the outcome of the analysis.\\
An example of such a consequence is ultraviolet/infrared mixing. The Moyal product is an example of a deformed product. Such models often suffer from a phenomenon called ultraviolet/infrared mixing~\cite{minwalla}. Divergencies at low energy scales cause divergencies at high energies and vice versa. This means that those have to be resolved at the same time and cannot be dealt with separately. The complicates the analysis considerably.\\

Using the correspondence of scalar function and Hermitean matrices, the $N\times N$-matrices in a matrix model partition function may be turned into functions. Usually, the limit $N\rightarrow\infty$ is used to define the partition function. And using the Moyal product the functions that correpond to the used matrices are obtained and the set of possible states is found. This is no formal requirement, since the entire theory may be reconstructed from the $n$-point functions, but it helps to keep close to the traditional treatment.\\

The Kontsevich model may be interpreted as a \textsc{qft} in any even number of dimensions. This model was studied nonperturbatively by Grosse and Steinacker~\cite{steinacker1, steinacker2, steinacker3}. A quartic model for complex scalar fields was exactly solved~\cite{langmann1} and found to be trivial.\\

Quantum field theories on the Moyal plane form explicit examples of spectral triples in noncommutative geometry~\cite{gayral1}. The simplest indication in this direction, which appeals to the intuition rather than mathematical understanding,  is the product $x^{\mu}\star x^{\nu}=x^{\mu}x^{\nu}+(i/2)\Theta$. This highlights that the noncommutativity of the Moyal space is parameterised by $\Theta$. In the limit of vanishing noncommutativity is is not difficult to see that the Moyal product reduces to the pointwise product commonly used in \textsc{qft}. However, often it appears that the limit of infinite noncommutativity is better behaved as \textsc{qft}. For \textsc{qft} methods at finite noncommutativity, see~\cite{wang1}.

\section{Grosse-Wulkenhaar model}
An exciting attempt to construct a \textsc{qft} is the $\Phi_{4}^{4}$-model on the Moyal plane, also called the Grosse-Wulkenhaar model. It may be thought of as a generalisation of the ordinary $\Phi^{4}$-theory (\ref{e:phi4})
\begin{align}
&\hspace{-8mm}S[\varphi]=\int \ud^{d} x\, \frac{1}{2}\varphi(x)\big(-\Delta+\Omega^{2}\Vert2\Theta^{-1}x\Vert^{2}+\mu^{2}\big)\varphi(x)+g\varphi^{\star 4}(x)\label{e:GW}\quad,
\end{align}
where the ordinary pointwise product is replaced by the Moyal product with deformation matrix $\Theta$ and the propagator is supplemented with a harmonic oscillator potential~\cite{wg0}. This action is then studied at the self-dual point $\Omega=1$~\cite{langmann2}.\\

The added harmonic oscillator term in the dynamic part results in a compact resolvent, so that a complete matrix treatment of this model possible. In this way the \textsc{uv}/\textsc{ir}-mixing is dealt with automatically.\\

A big difference with conventional $\varphi^{4}$-theory is the vanishing of the $\beta$-function~\cite{disertori}. Usually, a positive $\beta$-function indicates a Landau pole. The coupling constant diverges at a finite energy scale. To overcome this, the coupling must be scaled to zero from the start, leading to free theory.\\
This theory has been studied intensively~\cite{wg1,wg2,wg3,wg4,wg5,wg6}. The Ward-identies for this model in combination with the Schwinger-Dyson equations yield a closed equation for the $2$-point function in the limit of infinite noncommutativity. Once this is solved, all $n$-points functions are obtained from the $2$-point functions through linear equations. This is the exact and complete solution of a toy model \textsc{qft}. Checking the Osterwalder-Schrader axioms reduces to a set of conditions on the $2$-point function. Essential is the question whether two $2$-point function with one boundary component is a Stieltjes function. It appears that this question may be answered from the vacuum sector of the theory.


\subsection{Locality}
Manifest locality, as in (\ref{e:phi4}) is given up in the definition of the $\Phi_{4}^{4}$-model on the Moyal plane. This follows directly from the Moyal product formula (\ref{e:m1}). A remarkable result of this model is the re-emergence of the Euclidean symmetry. This is needed to reconstruct the Wightman \textsc{qft} from the $n$-point functions.\\

This is also an intriguing example of the difference between the interpretation of a theory based on the action and based on (the reconstruction from) the $n$-point functions.

\section{Summary of this thesis}

The natural aim of this thesis would be to compute the free energy density of the Grosse-Wulkenhaar model for certain parameter configurations. After renormalisation this could yield a solved model, meaning that all $n$-points functions are explicitly computed. From these it can determined whether this model is a \textsc{qft} in the sense of Osterwalder-Schrader.\\
However, much less is needed. The Schinger-Dyson methods have yielded a tremendous amount of information about this theory. All missing information can be obtained from the vacuum sector of the free energy density, meaning that the source term is set to zero from the start. The focus lies thus on the computational methods to determine the partition function for the vacuum sector of the $\Phi^{4}$-theory on the Moyal plane, which restricts us to an even number of dimensions.\\

The main object of study is the partition function for weak coupling before renormalisation. Here it will be assumed that the dynamic eigenvalues lie close together. This may be thought of as a large mass and small kinetic energies, although there is no meaningful interpretation of the model before renormalisation and before the Schwinger functions are obtained. The partition function is an integral over the space of Hermitean matrices, which can be reduced to integrals over their eigenvalues. Needed for the free energy density is the factorisation of these eigenvalue intergrals. This is obtained by the computation of the volume of the polytope of symmetric stochastic matrices.\\

The main tool to make this all work is the asymptotic approximation. The model is formulated as an $N\times N$-matrix integral, where $N$ is a regulator. The full fields are recovered in the limit $N\rightarrow\infty$. As soon as no exact solution is available, this limit will be applied to force one. The polytope volume mentioned before is an example of this.\\

The other obstacle is the remaining determinant of the $N\times N$-matrix with entries depending on the matrix eigenvalues. Such determinants are difficult to compute. Therefore, techniques and approximations for such determinants are discussed.\\

Besides this, methods to determine the partition function for strong coupling will be discussed too. Also here, the determinant calculation is the decisive hurdle.

\chapter{Computational techniques\label{sec:DHm}}

The treatment of partition functions requires computational techniques. Several of those will be recalled in this chapter. Although several of these will be well known, repeating it fixes conventions.

\section{The Vandermonde determinant in matrix models\label{sec:Vdm in MM}}
The appearance of the Vandermonde determinant 
\begin{equation*}
\Delta(\lambda_{1},\ldots,\lambda_{N})=\prod_{1\leq k<l\leq N}(\lambda_{l}-\lambda_{k})
\end{equation*}
in Hermitean matrix models
\begin{equation*}
\int \ud X\,f(X)\quad.
\end{equation*}
is well known. Since the diagonalisation is through conjugation by unitary matrices, this only works for functions invariant under conjugation by these. Let $f$ be such a function.\\
Let $X=U^{*}\bullet \Lambda\bullet  U$ be a Hermitean matrix, diagonalised by a unitary $U=e^{iT}$ with $T$ Hermitean. Here and everywhere the $\bullet $-product will indicate matrix multiplication. Since $U$ changes infinitessimally as $i\ud T\bullet  U$, the measure transforms as
\begin{align}
&\hspace{-8mm}\ud X_{ab}=U_{ah}^{*}\bullet \left(-i\ud T\bullet \Lambda+\ud\Lambda+ i\Lambda\bullet \ud T\right)_{hj}\bullet U_{jb}\nonumber\\
&=U_{ah}^{*}\bullet \left(\ud \Lambda_{hj}+ i(\lambda_{h}-\lambda_{j})\ud T_{hj}\right)\bullet U_{jb}\label{e:VdM-jacobian}\\
&=U_{ah}^{*}\bullet \left(\ud \Lambda_{hj}+ i(\lambda_{h}-\lambda_{j})\cdot(\ud T_{hj}^{(r)}+i\ud T_{hj}^{(i)})\right)\bullet U_{jb}\quad,\label{e:vdmdecom}
\end{align}
where $\Lambda=\diag(\lambda_{1},\ldots,\lambda_{N})$ and the Hermitean matrices were split in real parts $T_{hj}=T^{(r)}_{hj}+iT^{(i)}_{hj}$. The determinant of this $N^{2}\times N^{2}$-Jacobian is
\begin{equation*}
\prod_{k\neq l}(\lambda_{l}-\lambda_{k})^{2}\quad.
\end{equation*}
This changes the Hermitean matrix integral to
\begin{align*}
&\hspace{-8mm}\mathpzc{U}\int\ud U\,\Big(\prod_{h=1}^{N}\int\ud \lambda_{h}\Big)\,\Delta(\lambda_{1},\ldots,\lambda_{N})^{2}f(\vec{\lambda})\nonumber\\
&=\mathpzc{U}\Big(\prod_{h=1}^{N}\int\ud \lambda_{h}\Big)\,\Delta(\lambda_{1},\ldots,\lambda_{N})^{2}f(\vec{\lambda})\\
&=\mathpzc{U}\int\ud\vec{\lambda}\,\Delta(\lambda_{1},\ldots,\lambda_{n})^{2}f(\vec{\lambda})\quad,
\end{align*}
where the final equality is just convenient rewriting. The conventions used here are made more explicit in (\ref{e:Umeastra1}) and (\ref{e:Umeastra2}) in the proof of the Harish-Chandra-Itzykson-Zuber integral in Paragraph~\ref{sec:HCIZ}.\\

The Vandermonde determinant is the determinant of the Vandermonde matrix
\begin{equation}
V(\lambda_{1},\ldots,\lambda_{n})=\left(\begin{array}{ccccc}
1 & \lambda_{1} & \lambda_{1}^{2} & \ldots & \lambda_{1}^{n-1} \\
1 & \lambda_{2} & \lambda_{2}^{2} & \ldots & \lambda_{2}^{n-1} \\
\vdots & \vdots & \vdots & \ddots & \vdots \\
1 & \lambda_{n} & \lambda_{n}^{2} & \ldots & \lambda_{n}^{n-1}
\end{array}\right)\label{e:Vdmmat}\quad.
\end{equation}
It is the symmetric polynomial in $\lambda$ of smallest degree such that it vanishes when two values coincide.

\subsection{The inverse of the Vandermonde-matrix\label{sec:invVdm}}
The Vandermonde matrix 
\begin{equation*}
V(\vec{\lambda})_{ij}=\lambda_{j}^{i-1}
\end{equation*}
has nonzero determinant if all $\lambda_{j}$'s are different. It is thus invertible. Supposing that $\tilde{V}=(\tilde{v}_{ki})$ is the inverse, it is found through
\begin{align*}
&\hspace{-8mm}\delta_{kj}=(\tilde{V}\bullet V)_{kj}=\sum_{i=1}^{n}\tilde{v}_{ki}x_{j}^{i-1}=\frac{1}{\prod_{\stackrel{t=1}{t\neq k}}^{n}(x_{t}-x_{k})}\prod_{\stackrel{m=1}{m\neq k}}^{n}(x_{m}-x_{j})\\
&=\sum_{i=1}^{n}x_{j}^{i-1}\times\frac{1}{\prod_{\stackrel{t=1}{t\neq k}}^{n}(x_{t}-x_{k})}\sum_{\stackrel{1\leq m_{1}<\cdots<m_{n-i}\leq n}{m_{1},\ldots,m_{n-i}\neq k}}(-1)^{i-1}x_{m_{1}}\cdots x_{m_{n-i}}\quad.
\end{align*}
The matrix $\tilde{V}=V^{-1}(\vec{\lambda})$ thus has entries
\begin{equation}
\tilde{v}_{ki}=\frac{1}{\prod_{\stackrel{t=1}{t\neq k}}^{n}(x_{t}-x_{k})}\sum_{\stackrel{1\leq m_{1}<\cdots<m_{n-i}\leq n}{m_{1},\ldots,m_{n-i}\neq k}}(-1)^{i-1}x_{m_{1}}\cdots x_{m_{n-i}}\quad.\label{e:invVdm}
\end{equation}

\section{Orthogonal polynomials\label{sec:opm}}
The Vandermonde determinants in matrix models can often be tackled by orthogonal polynomials~\cite{chihara}. Also many other determinants can be computed from orthogonal polynomials. This makes it worthwhile to present some general properties. The starting point is a real, positive weight function $w:\mathbb{R}\rightarrow\mathbb{R}$ with respect to which the monic orthogonal polynomials in one variable are constructed.

\begin{thrm}\label{thrm:op}
For a positive definite weight function $w(\lambda)\in L^{1}(a,b)$ there exists a unique set of monic orthogonal polynomials $\{P_{n}\}$, constructed by
\begin{align*}
&P_{0}=1\; ;\quad P_{1}=\lambda-\alpha_{1}\qquad\text{with}\quad \alpha_{1}=\frac{\int \ud \lambda\,w(\lambda)\lambda}{\int \ud \lambda\, w(\lambda)}\quad\text{and}\\
&P_{n}=(\lambda-\alpha_{n})P_{n-1}-R_{n-1}P_{n-2}\qquad\text{for}\quad n\geq 2\quad\text{, where}\\
&\hspace{1cm}\alpha_{n}=\frac{\int\ud \lambda\,w(\lambda) \lambda P_{n-1}^{2}(\lambda)}{\int\ud \lambda\,w(\lambda)P_{n-1}^{2}(\lambda)}\qquad\text{and}\\
&\hspace{1cm}R_{n}=\frac{\int\ud \lambda\, w(\lambda)\lambda P_{n}(\lambda)P_{n-1}(\lambda)}{\int\ud \lambda\,w(\lambda)P_{n-1}^{2}(\lambda)}\quad.
\end{align*}
\end{thrm}
\begin{proof}
The requirements demand that $p_{0}(\lambda)=1$ and $p_{1}(\lambda)=\lambda-\alpha_{1}$. They are perpendicular if $\langle p_{1},p_{0}\rangle_{w}=0$, which determines $\alpha_{1}$. The monotonicity forces $p_{n\geq2}$ to be of the form \mbox{$p_{n}=(\lambda-\alpha_{n})p_{n-1}-R_{n-1}p_{n-2}-\sum_{m=0}^{n-3}\gamma_{m}p_{m}$}. Requiring $\langle p_{n},p_{m<n}\rangle_{w}=0$ fixes $\alpha_{n}$ for $m=n-1$, $R_{n-1}$ for $m=n-2$ and $\gamma_{m}=0$ for $m\leq n-3$.\\
To proof uniqueness we suppose the contrary, there is another set $\{q_{n}\}$ of monic orthogonal polynomials. Then, $(p_{n}-q_{n})$ is a polynomial of degree $(n-1)$ and can be decomposed into $\sum_{i=0}^{n-1}a_{i}p_{i}$ or $\sum_{i=0}^{n-1}b_{i}q_{i}$ and has therefore zero inner product with $p_{n}$ and $q_{n}$. This implies $\langle p_{n}-q_{n},p_{n}-q_{n}\rangle_{w}=0$, hence $p_{n}=q_{n}$.
\end{proof}

Another approach uses the moments
\begin{equation}
\rho_{n}=\mathpzc{l}[x^{n}]=\int \ud x\,w(x)x^{n}\quad.
\end{equation}
In this notation $\langle P_{m},P_{n}\rangle=\mathpzc{l}[P_{m}P_{n}]=h_{m}\delta_{m,n}$ for an orthogonal polynomial sequence $\{P_{n}\}_{n=0}^{\infty}$. This generalisation from the inner product to a linear functional automatically incorporates discrete measures. It is simple to see that for each weight (or sequence of moments) there exists at most one monic orthogonal polynomial sequence.

\begin{prps}\label{p:opex}
A moment functional $\mathpzc{l}$ with moment sequence $\{\rho_{n}\}_{n=0}^{\infty}$ defines a monic orthogonal polynomial sequence if and only if the Hankel determinant 
\begin{equation*}
\Delta_{n}=\det(\rho_{i+j})_{i,j=0}^{n}=\left|\begin{array}{cccc}\rho_{0}&\rho_{1}&\ldots&\rho_{n}\\\rho_{1}&\rho_{2}&\ldots&\rho_{n+1}\\\vdots&\vdots&\ddots&\vdots\\\rho_{n}&\rho_{n+1}&\ldots&\rho_{2n}\\\end{array}\right|\neq 0\qquad,\quad\forall n\in\mathbb{N}_{0}\quad.
\end{equation*}
\end{prps}
\begin{proof}
The orthogonality conditions for $P_{n}(x)=\sum_{i=0}^{n}a_{ni}x^{i}$ are given by $\mathpzc{l}[x^{m}P_{n}(x)]=\sum_{i=0}^{n}a_{ni}\rho_{m+i}=h_{n}\delta_{mn}$ for $m\leq n$ and $h_{n}\neq0$. This is equivalent to the system
\begin{equation*}
\left(\begin{array}{cccc}\rho_{0}&\rho_{1}&\ldots&\rho_{n}\\\rho_{1}&\rho_{2}&\ldots&\rho_{n+1}\\\vdots&\vdots&\ddots&\vdots\\\rho_{n}&\rho_{n+1}&\ldots&\rho_{2n}\\\end{array}\right)\left(\begin{array}{c}a_{n0}\\a_{n1}\\\vdots\\a_{nn}\end{array}\right)=\left(\begin{array}{c}0\\\vdots\\0\\h_{n}\end{array}\right)\quad.
\end{equation*}
We can clearly solve this for $n=0$. We can then proceed by induction.\\
If $\Delta_{n}\neq0$, the matrix can be inverted and a unique solution is obtained, where $h_{n}=\Delta_{n}/\Delta_{n-1}$ ensures that this new orthogonal polynomial is monic.
\end{proof}
\begin{cor}\label{c:opex}
For a moment sequence $\{\rho_{n}\}_{n=0}^{\infty}$ the $n$-th monic orthogonal polynomial is given by
\begin{equation*}
P_{n}(x)=\Delta_{n-1}^{-1}\left|\begin{array}{ccccc}\rho_{0}&\rho_{1}&\ldots&\rho_{n-1}&\rho_{n}\\\rho_{1}&\rho_{2}&\ldots&\rho_{n}&\rho_{n+1}\\\vdots&\vdots&&\ddots&\vdots\\\rho_{n-1}&\rho_{n}&\ldots&\rho_{2n-2}&\rho_{2n-1}\\1&x&\ldots&x^{n-1}&x^{n}\end{array}\right|
\end{equation*}
\end{cor}

\begin{dfnt}
A moment functional $\mathpzc{l}$ is positive definite, if $\mathpzc{l}[p(x)]>0$ for every nonzero polynomial that is nonnegative for all $x\in E$, where $E$ is a nonempty subset of $\mathbb{R}$, called the supporting set.
\end{dfnt}

If not specified, $E=\mathbb{R}$. For a positive definite $\mathpzc{l}$ all moments of even order are positive and all moments of odd order at least real. This last is seen inductively from $0<\mathpzc{l}[(x+1)^{2n}]=\sum_{k=0}^{2n}\binom{2n}{k}\mu_{2n-k}$. Subtract all the even moments. The result is real and so is it complex conjugate. It follows then that all moments are real.\\
It follows from Theorem~\ref{thrm:op} that the corresponding monic orthogonal polynomials have real coefficients. This also shows that for a positive definite moment functional a sequence of orthogonal polynomials exists, provided that all moments are finite. And if a sequence of real orthogonal polynomials exists, it must be given by Corollary~\ref{c:opex}, so for all $n\in\mathbb{N}_{0}$: $\Delta_{n}>0$.\\

If $p:x\mapsto p(x)$ is a polynomial that is nonnegative for all real $x$ and not zero, there exist real polynomials $q_{1,2}$, such that $p(x)=q^{2}_{1}(x)+q^{2}_{2}(x)$. To see this, notice that real zeroes must have even multiplicity and the other must occur in cojugate pairs. Hence, it can be written as
\begin{align*}
&\hspace{-8mm}p(x)=r^{2}(x)\times\prod_{k=1}^{m}(x-c_{k}-d_{k}i)(x-c_{k}+d_{k}i)\\
&=r^{2}(x)(C(x)+iD(x))(C(x)-iD(x))
\end{align*}
with the $c$'s and $d$'s real parameters.\\
If all $\Delta_{n}>0$, then there exists a sequence of real orthogonal polynomials, in which $q_{1,2}(x)=\sum_{k=0}^{n}a^{(1,2)}_{k}P_{k}(x)$ can be desomposed for some $n>0$ and with all $a$'s real. This shows that $\mathpzc{l}[p]>0$. Thus it follows that $\mathpzc{l}$ is positive definite if and only if all $\Delta_{n}>0$, or equivalently, all $P_{n}$ real and $h_{n}>0$.

\section{Miscellaneous}

\begin{lemma}\{\emph{Stirling's approximation}\}\label{l:SA}\\
For the factorial $n!=\prod_{k=0}^{n-1} (n-k)$ and double factorial $n!!=\prod_{k=0}^{\lfloor\frac{n-1}{2}\rfloor}(n-2k)$ with $n\in\mathbb{N}$ the following asymptotic expansions hold as  $n\rightarrow\infty$:
\begin{enumerate}
 \item $n!=\sqrt{2\pi n}n^{n}\exp[-n]\times(1+\frac{1}{12n}+\mathcal{O}(n^{-2}))\quad;$
 \item $(2n-1)!!=2^{\frac{1}{2}+n}n^{n}\exp[-n]\times\big(1-\frac{1}{24n}+\mathcal{O}(n^{-2})\big)\quad.$
 \item $\prod_{m=0}^{n-1}\frac{(n+m-1)!}{(n-1)!}=\big(\frac{16(n-1)}{e}\big)^{\binom{n}{2}}\exp[-\frac{3}{4}n^{2}+n-\frac{1}{3}]2^{\frac{5}{12}}\times\big(1+\mathcal{O}(\frac{1}{n})\big)\quad.$
\end{enumerate}
\end{lemma}
\begin{proof}
For \emph{1.} we apply Laplace's method for integrals, Lemma~\ref{l:Lmi}, to the Gamma function. This yields
\begin{align*}
&\hspace{-8mm}n!=\Gamma(n\!+\!1)=\!\int_{0}^{\infty}\!\!\!\!\!\ud t\,e^{-t}t^{n}=n\exp[n\log(n)\!-\!n]\!\int_{-1}^{\infty}\!\!\!\!\!\ud t\,\exp[-n(t\!-\!\log(1\!+\!t))]\\
&=\sqrt{2n}\exp[n\log(n)-n]\int_{-\sqrt{n/2}}^{\infty}\ud t\,\exp[-t^{2}-\sum_{m=3}\frac{(-\sqrt{2}t)^{m}}{m\cdot n^{\frac{m-2}{2}}}]\\
&=\sqrt{2n}\exp[n\log(n)-n]\times\Big(\int_{-\infty}^{\infty}\ud t\,\exp[-t^{2}]\times\big[1 + \frac{2\sqrt{2}}{3\sqrt{n}}t^{3} \\
&-\frac{t^{4}\!-\!\frac{4}{9}t^{6}}{n}+\frac{2\sqrt{2}}{405n^{\frac{3}{2}}} (162t^{5}\!-\!135t^{7}\!+\!20t^{9})+\mathcal{O}(\frac{1}{n^{2}})\big]+\mathcal{O}(\frac{n}{e^{\frac{n}{2}}})\Big)\,.
\end{align*}
Integrating this yields the stated result. Including more terms in the asymptotic expansion yields additional correction terms.\\
For \emph{2.} we apply \emph{1.} to
\begin{align*}
&\hspace{-8mm}(2n\!-\!1)!!=\frac{(2n)!}{(2n)!!}=\frac{(2n)!}{2^{n}(n!)}=\sqrt{2}2^{-n}\exp[n\!-\!n\log(n)\!+\!2n\log(2n)\!-\!2n]\\
&\times\Big(\frac{1+\frac{1}{24n}}{1+\frac{1}{12n}}+\mathcal{O}(n^{-2})\Big)\\
&=\exp[n\log(n)-n+2n\log(2)]\sqrt{n}2^{-n}\times\big(1-\frac{1}{24n}+\mathcal{O}(n^{-2})\big)\quad.
\end{align*}

For \emph{3.} some simple approximations are made first. Besides \emph{1.} the main tool used is the approximations of a sum by integrals. From
\begin{equation*}
\int_{n-\frac{1}{2}}^{n+\frac{1}{2}}\ud x\,\log(x)=\log(n)-\frac{1}{24n^{2}}+\mathcal{O}(n^{-4})
\end{equation*}
it follows that
\begin{equation*}
\log(n)+\mathcal{O}(n^{-3})=\int_{n-\frac{1}{2}}^{n+\frac{1}{2}}\ud x\,\log(x)+\frac{1}{12}x^{-2}\quad.
\end{equation*}
This is more than sufficient to see that
\begin{align*}
&\hspace{-8mm}\sum_{m=0}^{n-1}\frac{1}{2}[\log(n+m-1)-\log(n-1)]\\
&=-\frac{n}{2}\log(n-1)+\frac{1}{2}\int_{n-\frac{3}{2}}^{2n-\frac{3}{2}}\ud z\,\log(z)+\mathcal{O}(\frac{1}{n})\\
&=(n-\frac{3}{4})\log(2)-\frac{n-1}{2}+\mathcal{O}(\frac{1}{n})\quad.
\end{align*}
The same strategy yields
\begin{equation*}
n\log(n)=\int_{n-\frac{1}{2}}^{n+\frac{1}{2}}\ud x\,x\log(x)-\frac{1}{24x}+\mathcal{O}(n^{-3})\quad,
\end{equation*}
so that
\begin{align*}
&\hspace{-8mm}\sum_{m=0}^{n-1}\!(n\!+\!m\!-\!1)[\log(n\!+\!m\!-\!1)\!-\!\log(n\!-\!1)]\\
&\approx-3\binom{n}{2}\!\log(n\!-\!1)\!+\!\int_{n-\frac{3}{2}}^{2n-\frac{3}{2}}\!\!\!\!\ud z\,z\log(z)-\frac{1}{24z}\\
&=(2n^{2}-3n+\frac{9}{8})\log(2)-\frac{3}{4}(n-1)^{2}-\frac{\log(2)}{24}+\mathcal{O}(\frac{1}{n})\quad.
\end{align*}
The final computation is 
\begin{align*}
&\hspace{-8mm}\frac{1}{12}\sum_{m=1}^{n-1}\frac{1}{n+m-1}-\frac{1}{n-1}=-\frac{1}{12}+\frac{1}{12}\int_{n-1}^{2(n-1)}\ud z\,\frac{1}{z}\\
&=\frac{\log(2)-1}{12}+\mathcal{O}(n^{-1})\quad.
\end{align*}
Applying \emph{1.} in combination with these approximations yields
\begin{align*}
&\hspace{-8mm}\prod_{m=0}^{n-1}\frac{(n+m-1)!}{(n-1)!}=\Big\{\prod_{m=0}^{n-1}(n-1)^{m}\exp[-m]\\
&\times\exp[(n+m-1)\times\big[\log(n+m-1)-\log(n-1)\big]]\sqrt{1+\frac{m}{n-1}}\\
&\times\exp[\frac{1}{12(n+m-1)}-\frac{1}{12(n-1)}]\Big\}\times\big(1+\mathcal{O}(n^{-1})\big)\\
&=\big(\frac{n-1}{e}\big)^{\binom{n}{2}}\exp[(2n^{2}-3n+\frac{9}{8})\log(2)-\frac{3}{4}(n-1)^{2}-\frac{1}{24}\log(2)]\\
&\times\exp[(n-\frac{3}{4})\log(2)-\frac{n-1}{2}]\exp[\frac{\log(2)-1}{12}]\times\big(1+\mathcal{O}(\frac{1}{n})\big)\\
&=\big(\frac{2^{4}(n-1)}{e}\big)^{\binom{n}{2}}\exp[-\frac{3}{4}n^{2}+n-\frac{1}{3}]\,2^{\frac{5}{12}}\times\big(1+\mathcal{O}(\frac{1}{n})\big)\quad.
\end{align*}
\end{proof}

\begin{rmk}
Because all error terms are only asymptotic, it is difficult to give an error for asymptotic approximations of products of factorials, such as $\prod_{k=1}^{n}k!$. These are known~\cite{kellner, gruenberg}, but to complicated to introduce here. On the basis of the expansions in Lemma~\ref{l:SA} asymptotics up to a finite multiplicative factor can be obtained.
\end{rmk}

\begin{lemma}\label{l:LTA}
The Taylor series $\sum_{k=0}^{N}(-\gamma N)^{k}/k!$ converges to $\exp[-\gamma N]$, if \\$\gamma\in[\,0,W_{L}(e^ {-1})\,)$, where $W_{L}$ is the Lambert $W$-function.
\end{lemma}
\begin{proof}
We start with the power series expansion of the exponential function $\exp[-\gamma x]$ at $x=N$. If we apply Stirling's approximation to the remainder in integral form, we obtain
\begin{equation*}
\big|e^{-\gamma N}-\sum_{k=0}^{N}\frac{1}{k!}(-\gamma N)^{k}\big|\leq|\frac{(\gamma N)^{N+1}}{N!}|\leq (\gamma e)^{N}\sqrt{\frac{\gamma N}{2\pi}}\quad,
\end{equation*}
so that the relative error is bounded by 
\begin{equation}
(\gamma N/(2\pi))^{1/2}(\gamma \exp[1+\gamma])^{N}\quad.\label{e:reb}
\end{equation}
This tends to zero for $\gamma \in [0,W_{L}(e^{-1})]$, where $W_{L}$ is the Lambert $W$-function. It is defined by $W_{L}(x\exp[x])=x$. For every positive real $x$, one can find a positive real $y$ such that $x=y\exp[y]$. This implies in turn that $x=ye^{y}=W_{L}(x)e^{W_{L}(x)}$.\\
Choosing $\gamma$ smaller than $W(e^{-1})\approx0.28$ guarantees that the approximation by the first $N$ terms becomes better and better as $N$ tends to infinity.
\end{proof}

\begin{lemma}\label{l:fcl}
For $n\in\mathbb{N}$
\begin{equation*}
C_{n}=\sum_{m=1}^{n}(-1)^{m+n}\sum_{\substack{\mu_{i}\geq 1\\ \sum_{i=1}^{m}\mu_{i}=n}}\left(\prod_{j=1}^{m}\frac{1}{\mu_{j}!}\right)=\frac{1}{n!}\quad.
\end{equation*}
\end{lemma}
\begin{proof}
It is clear that the claim holds for $n=1$. We proceed by induction, so suppose that the claim holds for integers $k<n$. Then,
\begin{align*}
&\hspace{-8mm}C_{n}=\frac{(-1)^{n+1}}{n!}+\sum_{m=2}^{n}(-1)^{m+n}\sum_{\substack{\mu_{i}\geq 1\\ \sum_{i=1}^{m}\mu_{i}=n}}\left(\prod_{j=1}^{m}\frac{1}{\mu_{j}!}\right)\\
&=\frac{(-1)^{n+1}}{n!}-\sum_{\mu_{1}=1}^{n-1}\frac{1}{\mu_{1}!}\sum_{m'=1}^{n-\mu_{1}}(-1)^{m'+n}\sum_{\substack{\mu_{i}\geq 1\\ \sum_{i=2}^{m'+1}\mu_{i}=n-\mu_{1}}}\left(\prod_{j=2}^{m'+1}\frac{1}{\mu_{j}!}\right)\\
&=\frac{(-1)^{n+1}}{n!}-\sum_{\mu_{1}=1}^{n-1}\frac{(-1)^{\mu_{1}}}{\mu_{1}!}\times\frac{1}{(n-\mu_{1})!}\\
&=-\frac{1}{n!}\sum_{\mu_{1}=1}^{n}\binom{n}{\mu_{1}}(-1)^{\mu_{1}}=(1-1)^{n}+\frac{1}{n!}=\frac{1}{n!}\quad.
\end{align*}
\end{proof}

\section{Methods for integrals\label{sec:Mfi}}

\subsection{The Dirac delta\label{sec:dirac}}
The Dirac delta $\delta$ satisfying
\begin{equation}
f(0)=\int \ud x\, f(x)\,\delta(x)\label{e:dirac0}
\end{equation}
for real $x$ is often introduced by the indefinite integral
\begin{equation}
\delta(x)=\lim_{M\rightarrow\infty}\int_{-M}^{M}\ud q\,e^{2\pi i xq}\quad.\label{e:dirac1}
\end{equation}
To see how this formulation is linked to (\ref{e:dirac0}), assume that $f$ is a holomorphic function on the complex plane. It follows that
\begin{align*}
&\hspace{-8mm} \int \ud x\,\int_{-M}^{M}\ud q\,f(x)\exp[2\pi iqx]=\int \ud x\,f(x)\frac{e^{2\pi iMx}-e^{-2\pi iMx}}{2\pi ix}\\
&=\lim_{\varepsilon\rightarrow 0}\int \ud x\,f(x)\frac{e^{2\pi iMx}-e^{-2\pi iMx}}{2\pi i(x-(1+i)\varepsilon)}\quad,
\end{align*}
where the pole is shifted slightly, which is possible by the continuity of $f$. This is only done to ensure that only the contour in the upper-right quarter of the complex plane for the positive phase is needed. Closing this contour $\mathcal{C}$ by a quarter-arc with radius $R$ yields
\begin{align*}
&\hspace{-8mm}f(0)+\mathcal{O}(\varepsilon)=\oint_{\mathcal{C}}\frac{\ud x}{2\pi i}\frac{f(x)}{x-(1+i)\varepsilon}e^{2\pi i Mx}\\
&=\int_{0}^{R}\frac{\ud x}{2\pi i}\frac{f(x)}{x-(1+i)\varepsilon}e^{2\pi iMx}\\
&+\int_{0}^{\pi/2}\frac{\ud \phi}{2\pi}f(Re^{i\phi})\frac{Re^{i\phi}}{Re^{i\phi}-(1+i)\varepsilon}\exp[2\pi iMRe^{i\phi}]\\
&+\int_{R}^{0}\frac{\ud y}{2\pi}\frac{f(iy)}{iy-(1+i)\varepsilon}\exp[-2\pi My]\quad.
\end{align*}
Choosing $M$ such that 
\begin{equation*}
\frac{f(0)}{M}\rightarrow0\quad;\quad\frac{f(R)}{M}\rightarrow0\quad\text{and}\quad f(x)e^{-RM}\rightarrow0\text{ for all }|x|\leq R\quad.\label{e:diraccond}
\end{equation*}
suffices to show that the last two lines do not contribute. By the same arguments we can add the contour for the negative phase part, which by the residue theorem is equal to zero, so that the original formulation is obtained. If $x$ should be integrated over negative values as well, then both phases are needed anyway. \\

These ideas can be extended in the opposite direction. A simple example is
\begin{equation}
f(0)+\mathcal{O}(M^{-\delta})=\int \ud x\,\int_{-M}^{M}\ud q\,f(x)\exp[2\pi iqx]g(q)\quad,\label{e:dirac2}
\end{equation}
when $|g(q)|=1+\mathcal{O}(M^{-1-2\delta})$ for $q\in[-M,M]$, when \\$\max_{x\in[-M,M]}|f(x)|\leq M^{\delta}$.

\subsection{The stationary phase method\label{sec:spm}}

\begin{lemma}\emph{Laplace's method}\label{l:Lmi}\\
Let $f\in C^{2}(a,b)$ with a unique maximum at $x_{0}$, where $f''(x_{0})<0$. Then
\begin{equation*}
\lim_{n\rightarrow\infty}\int_{a}^{b}\ud x\, e^{nf(x)}=e^{nf(x_{0})}\sqrt{\frac{2\pi}{-nf''(x_{0})}}\quad.
\end{equation*}
\end{lemma}
\begin{proof}
Apply the Taylor expansion
\begin{equation*}
f(x)\approx f(x_{0})+\frac{1}{2}f''(x_{0})(x-x_{0})^{2}
\end{equation*}
to the integral. We can extend $(a,b)$ to the real line, because the extra part is strongly suppressed. The integral becomes Gaussian and we obtain the claim.
\end{proof}

A simple generalisation of this is the following corollary.
\begin{cor}\label{c:Lmi}
Let $f\in C^{2}(a,b)$ with a unique maximum at $x_{0}$, where \\$f''(x_{0})<0$. And let $g$ be a funcion that is twice continuously differentiable in a neighbourhood of $x_{0}$ and that satisfies $g(x_{0})\neq 0$. Then
\begin{equation*}
\lim_{n\rightarrow\infty}\int_{a}^{b}\ud x\, g(x)e^{nf(x)}=g(x_{0})e^{nf(x_{0})}\sqrt{\frac{2\pi}{-nf''(x_{0})}}\quad.
\end{equation*}
\end{cor}

A similar result exists for complex integral. It goes under the name of steepest descent method and can be used to estimate certain complex integrals. The well-known Gaussian integral
\begin{equation*}
\int_{-\infty}^{\infty}\ud z\, e^{-az^{2}}=\sqrt{\frac{\pi}{a}}
\end{equation*}
for positive real $a$ is well known. Its complex generalisation $\int \ud z\, e^{-iaz^{2}}$ is not straightforward. Loosely, we expect that the main contribution comes from the region around $z=0$, since the phase is stationary around there, so that the contributions add coherently.\\
To make this more precise, we interpret the integrand as a function on the complex plane $z=x+iy$, so that we obtain
\begin{equation*}
e^{-ia(x^{2}-y^{2}+2ixy)}\quad.
\end{equation*}
This integrand had minimal real part in the regions $xy<0$, which suggest a contour at an angle $3\pi/4$. Since there are no poles inside the contour and the contribution from the arcs vanishes, this yields the real Gaussian integral, so that we obtain
\begin{equation*}
\int_{-\infty}^{\infty}\ud z\, e^{-iaz^{2}}=\sqrt{\frac{\pi}{ia}}\quad.
\end{equation*}
This is a particularly nice example of a wider applicable strategy for integrals of the form $\int \ud z\, e^{f(z)}$. One modifies the contour in such a manner that the amplitude is as small as possible. Imagining the real part of the funcion as a hilly landscape, this means that the path should go through the valley. Where this is not possible, one takes the lowest point to cross a mountain ridge, i.e. the saddle point. Furthermore, one takes the steepest paths up and down, since these are of stationary phase, so that the contributions add coherently.\\
Analytic function can be expanded around the saddle point, so that the integral becomes a Gaussian again. If, in addition, some parameters becomes large, the approximation may becomes exaxt, as in Lemma~\ref{l:Lmi}.

\subsection*{Asymptotic stationary phase methods}
Many of the integrands considered will be dominated by a Gaussian factor. These integrals can be approximated well. If the phase factor oscillates rapidly, this is no longer the case and the computation becomes more intricate. To get some feeling for this, a few results are explicitly computed.
\begin{lemma}\label{l:aspm}
Let $A,\tilde{A},B,C,D$ be parameters that scale as $A=\mathcal{O}(N^{\frac{1}{2}})$, $\tilde{A}=\mathcal{O}(1)$, $B=\mathcal{O}(1)$, $C=\mathcal{O}(N^{-\frac{1}{2}})$ and $D=\mathcal{O}(N^{-1})$ as $N\rightarrow\infty$. If the convergence is guaranteed by $\Re(B),\Re(D)>0$, then the integrals
\begin{align*}
&\hspace{-8mm}\fbox{1.}\; \int_{-\infty}^{\infty}\ud x\,\exp[i\tilde{A}x-Bx^{2}+iCx^{3}-Dx^{4}]\\
&=\sqrt{\frac{\pi}{B}}\exp[-\frac{\tilde{A}^{2}}{4B}]\exp[\frac{C\tilde{A}^{3}}{8B^{3}}-\frac{D\tilde{A}^{4}}{16B^{4}}-\frac{9C^{2}\tilde{A}^{4}}{64B^{5}}]\\
&\times\exp[\frac{-3\tilde{A}C}{4B^{2}}+\frac{9\tilde{A}^{2}C^{2}}{8B^{4}}+\frac{3\tilde{A}^{2}D}{4B^{3}}]\exp[-\frac{3D}{4B^{2}}-\frac{15C^{2}}{16B^{3}}]\times\Big(1+\mathcal{O}(N^{-\frac{3}{2}})\Big)\qquad;
\end{align*}
\begin{align*}
&\hspace{-8mm}\fbox{2.}\; \int_{-\infty}^{\infty}\ud x\,\exp[iAx-Bx^{2}+iCx^{3}]\\
&=\exp[B'r^{2}-iCr^{3}]\sqrt{\frac{\pi}{B'}}\exp[-\frac{15C^{2}}{16(B')^{3}}]\times\Big(1+\mathcal{O}(N^{-\frac{3}{2}})\Big)\quad,
\end{align*}
where $r=\frac{-2B\pm\sqrt{4B^{2}+12AC}}{6iC}$ and $B'=-B-\frac{iA}{r}$;
\begin{align*}
&\hspace{-8mm}\fbox{3.}\; \int_{-\infty}^{\infty}\ud x\,\exp[iAx-Bx^{2}+iCx^{3}-Dx^{4}]\\
&=\exp[B's^{2}-iC's^{3}+Ds^{4}]\sqrt{\frac{\pi}{B'}}\exp[-\frac{15(C')^{2}}{16(B')^{3}}-\frac{3D}{4(B')^{2}}]\times\Big(1+\mathcal{O}(N^{-\frac{3}{2}})\Big)\quad,
\end{align*}
where $s=\mathcal{O}(N^{\frac{1}{2}})$ is a solution of $0=iA+2Bs+3iCs^{2}+4Ds^{3}$, \\\mbox{$B's=4Ds^{3}+(3/2)iCs^{2}-iA/2$} and $C'=C-4iDs$.
\end{lemma}
\begin{proof}
We apply the saddle point method to the first integral. First we shift the integral to the point of stationary phase $x=y+\tilde{x}$, where
\begin{equation*}
0=i\tilde{A}-2B\tilde{x}+3iC\tilde{x}^{2}-4D\tilde{x}^{3}
\end{equation*}
is approximately solved by
\begin{equation*}
\tilde{x}=\frac{i\tilde{A}}{2B}+\frac{3(iC)(i\tilde{A})^{2}}{8B^{3}}+\frac{(-D)(i\tilde{A})^{3}}{4B^{4}}+\frac{9(iC)^{2}(i\tilde{A})^{3}}{16B^{5}}+\mathcal{O}(N^{-\frac{3}{2}})\quad.
\end{equation*}
This shift turns the integral into
\begin{align*}
&\hspace{-8mm}\frac{\exp[i\tilde{A}\tilde{x}\!-\!B\tilde{x}^{2}\!+\!iC\tilde{x}^{3}\!-\!D\tilde{x}^{4}]}{\sqrt{B-3iC\tilde{x}+6D\tilde{x}^{2}}}\!\int_{-\infty}^{\infty}\!\!\!\!\!\ud y\,\exp[-y^{2}\!+\!y^{3}\frac{iC-4D\tilde{x}}{(B\!-\!3iC\tilde{x}\!+\!6D\tilde{x}^{2})^{\frac{3}{2}}}]\\
&\times\exp[-y^{4}\frac{D}{(B-3iC\tilde{x}+6D\tilde{x}^{2})^{2}}]\times\Big(1+\mathcal{O}(N^{-\frac{3}{2}})\Big)\\
&=\frac{\exp[i\tilde{A}\tilde{x}\!-\!B\tilde{x}^{2}\!+\!iC\tilde{x}^{3}\!-\!D\tilde{x}^{4}]}{\sqrt{B-3iC\tilde{x}+6D\tilde{x}^{2}}}\Big\{\Gamma(\frac{1}{2})\!-\!\Gamma(\frac{5}{2})\frac{D}{(B\!-\!3iC\tilde{x}\!+\!6D\tilde{x}^{2})^{2}}\\
&+\frac{1}{2}\Gamma(\frac{7}{2})\big(\frac{iC-4D\tilde{x}}{(B-3iC\tilde{x}+6D\tilde{x}^{2})^{\frac{3}{2}}}\big)^{2}\Big\}\times\Big(1+\mathcal{O}(N^{-\frac{3}{2}})\Big)\\
&=\sqrt{\frac{\pi}{B}}\exp[-\frac{\tilde{A}^{2}}{4B}]\exp[\frac{C\tilde{A}^{3}}{8B^{3}}-\frac{D\tilde{A}^{4}}{16B^{4}}-\frac{9C^{2}\tilde{A}^{4}}{64B^{5}}]\exp[\frac{-3\tilde{A}C}{4B^{2}}]\\
&\times\exp[\frac{9\tilde{A}^{2}C^{2}}{8B^{4}}+\frac{3\tilde{A}^{2}D}{4B^{3}}]\exp[-\frac{3D}{4B^{2}}-\frac{15C^{2}}{16B^{3}}]\times\Big(1+\mathcal{O}(N^{-\frac{3}{2}})\Big)\quad.
\end{align*}
For the second integral we notice that
\begin{equation*}
iAx-Bx^{2}+iCx^{3}=-B'(x+r)^{2}+iC(x+r)^{3}+B'r^{2}-iCr^{3}
\end{equation*}
for suitable $B'$ and $r$ that satisfy
\begin{align*}
&\hspace{-8mm}-B=-B'+3iCr\quad\text{and}\quad iA=3iCr^{2}-2B'r\quad,
\end{align*}
which is the same as
\begin{align*}
&\hspace{-8mm}B'=-B-\frac{iA}{r}\quad\text{and}\quad 0=iA+2Br+3iCr^{2}\quad.
\end{align*}
The obvious solution is
\begin{equation*}
r=\frac{-2B\pm\sqrt{4B^{2}+12AC}}{6iC}=\mathcal{O}(N^{\frac{1}{2}})
\end{equation*}
and $B'=-B-\frac{iA}{r}=\mathcal{O}(1)$. The resulting integral
\begin{align*}
&\hspace{-8mm}\int_{-\infty}^{\infty}\!\!\!\!\!\ud x\,\exp[iAx\!-\!Bx^{2}\!+\!iCx^{3}]=\exp[B'r^{2}\!-\!iCr^{3}]\int_{-\infty}^{\infty}\!\!\!\!\!\ud x\,\exp[-B'y^{2}\!+\!iCy^{3}]\\
&=\exp[B'r^{2}-iCr^{3}]\sqrt{\frac{\pi}{B'}}\exp[-\frac{15C^{2}}{16B'^{3}}]\times\Big(1+\mathcal{O}(N^{-\frac{3}{2}})\Big)\quad.
\end{align*}
is now easily solved applying $\emph{1.}$.\\

The same trick is used for $\emph{3.}$. We solve
\begin{align*}
&\hspace{-8mm}iAx-Bx^{2}+iCx^{3}-Dx^{4}\\
&=-B'(x+s)^{2}+iC'(x+s)^{3}-D(x+s)^{4}+\big(B's^{2}-iC's^{3}+Ds^{4}\big)
\end{align*}
by picking a solution $s$ of $0=iA+2Bs+3iCs^{2}+4Ds^{3}$. The conditions to satisfy are 
\begin{align*}
&iA=-2B's+3iC's^{2}-4Ds^{3}\quad;\\
&B=B'-3iC's+6Ds^{2}\quad;\\
&iC=iC'-4Ds\quad.
\end{align*}
Setting $B's=4Ds^{3}+(3/2)iCs^{2}-iA/2$ and $iC'=iC+4Ds$ solves this. The other steps are the same as before.
\end{proof}

\section{The Harish-Chandra-Itzykson-Zuber integral\label{sec:HCIZ}}
Integration over unitary integrals is only in some cases possible. One of the few possibilities is the Harish-Chandra-Itzykson-Zuber integral~\cite{harish-chandra,itzykson-zuber}.
\begin{thrm}\emph{Harish-Chandra-Itzykson-Zuber integral}\label{thrm:hciz}\\
The unitary integral
\begin{align*}
&\hspace{-8mm}I(X,Y,t)=\int_{U(N)}\ud U\,e^{t\Tr(X\bullet U^{*}\bullet Y\bullet U)}\nonumber\\
&=\frac{\prod_{m=0}^{N-1}m!}{t^{\binom{N}{2}}}\frac{1}{\Delta(x_{1},\ldots,x_{N})\Delta(y_{1},\ldots,y_{N})}\det_{1\leq k,l\leq N}\left(e^{tx_{k}y_{l}}\right)\quad,
\end{align*}
where the Hermitean matrices $X,Y$ have eigenvalues $x_{k}$ and $y_{k}$ for $1\leq k\leq N$ and the integration is against the Haar measure with 
\begin{equation*}
\int_{U(N)} \ud U=\vol\, U(N)=1\quad.
\end{equation*}
\end{thrm}
\begin{proof}
Various ways to prove this are given in~\cite{zinnP}. Below we will give the proof based on the heat-equation.\\

The function
\begin{equation*}
K(M_{A},M_{B},s)=\frac{2^{\binom{N}{2}}}{(4\pi s)^{\frac{N^{2}}{2}}}e^{-\frac{1}{4 s}\Tr\big((M_{A}-M_{B})^{2}\big)}
\end{equation*}
depends on two Hermitean $N\times N$-matrices $M_{A}$ and $M_{B}$ and a constant $s$. It satisfies the heat equation 
\begin{align*}
&\hspace{-8mm}0=(\frac{\partial}{\partial s}-\sum_{k}\frac{\partial^{2}}{\partial (M_{A}^{(r)})_{kk}^{2}}-\frac{1}{2}\sum_{k<l}\frac{\partial^{2}}{\partial (M_{A}^{(r)})_{kl}^{2}}-\frac{1}{2}\sum_{k<l}\frac{\partial^{2}}{\partial (M_{A}^{(i)})_{kl}^{2}})K(M_{A},M_{B},s)\\
&=\big(\partial_{s}-\Delta_{M_{A}}\big)K(M_{A},M_{B},s)\quad.
\end{align*}
Using that the trace of the square of a Hermitean matrix  with real components $M_{kl}=M^{(r)}_{kl}+iM^{(i)}_{kl}$ for $k<l$ and $M_{kk}=M^{(r)}_{kk}$ is given by
\begin{equation*}
\Tr\,M^{2}=\sum_{m=1}^{N} (M^{(r)}_{mm})^{2}+2\sum_{k<l}(M^{(r)}_{kl})^{2}+(M^{(i)}_{kl})^{2}\quad,
\end{equation*}  
the boundary condition $K(M_{A},M_{B},s)\rightarrow\delta(M_{A}-M_{B})$ for $s\rightarrow0$. Diagonalising the matrices $M_{A}=U_{A}\bullet A\bullet U_{A}^{\ast}$ and $M_{B}=U_{B}\bullet B\bullet U_{B}^{\ast}$ yields the function
\begin{align*}
&\hspace{-8mm}\tilde{K}(A,B,s)= \int\ud W\,K(M_{A},W\bullet M_{B}\bullet W^{\ast},s)=\int\ud V\,K(A,V\bullet B\bullet V^{\ast},s)\\
&=\frac{2^{\binom{N}{2}}}{(4\pi s)^{\frac{N^{2}}{2}}}e^{-\frac{1}{4s}\Tr(A^{2}+B^{2})}I(A,B,\frac{1}{2s})\quad,
\end{align*}
This is again a solution of the heat equation and depends symmetrically on the eigenvalues of $A$ and $B$. This symmetry follows from the fact that permutation matrices are also unitary matrices. The next step is to change the heat equation from matrix variables to eigenvalue and angular variables. For this we must determine the metric associated with this transformation. The computation needed for this is very similar to that in Paragraph~\ref{sec:Vdm in MM}. The absolute square of (\ref{e:VdM-jacobian}) yields
\begin{align*}
&\hspace{-8mm}\Tr \,\ud M_{ab}^{2}=\sum_{k}\big(\ud A_{kk}^{2}+\sum_{k<j}(a_{j}-a_{k})^{2}\cdot\big((\ud T^{(r)}_{kj})^{2}+(\ud T^{(i)}_{kj})^{2}\big)\quad.
\end{align*}
The eigenvalues $a_{j}$ are put here in a diagonal matrix $A$. This determines the metric $g$ and shows that $\sqrt{\det g}=\prod_{k<l}(a_{l}-a_{k})^{2}=\Delta(a_{1},\ldots,a_{N})^{2}$. The unlucky situation occurs here that we are using the symbol '$\Delta$' for both the Vandermonde determinant and the Laplacian. They can be separated by their arguments.\\
To reach (\ref{e:HCIZLapl}) we write $\Delta(a_{1},\ldots,a_{N})=\tilde{\Delta}$ and $\partial_{j}=\frac{\partial}{\partial a_{j}}$. Afterwards we will refrain from this. The Laplacian is found through
\begin{align}
&\hspace{-8mm}\Delta_{M}=\frac{1}{\sqrt{\det g}}\frac{\partial}{\partial\xi^{i}}\big(\sqrt{\det g}g^{ij}\frac{\partial}{\partial\xi^{j}}\big)=\sum_{j}\frac{1}{\tilde{\Delta}^{2}}\partial_{j}\tilde{\Delta}^{2}\partial_{j}+\Delta_{T}\nonumber\\
&=\sum_{j}\frac{1}{\tilde{\Delta}^{2}}\big(\tilde{\Delta}\partial_{j}+[\partial_{j},\tilde{\Delta}]\big)\cdot\big(\partial_{j}\tilde{\Delta}-[\partial_{j},\tilde{\Delta}]\big)+\Delta_{T}\nonumber\\
&=\!\sum_{j}\!\frac{1}{\tilde{\Delta}}\partial_{j}^{2}\tilde{\Delta}\!+\!\sum_{j}\!\frac{1}{\tilde{\Delta}^{2}}[\partial_{j},\tilde{\Delta}]\partial_{j}\tilde{\Delta}\!-\!\sum_{j}\!\frac{1}{\tilde{\Delta}^{2}}[\partial_{j},\tilde{\Delta}]^{2}\!-\!\sum_{j}\!\frac{1}{\tilde{\Delta}}\partial_{j}[\partial_{j},\tilde{\Delta}]\!+\!\Delta_{T}\nonumber\\
&=\sum_{j}\frac{1}{\tilde{\Delta}}\partial_{j}^{2}\tilde{\Delta}+\sum_{j}\frac{1}{\tilde{\Delta}}[\partial_{j},\tilde{\Delta}]\partial_{j}-\sum_{j}\frac{1}{\tilde{\Delta}}\partial_{j}[\partial_{j},\tilde{\Delta}]+\Delta_{T}\nonumber\\
&=\sum_{j}\frac{1}{\tilde{\Delta}}\partial_{j}^{2}\tilde{\Delta}+\sum_{j}\frac{1}{\tilde{\Delta}}[\partial_{j},[\partial_{j},\tilde{\Delta}]]+\Delta_{T}\nonumber\\
&=\sum_{j}\frac{1}{\tilde{\Delta}}\partial_{j}^{2}\tilde{\Delta}+\Delta_{T}\quad.\label{e:HCIZLapl}
\end{align}
For the last it must be checked that
\begin{align*}
&\hspace{-8mm}0=\sum_{k=1}^{N}[\partial_{k},[\partial_{k},\tilde{\Delta}]]=\sum_{k=1}^{N}[\partial_{k},\tilde{\Delta}\sum_{l\neq k}\frac{1}{a_{k}-a_{l}}]\\
&=\tilde{\Delta}\sum_{k=1}^{N}(\sum_{l\neq k}\frac{1}{x_{k}-x_{l}})^{2}-\tilde{\Delta}\sum_{k=1}^{N}\sum_{l\neq k}\frac{1}{(x_{k}-x_{l})^{2}}\\
&=\tilde{\Delta}\sum_{k=1}^{N}\sum_{l\neq k}\sum_{m\neq k,l}\frac{1}{x_{k}-x_{l}}\frac{1}{x_{k}-x_{m}}=0\quad.
\end{align*}
Supposing that $N=3$, this is easily checked. For other $N$ summing over subsets of $\{1,\ldots,N\}$ of cardinality $3$ and taking $k,l$ and $m$ from this subset demonstrates the identity.\\

The measure transformation
\begin{align}
&\hspace{-8mm}\ud M=\prod_{i}\ud M^{(r)}_{ii}\prod_{i<j}\ud M^{(r)}_{ij}\ud M^{(i)}_{ij}=\sqrt{\det(g)}\prod\ud\xi^{\alpha}\nonumber\\
&=\Delta(\lambda_{1},\ldots,\lambda_{N})^{2}\prod_{i}\ud \lambda_{i}\prod_{i,j}\ud T_{ij}=\mathpzc{U}\Delta(\lambda_{1},\ldots,\lambda_{N})^{2}\ud \Lambda\,\ud U\label{e:Umeastra1}
\end{align}
involves the constant $\mathpzc{U}$. This constant may be determined by the test calculation
\begin{align}
&\hspace{-8mm}1=\int \ud M\,\frac{e^{-\frac{1}{2}\Tr(M^{2})}}{(2\pi)^{\frac{N}{2}}\pi^{\binom{N}{2}}}=\frac{\mathpzc{U}}{(2\pi)^{\frac{N}{2}}\pi^{\binom{N}{2}}}\int\ud \Lambda\,e^{-\frac{1}{2}\Tr(\Lambda^{2})}\Delta(\lambda_{1},\ldots,\lambda_{N})^{2}\nonumber\\
&=\frac{\mathpzc{U}}{(2\pi)^{\frac{N}{2}}\pi^{\binom{N}{2}}}\int \ud\Lambda\,\small{\left|\begin{array}{ccc}
H_{0}(\lambda_{1}) & \ldots & H_{N-1}(\lambda_{1}) \\
\vdots & \ddots & \vdots \\
H_{0}(\lambda_{N}) &  \ldots & H_{N-1}(\lambda_{N})
\end{array}\right|^{2}}\exp[-\frac{1}{2}\sum_{j=1}^{N}\lambda_{j}^{2}]\nonumber\\
&=\frac{\mathpzc{U}}{\pi^{\binom{N}{2}}}\prod_{m=0}^{N}m!\quad.\label{e:Umeastra2}
\end{align}
Due to the orthogonality properties of the Hermite polynomials, which are the unique monic orthogonal polynomials for this weight function, only terms of the form\\$H_{0}^{2}(\lambda_{\sigma(1)})H_{1}^{2}(\lambda_{\sigma(2)})\ldots H_{N-1}^{2}(\lambda_{\sigma(N)})$ contribute, where $\sigma$ is a permutation of the indices. Noticing that $\langle H_{m},H_{n}\rangle=\sqrt{2\pi}\delta_{mn}(n!)$ and that there are $N!$ such permutations the answer is immediate.\\

This shows that $\tilde{K}$ satisfies
\begin{equation*}
0=\big(\frac{\partial}{\partial s}-\sum_{j=1}^{N}\frac{\partial^{2}}{\partial a_{j}^{2}}\big)\Delta(a_{1},\ldots,a_{N})\tilde{K}(A,B,s)\quad
\end{equation*}
and we conclude that $\Delta(a_{1},\ldots,a_{N})\Delta(b_{1},\ldots,b_{N})\tilde{K}(A,B,s)$ is an antisymmetric function in the $a$'s and $b$'s and a solution of the flat heat equation (formulated in either $a$'s or $b$'s). The boundary condition
\begin{equation*}
\lim_{s\rightarrow0}\Delta(a_{1},\ldots,a_{N})\Delta(b_{1},\ldots,b_{N})\tilde{K}(A,B,s)=\mathpzc{K} \sum_{\rho\in\mathcal{S}_{N}}\sgn(\rho)\prod_{j=1}^{N}\delta(a_{j}-b_{\rho(j)})
\end{equation*}
is completely determined by
\begin{align*}
&\hspace{-8mm}\Delta(b_{1},\ldots,b_{N})=\Delta(b_{1},\ldots,b_{N})\int \ud M_{A} K(M_{A},M_{B},0)\\
&=\mathpzc{U}\mathpzc{K}\int \ud A\,\Delta(a_{1},\ldots,a_{N})\big\{\sum_{\rho\in\mathcal{S}_{N}}\sgn(\rho)\prod_{j=1}^{N}\delta(a_{j}-b_{\rho(j)})\big\}\\
&=(N!)\mathpzc{U}\mathpzc{K}\Delta(b_{1},\ldots,b_{N})\quad.
\end{align*}
Taking the asymptotics of $\Delta(a_{1},\ldots,a_{N})\Delta(b_{1},\ldots,b_{N})\tilde{K}$ and its symmetry into account there remains only one candidate:
\begin{equation*}
\tilde{K}(A,B,s)=\mathpzc{K}\Delta(a_{1},\ldots,a_{N})^{-1}\Delta(b_{1},\ldots,b_{N})^{-1}\big(\frac{1}{4\pi s}\big)^{\frac{N}{2}}\det_{m,n}\big(e^{-\frac{1}{4s}(a_{m}-b_{n})^{2}}\big)\,,
\end{equation*}
which satisfies all the requirements. This proves the Harish-Chandra-\\Itzykson-Zuber integral.
\end{proof}

The \textsc{hciz}-integral can be checked numerically. To this end the integral over the unitary group must be parametrised. Indirect evidence or this formula can be found in Example~\ref{exm:constant2}.

\subsection{Symmetrisation of the HCIZ-integral\label{sec:symHCIZ}}
A common application of the \textsc{hciz}-integral is for integrals over Hermitean matrices, such as
\begin{equation*}
Z=\int \ud M\,\exp[-\Tr(Q\bullet M^{2})]\quad,
\end{equation*}
where both $Q$ and $M$ are Hermitean $N\times N$-matrices. Diagonalising \\$M=U\bullet \Lambda\bullet  U^{\ast}$ and integrating over the unitaries with the \textsc{hciz}-integral from Theorem~\ref{thrm:hciz} yields
\begin{equation}
Z=\!\mathpzc{U}\int_{\mathbb{R}^{N}}\!\!\!\!\!\!\ud^{N}\vec{\lambda}\,\frac{(-1)^{\binom{N}{2}}\big(\prod_{k=0}^{N-1}k!\big)\Delta(\lambda_{1},\ldots,\lambda_{N})^{2}}{\Delta(\lambda_{1}^{2},\ldots,\lambda_{N}^{2})\Delta(q_{1},\ldots,q_{N})}\!\!\det_{1\leq k,l\leq N}\!\left(e^{-\lambda_{k}^{2}q_{l}}\right)\;,\label{e:shciz1}
\end{equation}
where $q_{j}$ and $\lambda_{j}$ are the $j$-th eigenvalue of $Q$ and $\Lambda$ respectively. This remains finite as two of the parameters approach each other. For simplicity this is shown for $N=2$, although it holds generally. Assume $\lambda_{2}=-\lambda_{1}+\varepsilon$, then
\begin{equation*}
\frac{e^{-q_{1}\lambda_{1}^{2}-q_{2}\lambda_{2}^{2}}-e^{-q_{2}\lambda_{1}^{2}-q_{1}\lambda_{2}^{2}}}{\lambda_{2}+\lambda_{1}}=\frac{e^{-(q_{1}+q_{2})\lambda_{1}^{2}}}{\varepsilon}(2(q_{2}-q_{1})\lambda_{1}\varepsilon)+\mathcal{O}(\varepsilon)\quad.
\end{equation*}

The integral (\ref{e:shciz1}) can be rewritten as 
\begin{align}
&\hspace{-8mm}Z=\mathpzc{U}\frac{(-1)^{\binom{N}{2}}\big(\prod_{k=0}^{N-1}k!\big)}{\Delta(q_{1},\ldots,q_{N})}\!\sum_{\sigma\in\mathcal{S}_{N}}\!\!\sgn(\sigma)\!\int_{\mathbb{R}^{N}}\!\!\!\!\!\ud^{N}\vec{\lambda}\,\big[\!\!\!\prod_{1\leq k<l\leq N}\frac{\lambda_{l}-\lambda_{k}}{\lambda_{l}+\lambda_{k}}\big]e^{-\sum_{j}\lambda_{j}^{2}q_{\sigma(j)}}\nonumber\\
&=\frac{(-\pi)^{\binom{N}{2}}}{\Delta(q_{1},\ldots,q_{N})}\int_{\mathbb{R}^{N}}\!\!\!\ud^{N}\vec{\lambda}\,\big[\!\!\!\prod_{1\leq k<l\leq N}\frac{\lambda_{l}-\lambda_{k}}{\lambda_{l}+\lambda_{k}}\big]e^{-\sum_{j}\lambda_{j}^{2}q_{j}}\quad.\label{e:shciz2}
\end{align}
In the last step we renamed in each term of the sum over the permutation group the integration variables $k\mapsto \sigma(k)$. This maps the Vandermonde determinant in the numerator
\begin{equation*}
\prod_{1\leq k<l\leq N}(\lambda_{l}-\lambda_{k})\mapsto \prod_{1\leq k<l\leq N}(\lambda_{\sigma(l)}-\lambda_{\sigma(k)})=\sgn(\sigma)\prod_{1\leq k<l\leq N}(\lambda_{l}-\lambda_{k})\quad.
\end{equation*}
At first this may seem impossible. In (\ref{e:shciz1}) it is not difficult to see that the integral remains finite as $q_{2}\rightarrow q_{1}+\varepsilon$ or $\lambda_{2}=-\lambda_{1}+\varepsilon$, while $\varepsilon\rightarrow0$. This is no longer manifest in (\ref{e:shciz2}). To see that no divergence appears, note that we have a single weight function in the sense of Theorem~\ref{thrm:op} for the integrals over $\lambda_{1}$ and $\lambda_{2}$ in this case. Rewriting the Vandermonde determinant in the monic orthogonal polynomials for this weight functions shows that the integral produces precisely the required zeroes.

\section{Quartic integrals and differential equations\label{sec:Qide}}
Quartic \textsc{qft}'s lead naturally to quartic integrals in the partition function. For regimes of strong coupling the approximations based on Gaussian integrals will not work. Other methods are required then. To this end some basic results will be demonstrated.\\
It is assumed that much of the partition functions' analytics are captured by the following functions
\begin{align}
&\mathpzc{k}_{n}(\mu)=\int_{0}^{\infty}\ud\lambda\,\lambda^{n-\frac{1}{2}}e^{-\lambda-\frac{\lambda^{2}}{\mu}}=\mu^{\frac{2n+1}{4}}\int_{-\infty}^{\infty}\ud\lambda\,\lambda^{2n}e^{-\sqrt{\mu}\lambda^{2}-\lambda^{4}}\label{e:defk1}\quad;\\
&\mathpzc{K}_{n}(\mu)=\int_{0}^{\infty}\ud\lambda\,\lambda^{n-\frac{1}{2}}e^{-\sqrt{\mu}\lambda-\lambda^{2}}\label{e:defk2}\quad;\\
&\mathpzc{h}_{n}(e_{j},g,u)=g^{-\frac{n+1}{4}}\mathpzc{h}_{n}(\frac{e_{j}}{\sqrt{g}},1,\frac{u}{g^{1/4}})=\!\int_{-\infty}^{\infty}\!\!\!\ud \lambda\,\lambda^{n}e^{-e_{j}\lambda^{2}-g\lambda^{4}-iu\lambda}\quad.\label{e:defh}
\end{align}

\begin{rmk}[Pearcey integral]\label{rmk:PI}
The Pearcey integral 
\begin{equation*}
P(x,y)=\int_{-\infty}^{\infty}\ud \lambda\,e^{i(\lambda^{4}+x\lambda^{2}+y\lambda)}
\end{equation*}
is closely related to $\mathpzc{h}_{0}(e,g,u)$. Writing the coupling as a complex parameter $|g|e^{i\vartheta}$, we obtain
\begin{equation*}
\mathpzc{h}_{0}(e,|g|e^{i\vartheta},u)=|g|^{-\frac{1}{4}}e^{-\frac{i}{4}(\vartheta+\frac{\pi}{2})}P(\frac{e}{\sqrt{|g|}}e^{-\frac{i}{2}(\vartheta-\frac{\pi}{2})},\frac{u}{|g|^{1/4}}e^{-\frac{i}{4}(\vartheta-\frac{7\pi}{2}})
\end{equation*}
as long as $|\vartheta|<\pi/2$. In paragraph~\ref{sec:PI} we discuss some techniques for this integral.
\end{rmk}

\begin{lemma}\label{l:mBe}
The function $K_{\alpha}:\mu\mapsto 2e^{-\mu/8}\frac{1}{\sqrt{\mu}}\mathpzc{k}_{0}(8\mu)$ is the second order solution of the modified Bessel's equation
\begin{equation}
0=\left(\mu^{2}\partial_{\mu}^{2}+\mu\partial_{\mu}-(\mu^{2}+\alpha^{2})\right)K_{\alpha}(\mu)\quad\label{e:mBe}
\end{equation}
with $\alpha=1/4$.
\end{lemma}

\begin{proof}
The functions
\begin{equation*}
k_{n}(\mu)=\int_{-\infty}^{\infty}\ud \lambda\,\lambda^{n}e^{-\lambda^{2}-\frac{1}{8\mu}\lambda^{4}}
\end{equation*}
satisfy the identities
\begin{equation*}
0=\int \ud\lambda\,\frac{\ud}{\ud \lambda}\lambda^{n} \,e^{-\lambda^{2}-\frac{1}{8\mu}\lambda^{4}}=nk_{n-1}-2k_{n+1}-\frac{1}{2\mu}k_{n+3}\quad.
\end{equation*}
For $n=1,3,5$ we find $0=k_{0}-2k_{2}-\frac{1}{2\mu}k_{4}$, $0=3k_{2}-2k_{4}-\frac{1}{2\mu}k_{6}$ and $0=5k_{4}-2k_{6}-\frac{1}{2\mu}k_{8}$, which we can sum in such a combination so that $k_{2}$ and $k_{6}$ disappear. In this way we obtain
\begin{equation*}
0=k_{0}-\frac{4}{3}k_{4}(\frac{1}{\mu}+1)+\frac{1}{12\mu^{2}}k_{8}\quad.
\end{equation*}
Because $(8\mu^{2}\frac{\ud}{\ud\mu})^{m}k_{0}=k_{4m}$, we can write this as
\begin{equation*}
0=\Big(\mu^{2}\frac{\ud^{2}}{\ud\mu^{2}}-2\mu^{2}\frac{\ud}{\ud \mu}+\frac{3}{16}\Big)k_{0}(\mu)=0\quad.
\end{equation*}
Writing $k_{0}(\mu)=f(\mu)\tilde{k}_{0}(\mu)$, this becomes
\begin{equation*}
0=f\Big(\mu^{2}\frac{\ud^{2}}{\ud \mu^{2}}+(2\mu^{2}\frac{f'}{f}-2\mu^{2})\frac{\ud}{\ud \mu}+\frac{3}{16}+\mu^{2}\frac{f''}{f}-2\mu^{2}\frac{f'}{f}\Big)\tilde{k}_{0}(\mu)=0\quad.
\end{equation*}
The first order condition yields $\frac{f'}{f}=\frac{1}{2\mu}+1$, so that we find $f(\mu)=\sqrt{\mu}e^{\mu}$, which yields the modified Bessel's equation with $\alpha=1/4$.
\end{proof}
\begin{rmk}\label{rmk:Bfo}
A second order differential equation has in general two independent solutions. In the case of Bessel functions the exponentially increasing one is the `first' and the decaying one the `second'. So the above is the modified Bessel function of the second kind.
\end{rmk}

By the same reasoning for general $n$ we find the following lemma.
\begin{lemma}\label{l:mTe}
For $n\in\mathbb{N}_{0}$ the function
\begin{align*}
&\hspace{-8mm}\mathpzc{k}_{n}(\mu)=\frac{1}{2}\mu^{\frac{2n+1}{4}}\sum_{k=0}^{\infty}\frac{(-\sqrt{\mu})^{k}}{k!}\Gamma(\frac{2(n+k)+1}{4})\\
&=2^{-\frac{2n+3}{2}}\mu^{\frac{2n+3}{4}}\Gamma(n+\frac{1}{2})U(\frac{2n+3}{4},\frac{3}{2},\frac{\mu}{4})\quad,
\end{align*}
where $U(a,b,z)$ is Tricomi's hypergeometric function.
\end{lemma}
\begin{proof}
Differentiating under the integral for we find
\begin{equation*}
0=(n+1/2)\mathpzc{k}_{n}(\mu)-\mathpzc{k}_{n+1}(\mu)-\frac{2}{\mu}\mathpzc{k}_{n+2}(\mu)\quad.
\end{equation*}
Combining this with the corresponding equations for $n+1$ and $n+2$ we can construct the equation
\begin{equation*}
0=(n+1/2)\mathpzc{k}_{n}-(\frac{2}{\mu}+\frac{1}{n+3/2}+\frac{2}{\mu}\frac{n+5/2}{n+3/2})\mathpzc{k}_{n+2}+\frac{4}{\mu^{2}}\frac{1}{n+3/2}\mathpzc{k}_{n+4}\quad.
\end{equation*}
Observing that $\mathpzc{k}_{n+2}(\mu)=\mu^{2}\partial_{\mu}\mathpzc{k}_{n}(\mu)$ we can turn this into the differential equation
\begin{equation*}
0=\frac{4}{n+3/2}\Big[\mu^{2}\partial_{\mu^{2}}+(-n-\frac{\mu}{4})\mu\partial_{\mu}+\frac{(n+1/2)(n+3/2)}{4}\Big]\mathpzc{k}_{n}(\mu)\quad.
\end{equation*}
Switching to the function $\tilde{\mathpzc{k}_{n}}(\mu)=\frac{4}{n+3/2}\frac{1}{f(\mu)}\mathpzc{k}_{n}(\mu)$ with $f(\mu)=e^{\mu/8}\mu^{\frac{n+1}{2}}$ this differential equation becomes
\begin{equation}
0=\Big[\mu^{2}\partial_{\mu}^{2}+\mu\partial_{\mu}-\big(\frac{\mu^{2}}{64}+n\frac{\mu}{8}+\frac{1}{16}\big)\Big]\tilde{\mathpzc{k}_{n}}(\mu)\quad,\label{e:bfh1}
\end{equation}
which only for $n=0$ yields the modified Bessel's equation. To see what this is for general $n$, we make a small detour. The general solution $w(z)$ of the Kummer equation
\begin{equation*}
0=\Big[z\partial_{z}^{2}+(b-z)\partial_{z}-a\Big]w(z)
\end{equation*}
is split $w(z)=z^{\frac{1-b}{2}}e^{z/2}v(z)$, so that
\begin{equation*}
0=\Big[z^{2}\partial_{z}^{2}+z\partial_{z}-\Big(\big(\frac{z}{2}\big)^{2}+\frac{z}{2}(2a-b)+\frac{(b-1)^{2}}{4}\Big)\Big] v(z)\quad,
\end{equation*}
which for $z=\mu/4$, $b=3/2$ and $a=(2n+3)/4$ yields (\ref{e:bfh1}). The two indepenent solutions of the Kummer equation are the Kummer function (or generalised Laguerre polynomial)
\begin{equation*}
M(a,b,z)=\sum_{n=0}^{\infty}\frac{z^{n}}{n!}\frac{a^{(n)}}{b^{(n)}}\quad,
\end{equation*}
where $a^{(n)}=\prod_{j=0}^{n-1}(a+j)$ with $a^{(0)}=1$ is the rising factorial, and the Tricomi function
\begin{equation*}
U(a,b,z)=\frac{\Gamma(1-b)}{\Gamma(a-b+1)}M(a,b,z)+\frac{\Gamma(b-1)}{\Gamma(a)}z^{1-b}M(a-b+1,2-b,z)\quad.
\end{equation*}
Inverting all the steps above, it is straightforward to find the general solution
\begin{equation*}
\mathpzc{k}_{n}(\mu)=\frac{2n+3}{4}2^{-\frac{3}{2}}\mu^{\frac{2n+3}{4}}\big(\alpha M(\frac{2n+3}{4},\frac{3}{2},\frac{\mu}{4})+\beta U(\frac{2n+3}{4},\frac{3}{2},\frac{\mu}{4})\big).
\end{equation*}
Since $\mathpzc{k}_{n}(\mu)\rightarrow 0$ for $\mu\rightarrow 0$ it follows that $\alpha=0$. To find $\beta$ we first rewrite
\begin{equation*}
U(a,\frac{3}{2},z)=\frac{4^{a-1}}{\sqrt{z}\Gamma(2a-1)}\sum_{k=0}^{\infty}\frac{(-2\sqrt{z})^{k}}{k!}\Gamma(\frac{2a+k-1}{2})\quad,
\end{equation*}
for which we have used some standard $\Gamma$-function properties and compute
\begin{equation*}
\mathpzc{k}_{n}(\mu)\sim \frac{1}{2}\Gamma(\frac{2n+1}{4})\mu^{\frac{2n+1}{4}}\qquad,\,\mu\rightarrow0\quad,
\end{equation*}
so that $\beta=2^{-n}\frac{4}{2n+3}\Gamma(\frac{2n+1}{2})$.\\
This proves the statement.
\end{proof}

\begin{rmk}\label{rmk:mTe}
A simpler proof is given by the expansion
\begin{align*}
&\hspace{-8mm}\mathpzc{k}_{n}(\mu)=\frac{1}{2}\mu^{\frac{2n+1}{4}}\sum_{k=0}^{\infty}\frac{(-\sqrt{\mu})^{k}}{k!}\int_{0}^{\infty}\ud y\,e^{-y}y^{\frac{2(n+k)-3}{4}}\\
&=\frac{1}{2}\mu^{\frac{2n+1}{4}}\sum_{k=0}^{\infty}\frac{(-\sqrt{\mu})^{k}}{k!}\Gamma(\frac{2(n+k)+1}{4})\quad.
\end{align*}
\end{rmk}

\begin{rmk}
From graphics and expansions $\mu\approx 0$ it is not difficult to see that 
\begin{equation}
\big|\frac{\mathpzc{K}_{n}(\mu)}{\mathpzc{K}_{n-1}(\mu)}\big|\leq\frac{\mathpzc{K}_{n}(0)}{\mathpzc{K}_{n-1}(0)}=\frac{\Gamma(\frac{2n+1}{4})}{\Gamma(\frac{2n-1}{4})}\leq\sqrt{n}\quad.\label{e:Kest}
\end{equation}
How can one prove that?
\end{rmk}

\begin{lemma}\label{l:WTf}
A solution $w_{\kappa,m}$ of the Whittaker equation is related through
\begin{equation*}
w_{\kappa,m}(z)=e^{-\frac{z}{2}}z^{m+\frac{1}{2}}u(z)
\end{equation*}
to a solution $u$ of the Kummer equation.
\end{lemma}
\begin{proof}
The Whittaker equation
\begin{equation*}
0=\Big[\partial_{z}^{2}+(\frac{\frac{1}{4}-m^{2}}{z^{2}}+\frac{\kappa}{z}-\frac{1}{4})\Big]w_{\kappa,m}(z)
\end{equation*}
turns after the substitution
\begin{equation*}
w_{\kappa,m}(z)=e^{-\frac{z}{2}}z^{m+\frac{1}{2}}u(z)
\end{equation*}
into
\begin{equation*}
0=e^{-\frac{z}{2}}z^{m-\frac{1}{2}}\Big[z\partial_{z}^{2}+(2m+1-z)\partial_{z}+(\kappa-m-\frac{1}{2})\Big]u(z)\quad,
\end{equation*}
which is the Kummer equation for $a=m-\kappa+\frac{1}{2}$ and $b=2m+1$.
\end{proof}
The conventions prescribe
\begin{align*}
&\hspace{-8mm}W_{\kappa,m}(z)=z^{m+\frac{1}{2}}e^{-\frac{z}{2}}U(m-\kappa+\frac{1}{2},2m+1,z)\qquad\text{ and }\\
&\hspace{-8mm}M_{\kappa,m}(z)=z^{m+\frac{1}{2}}e^{-\frac{z}{2}}M(m-\kappa+\frac{1}{2},2m+1,z)\quad.
\end{align*}

\begin{cor}
In terms of the Whittaker function, the integral
\begin{equation*}
\mathpzc{k}_{n}(\mu)=\Gamma(n+\frac{1}{2})\big(\frac{\mu}{4}\big)^{\frac{n}{2}}e^{\frac{\mu}{2}}W_{-\frac{n}{2},\frac{1}{4}}(\mu/4)\quad.
\end{equation*}
\end{cor}

\begin{lemma}\label{l:heatsol}
The function $\mathpzc{h}_{n}(x,g,u)$ is given by
\begin{equation*}
\mathpzc{h}_{n}(e_{j},g,u)=\frac{(-i)^{n}}{\sqrt{8sg}}\int_{0}^{\infty}\ud y\,y^{-\frac{n+1}{2}}H_{n}(\frac{u}{\sqrt{2y}})e^{-\frac{u^{2}}{4y}}e^{-\frac{(y-e_{j})^{2}}{4sg}}\quad
\end{equation*}
with $s=-1$.
\end{lemma}
\begin{proof}
We recall (\ref{e:defh}) from which it follows that
\begin{equation*}
\big(\frac{\ud}{\ud g}-s\frac{\ud^{2}}{\ud x^{2}}\big)\mathpzc{h}_{n}(x,g,u)=0\quad,
\end{equation*}
with $s=-1$. The fundamental solution of this is given by
\begin{equation*}
\mathpzc{h}_{n}(x,g,u)=\frac{1}{\sqrt{4\pi st}} \int_{-\infty}^{\infty}\ud y\, \psi_{n}(y)\,e^{-\frac{(y-x)^{2}}{4st}}
\end{equation*}
with the initial condition
\begin{align*}
&\hspace{-8mm}\psi_{n}(y)=\mathpzc{h}_{n}(y,0,u)=i^{n}\frac{\partial^{n}}{\partial u^{n}}\sqrt{\frac{\pi}{y}}e^{-\frac{u^{2}}{4y}}\vartheta(y)\\
&=\big(\frac{-i}{\sqrt{2y}}\big)^{n}\sqrt{\frac{\pi}{y}}H_{n}(\frac{u}{\sqrt{2y}})e^{-\frac{u^{2}}{4y}}\vartheta(y)\quad,
\end{align*}
where $\vartheta$ is the Heaviside step function and $H_{n}$ is the $n$-th Hermite polynomial (\ref{e:hpdef}). The convolution with the heat kernel yields the result immediately.
\end{proof}

\section{Pearcey's Integral\label{sec:PI}}

In Remark~\ref{rmk:PI} Pearcey's integral was already mentioned. It may play a major role in the evaluation of the partition function for the Grosse-Wulkenhaar model for strong coupling. To see how this may work, the function and a rapid converging evaluation method are presented. Afterwards these techniques are discussed in the context of complex-valued arguments.\\

On the basis of the texts~\cite{paris1,wright} and~\cite{paris2} some basic properties of Pearcey's integral~\cite{pearcey}
\begin{align}
&\hspace{-8mm}P(x,y)=\int_{-\infty}^{\infty}\ud \lambda\,e^{if(\lambda)}\quad,\text{ where }\label{e:pi}\\
&f(\lambda)=\lambda^{4}+x\lambda^{2}+y\lambda\quad\label{e:pi2}
\end{align}
for real-valued arguments $x$ and $y$ are presented.\\

The straightforward method is to rotate the parameter and simply integrate over the real line. This yields a converging series of hypergeometric functions. However, it is alternating, making the convergence rather slow. It is for this reason that another integration scheme is considered.\\

This integrand is interpreted as a function on the complex plane. To calculate or approximate it one makes use of the steepest descent method from paragraph~\ref{sec:Mfi}. Viewing the real part of this integral as a landscape, this method tells us that we have to cross the mountain ridges at the lowest points, the saddle points, and approach these saddle points via the path of steepest ascent. These are also the path of constant phase, so that all contributions are added coherently.\\

Setting $\lambda=Re^{i\varphi}$ with $R$ large and positive it follows that there are four valleys and mountain ridges at the periphery of the plane, where the valleys correspond to large positive imaginary part of $f(\lambda)$ and the hills to large negative. See also Figure~\ref{f:3pI}.\\

If the discriminant $\Delta$ of the polynomial $f'(\lambda)=4\lambda^{3}+2x\lambda+y$ with real coefficients is negative, its saddles are real and distinct and hence the integration contour consists of three different paths. With the above remark it is easy to see that the contour in this case goes first from $-\infty$ to $-e^{i\pi/8}\infty$, then one picks the path with the phase correpondings to the left saddle, crosses it and continues this path to the edge of the plane. This must be in a valley again, since the saddle is the highest point of this path of constant phase. This is repeated for the other two saddles, until the path ends up in $e^{i\pi/8}\infty$, from which there is a vanishing path to $+\infty$. In Figure~\ref{f:3pR} there is a height line from 8 to 11 o'clock and then one at a different height from there to 5 o'clock and finally to 2 o'clock.\\

As the discriminant become smaller, two of the real roots approach each other, until they touch at $\Delta=0$. As soon as it becomes negative, the roots seperate again and move in the imaginary direction. When the second and third root coalesce, the paths from the second root to $e^{-3i\pi/8}\infty$ and from there to the third root cancel each other, so that a contour composed of two such height lines is obtained. The height lines for this situation are shown in Figure~\ref{f:2pR}.\\
The line corresponding to $\Delta=0$ is given by the equations $f'(\lambda)=f''(\lambda)=0$. Solving this yields the caustic
\begin{equation}
y^{2}+(\frac{2x}{3})^{3}=0\quad.\label{e:caustic}
\end{equation}
This is the edge of the orange region in Figure~\ref{f:caus}, where the cubic polynomial has three real roots. This is also the line that would have been founded from solving $\Delta=0$.

\begin{figure}[!thb]
\begin{subfigure}[t]{0.5\linewidth}
\includegraphics[width=0.9\textwidth]{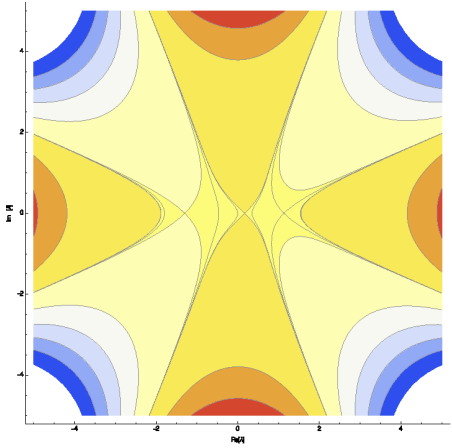}
\caption{The saddle points are marked by\\ the crossing height lines. Red is a \\positive real part, blue a negative.\label{f:3pR}}
\end{subfigure}%
\begin{subfigure}[t]{0.5\linewidth}
\includegraphics[width=0.9\textwidth]{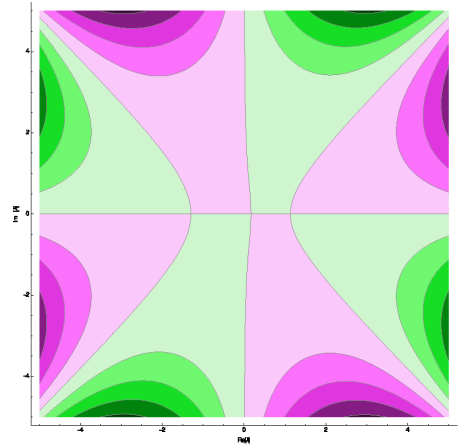}
\subcaption{The saddle points are the crosses at which the colours alternate. Purple corresponds to a positive imaginary part, green to a negative.\label{f:3pI}}
\end{subfigure}
\caption{The real and imaginary phase diagram of $f$ in (\ref{e:pi2}) for $x=-3$ and $y=1$.\label{f:3p}}
\end{figure}

\begin{figure}[!ht]
\begin{subfigure}[t]{0.5\linewidth}
\includegraphics[width=0.9\textwidth]{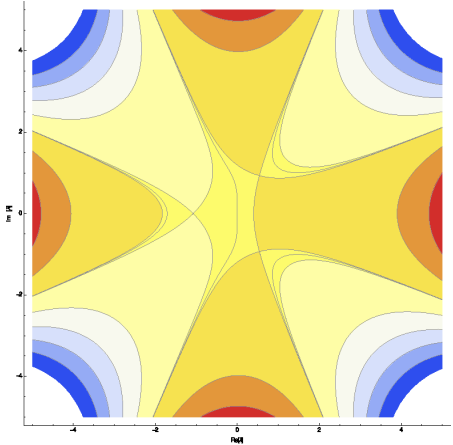}
\caption{The real phase diagram for $x=0$\\ and $y=5$.\label{f:2pR}}
\end{subfigure}%
\begin{subfigure}[t]{0.5\linewidth}
\includegraphics[width=0.9\textwidth]{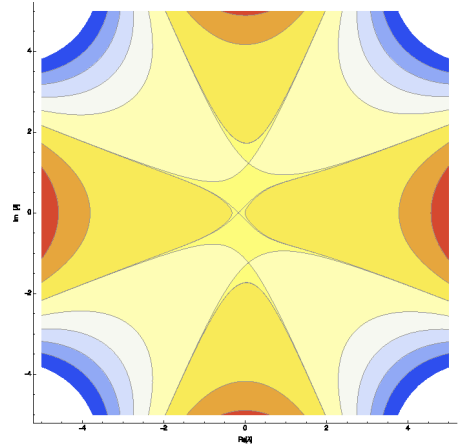}
\subcaption{The real phase diagram for $x=3$\\ and $y=1$.\label{f:1pR}}
\end{subfigure}%
\caption{Two examples of phase diagrams of $\re[f(\lambda)]$ for two and three integration curves.\label{f:2p}}
\end{figure}

The caustic above was obtained from the condition that two height lines of the phase merge. This can happen once more. The real parameters $x$ and $y$ imply that there is one real root, called $\lambda_{1}$. The other saddle points is denoted by $\lambda_{2}$. The lines of steepest descent are also lines of constant phase. If two curves of constant phase share a point (asymptotically) and have the same phase, they are identical. This shows that the condition at which they become one is
\begin{equation}
\re[f(\lambda_{1})-f(\lambda_{2})]=0\quad.\label{e:2n1cond}
\end{equation}
This condition yields another caustic line, called the Stokes line.

\subsection{The Stokes line for real parameters\label{sec:PIrp}}
To find the Stokes line, defined by (\ref{e:2n1cond}), the phases of the stationary points must be considered. The function $f:\lambda\mapsto \lambda^{4}+x\lambda^{2}+y\lambda$ has three stationary points $\lambda_{s}$ with $s=1,2,3$. These are the solutions of $f'(\lambda_{s})=0$. Assuming that $x,y\in\mathbb{R}$ ensures that there is one real solution, which will be called $\lambda_{1}$. For all three points the function's value at $\lambda_{s}$ is given by
$f(\lambda_{s})=f(\lambda_{s})-\lambda_{s}f'(\lambda_{s})/4=x\lambda_{s}^{2}/2+3y\lambda_{s}/4$. The comparison of the real parts of the phases yields the coincidence relation
\begin{equation}
0=2x(\lambda_{1}^{2}-\lambda_{2,R}^{2}+\lambda_{2,I}^{2})+3y(\lambda_{1}-\lambda_{2,R})\quad,\label{e:2n1condb}
\end{equation}
where the subscripts $R$ and $I$ indicate real and imaginary parts respectively.\\

Pearcey's integral is symmetric in the second argument, $P(x,y)=P(x,-y)$, so that without loss of generalisation it may be assumed that $y\geq0$. Furthermore, this implies that that it is safe to assume that at least one root lies on the negative real axis.\\

The well-known relation
\begin{align*}
&\hspace{-8mm}\sin^{3}(\varphi)-\frac{3}{4}\sin(\varphi)+\frac{1}{4}\sin(3\varphi)=\sin^{3}(\varphi)-\frac{3}{4}\sin(\varphi)\\
&+\frac{1}{4}\big(\sin(\varphi)\cos(2\varphi)+\sin(2\varphi)\cos(\varphi)\big)\\
&=\frac{3}{4}\sin(\varphi)\big(\sin^{2}(\varphi)+\cos^{2}(\varphi)\big)-\frac{3}{4}\sin(\varphi)=0
\end{align*}
allows to solve the equation $4\lambda^{3}+2x\lambda+y=0$. With $\lambda=i\sqrt{\frac{2x}{3}}\sin(\varphi)$ this yields
\begin{equation*}
0=-4i\big(\frac{2x}{3}\big)^{\frac{3}{2}}\Big(\sin^{3}(\varphi)-\frac{3}{4}\sin(\varphi)+\frac{1}{4}\sin(3\varphi)\Big)\quad,
\end{equation*}
where $\sin(3\varphi)=iy\big(\frac{3}{2x}\big)^{\frac{3}{2}}$. There is one real root, $\lambda_{1}$ and two complex roots $\lambda_{2,3}$, but these two have the same real value, since they are each others conjugated. This means that it is safe to take $\lambda_{1}=\sin(\varphi)$, which is real, and $\lambda_{2}=\sin(\varphi+\frac{2\pi}{3})$.\\

This can be split, using
\begin{equation*}
\sin(a+bi)=\sin(a)\cosh(b)+i\cos(a)\sinh(b)
\end{equation*}
for real $a$ and $b$. Assuming that both $x$ and $y$ are real and furthermore $x$ is positive, it is clear that $\varphi=i\vartheta$ is purely imaginary. What happens for negative $x$ will be treated later. The stationary points are given by
\begin{align*}
&\lambda_{1}=-\sqrt{\frac{2x}{3}}\sinh(\vartheta)\\
&\lambda_{2}=-\frac{1}{2}\sqrt{\frac{2x}{3}}\sinh(\vartheta)+i\frac{\sqrt{3}}{2}\sqrt{\frac{2x}{3}}\cosh(\vartheta)\quad,
\end{align*}
so that the coalescence condition becomes
\begin{equation}
0=\cosh^{2}(\vartheta)+\sinh^{2}(\vartheta)-A\sinh(\vartheta)\quad,\label{e:ccon1}
\end{equation}
where $\sinh(3\vartheta)=A/2$ and $A=(y/\sqrt{2})(3/x)^{3/2}$.\\

The next step is now to remove $\vartheta$ from these equations, so that $A$ can be determined. From this the relation between $x$ and $y$, for which the number of paths in the integration contour reduces to one, will follow. Standard trigonometric identities yield
\begin{equation}
\frac{A}{2}=\sinh(3\vartheta)=3\sinh(\vartheta)+4\sinh^{3}(\vartheta)\;\Rightarrow\; \sinh^{3}(\vartheta)=\frac{A}{8}-\frac{3}{4}\sinh(\vartheta)\;\label{e:ccon2}
\end{equation}
and the following reformulation of (\ref{e:ccon1})
\begin{equation}
0=\sinh^{2}(\vartheta)-\frac{A}{2}\sinh(\vartheta)+\frac{1}{2}\quad.\label{e:ccon3}
\end{equation}
Removing the term $\sinh^{3}(\vartheta)$ from (\ref{e:ccon2}) using (\ref{e:ccon3}) multiplied by $\sinh(\vartheta)$, the quadratic equation
\begin{equation}
0=A\sinh^{2}(\vartheta)+\frac{1}{2}\sinh(\vartheta)-\frac{A}{4}\quad\label{e:ccon4}
\end{equation}
is obtained. From this the equations
\begin{equation}
(A^{2}+1)\sinh(\vartheta)=\frac{3A}{2}\quad\text{and}\quad (A^{2}+1)\sinh^{2}(\vartheta)=\frac{A^{2}}{4}-\frac{1}{2}\label{e:ccon5}\quad
\end{equation}
are obtained. Eliminating $\sinh(\vartheta)$ from this gives the equation of interest
\begin{equation}
0=A^{4}-10A^{2}-2\quad.\label{e:ccon6}
\end{equation}
This is quadratic in $A^{2}$ and implies
\begin{equation}
\frac{27}{2}y^{2}=x^{3}(5\pm \sqrt{27})\quad.\label{e:ccon7}
\end{equation}
Since the parameters are real, both sides must be postive and the plus sign must be chosen. The minus sign requires $x\leq0$, but the entire line lies then to the left of the caustic, so that its meaning is lost. Therefore, the Stokes line is given by
\begin{equation}
0=\frac{27}{2}y^{2}-x^{3}(5+\sqrt{27})\quad.\label{e:stokesl}
\end{equation}
In the grey region in Figure~\ref{f:caus} only one height line contributes to the integration contour. An example of a single height line that forms the integration contour for the Pearcey integral is shown in Figures~\ref{f:1pR}.

\begin{figure}[!htb]
\includegraphics[width=0.8\textwidth]{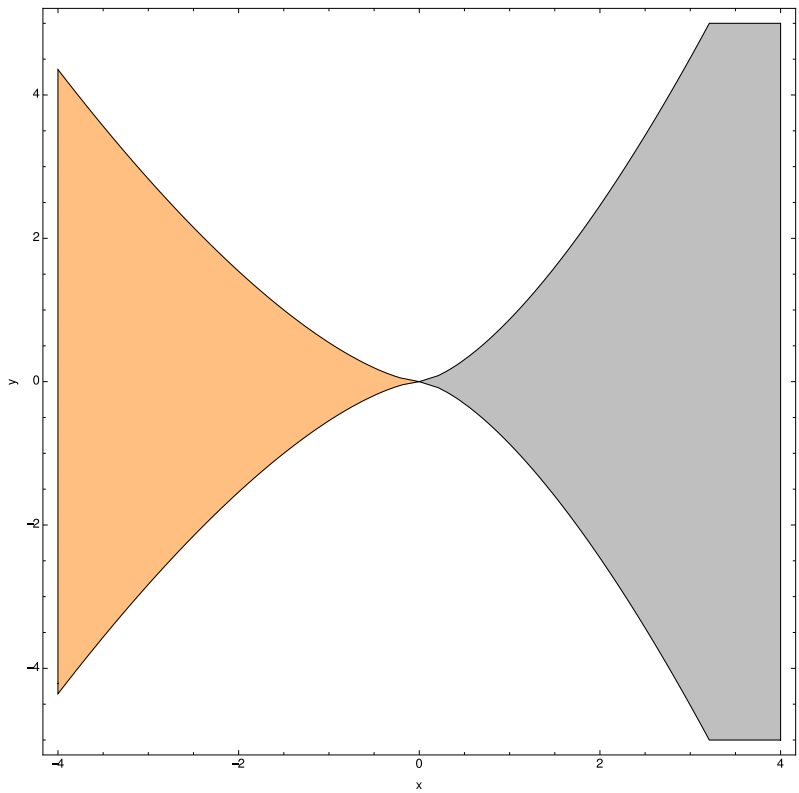}
\caption{The caustic (\ref{e:caustic}) in orange and Stokes line (\ref{e:stokesl}) in grey for (\ref{e:pi}).\label{f:caus}}
\end{figure}

\subsection{Complex parameters\label{sec:Cp}}
For application to quartic matrix models the case of real parameters is not enough. Complex parameters are needed, but the more general case is more difficult. The same integration scheme of paths of steepest descent is still applicable, but exact results for the number of paths needed are not easily obtained. A straightforward option is to repeat the same steps and produce results for certain parameter regimes. A suitable expansion parameter is $A=(\frac{3y}{2x})^{\frac{3}{2}}$. From assumptions on the size and argument of $A$ approximate results on the number of contours needed can be determined. This yields rapidly converging asymptotic expansions of the Pearcey integral.\\

Another option is to consider solutions of differential equations satisfied by $P(x,y)$. These and some other methods have been studied extensively recently~\cite{lopez1, lopez2, lopez3}. This makes an application to quartic matrix models for large coupling worth trying.

\chapter{Determinant techniques\label{sec:Det_tech}}

It has been shown in Chapter\ref{sec:DHm} that determinants appear naturally in matrix integrals. This makes it necessary to study determinants and approximation schemes for them. In this chapter $N\times N$-matrices $M=(a_{kl})$ with entries of the form $a_{kl}=f(k,l)$ for some function $f$ are considered, $N$ will be assumed large. In particular, most examples will deal with exponential functions and integrals over those. It is clear that this setup cannot be dealt with in full generality, but in certain cases it is possible to provide some structure and results.

\section{Power series decomposition of the determinant}

The first goal is the Lemma~\ref{l:detcom}. A part of this Lemma, $\gamma\rightarrow0$, can be seen directly from the Harish-Chandra-Itzykson-Zuber integral from Theorem~\ref{thrm:hciz}. However, the computations used will be useful later on. Therefore, a couple of technical lemmas are presented.

\begin{lemma}\label{l:permuf}
For $p,n\in\mathbb{N}_{0}$ and arguments $x_{j}\in\mathbb{C}$ that satisfy $x_{j}\neq x_{k}$ for all $1\leq j<k\leq n$, the function
\begin{equation*}
F_{n,p}(\vec{x})=\sum_{k=1}^{n}x_{k}^{p}\cdot\big[\prod_{\stackrel{t=1}{t\neq k}}^{n}\frac{1}{x_{t}-x_{k}}\big]
\end{equation*}
is given by
\begin{align*}
&\hspace{-4mm}\bullet\, F_{n,p}(\vec{x})=0
\phantom{(-1)^{n-1}\!\!\!\!\!\!\!\!\!\!\!\!\sum_{1\leq j_{1}\leq j_{2}\leq\ldots\leq j_{m}\leq n}\!\!\!\!\!\!\!\!\!\!\!\!\!\!\!x_{j_{1}}x_{j_{2}}\ldots x_{j_{m}} }
\!\!\!\!\!\!\qquad\text{, if }p=0,\ldots,n-2\quad;\\
&\hspace{-4mm}\bullet\, F_{n,p}(\vec{x})=(-1)^{n-1}
\phantom{0\!\!\!\!\!\!\!\!\!\!\!\!\sum_{1\leq j_{1}\leq j_{2}\leq\ldots\leq j_{m}\leq n}\!\!\!\!\!\!\!\!\!\!\!\!\!\!\!x_{j_{1}}x_{j_{2}}\ldots x_{j_{m}} }
\!\!\!\!\!\!\qquad\text{, if }p=n-1\quad;\\
&\hspace{-4mm}\bullet\, F_{n,p}(\vec{x})=(-1)^{n-1}\!\!\!\!\!\!\!\!\!\!\!\!\sum_{1\leq j_{1}\leq j_{2}\leq\ldots\leq j_{m}\leq n}\!\!\!\!\!\!\!\!\!\!\!\!\!\!\!x_{j_{1}}x_{j_{2}}\ldots x_{j_{m}} 
\phantom{0}
\!\!\!\!\!\!\qquad\text{, if }p=n\!-\!1\!+\!m \;\&\;m\in\mathbb{N}\;.
\end{align*}
For $p=n-1+m$ for $m,n\in\mathbb{N}_{0}$ $F_{n,p}$ consists of $\binom{n+m}{m}=\binom{p+1}{n}$ terms.
\end{lemma}

\begin{proof}
The function $F_{n,p}$, considered as a function of $x_{n}$ only with the other arguments fixed, is easily seen to be meromorphic. It is in fact holomorphic, because the poles in the various terms cancel each other.\\
This function is defined on the disk with radius $R$. For $p<n-1$ all terms of $f$ decay as $R^{-1}$ or faster as $R\rightarrow\infty$. Increasing $R$, the supremum of $|F_{n,p}|$ on the boundary can be made arbitrarily small. By the maximum modulus principle it follows that $F_{n,p}=0$.\\
For $p=n-1$ the situation is only slightly different. The first $n-1$ terms vanish at the boundary. As $|x_{n}|\rightarrow\infty$, the last term becomes $(-1)^{n-1}$. This shows that $F_{n,n-1}=(-1)^{n-1}$.\\

For the case $p\geq n$ is is needed that
\begin{equation*}
\frac{x_{1}^{l}-x_{2}^{l}}{x_{1}-x_{2}}=\sum_{j=0}^{l-1}x_{1}^{j}x_{2}^{l-1-j}\qquad\text{for }l\in\mathbb{N}\quad,
\end{equation*}
which proves the case $n=2$ directly for all $p\geq n$. We proceed by induction on $n$. Subtracting multiples of $F_{n,p}=0$, it follows that
\begin{align*}
&\hspace{-8mm}F_{n,p}=\sum_{k=1}^{n}\big[\prod_{\stackrel{t=1}{t\neq k}}^{n}\frac{1}{x_{t}-x_{k}}\big]x_{k}^{p}-x_{n}^{p}F_{n,0}\\
&=-\sum_{k=1}^{n-1}\big[\prod_{\stackrel{t=1}{t\neq k}}^{n-1}\frac{1}{x_{t}-x_{k}}\big]\sum_{l=0}^{p-1}x_{k}^{l}x_{n}^{p-1-l}\\
&=-\sum_{k=1}^{n-1}\big[\prod_{\stackrel{t=1}{t\neq k}}^{n-1}\frac{1}{x_{t}-x_{k}}\big]\sum_{l=0}^{p-1}x_{k}^{l}x_{n}^{p-1-l}+\sum_{m=0}^{n-3}x_{N}^{p-1-m}F_{n-1,m}\\
&=-\sum_{k=1}^{n-1}\big[\prod_{\stackrel{t=1}{t\neq k}}^{n-1}\frac{1}{x_{t}-x_{k}}\big]\sum_{l=n-2}^{p-1}x_{k}^{l}x_{n}^{p-1-l}\\
&=-\sum_{l=n-2}^{p-1}x_{n}^{p-1-l}F_{n-1,l}\quad.
\end{align*}
This recursion relation is satisfied by
\begin{equation*}
F_{n,p}=(-1)^{n-1}\!\!\!\!\!\!\!\!\sum_{1\leq j_{1}\leq j_{2}\leq\ldots\leq j_{m}\leq n}\!\!\!\!\!\!\!\!\!\!x_{j_{1}}x_{j_{2}}\ldots x_{j_{m}}\quad.
\end{equation*}
The number of terms of $F_{n,p}$ with $p=n-1+m$ is given by
\begin{align*}
&\hspace{-8mm}\tilde{f}_{n,m}=\!\!\!\sum_{1\leq j_{1}\leq j_{2}\leq\ldots\leq j_{m}\leq n}\!\!\!1=\!\sum_{k=1}^{n}f_{n-k+1,m-1}=f_{n-1,m}\!+\!f_{n,m-1}=\binom{n+m}{m}\,,
\end{align*}
which follows from the observation that this is precisely the recursion relation of Pascal's triangle. This concludes the proof.
\end{proof}

\begin{lemma}\label{l:permuf2}
For $p,n\in\mathbb{N}_{0}$ and arguments $x_{j}\in\mathbb{C}$ that satisfy $x_{j}\neq x_{k}$ for all $1\leq j<k\leq n$, the function
\begin{align*}
&\hspace{-8mm}G_{n,p}(\vec{x};z)=\sum_{k=1}^{n}x_{k}^{p}\big[\prod_{\stackrel{t=1}{t\neq k}}^{n}\frac{x_{t}-z}{x_{t}-x_{k}}\big]\\
&=\sum_{m=1}^{n}\sum_{\{j_{1},\ldots,j_{m}\}\subset \{1,\ldots, n\}}\sum_{k=0}^{m-1}(-z)^{m-1-k}\cdot F_{m,p+k}(x_{j_{1}},\ldots,x_{j_{m}})\quad.
\end{align*}
If $p=0$, then $G_{n,0}(\vec{x};z)=1$.
\end{lemma}

\begin{proof}
All fractions in $G_{n,p}$ are rewritten as
\begin{equation*}
\frac{x_{t}-z}{x_{t}-x_{k}}=1+\frac{x_{k}-z}{x_{t}-x_{k}}\quad.
\end{equation*}
The central idea is to expand products of these fractions and reorganise them into multiples of $F_{n,p}$ for various $n$, $p$ and sets of arguments $(x_{j_{1}},\ldots,x_{j_{m}})$.\\

The first step is the expansion of the product 
\begin{equation*}
\prod_{\stackrel{t=1}{t\neq k}}^{n}(1+\frac{x_{k}-z}{x_{t}-x_{k}})
\end{equation*}
contains $\binom{n-1}{m-1}$ terms that depend on precisely $m\geq 2$ of the arguments $x_{1},\ldots,x_{n}$. The product $1^{n-1}=\binom{n-1}{0}$ depends on none. This will be taken as $m=1$. All numerators in this expansion are of the form $x_{k}^{s-r}(-z)^{r}$ and can easily be multiplied by $x_{k}^{p}$.\\

The second step is to recognise that there is a sum of $n$ such products and there are $m$ terms needed to compose $F_{m,m-1}=(-1)^{m-1}$. That one can find these terms in the expansion follows from symmetry. This leads to
\begin{equation*}
G_{n,p}(\vec{x};z)=\sum_{m=1}^{n}\sum_{\{j_{1},\ldots,j_{m}\}\subset \{1,\ldots, n\}}\sum_{k=0}^{m-1}(-z)^{m-1-k}\cdot F_{m,p+k}(x_{j_{1}},\ldots,x_{j_{m}})\;.
\end{equation*}
Setting $p=0$, it follows that $F_{m,k}=0$, unless $k=m-1$.\\

This leads to $G_{n,0}(\vec{x};z)=-\sum_{m=1}^{n}\binom{n}{m}(-1)^{m}=1$.
\end{proof}

There is another way to explain the identies equalling one or zero. They are the equations that show that the (relatively simple) left-inverse (\ref{e:invVdm})
\begin{align}
&\hspace{-8mm}\delta_{kj}
=\!\sum_{i=1}^{n}x_{j}^{i-1}\times\frac{1}{\prod_{\stackrel{t=1}{t\neq k}}^{n}(x_{t}\!-\!x_{k})}\!\!\sum_{\stackrel{1\leq m_{1}<\cdots<m_{n-i}\leq n}{m_{1},\ldots,m_{n-i}\neq k}}\!\!\!\!\!\!(-1)^{i-1}x_{m_{1}}\cdots x_{m_{n-i}}\;.\label{e:invVdm2}
\end{align}
of the Vandermonde matrix $V_{ij}=x_{j}^{i-1}$ is also its right-inverse.

The simplest for of the matrix determinant decomposition is given by the following lemma.\\
\begin{lemma}\label{l:detcom}
Let $\gamma$ be a positive real number smaller than $N^{3}W_{L}(e^{-1})$. If all pairs of arguments $x_{k},y_{l}\in\mathbb{R}$ satisfy
\begin{equation*}
0\leq x_{k}y_{l}\leq \frac{\gamma}{N^{2}}\qquad,\text{for }1\leq k,l\leq N\quad,
\end{equation*}
then the $N\times N$-matrix
\begin{align*}
&\hspace{-8mm}\det_{1\leq k,l\leq N}\big(e^{c\cdot x_{k}y_{l}}\big)=\frac{c^{\binom{N}{2}}}{\prod_{m=0}^{N-1}m!}\det_{1\leq k,m\leq N}\big(x_{k}^{m-1}\big)\det_{1\leq m,l\leq N}\big(y_{l}^{m-1}\big)\times\Big(1+\mathcal{O}(\gamma)\Big)\,,
\end{align*}
where $c=\pm 1$ or $c=\pm i$
\end{lemma}
\begin{proof}
It is without consequence to assume that $\max_{k}|y_{k}|\leq \sqrt{\gamma /N^{2}}$ and $\max_{k}|x_{k}|\leq \sqrt{\gamma /N^{2}}$. If this were not the case, then we could multiply all $y_{k}$'s by $\frac{\sqrt{\gamma /N^{2}}}{\max_{k}|y_{k}|}$ and divide all $x_{k}$'s by it. Below the proof for $c=i$ is given, but this is nowhere essential. All the steps work equally well for the other choices.\\
It follows from Lemma~\ref{l:LTA} that the Taylor series of 
\begin{equation}
\exp[i\frac{\gamma}{N^{2}}]-\sum_{k=0}^{N-1}(i\frac{\gamma}{N^{2}})^{k}/(k!)=\sum_{k=N}^{\infty}(i\frac{\gamma}{N^{2}})^{k}/(k!)\label{e:relerr}
\end{equation}
converges rapidly. Next we introduce some large $\tilde{N}$ and write the matrix
\begin{equation*}
\big(e^{ix_{m}y_{n}}\big)\approx\sum_{j=1}^{N+\tilde{N}}\Big(\frac{(ix_{m})^{j-1}}{(j-1)!}\Big)\Big(y_{n}^{j-1}\Big)=\Big(A_{1}\,A_{2}\Big)\left(\begin{array}{c}\!B_{1}\!\\\!B_{2}\!\end{array}\right)
\end{equation*}
as a matrix product. The matrices $A_{1}$ and $B_{1}$ are $N\times N$ matrices given by
\begin{equation*}
(A_{1})_{mj}=\frac{(ix_{m})^{j-1}}{\sqrt{(j-1)!}}\quad\text{and}\quad(B_{1})_{jn}=\frac{y_{n}^{j-1}}{\sqrt{(j-1)!}}\quad.
\end{equation*}
Apart from multiplication by a diagonal matrix, one is a Vandermonde-matrix and the other is its transpose. The matrix $A_{2}$ is of size $N\times \tilde{N}$ and $B_{2}$ of $\tilde{N}\times N$ and they extend the matrices $A_{1}$ and $B_{1}$ respectively. Choosing for each $N\in\mathbb{N}$ an $\tilde{N}$, such that the relative errors are smaller than $\frac{2^{-N}}{(N+1)!}$ guarantees that $\det(A_{1}\bullet B_{1}+A_{2}\bullet B_{2})\rightarrow \det(\exp[ix_{m}y_{n}])$. In the matrix $(A_{1}\bullet B_{1}+A_{2}\bullet B_{2})$, the $(n,m)$-entry is an approximation of $\exp[ix_{m}y_{n}]$ and can thus estimated by $\exp[ix_{m}y_{n}](1+R_{mn})$ with $R_{mn}$ the relative error. The matrix $\big(A_{1}\bullet B_{1}+A_{2}\bullet B_{2}\big)_{mn}$ is thus the same as $\big(\exp[ix_{m}y_{n}](1+R_{mn})\big)_{mn}$. This has $N!$ terms of $N$ factors of two terms, so $2^{N}(N!)$ in total. The absolute error of the approximation is smaller than $2^{N}(N!)\cdot\max_{m,n}|R_{mn}|\rightarrow 0$.\\

By the matrix determinant lemma and Sylvester's determinant theorem this can be written as
\begin{align*}
&\hspace{-8mm}\det_{m,n}\big(e^{ix_{m}y_{n}}\big)\leftarrow\det(A_{1}\bullet B_{1}+A_{2}\bullet B_{2})\\
&=\det(A_{1}\bullet B_{1})\det(I_{N}+A_{1}^{-1}\bullet A_{2}\bullet B_{2}\bullet B_{1}^{-1})\quad.
\end{align*}
The needed inverses of these Vandermonde-matrices are given by
\begin{equation*}
(B_{1}^{-1})_{kl}=\sqrt{(l-1)!}(-1)^{l-1}\cdot\big\{\prod_{\stackrel{t=1}{t\neq k}}^{N}(y_{t}-y_{k})\big\}^{-1}\!\!\!\!\!\sum_{\stackrel{1\leq m_{1}<\cdots<m_{N-l}\leq n}{m_{1},\ldots,m_{N-l}\neq k}}\!\!\!\!\!\!\!\!y_{m_{1}}\ldots y_{m_{N-l}}
\end{equation*}
and
\begin{equation*}
(A_{1}^{-1})_{lk}=\sqrt{(l-1)!}\cdot i^{l-1}\cdot\big\{\prod_{\stackrel{t=1}{t\neq k}}^{N}(x_{t}-x_{k})\big\}^{-1}\!\!\!\!\!\sum_{\stackrel{1\leq m_{1}<\cdots<m_{N-l}\leq n}{m_{1},\ldots,m_{N-l}\neq k}}\!\!\!\!\!\!\!\!x_{m_{1}}\ldots x_{m_{N-l}}\quad.
\end{equation*}
The elements of these matrices follow from (\ref{e:invVdm2}) and are given by
\begin{align*}
&\hspace{-8mm}(A_{1}^{-1}\bullet A_{2})_{lj}=\sum_{k=1}^{N}\frac{i^{N+j-1}(-1)^{l-1}}{\prod_{t\neq k}(x_{t}-x_{k})}\\
&\times\frac{\sqrt{(l-1)!}}{\sqrt{(N+j-1)!}}\sum_{\stackrel{1\leq m_{1}<\cdots<m_{N-l}\leq N}{m_{1},\ldots,m_{N-l}\neq k}}x_{m_{1}}\cdots x_{m_{N-l}}\cdot x_{k}^{N+j-1}
\end{align*}
and
\begin{align}
&\hspace{-8mm}(B_{2}\bullet B_{1}^{-1})_{jl}=\sum_{k=1}^{N}\frac{(-1)^{l-1}}{\prod_{t\neq k}(y_{t}-y_{k})}\nonumber\\
&\times\frac{\sqrt{(l-1)!}}{\sqrt{(N+j-1)!}}\sum_{\stackrel{1\leq m_{1}<\cdots<m_{N-l}\leq N}{m_{1},\ldots,m_{N-l}\neq k}}y_{m_{1}}\cdots y_{m_{N-l}}\cdot y_{k}^{N+j-1}\nonumber\\
&=\frac{\sqrt{(l-1)!}}{\sqrt{(N+j-1)!}}\times\frac{1}{(l-1)!}\,\left.\partial_{z}^{l-1}\right)_{z=0}\,G_{N,N+j-1}(\vec{y};z)\quad,\label{e:VdMinvest}
\end{align}
which follows from direct evaluation of the right-hand side. It follows from Lemma~\ref{l:permuf} that these elements do not diverge as several of the elements, say $y_{1},\ldots,y_{n}$, approach each other. Using the function $G_{N,p}(\vec{y};z)$ from Lemma~\ref{l:permuf2}, it follows that
\begin{align}
&\hspace{-8mm}\big|(B_{2}\bullet B_{1}^{-1})_{jl}\big|=\big|\frac{\sqrt{(l-1)!}}{\sqrt{(N+j-1)!}}\frac{1}{(l-1)!}\,\left.\partial_{z}^{l-1}\right)_{z=0}\,G_{N,N+j-1}(\vec{y};z)\big|\nonumber\\
&=\big|(-1)^{l-1}\frac{\sqrt{(l-1)!}}{\sqrt{(N+j-1)!}}\nonumber\\
&\times\sum_{m=1}^{N}\sum_{\{h_{1},\ldots,h_{m}\}\subset\{1,\ldots,N\}}\sum_{k=0}^{m-1}\delta_{m-1-k,l-1}\,F_{m,N+j-1+k}(y_{h_{1}},\ldots,y_{h_{m}})\big|\nonumber\\
&=\big|(-1)^{l-1}\frac{\sqrt{(l\!-\!1)!}}{\sqrt{(N\!+\!j\!-\!1)!}}\sum_{m=l}^{N}\sum_{\stackrel{\{h_{1},\ldots,h_{m}\}}{\subset\{1,\ldots,N\}}}\!\!\!F_{m,N+j-1+m-l}(y_{h_{1}},\ldots,y_{h_{m}})\big|\nonumber\\
&\leq \frac{\sqrt{(l-1)!}}{\sqrt{(N+j-1)!}}\sum_{m=l}^{N}\binom{N}{m}\binom{N+j-l+m}{m}(\sqrt{\frac{\gamma}{N^{2}}})^{N+j-l}\quad,\label{e:BBest}
\end{align}
where it is used that the number of terms of $F_{n,n-1+m}$ is given by $\binom{n+m}{m}$. Estimating
\begin{align}
&\hspace{-8mm}\sum_{m=l}^{N}\binom{N}{m}\binom{N+j-l+m}{m}\leq\sum_{m=1}^{N}\binom{N}{m}\binom{N+j-l+m}{m}\nonumber\\
&=-1+{}_{2}F_{1}(-N,j-l+N;1;-1)\qquad 1\leq l\leq N\quad,\label{e:hgfest}
\end{align}
where the hypergeometric function 
\begin{equation*}
{}_{2}F_{1}(a,b;c;z)=\sum_{s=0}^{\infty}\frac{z^{s}}{s!}\frac{(a)_{s}(b)_{s}}{(c)_{s}}\quad
\end{equation*}
is given in terms of the Pochhammer symbol $(a)_{n}=\prod_{t=0}^{n-1}(a+t)$. Now, the hypergeometric function can be expressed~\cite{temme1} as a Jacobi polynomial $P_{n}^{(a,b)}$ via
\begin{align*}
&\hspace{-8mm}{}_{2}F_{1}(-n,n+a+b+1;a+1;z)=\frac{n!}{(a+1)_{n}}P_{n}^{(a,b)}(1-2z)\\
&=\sum_{s=0}^{n}\binom{n+a}{s}\binom{n+b}{n-s}\big(\frac{z+1}{2}\big)^{s}\big(\frac{z-1}{2}\big)^{n-s}\quad,
\end{align*}
when $n,n+a,n+b,n+b+a$ are nonnegative integers. This yields the upper bound
\begin{align}
&\hspace{-8mm}{}_{2}F_{1}(-N,j+N;1;-1)=P_{N}^{(0,j-1)}(3)\leq 2^{N}\sum_{s=0}^{N}\binom{N}{s}\binom{N+j-1}{N-s}\nonumber\\
&=2^{N}\binom{2N+j-1}{N}\leq2^{N}\frac{2^{2N+j-1}}{\sqrt{\pi N}}\quad,\label{e:hgfest2}
\end{align}
from which it follows that (\ref{e:BBest}) is maximal for $l=N$ for large $N$.\\

For the other half of the correction matrix, the same steps give
\begin{align}
&\hspace{-8mm}\big|(A_{1}^{-1}\!\bullet\! A_{2})_{kj}\big|\leq \sqrt{\frac{(k-1)!}{(N\!+\!j\!-\!1)!}}\sum_{n=k}^{N}\!\binom{N}{n}\binom{N\!+\!j\!-\!k\!+\!n}{n}(\frac{\gamma}{N^{2}})^{\frac{N+j-k}{2}}\,.\label{e:mas}
\end{align}
For the combination of the two, it follows that
\begin{align*}
&\hspace{-8mm}\big|\big(A_{1}^{-1}\bullet A_{2}\bullet B_{2}\bullet B_{1}^{-1}\big)_{kl}\big|\leq \sum_{j=1}^{\infty}\frac{(N-1)!}{(N+j-1)!}\binom{N+j}{N}\binom{N+j}{N}(\frac{\gamma}{N^{2}})^{j}\\
&\leq \sum_{j=1}^{\infty}(1+\frac{j}{N})\frac{1}{(j!)^{2}}\big(\frac{\gamma}{N}\big)^{j}\cdot\big[\prod_{t=1}^{j}(1+\frac{t}{N})\big]=\smallo(\frac{1}{N})\quad.
\end{align*}
The diagonal is dominant in $I_{N}+A_{1}^{-1}\bullet A_{2}\bullet B_{2}\bullet B_{1}^{-1}$. The product of the diagonal elements yields a factor $(1+\smallo(N^{-1}))^{N}\rightarrow 1$. The off-diagonal corrections $\sum_{m=2}^{N}\binom{N}{m}\smallo(N^{-m})$ also vanish, so that
\begin{equation*}
\det(I_{N}+A_{1}^{-1}\bullet A_{2}\bullet B_{2}\bullet B_{1}^{-1})\rightarrow \det(I_{N})\quad.
\end{equation*}
This completes the proof.
\end{proof}

\begin{lemma}\{\emph{Cauchy-Binet formula}\}\label{l:CB}\\
Let $A$ be a $m\times n$-matrix, $B$ a $n\times m$-matrix and $S$ the set of $m$-subsets of $\{1,\ldots, n\}$
\begin{equation*}
\mathcal{S}=\{(i_{1},\ldots,i_{m})|1\leq i_{j}\leq n\;\&\; i_{j}<i_{j+1}\;\&\; 1\leq j\leq m\}\quad.
\end{equation*}
The determinant of $A\bullet B$ is given by
\begin{equation*}
\det(A\bullet B)=\sum_{I\in \mathcal{S}}|A|_{\textbf{1}I}|B|_{I\textbf{1}}\quad,
\end{equation*}
where $\textbf{1}=(1,\ldots,m)$ is the trivial index set and $|A|_{\textbf{1}I}=\det(\mathcal{M}_{I})$ is the determinant of the $m\times m$-matrix $\mathcal{M}_{I}$ with entries $(\mathcal{M}_{I})_{kl}=A_{ki_{l}}$.
\end{lemma}

\begin{rmk}\label{rmk:ap}
A simpler proof of an even stronger result is possible. This makes use of the Cauchy-Binet formula~\cite{shafarevich}, see Lemma~\ref{l:CB}. The entry can be written as a power series
\begin{equation*}
f(x,y)=\exp[xy]=\sum_{k=0}^{\infty}\frac{(xy)^{k}}{k!}\quad.
\end{equation*}
For the determinants this implies
\begin{align*}
&\hspace{-8mm}\det_{k,l}\exp[x_{k}y_{l}]=\sum_{J\in\mathcal{S}_{\Lambda, N}}\big|x_{k}{j-1}\big|_{\textbf{1}J}\cdot\big|\frac{1}{(j-1)!}\big|_{JJ}\cdot\big|y_{l}^{j-1}\big|_{J\textbf{1}}\\
&=\sum_{J\in\mathcal{S}_{\Lambda, N}}\big|X\big|_{\textbf{1}J}\cdot\big|M\big|_{JJ}\cdot\big|Y\big|_{J\textbf{1}}\quad.
\end{align*}
We only need to sum over one index set, because the matrix of derivatives 
\begin{equation*}
M_{jm}=\frac{1}{(j!)(m!)}f^{(j,m)}(0,0)=\frac{\delta_{j,m}}{j!}
\end{equation*}
is diagonal. To an index set $J$ the size $|J|=\sum_{j\in J}j\geq \binom{N}{2}$ is assigned. For every $J$ the determinant $|Y|_{J\textbf{1}}$ is a multiple of the Vandermonde-determinant $\Delta(y_{1},\ldots,y_{N})$. It is not difficult to estimate
\begin{equation*}
\Big| \big|X\big|_{\textbf{1}J}\cdot\big|Y\big|_{J\textbf{1}} \Big|\leq \Big|\Delta(x_{1},\ldots,x_{N})\cdot\Delta(y_{1},\ldots,y_{N})\Big| \cdot (N^{2}\max_{k}|x_{k}|\cdot\max_{l}|y_{l}|)^{|J|-\binom{N}{2}}
\end{equation*}
and
\begin{equation*}
\big|M\big|_{JJ}\leq \big|M\big|_{\textbf{1}\textbf{1}}\big(\frac{e}{N}\big)^{|J|-\binom{N}{2}}
\end{equation*}

The number of elements in $\mathcal{S}_{\infty,N}$ with size $\Lambda$ is the number $p_{N}(\Lambda)$. It satisfies the recursion $p_{m}(n)=p_{m-1}(n-m+1)+p_{m}(n-m)$. Too see this, note that every configuration of different nonnegative integers $0\leq i_{1}<\ldots<i_{m}$ that sums to $n$ adds one to $p_{m}(n)$. Now, either $i_{1}=0$ or $i_{1}\geq1$. In the first case, this configuration corresponds to one contributing to $p_{m-1}(n-m+1)$. In the other case, one can be subtraced from every integer, so that it corresponds to $p_{m}(n-m)$. It is not difficult to see that $p_{1}(n)=1$ and $p_{2}(n)=\lfloor\frac{n+1}{2}\rfloor$. The first few of these coefficients are given in Table~\ref{t:scex3}. Although it is (probably) possible~\cite{christopher, tani} to compute the asymptotics of these number, a simple estimate will suffice here. The estimate
\begin{equation*}
p_{m}(n)\leq \alpha_{m}2^{n-\binom{m}{2}}\quad\text{with}\quad \alpha_{m}=\prod_{t=1}^{m}\frac{1}{1-2^{-t}}
\end{equation*}
holds. Filling this in in the recursion relation yields
\begin{align*}
&\hspace{-8mm}p_{m}(n)=p_{m\!-\!1}(n\!-\!m\!+\!1)+p_{m}(n\!-\!m)\leq \alpha_{m\!-\!1}2^{n\!-\!m\!+\!1\!-\!\binom{m\!-\!1}{2}}\!+\!\alpha_{m}2^{n\!-\!m\!-\!\binom{m}{2}}\\
&=(\alpha_{m-1}+2^{-m}\alpha_{m})2^{n-\binom{m}{2}}=\alpha_{m}2^{n-\binom{m}{2}}\quad.
\end{align*}
It is not difficult to see that this holds for $m=1$ and $n\leq \binom{m}{2}$, which proves the estimate. For simplicity it may be uses that $\alpha_{m}<\alpha_{\infty}<3.47$.
\begin{table}[!h]
\hspace*{3cm}\footnotesize{\begin{tabular}{c||cccc}
$n\backslash m$ & $1$ & $2$ & $3$ & $4$\\\hline
$1$ & $1$ & $1$ & $0$ & $0$\\
$2$ & $1$ & $1$ & $0$ & $0$\\
$3$ & $1$ & $2$ & $1$ & $0$\\
$4$ & $1$ & $2$ & $1$ & $0$\\
$5$ & $1$ & $3$ & $2$ & $0$\\
$6$ & $1$ & $3$ & $3$ & $1$\\
$7$ & $1$ & $4$ & $4$ & $1$\\
$8$ & $1$ & $4$ & $5$ & $2$\\
$9$ & $1$ & $5$ & $7$ & $3$\\
$10$ & $1$ & $5$ & $8$ & $5$
\end{tabular}}
\caption{The number of partitions $p_{m}(n)$ of $n$ in $m$ distinct nonnegative integers.\label{t:scex3}}
\end{table}
Putting things together shows that
\begin{align*}
&\hspace{-8mm}\det_{k,l}e^{x_{k}y_{l}}=|X|_{\textbf{1}\textbf{1}}\cdot|M|_{\textbf{1}\textbf{1}}\cdot|Y|_{\textbf{1}\textbf{1}}\times\big(1\!+\!\mathcal{O}(\frac{2eN\max_{k}\!|x_{k}|\cdot\max_{l}\!|y_{l}|}{1\!-\!2eN\max_{k}\!|x_{k}|\cdot\max_{l}\!|y_{l}|})\big)\quad,
\end{align*}
so that 
\begin{equation}
N\max_{k}|x_{k}|\cdot\max_{l}|y_{l}|\rightarrow0\label{e:nmatdemcond}
\end{equation}
is a sufficient condition for convergence.
\end{rmk}

\begin{exm}
Lemma~\ref{l:detcom} is best tested and demonstrated by an example. For $n=3,4,5,6,7$ define $x_{k}=y_{k}=k\cdot n^{-1.75}$ and construct two $n\times n$-matrices:
\begin{equation}
Q_{kl}^{(n)}=\exp[x_{k}\cdot y_{l}]\qquad\text{and}\qquad R_{kl}^{(n)}=\sum_{m=0}^{n-1}\frac{(x_{k}y_{l})^{m}}{m!}\quad.\label{e:xQR}
\end{equation}
The determinants of these matrices are given in Table~\ref{t:matdecex}.
\begin{table}[!h]
\hspace*{3cm}\footnotesize{\begin{tabular}{l||l|l|l}
$n$& $\det Q^{(n)}$&$\det R^{(n)}$&ratio\\\hline
$3$& $2.53E-05$&$1.96E-05$&$1.30$\\
$4$& $3.32E-12$&$2.73E-12$&$1.22$\\
$5$& $1.16E-22$&$9.90E-23$&$1.18$\\
$6$& $5.57E-37$&$4.85E-37$&$1.15$\\
$7$& $2.17E-55$&$2.11E-55$&$1.03$
\end{tabular}}
\caption{The determinants of the $n\times n$-matrices $Q^{(n)}$ and $R^{(n)}$ from (\ref{e:xQR}). The notation $1.37E-3=1.37\times 10^{-3}$ is used here.\label{t:matdecex}}
\end{table}
\end{exm}

\begin{exm}
Another example is discussed. Consider the $n\times n$-matrices
\begin{align*}
&Q_{kl}(\varepsilon)=\exp[\varepsilon ky_{l}]\quad;\\
&R_{kl}(\varepsilon)=\exp[\varepsilon ky_{l}+(\varepsilon ky_{l})^{2}]\quad;\\
&S_{kl}(\varepsilon)=\exp[\varepsilon (k+\sqrt{\varepsilon k^{2}})(y_{l}+\sqrt{\varepsilon}y_{l}^{2})]\quad
\end{align*}
and rescale their determinants by $\varepsilon^{-\binom{n}{2}}$. The rescaled determinant of $Q(0)$ is given by
\begin{equation*}
q = \Delta(y_{1},\ldots y_{n})=\prod_{1\leq k,l\leq n}(y_{l}-y_{k})\quad.
\end{equation*}
Figure~\ref{f:detexm2} shows that $\varepsilon^{-\binom{n}{2}}\det(S(\varepsilon))\rightarrow q$, but the same is not true for $R$.
\begin{figure}[!htb]
\hspace{1.5cm}\includegraphics[width=0.6\textwidth]{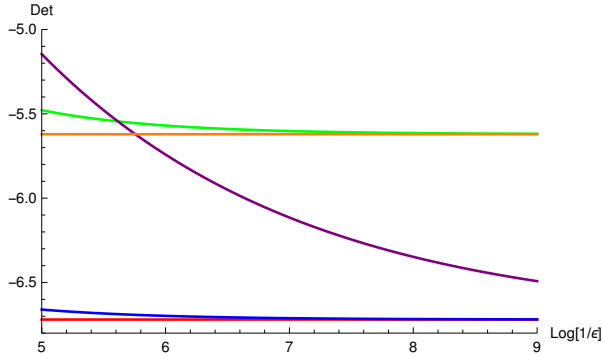}
\caption{The logarithm of the rescaled determinant of $Q(\varepsilon)$ (blue), $R(\varepsilon)$ (green) and $S(\varepsilon)$ (purple) and the logarithm of the asymptotic value $q$ of $\varepsilon^{-\binom{n}{2}}\det(Q(\varepsilon))$ for $n=3$, where the entries are given by $y_{m}=1+\sqrt{m}/8+m^{1/4}/4$. The logarithm of the asymptotic rescaled value of $S$ is given in orange.\label{f:detexm2}}
\end{figure}
The difference between $R$ and $S$ can be explained. The entries of $S$ satisfy the conditions of Lemma~\ref{l:detcom}, whereas those of $R$ satisfy the composition
\begin{equation*}
\exp[(\varepsilon xy)+ (\varepsilon xy)^{2}]=\sum_{k=0}^{\infty}\frac{(\varepsilon xy)^{k}}{k!}\sum_{j=\lceil k/2\rceil}^{k}\frac{k!}{(k-j)!\cdot (2j-k)!}\quad.
\end{equation*}
This shows that the rescaled determinant of $S(0)$ is given by
\begin{equation*}
q = \big[\prod_{k=0}^{n-1}\sum_{j=\lceil k/2\rceil}^{k}\frac{k!}{(k-j)!\cdot (2j-k)!}\big]\cdot\big\{\prod_{1\leq k,l\leq n}(y_{l}-y_{k})\big\}\quad.
\end{equation*}
\end{exm}

Generalising the function $f(x,y)=\exp[xy]$ would turn the Lemma~\ref{l:detcom} into the following:
\begin{lemma}\label{l:detcom2}
Let $f$ be an analytic function on an open set containing the origin and let $x_{j},y_{j}\in\mathbb{C}$. If the derivatives $|f^{k}(0)|>0$ for all $k\in\mathbb{N}_{0}$ are nonzero, then there exists an $\varepsilon>0$, such that the determinant
\begin{align*}
&\hspace{-8mm}\det_{1\leq m,n\leq N}\big(f(x_{m}y_{n})\big)=\det_{1\leq m,k\leq N}\big(x_{m}^{k-1}\big)\det_{1\leq k,l\leq N}\big(\frac{f^{(k-1)}(0)}{(k-1)!}\delta_{k,l}\big)\\
&\det_{1\leq l,n\leq N}\big(y_{n}^{l-1}\big)\times\big(1+\mathcal{O}(N^{-\frac{1}{2}})\big)\quad,
\end{align*}
if $|x_{k}y_{l}|<\varepsilon$ for all $1\leq k,l\leq N$. 
\end{lemma}

Many more factors than one would naively expect, contribute to the determinant. To demonstrate this, we compute the asymptotic determinant of the matrix with entries $\exp[x_{k}y_{l}+\alpha (x_{k}y_{l})^{n}]$. For a very large matrix size $N$ we may ignore the second term if $\frac{\alpha^{N/n}}{(N/n)!}\ll\frac{1}{N!}$, which translates to $\alpha \ll (e/N)^{n-1}/n$.

\section{Construction of the matrix of derivatives\label{sec:cmd}}

Not all determinants of interest are covered by Lemmas~\ref{l:detcom} and~\ref{l:detcom2}. For bivariate functions the matrix of derivatives $M$ will not be diagonal. This also means that it is not as simple to see how it should be constructed or related to various parts of the function. For functions $f(x,y)=\exp[\mathpzc{P}(x,y)]$, where $\mathpzc{P}$ is a degree $n$ bivariate polynomial, this construction will be demonstrated explicitly.\\
The starting point is the promotion of the Taylor series to a matrix equation. This is not as straightforward as in the case of diagonal $M$. It will be helpful to introduce double indices $\underline{q}=\binom{q^{(1)}}{q^{(2)}}$, which will be treated as vectors of length $2$. They will be the indices of the matrices in the decomposition. Each component runs from $0$ tot $N-1$, so that the matrices will be of size $N^{2}\times N^{2}$.

\begin{exm}{\label{exm:ee}}
An explicit example in terms of double indices is the expansion of the function
\begin{align*}
&\hspace{-8mm}\mathpzc{f}(x,y)=e^{uxy+vxy^{2}}\!\!=\!\!\!\!\sum_{k_{1},k_{2}=0}\!\!x^{k_{1}+k_{2}}\,\Big(\frac{u^{k_{1}}}{k_{1}!}\Big)_{\binom{k_{1}+k_{2}}{0}\binom{k_{2}}{k_{1}}}\cdot\Big(\frac{v^{k_{2}}}{k_{2}!}\Big)_{\binom{k_{2}}{k_{1}}\binom{0}{k_{1}+2k_{2}}}\,y^{k_{1}+2k_{2}}\;,
\end{align*}
which is structurally already matrix multiplication.
\end{exm}

If the polynomial 
\begin{equation}
\mathpzc{P}(x,y)=\sum_{i=1}^{n}u_{i}x^{m_{i}^{(1)}}y^{m_{i}^{(2)}}\quad\label{e:mpol}
\end{equation}
with $m^{(1,2)}\in\mathbb{N}$ is expanded through summation variables $k_{1},\ldots,k_{n}\in\mathbb{N}_{0}$, then index components for $j=1,\ldots,n+1$ are defined through
\begin{equation}
q_{j}^{(1)}=\sum_{i=j}^{n}m_{i}^{(1)}k_{i}\quad\text{and}\quad q_{j}^{(2)}=\sum_{i=1}^{j-1}m_{i}^{(2)}k_{i}\quad.\label{e:doubleindex}
\end{equation}
A consequence of this definition is that $q_{n+1}^{(1)}=q_{1}^{(2)}=0$. It follows then that
\begin{align}
&\hspace{-8mm}\exp\big[\sum_{j=1}^{n}u_{j}x_{\alpha}^{m^{(1)}_{j}}y_{\beta}^{m_{j}^{(2)}}\big]\nonumber\\
&=\sum_{\stackrel{k_{j}=0}{1\leq j\leq n}}x_{\alpha}^{\sum_{j=1}^{n}k_{j}m_{j}^{(1)}}\big(\prod_{i=1}^{n}\frac{u_{i}^{k_{i}}}{k_{i}!}\big)y_{\beta}^{\sum_{j=1}^{n}k_{j}m_{j}^{(2)}}\nonumber\\
&=\Big(x_{\alpha}^{q^{(1)}_{1}}\Big)_{\alpha \underline{q}_{1}}\bullet \Big(\sum_{t_{1}=0}^{N-1}\frac{u_{1}^{t_{1}}}{t_{1}!}\cdot\delta_{t_{1}m_{1}^{(1)},q_{1}^{(1)}-q_{2}^{(1)}}\cdot\delta_{t_{1}m_{1}^{(2)},q_{2}^{(2)}-q_{1}^{(2)}}\Big)_{\underline{q}_{1}\underline{q}_{2}}\nonumber\\
&\bullet \Big(\sum_{t_{2}=0}^{N-1}\frac{u_{2}^{t_{2}}}{t_{2}!}\cdot\delta_{t_{2}m_{2}^{(1)},q_{2}^{(1)}-q_{3}^{(1)}}\cdot\delta_{t_{2}m_{2}^{(2)},q_{3}^{(2)}-q_{2}^{(2)}}\Big)_{\underline{q}_{2}\underline{q}_{3}}\bullet\cdots\nonumber\\
&\bullet \!\Big(\!\sum_{t_{n}=0}^{N-1}\!\frac{u_{n}^{t_{n}}}{t_{n}!}\!\cdot\!\delta_{t_{n}m_{n}^{(1)},q_{n}^{(1)}\!-q_{n+1}^{(1)}}\!\cdot\!\delta_{t_{n}m_{n}^{(2)},q_{n+1}^{(2)}\!-q_{n}^{(2)}}\!\Big)_{\underline{q}_{n}\underline{q}_{n+1}}\!\!\!\!\!\!\bullet \Big(\!y_{\beta}^{q_{n+1}^{(2)}}\!\Big)_{\underline{q}_{n+1}\beta}\,.\label{e:mdc}
\end{align}
A single matrix entry is already written here as a matrix product. The only thing left to do, is to promote the indices $\alpha,\beta$ to double indices $\underline{\alpha},\underline{\beta}$ and extend the matrices in such a way that 
\begin{equation*}
\det_{\alpha,\beta}\exp\big[\sum_{j=1}^{n}u_{j}x_{\alpha}^{m^{(1)}_{j}}y_{\beta}^{m_{j}^{(2)}}\big] = \det_{\underline{\alpha},\underline{\beta}}\exp\big[\sum_{j=1}^{n}u_{j}x_{\underline{\alpha}}^{m^{(1)}_{j}}y_{\underline{\beta}}^{m_{j}^{(2)}}\big]
\end{equation*}
and 
\begin{equation*}
\det_{\alpha,q_{1}^{(1)}} x_{\alpha}^{q_{1}^{(1)}}=\det_{\underline{\alpha},\underline{q}_{1}}x_{\underline{\alpha}}^{\underline{q}_{1}}\qquad\&\qquad\det_{q_{n+1}^{(1)},\beta} y_{\beta}^{q_{n+1}^{(2)}}=\det_{\underline{q}_{n+1},\underline{\beta}}y_{\underline{\beta}}^{\underline{q}_{n+1}}\quad.
\end{equation*}

Let's start with the generalised Vandermonde-matrix for $x$. For $N\in\mathbb{N}$ the indices $\underline{\alpha},\underline{q}_{1}$ are elements of $\{(k,l)|0\leq k,l\leq N-1\}$. The indices are ordered through
\begin{equation}
\underline{p}\leq \underline{q} \Leftrightarrow p^{(1)}+N\cdot p^{(2)} \leq q^{(1)}+N\cdot q^{(2)}\quad.\label{e:indexorder}
\end{equation}
This implies the following order for double indices
\begin{equation*}
\binom{0}{0}\binom{1}{0}\ldots\binom{N-1}{0}\binom{0}{1}\ldots\binom{N-1}{1}\binom{0}{2}\ldots\binom{N-1}{N-1}\quad.
\end{equation*}
Defining the matrix elements by
\begin{equation*}
X=\big(x_{\underline{\alpha}}^{\underline{q}}\big)_{\underline{\alpha}\underline{q}}=\left\{\begin{array}{ll}
x_{\alpha^{(1)}}^{q^{(1)}} & \text{, if }q^{(2)}=\alpha^{(2)}=0\\
1 & \text{, if }\underline{q}=\underline{\alpha} \;\&\; \alpha^{(2)}>0\\
0 & \text{, if otherwise}
\end{array}\right.\quad,
\end{equation*}
this is the Vandermonde-matrix in the upper-left corner and extended trivially on the diagonal, so that its determinant is given by the Vandermonde-determinant $\Delta(x_{0},\ldots,x_{N-1})$. For the generalised Vandermonde-matrix of the $y$'s, the choice
\begin{equation*}
Y=\big(y_{\underline{\beta}}^{\underline{q}}\big)_{\underline{q}\underline{\beta}}=\left\{\begin{array}{ll}
y_{\beta^{(1)}}^{q^{(2)}} & \text{, if }q^{(1)}=\beta^{(2)}=0\\
1 & \text{, if }\underline{q}=\underline{\beta} \quad\&\quad \beta^{(1)},\beta^{(2)}>0\\
1 & \text{, if } \beta^{(1)}=q^{(2)}=0 \quad\&\quad \beta^{(2)}=q^{(1)}>0\\
0 & \text{, if otherwise}
\end{array}\right.\quad,
\end{equation*}
has determinant $(-1)^{N-1}\Delta(y_{0},\ldots,y_{N-1})$. For $N=3$ this results in
\begin{equation*}
\big(y_{\underline{\beta}}^{\underline{q}}\big)=\tiny\left(
\begin{array}{ccccccccc}
1&1&1&0&0&0&0&0&0\\                                                       
0&0&0&1&0&0&0&0&0\\
0&0&0&0&0&0&1&0&0\\
y_{0}&y_{1}&y_{2}&0&0&0&0&0&0\\
0&0&0&0&1&0&0&0&0\\
0&0&0&0&0&1&0&0&0\\
y_{0}^{2}&y_{1}^{2}&y_{2}^{2}&0&0&0&0&0&0\\
0&0&0&0&0&0&0&1&0\\
0&0&0&0&0&0&0&0&1
\end{array}
\right)\normalsize\quad.
\end{equation*}

The function specific matrices in (\ref{e:mdc}) are given by
\begin{align}
&\hspace{-8mm}U_{l}=\!\Big(\!\sum_{t_{l}=0}^{N-1}\!\frac{u_{l}^{t_{l}}}{t_{l}!}\!\cdot\!\delta_{t_{l}m_{l}^{(1)},q_{l}^{(1)}\!-q_{l+1}^{(1)}}\!\cdot\!\delta_{t_{l}m_{l}^{(2)},q_{l+1}^{(2)}\!-q_{l}^{(2)}}\Big)_{\underline{q}_{l}\underline{q}_{l+1}}\quad\text{, for }1\!\leq\!l\!\leq\!n\;,\label{e:Ulmat}
\end{align}
where the double indices $\underline{q}_{l}$ are defined by (\ref{e:doubleindex}). Because $m_{l}^{(2)}>0$, it follows from the ordering (\ref{e:indexorder}) that these matrices are all upper triangular and have only unit entries on the diagonal. This implies that all nontrivial entries below the diagonal in the resulting matrix are the result of the lower triangular entries $y_{m}^{n\geq1}$ in the final Vandermonde-matrix $Y$. Besides that, it is obvious that
\begin{equation*}
\det( X\bullet \prod_{j}U_{j}\bullet Y)=\Delta(x_{0},\ldots,x_{N-1})\Delta(y_{0},\ldots,y_{N-1})\quad.
\end{equation*}
\begin{exm}\label{exm:Umatexpl}
Following (\ref{e:Ulmat}) for the function $\exp[uxy]$ with $N=3$ the matrix
\begin{equation*}
U=\tiny\left(
\begin{array}{ccccccccc}
1&0&0&0&0&0&0&0&0\\
0&1&0&u&0&0&0&0&0\\
0&0&1&0&u&0&\frac{u^{2}}{2}&0&0\\
0&0&0&1&0&0&0&0&0\\
0&0&0&0&1&0&u&0&0\\
0&0&0&0&0&1&0&u&0\\
0&0&0&0&0&0&1&0&0\\
0&0&0&0&0&0&0&1&0\\
0&0&0&0&0&0&0&0&1
\end{array}
\right)\normalsize\quad.
\end{equation*}
\end{exm}

Although this is slightly different than anticipated, not all hope is lost. The function entries are contained in this matrix. For an insight in the matrix $\mathcal{F}=X\bullet \prod_{l}U_{l}\bullet Y$ the product of $X$ and the upper triangular matrix 
\begin{equation*}
V=U_{1}\bullet \cdots\bullet  U_{n}=\left(
\begin{array}{ccccc}
v_{\binom{0}{0}\binom{0}{0}}&v_{\binom{0}{0}\binom{1}{0}}&v_{\binom{0}{0}\binom{2}{0}}&\ldots&v_{\binom{0}{0}\binom{N-1}{N-1}}\\
0&v_{\binom{1}{0}\binom{1}{0}}&v_{\binom{1}{0}\binom{2}{0}}&\ldots&v_{\binom{1}{0}\binom{N-1}{N-1}}\\
0&0&v_{\binom{2}{0}\binom{2}{0}}&&v_{\binom{2}{0}\binom{N-1}{N-1}}\\
\vdots&\vdots&&\ddots&\vdots\\
0&0&0&\ldots&v_{\binom{N-1}{N-1}\binom{N-1}{N-1}}
\end{array}
\right)
\end{equation*}
is examined first, where all diagonal $v_{\underline{q}\underline{q}}=1$. Schematically, the product matrix is given by
\begin{equation*}
X\bullet U_{1}\bullet \cdots\bullet  U_{n}=\tiny\left(
\begin{array}{ccccccc}
p_{\binom{0}{0}}(x_{0})&&&\ldots&&&p_{\binom{N-1}{N-1}}(x_{0})\\
\vdots&&&\ldots&&&\vdots\\
p_{\binom{0}{0}}(x_{N-1})&&&\ldots&&&p_{\binom{N-1}{N-1}}(x_{N-1})\\
0&\ldots&0&1&v_{\binom{0}{1}\binom{1}{1}}&\ldots&v_{\binom{0}{1}\binom{N-1}{N-1}}\\
0&\ldots&&0&1&\ldots&v_{\binom{1}{1}\binom{N-1}{N-1}}\\
\vdots&&&&0&\ddots&\vdots\\
0&\ldots&&&&0&1
\end{array}
\right)\normalsize\;,
\end{equation*}
$p_{\underline{q}}$ is a polynomial of degree $N-1$ in one variable. Multiplying from the left by elementary matrices $\mathcal{E}$, i.e. subtracting row $i$ from $j$, this matrix can be reduced to
\begin{align*}
&\hspace{-8mm}\mathcal{E}\bullet X\bullet U_{1}\bullet \cdots\bullet  U_{n}\\[2mm]
&\!\!\!\!\!\!\!\!\!\!=\!\tiny\left(
\begin{array}{cccccccccccc}
\!\!\!\!\!p_{\binom{0}{0}}(x_{0})\!\!\!&\!\!p_{\binom{1}{0}}(x_{0})\!\!\!&\!\!\!\ldots\!\!\!&\!\!\!p_{\binom{N\!-\!1}{0}}(x_{0})\!\!&\!\!p_{\binom{0}{1}}(x_{0})\!\!&\!\!0\!&\!\ldots\!&\!0\!\!&\!\!\!p_{\binom{0}{N\!-\!1}}(x_{0})\!\!\!\!&\!\!\!0\!&\!\ldots\!&\!0\!\!\!\\
\!\vdots\!\!&\!\vdots\!&&\!\vdots\!&\!\vdots\!&\!\vdots\!&&\!\vdots\!&\!\vdots\!&\!\vdots\!&&\!\vdots\!\!\!\\
\!\!\!\!\!p_{\binom{0}{0}}(x_{N\!-\!1})\!\!\!&\!\!p_{\binom{1}{0}}(x_{N\!-\!1})\!\!\!&\!\!\!\ldots\!\!\!&\!\!\!p_{\binom{N\!-\!1}{0}}(x_{N\!-\!1})\!\!&\!\!\!p_{\binom{0}{1}}(x_{N\!-\!1})\!\!\!&\!0\!&\!\ldots\!&\!0\!\!\!&\!\!\!p_{\binom{0}{N\!-\!1}}(x_{N\!-\!1})\!\!\!\!&\!\!\!0\!&\!\ldots\!&\!0\!\!\!\\
\!0\!\!&\!0\!\!&\!\!\ldots\!\!&\!\!0\!&\!1\!&\!0\!&\!0\!&&\!\ldots\!&&&\!0\!\!\!\\
\!0\!\!&\!0\!\!&\!\!\ldots\!\!&\!\!0\!&\!0\!&\!1\!&\!0\!&&\!\ldots\!&&&\!0\!\!\!\\
\!\vdots\!\!&&&\!\vdots\!&&&\!\ddots\!&&&&&\!\vdots\!\!\!\\
\!&&&&&&&&&&&\!\!\\
\!0\!\!&\!0\!\!&\!\!\ldots\!\!&\!\!0\!&\!0\!&&&&\!1\!\!&\!\!0\!&\!\ldots\!&\!0\!\!\!\\
\!&&&&&&&&&&&\!\!\\
\!\vdots\!\!&&&\!\vdots\!&&&&&\!&\!\!\ddots\!&&\!\vdots\!\!\!\\
\!0\!\!&\!0\!\!&\!\!\ldots\!\!&\!\!0\!&&&&&\!&\!\!0\!&\!1\!&\!0\!\!\!\\
\!0\!\!&\!0\!\!&\!\!\ldots\!\!&\!\!0\!&&&&&&&\!0\!&\!1\!\!\!
\end{array}
\right)\normalsize\,.
\end{align*}
From either the structure of the matrices $U_{l}$ or the requirement that its determinant is $\prod_{k<l}(x_{l}-x_{k})$ it can be seen that $p_{\binom{m}{0}}$ is a monic polynomial of degree $m$.\\

The next step is to invert the matrix $\big(p_{\binom{0}{m}}(x_{l})\big)_{lm}$ with $1\leq l,m\leq N-1$. This yields
\begin{equation*}
\big(\mathcal{R}\big)_{kl}\bullet \big(p_{\binom{0}{m}}(x_{l})\big)_{lm}=\big(I_{N-1}\big)_{km}\quad.
\end{equation*}
Now, add to line $\binom{0}{k}$ of the above matrix $-\mathcal{R}_{kl}$ times line $\binom{l}{0}$, for all combinations $1\leq k,l\leq N-1$. This can be implemented through extra elementary matrices, so that
\begin{align*}
&\hspace{-8mm}\mathcal{E}'\bullet X\bullet U_{1}\bullet \cdots\bullet  U_{n}\\[2mm]
&\!\!\!\!\!\!\!\!=\!\tiny\left(
\begin{array}{cccccccccccc}
\!\!\!\!\!p_{\binom{0}{0}}(x_{0})\!\!\!&\!\!\!p_{\binom{1}{0}}(x_{0})\!\!\!&\!\!\!\ldots\!\!\!&\!\!\!p_{\binom{N\!-\!1}{0}}(x_{0})\!\!\!&\!\!\!p_{\binom{0}{1}}(x_{0})\!\!\!&\!\!\!0\!\!&\!\!\ldots\!\!&\!\!0\!\!\!&\!\!\!p_{\binom{0}{N\!-\!1}}(x_{0})\!\!\!&\!\!\!0\!&\!\ldots\!&\!0\!\!\!\\
\!\vdots\!&\!\vdots\!&&\!\vdots\!&\!\vdots\!&\!\vdots\!&&\!\vdots\!&\!\vdots\!&\!\vdots\!&&\!\vdots\!\!\!\\
\!\!\!\!p_{\binom{0}{0}}(x_{N\!-\!1})\!\!\!&\!\!\!p_{\binom{1}{0}}(x_{N\!-\!1})\!\!\!&\!\!\!\ldots\!\!\!&\!\!\!p_{\binom{N\!-\!1}{0}}(x_{N\!-\!1})\!\!\!&\!\!\!p_{\binom{0}{1}}(x_{N\!-\!1})\!\!\!&\!\!\!0\!\!\!&\!\!\!\ldots\!\!\!&\!\!\!0\!\!\!&\!\!\!p_{\binom{0}{N\!-\!1}}(x_{N\!-\!1})\!\!\!&\!\!\!0\!\!\!&\!\!\!\ldots\!\!\!&\!\!\!0\!\!\!\\
\!\!\!0\!\!\!&\!\!\!r_{\binom{0}{1}\binom{1}{0}}\!\!\!&\!\!\!\ldots\!\!\!&\!\!\!r_{\binom{0}{1}\binom{N-1}{0}}\!\!\!&\!\!\!0\!\!\!&\!\!\!0\!\!\!&\!\!\!0\!\!\!&\!&\!\!\!\ldots\!\!&&&\!\!0\!\!\!\\
\!\!\!0\!\!\!&\!\!\!0\!\!\!&\!\!\!\ldots\!\!\!&\!\!\!0\!\!\!&\!\!\!0\!\!\!&\!\!\!1\!\!\!&\!\!\!0\!\!\!&\!&\!\!\!\ldots\!\!\!&\!&\!&\!\!\!0\!\!\!\\
\!\!\!\vdots\!\!\!&\!&\!&\!\!\!\vdots\!\!\!&\!&\!&\!\!\!\ddots\!\!\!&\!&\!&\!&\!&\!\!\!\vdots\!\!\!\\
\!&\!&\!&\!&\!&\!&\!&\!&\!&\!&\!&\!\\
\!\!\!0\!\!\!&\!\!\!r_{\binom{0}{N-1}\binom{1}{0}}\!\!\!&\!\!\!\ldots\!\!\!&\!\!\!r_{\binom{0}{N-1}\binom{N-1}{0}}\!\!\!&\!\!\!0\!\!\!&\!&\!&\!&\!\!\!0\!\!\!&\!\!0\!\!&\!\!\ldots\!\!&\!\!0\!\!\!\\
\!&\!&\!&\!&\!&\!&\!&\!&\!&\!&\!&\!\\
\!\!\vdots\!\!&\!&\!&\!\!\vdots\!\!&\!&\!&\!&\!&\!&\!\!\ddots\!\!&\!&\!\!\vdots\!\!\\
\!\!0\!\!&\!\!0\!\!&\!\!\ldots\!\!&\!\!0\!\!&\!&\!&\!&\!&\!&\!\!0\!\!&\!\!1\!\!&\!\!0\!\!\\
\!\!0\!\!&\!\!0\!\!&\!\!\ldots\!\!&\!\!0\!\!&\!&\!&\!&\!&\!&\!&\!\!0\!\!&\!\!1\!\!
\end{array}
\right)\normalsize\,,
\end{align*}
where
\begin{equation*}
r_{\binom{0}{k}\binom{m}{0}}=-\sum_{l=1}^{N-1}\mathcal{R}_{kl}\,p_{\binom{m}{0}}(x_{l})\qquad\text{and}\qquad\delta_{km}=\sum_{l=1}^{N-1}\mathcal{R}_{kl}\,p_{\binom{0}{m}}(x_{l})\quad.
\end{equation*}
Multiplying this from the right by $Y$, and swapping the rows $\binom{0}{k}\leftrightarrow\binom{k-1}{1}$ and columns $\binom{0}{k}\leftrightarrow\binom{k-1}{1}$ for $1\leq k\leq N-1$ the matrix
\begin{equation*}
\left(\begin{array}{ccc}
\mathcal{F}&P'&0\\
0&\mathcal{M}&0\\
0&0&I_{(N-1)^{2}}
\end{array}\right)
\end{equation*}
is obtained. The $N\times N$-matrix
\begin{equation*}
\mathcal{F}_{\alpha\beta}=\exp[\sum_{j=1}^{n}u_{j}x_{\alpha}^{m_{j}^{(1)}}y_{\beta}^{m_{j}^{(2)}}]+\mathcal{O}(x_{\alpha}^{N})+\mathcal{O}(y_{\beta}^{N})
\end{equation*}
is our target matrix. The $N\times(N-1)$-matrix
\begin{equation*}
P'_{kl}=p_{\binom{l}{0}}(x_{k})\qquad,\text{ where}\quad 0\!\leq\! k\!\leq\! N\!\!-\!1\quad\&\quad 1\!\leq\! l\!\leq\! N\!\!-\!1
\end{equation*}
is not relevant, because it does not contribute to the determinant. The generalised $(N-1)\times(N-1)$-matrix of derivatives is given by $\mathcal{M}_{km}=r_{km}$ and its determinant is given by 
\begin{align*}
&\hspace{-8mm}\det(\mathcal{M})=(-1)^{N-1}\det(\mathcal{R})\cdot\Delta(x_{1},\ldots,x_{N-1})\cdot\big(\prod_{j=1}^{N-1}x_{j}\big)\\
&=(-1)^{N-1}\frac{\Delta(x_{1},\ldots,x_{N-1})\cdot\big(\prod_{j=1}^{N-1}x_{j}\big)}{\det_{1\leq k,l\leq N-1}\big(p_{\binom{0}{k}}(x_{l})\big)}\quad.
\end{align*}

In the definition of the polynomial (\ref{e:mpol}) it was demanded that $m^{(1,2)}>0$. The consequence of this is that the matrices $U_{l}$ all have a trivial first row, see example~\ref{exm:Umatexpl}. A direct consequence of this is that $v_{\binom{0}{0}\underline{q}}=0$, when $\underline{q}\neq\binom{0}{0}$. Using this one can write
\begin{align*}
&\hspace{-8mm}\det_{1\leq k,l\leq N-1}\big(p_{\binom{0}{k}}(x_{l})\big)=\det_{1\leq k,l\leq N-1}\big(\sum_{i=0}^{N-1}x_{l}^{i}\,v_{\binom{i}{0}\binom{0}{k}}\big)\\
&=\Delta(x_{1},\ldots,x_{N-1})\cdot\big(\prod_{j=1}^{N-1}x_{j}\big)\cdot\det_{1\leq k,i\leq N-1}\big(v_{\binom{i}{0}\binom{0}{k}}\big)\quad.
\end{align*}
This shows that the determinant of an $N\times N$-matrix with entries $\exp[\mathpzc{P}(x_{\alpha},y_{\beta})]$ can be approximated by $\Delta(x_{1},\ldots,x_{N})\cdot\det(M)\cdot\Delta(y_{1},\ldots,y_{N})$, where 
\begin{equation*}
\det(M)=\frac{1}{\det(\mathcal{M})}=\det_{1\leq k,i\leq N-1}\big(v_{\binom{i}{0}\binom{0}{k}}\big)=\det_{0\leq k,i\leq N-1}\big(v_{\binom{i}{0}\binom{0}{k}}\big)\quad.
\end{equation*}
Going back to (\ref{e:mdc}) or example~\ref{exm:ee} this result is not surprising. However, the information gained is the connection between the individual terms of the polynomial $\mathpzc{P}$ and the composition of the matrix $V$ and its minor $\det(M)$. The natural continuation is to compute it. A connection between determinants and minors is needed for this.

\begin{dfnt}\{Minor\}\label{dfnt:minor}\\
The minor of an $m\times n$-matrix $A$ that corresponds to selecting $k$ rows \mbox{$I=(i_{1},\ldots,i_{k})$} and $k$ columns $J=(j_{1},\ldots,j_{k})$, where $1\!\leq\! k\!\leq\! m,n$ and where the indices may be assumed strictly increasing, is denoted by
\begin{equation*}
|A|_{IJ}=\det_{1\leq l_{1},l_{2}\leq k}\big(A_{i_{l_{1}}j_{l_{2}}}\big)\quad.
\end{equation*}
\end{dfnt}

Now, suppose that $A$ is a $p\times n$-matrix, $B$ an $n\times q$-matrix and $I=(i_{1},\ldots,i_{k})$, $J=(j_{1},\ldots,j_{k})$ index subsets, where $1\leq k\leq p,q,n$. Denote by $\mathcal{K}$ the set of subsets of $\{1,\ldots,n\}$ with $k$ elements. The Cauchy-Binet formula in Lemma~\ref{l:CB} is
\begin{equation*}
|AB|_{IJ}=\sum_{K\in\mathcal{K}}|A|_{IK}|B|_{KJ}\quad.
\end{equation*}

This is a very helpful tool to compute or approximate $\det(M)$ in concrete examples. In the simple case $\mathpzc{P}(x,y)=uxy$, it can be checked directly that the matrix $M$ is diagonal. Also the case
\begin{equation*}
\mathpzc{P}(x,y)=\sum_{j}u_{j}x^{m_{j}^{(1)}}y^{m_{j}^{(2)}}\qquad\text{, where }m_{j}^{(2)}\geq m_{j}^{(1)}>0
\end{equation*}
can be seen immediately. The matrix $M$ is upper triangular in this case and the determinant depends only on the diagonal terms, i.e. those with $m_{j}^{(1)}= m_{j}^{(2)}$. To show that this converges to the determinant of the original matrix some additional constraints are needed.

\subsection{Matrix decomposition for triangular matrices\label{sec:Mdutm}}
The construction in Paragraph~\ref{sec:cmd} does not explain how the determinant is computed for matrices, whose entries are a bivariate function of the indices. Before making this more explicit in the next Paragraph, a simple test on triangular matrices is performed.\\

In Lemma~\ref{l:detcom} a matrix under the determinant is decomposed as $F=X\bullet M\bullet Y$, where $F$ is a square matrix with entries $F_{kl}=f(x_{k},y_{l})=\tilde{f}(x_{k}y_{l})$, $X$ and $Y$ are Vandermonde matrices and
\begin{equation*}
M_{kl}=\frac{f^{(k-1,l-1)}(0,0)}{((k-1)!)\cdot((l-1)!)}\quad
\end{equation*}
is the matrix of derivatives $M$. The simplest nondiagonal case is a triangular matrix of derivatives, as would follow from a function like 
\begin{equation*}
f(x,y)=\exp[x\cdot(\sum_{j=1}^{n}\beta_{j}y^{j})]\quad.
\end{equation*}
Expanding this function using partial exponential Bell polynomials $B_{l,m}$ yields
\begin{equation*}
\frac{1}{m!}\big(\sum_{j=1}^{n}\beta_{j}\frac{z^{j}}{j!}\big)^{m}=\sum_{l=0}^{\infty}\frac{y^{l}}{l!}\,B_{l,m}(\beta_{1},\ldots,\beta_{n})\quad
\end{equation*}
for the generating function. Naively one could expect that the off-diagonal parts are irrelevant, since they do not contribute to the determinant. To show that this is not the case, a simple test is conducted.\\

It is not difficult to go one step further and try $\exp[p_{1}(x)\cdot p_{2}(y)]$ with $p_{1}(x)=\sum_{i=1}^{m}\alpha_{i}x^{i}$ and $p_{2}(y)=\sum_{j=1}^{n}\beta_{j}y^{j}$. It will be shown in Lemma~\ref{l:scex2} that $N^{2}p_{1}(x_{k})p_{2}(y_{l})\rightarrow0$ is sufficient for the decomposition
\begin{equation*}
\det\big(\exp[p_{1}(x_{k})\cdot p_{2}(y_{l})]\big)=\det \big(p_{1}(x_{k})^{i-1}\big)\bullet \big(\frac{\delta_{ij}}{(j-1)!}\big)\bullet \big(p_{2}(y_{l})^{j-1}\big)\quad.
\end{equation*}

The relevance of the off-diagonal parts can be checked numerically now. The polynomials $p_{c}(x)=x+cx^{2}$ and $r_{g}(y)=y+g(y^{2}+y^{3})$ at the grid of points $x_{k}=y_{k}=\frac{k-1}{4}$ with $k=0,1,2,3$ generate a $4\times 4$-matrix $Q(c,g,\varepsilon)_{kl}=(\exp[p_{c}(\varepsilon x_{k})r_{g}(y_{l})])$. When $\varepsilon$ becomes smaller, the Taylor series approximation using the matrix decomposition becomes better. This would suggest that the determinant
\begin{equation*}
\lim_{\varepsilon\rightarrow0} \varepsilon^{-\binom{4}{2}}\det Q(c,g,\varepsilon) = \mathpzc{q}\quad
\end{equation*}
is independent of $c$ and $g$. Some numerical evidence for this is given in Table~\ref{t:ndmc}.
\begin{table}[!h]
\hspace*{3cm}\footnotesize{\begin{tabular}{l||l|l}
$\varepsilon$& $\varepsilon^{-6}\det Q(1,1,\varepsilon)$&$\varepsilon^{-6}\det Q(0,1,\varepsilon)$\\\hline
$4^{-1}$& $3.05E-4$&$1.06E-4$\\
$4^{-2}$& $1.13E-4$&$8.61E-5$\\
$4^{-3}$& $8.77E-5$&$8.17E-5$\\
$4^{-4}$& $8.21E-5$&$8.07E-5$\\
$4^{-5}$& $8.08E-5$&$8.04E-5$\\
$4^{-6}$& $8.05E-5$&$8.04E-5$
\end{tabular}}
\caption{The determinants of the $4\times 4$-matrices $Q(1,1,\varepsilon)$ and $Q(0,1,\varepsilon)$ for small values of $\varepsilon$. The notation $1.37E-3=1.37\times 10^{-3}$ is used here.\label{t:ndmc}}
\end{table}
The expected asymptotic value for this determinant is given by
\begin{align*}
&\hspace{-8mm}\mathpzc{q}=\lim_{\varepsilon\rightarrow0}\varepsilon^{-6}\det_{k,l}\big(Q(c,1,\varepsilon)\big)\\
&=\det_{k,j}\big(x_{k}^{j}\big)\det_{j,t}\big(\frac{t!}{j!}B_{j,t}(1,2c)\big)\det_{t,m}\big(\frac{1}{m!}B_{m,t}(1,2,6)\big)\det_{m,l}\big(y_{l}^{m}\big)\\
&=\frac{12}{4^{6}}\cdot1\cdot\{\prod_{t=0}^{3}\frac{1}{t!}\}\frac{12}{4^{6}}=7.15E-7\quad.
\end{align*}
This is not even close. Repeating the process for $g=0.01$ yields Table~\ref{t:ndmc2}. 
\begin{table}[!h]
\hspace*{3cm}\footnotesize{\begin{tabular}{l||l|l}
$\varepsilon$& $\varepsilon^{-6}\det Q(1,0.01,\varepsilon)$&$\varepsilon^{-6}\det Q(0,0.01,\varepsilon)$\\\hline
$4^{-1}$& $2.52E-6$&$8.91E-7$\\
$4^{-2}$& $1.05E-6$&$8.01E-7$\\
$4^{-3}$& $8.36E-7$&$7.80E-7$\\
$4^{-4}$& $7.88E-7$&$7.75E-7$\\
$4^{-5}$& $7.77E-7$&$7.73E-7$\\
$4^{-6}$& $7.74E-7$&$7.73E-7$
\end{tabular}}
\caption{The determinants of the $4\times 4$-matrices $Q(1,0.01,\varepsilon)$ and $Q(0,0.01,\varepsilon)$ for small values of $\varepsilon$. The notation $1.37E-3=1.37\times 10^{-3}$ is used here.\label{t:ndmc2}}
\end{table}
This is much closer to the asymptotic value. The reason for this must lie in the off-diagonal terms. Because of the small $\varepsilon$ the first few terms yield arbitrary good approximations, but the off-diagonal terms do not satisfy the conditions for the upper triangular matrix decomposition. This explains the difference.

\section{Examples and approximations}
In this paragraph some examples of determinant computation for `nondiagonal' functions will be discussed. Nondiagonal refers here to the matrix of derivatives for these bivariate functions, which is nondiagonal in these cases. The emphasis will be on examples that will be relevant later.

\begin{exm}\label{exm:dc1}
Consider the determinant defined by
\begin{equation}
\mathcal{Q}=\det_{1\leq k,l\leq n}\int_{-\infty}^{\infty}\ud x\,x^{k-1}e^{iax-b_{l}x^{2}}\quad.\label{e:detexm3Q}
\end{equation}
Multiplying $\mathcal{Q}$ by 
\begin{equation*}
\frac{\prod_{m=1}^{n}(m-1)!}{\varepsilon^{\binom{n}{2}}}\det_{l,m}\big(\frac{(m\varepsilon)^{l-1}}{(l-1)!}\big)
\end{equation*}
yields with the help of Lemma~\ref{l:detcom}
\begin{align}
&\hspace{-8mm}\mathcal{Q}\leftarrow\varepsilon^{-\binom{n}{2}}\det_{k,l}\int_{-\infty}^{\infty}\ud x\,\exp[(ia+k\varepsilon)x-b_{l}x^{2}]\label{e:detexm3R}\\
&=\varepsilon^{-\binom{n}{2}}\{\prod_{m=1}^{n}\sqrt{\frac{\pi}{b_{m}}}\exp[\frac{-a^{2}}{4b_{m}}]\}\det_{k,l}\exp[\frac{2iak\varepsilon+k^{2}\varepsilon^{2}}{4b_{l}}]\nonumber\\
&=\big(\frac{ia}{2}\big)^{\binom{n}{2}}\big[\prod_{m=1}^{n}\sqrt{\frac{\pi}{b_{m}}}\big]\exp[\sum_{m=1}^{n}\frac{-a^{2}}{4b_{m}}]\big\{\!\!\!\!\!\prod_{1\leq k<l\leq n}\!\!\!(b_{l}^{-1}-b_{k}^{-1})\big\}\quad.\label{e:calQ}
\end{align}
An alternative method, leading to the same outcome, is to remove $k\varepsilon$ from the integral by a shift $\lambda\rightarrow\lambda+\tilde{\lambda}$, where $\tilde{\lambda}$ solves
\begin{equation*}
0=k\varepsilon-2B_{l}\tilde{\lambda}\quad.
\end{equation*}
These results are easily checked numerically. Take $n=6$, $a=\sqrt{n}$ and $b_{k}=1+\sqrt{k/n}$, then direct evaluation of $\mathcal{Q}$ and (\ref{e:calQ}) yield $-5.25E-13$ and $-5.24E-13$ respectively. The convergence can be seen in Figure~\ref{f:detexm3}.
\begin{figure}[!ht]
\hspace{1.5cm}\includegraphics[width=0.6\textwidth]{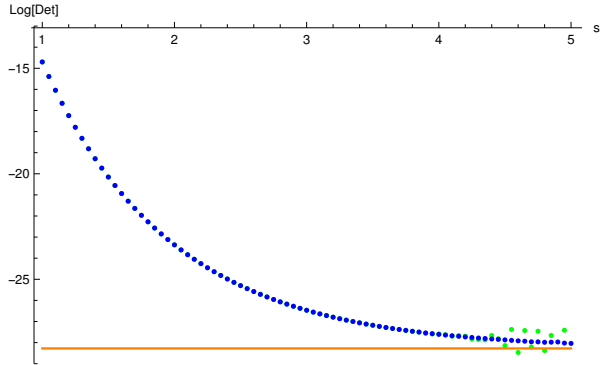}
\caption{The real part of the logarithm of $Q(\varepsilon)$ (green) (\ref{e:detexm3Q}), (\ref{e:detexm3R}) (blue) and the real part of the logarithm of the asymptotic value (\ref{e:calQ}) (orange) for $n=6$ for $\varepsilon=\exp[-s]$, where $B_{m}=1+\sqrt{m/n}$ and $A=\sqrt{n}$.\label{f:detexm3}}
\end{figure}
Writing
\begin{equation*}
\int_{-\infty}^{\infty}\ud x\,x^{k-1}e^{iax-b_{l}x^{2}}=(-i)^{\binom{N}{2}}\partial_{a}^{k-1}\int_{-\infty}^{\infty}\ud x\,e^{iax-b_{l}x^{2}}
\end{equation*}
yields the same result just as quickly. For more complicated functions than $\exp[-\frac{a^{2}}{4b}]$ this is not necessarily the case.
\end{exm}

For functions of form $f(x,y)=g\big((1+x)y\big)$. the previous methods are still sufficient.
Take the function $f(z)=\sum_{m=0}\mathcal{C}_{m}z^{m}$. Its determinant can be computed through
\begin{align*}
&\hspace{-8mm}\lim_{\varepsilon\rightarrow0}\varepsilon^{-\binom{N}{2}}\det_{k,l}\big(f((1+k\varepsilon)y_{l})\big)\\
&=\lim_{\varepsilon\rightarrow0}\varepsilon^{-\binom{N}{2}}\det_{k,l}\Big(\big((k\varepsilon)^{m-1}\big)\bullet\big(\binom{n-1}{m-1}\big)\bullet\big(\mathcal{C}_{n-1}\delta_{n,q}\big)\bullet\big(y_{l}^{q-1}\big)\Big)\\
&=\lim_{\varepsilon\rightarrow0}\varepsilon^{-\binom{N}{2}}\det_{k,l}\Big(\big((k\varepsilon)^{m-1}\big)\bullet\big(\sum_{q=0}\binom{q}{m-1}\mathcal{C}_{q}y_{l}^{q}\big)\Big)\\
&=\lim_{\varepsilon\rightarrow0}\varepsilon^{-\binom{N}{2}}\det_{k,l}\Big(\big((k\varepsilon)^{m-1}\big)\bullet\big(\frac{y_{l}^{m-1}}{(m-1)!}\frac{\partial^{m-1}}{\partial y_{l}^{m-1}}f(y_{l})\big)\Big)\\
&=\det_{m,l}\big(y_{l}^{m-1}\frac{\partial^{m-1}}{\partial y_{l}^{m-1}}f(y_{l})\big)\quad.
\end{align*}

The examples above are particularly useful for stationary phase integrals. When the entry functions are combinatorial factors, two powerful techniques are combinatorics and the orthogonal polynomials, see Paragraph~\ref{sec:opm}. The following two lemmas are examples of this.
\begin{lemma}\label{l:sfacdet1}
Let $B(x,y)$ denote the Beta function, given by
\begin{equation*}
B(x,y)=\int_{0}^{1}\ud t\,t^{x-1}(1-t)^{y-1}\quad.
\end{equation*}
The determinant of the $n\times n$-matrix with integer Beta function entries is given by
\begin{equation*}
\det_{1\leq k,l\leq n}B(k,l)=(-1)^{\binom{n}{2}}4^{1-n}\prod_{k=1}^{n-1}\frac{1}{2k+1}\binom{2k-1}{k}^{-2}\quad.
\end{equation*}
\end{lemma}
\begin{proof}
The determinant can be written as
\begin{align*}
&\hspace{-8mm}\det_{k,l}B(k,l)=\int_{[0,1]^{n}}\!\!\!\!\!\ud \vec{t}\,\big[\prod_{l=1}^{n}(1-t_{l})^{l-1}\big]\det_{k,l}(t_{l}^{k-1})\\
&=\!\int_{[0,1]^{n}}\!\!\!\!\!\!\ud \vec{t}\,\big[\prod_{l=1}^{n}(1\!-\!t_{l})^{l\!-\!1}\big]\det_{k,l}\big(Q_{k\!-\!1}(t_{l})\big)\!=\!(-1)^{\binom{n}{2}}\prod_{l=1}^{n}\!\int_{0}^{1}\!\!\!\ud t\,Q_{l\!-\!1}(t)^{2}\;,
\end{align*}
where $\{Q_{k}|k\in\mathbb{N}_{0}\}$ is the unique family of monic orthogonal polynomials with respect to the unit weight on $[0,1]$ described in Theorem~\ref{thrm:op}. These are related to the Legendre polynomials $P_{k}$ through 
\begin{equation*}
Q_{k}(x)=c_{k}P_{k}(2x-1)\qquad,\text{ where }c_{k}=\frac{\Gamma(k+1)\Gamma(\frac{1}{2})}{4^{k}\Gamma(k+\frac{1}{2})}\quad.
\end{equation*}
Using that
\begin{equation*}
\int_{-1}^{1}\ud x\,P_{k}(x)P_{l}(x)=\frac{2}{c_{k}c_{l}}\int_{0}^{1}\ud t\,Q_{k}(t)Q_{l}(t)=\frac{2\delta_{kl}}{2k+1}
\end{equation*}
it follows that
\begin{align*}
&\hspace{-8mm}\det_{k,l}B(k,l)=(-1)^{\binom{n}{2}}2^{-N}\cdot\big[\prod_{m=1}^{n-1}c_{m}^{2}\big]\cdot\big[\prod_{j=0}^{n-1}\frac{2}{2j+1}\big]\\
&=(-1)^{\binom{n}{2}}4^{1-n}\prod_{k=1}^{n-1}\frac{1}{2k+1}\binom{2k-1}{k}^{-2}\quad,
\end{align*}
where the Gamma function duplication formula 
\begin{equation*}
\Gamma(k)\Gamma(k+1/2)=\sqrt{\pi}2^{1-2k}\Gamma(2k)
\end{equation*}
was used.
\end{proof}

\begin{rmk}
Using that
\begin{equation*}
B(k,l)=\!\int_{0}^{1}\!\!\!\ud t\,t^{k\!-\!1}(1\!-\!t)^{l\!-\!1}=\!\!\int_{0}^{1}\!\!\!\ud t\,t^{k\!-\!1}\sum_{s=0}^{l-1}\!\binom{l\!-\!1}{s}(-t)^{s}=\!\sum_{s=0}^{l-1}\!\binom{l\!-\!1}{s}\frac{(-1)^{s}}{s\!+\!k}\;,
\end{equation*}
Lemma~\ref{l:sfacdet1} can also be used to compute $\det_{k,l}(k+l)^{-1}$.
\end{rmk}

\begin{lemma}\label{l:sfacdet2}
The determinant of the $n\times n$-matrix $M$ with entries \\\mbox{$M_{kl}=\frac{1}{(2k-l)!}\vartheta(2k-l\geq0)$} for $k,l=0,\ldots,n-1$, where $\vartheta$ is the Heaviside stepfunction, is given by $\prod_{t=1}^{n-1}\frac{1}{(2t-1)!!}$.
\end{lemma}
\begin{proof}
We compute
\begin{equation*}
Y=\det_{0\leq k,l\leq n}\Big(\binom{2k}{l}\Big)=\big\{\prod_{m=0}^{n-1}\frac{(2m)!}{m!}\big\}\det_{0\leq k,l\leq n}\Big(\frac{1}{(2k-l)!}\Big)\quad.
\end{equation*}
The coefficient of $x^{l}$ in $(1+x)^{2k}$ is $\binom{2k}{l}=\sum_{m=0}^{l}\binom{k}{m}\binom{m}{l-m}2^{2m-l}$. This follows from
\begin{equation*}
(1+x)^{2k}=\sum_{m=0}^{k}\binom{k}{m}x^{m}(2+x)^{m}=\sum_{m=0}^{k}\sum_{r=0}^{m}\binom{k}{m}\binom{m}{r}x^{m+r}2^{m-r}\quad.
\end{equation*}
This implies that
\begin{equation*}
Y=\det_{1\leq k,l\leq n}\Big(\big(\binom{k}{m}\big)\bullet\big(2^{2m-l}\binom{m}{l-m}\big)\Big)=\prod_{m=0}^{n-1}2^{m}\quad,
\end{equation*}
because both matrices are triangular. Putting the two expressions for $Y$ together yields
\begin{equation*}
\det_{0\leq k,l\leq n}\Big(\frac{1}{(2k-l)!}\Big)=\prod_{m=0}^{n-1}\frac{2^{m} \cdot m!}{(2m)!}=\prod_{m=1}^{n-1}\frac{1}{(2m-1)!!}\quad.
\end{equation*}
\end{proof}

The final pages of this paragraph are dedicated to approximations for determinants. They give sufficient conditions to neglect small terms in a bivariate polynomial in the exponent with respect to the $(1,1)$-term.

\begin{lemma}\label{l:scex1}
For some fixed $\mathpzc{M}\geq 2$ and parameters $\alpha_{n}$ for $2\leq n\leq \mathpzc{M}$ that satisfy $\alpha_{n}\ll N^{-n-1}$ the limit
\begin{equation*}
\lim_{N\rightarrow\infty} \det_{1\leq k,l\leq N}\exp[\alpha_{1} x_{k}y_{l}+\sum_{n=2}^{\mathpzc{M}}\alpha_{n}(x_{k}y_{l})^{n}]=\lim_{N\rightarrow\infty} \det_{1\leq k,l\leq N}\exp[\alpha_{1} x_{k}y_{l}]\quad
\end{equation*}
holds.
\end{lemma}
\begin{proof}
The lemma will first be proved for a single $n\geq2$. We apply the determinant decomposition from Lemma~\ref{l:detcom} to $\exp[ x_{k}y_{l}+\alpha_{n}(x_{k}y_{l})^{n}]$ for an $(N\!+\!1)\times (N\!+\!1)$-matrix. Assuming for simplicity that $N$ is a multiple of $n$, the diagonal matrix of derivatives has as final entry
\begin{equation*}
\frac{1}{N!}+\frac{\alpha_{n}}{(N-n)!}+\ldots+\frac{\alpha_{n}^{\frac{N}{n}}}{(N/n)!}\quad.
\end{equation*}
It follows from Stirling's approximation that $\alpha_{n}$ may be ignored, if 
\begin{equation*}
\alpha_{n}\ll\frac{1}{N^{n}}\quad\text{and}\quad\alpha_{n}\ll \frac{e^{n-1}}{nN^{n-1}}\quad.
\end{equation*}
Because the determinant of this matrix consists of $N$ factors
\begin{equation*}
\alpha_{n}\ll N^{-n-1}
\end{equation*}
is a sufficient condition. It follows directly that adding finitely many such terms does not change the argument.
\end{proof}

\begin{lemma}\label{l:scex2}
For some fixed $\mathpzc{M}\geq 2$ and parameters $\beta_{n}$ for $2\leq n\leq \mathpzc{M}$ that satisfy 
\begin{equation*}
\beta_{n}\ll \frac{1}{N^{2}\max_{j}|y_{j}|^{n-1}}
\end{equation*}
the limit
\begin{equation*}
\lim_{N\rightarrow\infty} \det_{1\leq k,l\leq N}\exp[x_{k}(y_{l}+\sum_{n=2}^{\mathpzc{M}}\beta_{n}y_{l}^{n})]=\lim_{N\rightarrow\infty} \det_{1\leq k,l\leq N}\exp[ x_{k}y_{l}]
\end{equation*}
holds.
\end{lemma}
\begin{proof}
The determinant decomposition from Lemma~\ref{l:detcom} is applied once more to the $N\times N$-matrix
\begin{align*}
&\hspace{-8mm}\det_{k,l}\exp[x_{k}(y_{l}+\sum_{n=2}^{\mathpzc{M}}\beta_{n}y_{l}^{n})]\,\leftarrow\, \det_{k,l}\big(x_{k}^{m-1}\big)_{km}\\
&\bullet\,\big(\frac{\delta_{mt}}{(m-1)!}\big)_{mt}\,\bullet\,\big(\,(y_{l}+\sum_{n=2}^{\mathpzc{M}}\beta_{n}y_{l}^{n})^{t-1}\,\big)_{tl}\quad.
\end{align*}
The coefficients $\beta_{n}$ may be ignored, if the final determinant can be replaced by the determinant of the Vandermonde-matrix with entries $y_{l}^{t-1}$. It follows from
\begin{align*}
&\hspace{-8mm}\det_{1\leq t,l\leq N}(y_{l}+\sum_{n=2}^{\mathpzc{M}}\beta_{n}y_{l}^{n})^{t-1}=\prod_{k<l}(y_{l}-y_{k})\cdot(1+\sum_{n=2}^{\mathpzc{M}}\beta_{n}\frac{y_{l}^{n}-y_{k}^{n}}{y_{l}-y_{k}})\\
&\leq\Big[\prod_{k<l}(y_{l}-y_{k})\Big]\cdot(1+\sum_{n=2}^{\mathpzc{M}}2n\beta_{n}\max_{1\leq m\leq N}|y_{m}|^{n-1})^{\binom{N}{2}}
\end{align*}
that $\beta_{j}$ may be ignored, provided that
\begin{equation*}
\beta_{j}\cdot\max_{m}|y_{m}|^{j-1}\ll\frac{\beta_{1}}{2jN^{2}}\quad.
\end{equation*}
\end{proof}

\begin{lemma}\label{l:scex3}
For some fixed $\mathpzc{M}\in\mathbb{N}$ and parameters $\gamma_{j}\in\mathbb{C}$, $m_{j},n_{j}$ with $m_{j}\neq n_{j}$ for $1\leq j\leq \mathpzc{M}$ the limit
\begin{equation*}
\lim_{N\rightarrow\infty} \det_{1\leq k,l\leq N}\exp[k\varepsilon y_{l}+\sum_{j=1}^{\mathpzc{M}}\gamma_{j}(k\varepsilon)^{m_{j}}y_{l}^{n_{j}}]=\lim_{N\rightarrow\infty} \det_{1\leq k,l\leq N}\exp[k\varepsilon y_{l}]
\end{equation*}
holds, if 
\begin{equation*}
m_{j}>n_{j}\qquad \text{, for all }1\leq j\leq \mathpzc{M}
\end{equation*}
\'or if 
\begin{equation*}
n_{j}>m_{j}\quad \&\quad \gamma_{j} \ll \frac{N^{-n_{j}}}{\max_{l}|y_{l}|^{n_{j}-m_{j}}}\qquad\text{, for all }1\leq j\leq \mathpzc{M}\;.
\end{equation*}
\end{lemma}
\begin{proof}
This will first be proved for $\mathpzc{M}=1$. Consider the $N\times N$-matrix $F$ consisting of the entries
\begin{equation*}
F_{kl}=f(k\varepsilon,y_{l})=\exp[(k\varepsilon)\cdot y_{l}+\gamma (k\varepsilon)^{m}y_{l}^{n}]\quad.
\end{equation*}
To ignore the $\gamma$-term as $N\rightarrow\infty$ we must assume that $\gamma$ and/or $y_{l}$ are small. First the case $n>m$ will be discussed. It is treated using the same strategy as in Remark~\ref{rmk:ap}. By the Cauchy-Binet formula it follows that
\begin{align}
&\hspace{-8mm}\lim_{\varepsilon\rightarrow0}\varepsilon^{-\binom{N}{2}}\det(F)\nonumber\\
&=\lim_{\varepsilon\rightarrow0}\varepsilon^{-\binom{N}{2}}\!\!\!\!\!\!\sum_{I,J\in\mathcal{S}_{\infty,N}}\!\!\!\!\!\!\big|(k\varepsilon)^{j\!-\!1}\big|_{\textbf{1}I}\cdot\big|\frac{f^{(j,m)}(0,0)}{((j\!-\!1)!)((m\!-\!1)!)}\big|_{IJ}\cdot\big|y_{l}^{m\!-\!1}\big|_{J\textbf{1}}\label{e:scex3e0}\\
&=\sum_{J\in\mathcal{S}_{\infty,N}}\big|k^{j-1}\big|_{\textbf{1}\textbf{1}}\cdot\big|\frac{f^{(j,m)}(0,0)}{((j-1)!)((m-1)!)}\big|_{\textbf{1}J}\cdot\big|y_{l}^{m-1}\big|_{J\textbf{1}}\nonumber\\
&=(\prod_{t=1}^{N-1}t!)\sum_{J\in\mathcal{S}}\big|M\big|_{\textbf{1}J}\cdot\big|Y\big|_{J\textbf{1}}\quad.\label{e:scex3e1}
\end{align}
If $\gamma\leq N^{-m}$, then the diagonal entries are the largest on each row. To an index set $J$ a size $|J|=\sum_{j\in J}j\geq \binom{N}{2}$ is assigned. For every $J$ the determinant $|Y|_{J\textbf{1}}$ is a multiple of the Vandermonde-determinant $\Delta(y_{1},\ldots,y_{N})$, so that
\begin{equation*}
\Big| \big|Y\big|_{J\textbf{1}} \Big|\leq \Delta(y_{1},\ldots,y_{N}) \cdot (N\max_{l}|y_{l}|)^{|J|-\binom{N}{2}}\quad.
\end{equation*}
The number of elements in $\mathcal{S}_{\infty,N}$ with size $\Lambda$ is the number $p_{N}(\Lambda)$. It was shown in Remark~\ref{rmk:ap} that it satisfies
\begin{equation*}
p_{m}(n)\leq \alpha_{m}2^{n-\binom{m}{2}}\quad\text{with}\quad \alpha_{m}=\prod_{t=1}^{m}\frac{1}{1-2^{-t}}\quad.
\end{equation*}
The determinant (\ref{e:scex3e1}) is given by
\begin{align*}
&\hspace{-8mm}\lim_{\varepsilon\rightarrow0}\varepsilon^{-\binom{N}{2}}\det(F)=(\prod_{t=1}^{N-1}t!)\big|M\big|_{\textbf{1}\textbf{1}}\cdot \big|Y\big|_{\textbf{1}\textbf{1}}\times\Big(1+\mathpzc{C}_{1}\Big)\quad,
\end{align*}
where the correction is estimated by
\begin{equation}
|\mathpzc{C}_{1}|\leq\alpha_{\infty}\sum_{d=1}^{\infty}(\gamma N^{m})^{d}(2N\max_{l}|y_{l}|)^{d(n-m)}\Big)\quad,\label{e:corrfacscex3}
\end{equation}
where the difference $|J|-\binom{N}{2}= d(n-m)$ is used to sum over all corrections. This correction term may be ignored, if 
\begin{equation*}
\gamma\max_{l}|y_{l}|^{n-m} \ll N^{-n}
\end{equation*}
In this case, the determinant is asymptotically given by the first term in the Cauchy-Binet formula, which does not depend on $\gamma$, since $M$ is triangular. This means that the $\gamma$-term may be ignored.\\
The proof used the fact that every time the correction term is encountered the size of the index set is increased by $n-m$. This implies that the correction term $\mathpzc{C}_{1}$ is generated by the number of times that $\gamma$ is seen. This remains the case for larger $\mathpzc{M}$ and shows that only the correction factor (\ref{e:corrfacscex3}) changes. In fact, the correction factor for index $j$ satisfies 
\begin{equation*}
|\mathpzc{C}_{1}^{(j)}|\leq\alpha_{\infty}\sum_{d=1}^{\infty}(\gamma_{j} N^{m_{j}})^{d}(2N\max_{l}|y_{l}|)^{d(n_{j}-m_{j})}\Big)
\end{equation*}
and the determinant (\ref{e:scex3e1}) is now given by
\begin{align*}
&\hspace{-8mm}\lim_{\varepsilon\rightarrow0}\varepsilon^{-\binom{N}{2}}\det(F)=(\prod_{t=1}^{N-1}t!)\big|M\big|_{\textbf{1}\textbf{1}}\cdot \big|Y\big|_{\textbf{1}\textbf{1}}\times\Big(1+\mathpzc{C}_{2}\Big)\quad,
\end{align*}
with 
\begin{equation*}
\big|1+\mathpzc{C}_{2}\big|\leq \prod_{j=1}^{M}\big|1+\mathpzc{C}_{1}^{(j)}\big|\rightarrow 1\quad.
\end{equation*}
The proves the statement, if all $n_{j}>m_{j}$.\\

The case $m>n$ is much simpler. In this case the matrix $M$ is lower triangular, so that $J=\textbf{1}$. Sending $\varepsilon\rightarrow0$ shows that $I=\textbf{1}$ too. And because $M$ is triangular, $\gamma$ falls out. The same argument works for larger $\mathpzc{M}$.
\end{proof}

There is one extension of this concept that needs discussing. Compared with Lemma~\ref{l:scex3} other off-diagonal terms with opposite powers are added.

\begin{lemma}\label{l:scex4}
For some fixed $\mathpzc{M}\in\mathbb{N}$ and parameters $\gamma_{l,j}\in\mathbb{C}$, $m_{l,j},n_{l,j}$ with $m_{l,j}< n_{l,j}$ for $1\leq j\leq \mathpzc{M}$ and $l=1,2$ the limit
\begin{align*}
&\hspace{-8mm}\lim_{N\rightarrow\infty} \det_{1\leq k,l\leq N}\exp[k\varepsilon y_{l}+\sum_{j=1}^{\mathpzc{M}}\gamma_{1,j}(k\varepsilon)^{m_{1,j}}y_{l}^{n_{1,j}}+\gamma_{2,j}(k\varepsilon)^{n_{2,j}}y_{l}^{m_{2,j}}]\\
&=\lim_{N\rightarrow\infty} \det_{1\leq k,l\leq N}\exp[k\varepsilon y_{l}]
\end{align*}
holds, if for all $1\leq j\leq M$
\begin{equation*}
\gamma_{1,j}\leq \frac{N^{-\frac{1}{2}-n_{1,j}}}{1+\max_{l}|y_{l}|^{n_{1,j}-m_{1,j}}}\qquad\text{and}\qquad\gamma_{2,j}\leq N^{-\frac{1}{2}-m_{2,j}}\quad.
\end{equation*}
\end{lemma}
\begin{proof}
For two arbitrary integers $1\leq j,\tilde{j}\leq \mathpzc{M}$ that are possibly identical, we define $\gamma_{1}=\gamma_{1,j}$, $n_{1}=n_{1,j}$, $m_{1}=m_{1,j}$ and $\gamma_{2}=\gamma_{2,\tilde{j}}$, $n_{2}=n_{2,\tilde{j}}$ and $m_{2}=m_{2,\tilde{j}}$. For these two terms the claim will be proved and the general result will follow.\\

The determinant of the matrix
\begin{equation}
F_{kl}=f(k\varepsilon,y_{l})=\exp[(k\varepsilon)\cdot y_{l}+\gamma_{1} (k\varepsilon)^{m_{1}}y_{l}^{n_{1}}+\gamma_{2} (k\varepsilon)^{n_{2}}y_{l}^{m_{2}}]\quad\label{e:scex4e-1}
\end{equation}
is decomposed once more, where $n_{1}>m_{1}\geq1$ and $n_{2}>m_{2}\geq1$ are postive integers. The Cauchy-Binet formula again yields (\ref{e:scex3e1}). The same arguments as above can be used to show when index sets $J\neq \textbf{1}$ are negligible. However, we can no longer claim that the determinant does not depend on $\gamma_{1}$ and $\gamma_{2}$. It does. Under which conditions are these contributions sufficiently small?\\
Consider the smallest power $p=d_{1}m_{1}+d_{2}n_{2}=d_{1}n_{1}+d_{2}m_{2}$ that can be written as a positive integer combination, i.e. $d_{1,2}\in\mathbb{N}$. Assuming that $N$ is larger than $p$, the contribution of this combination is much smaller, if $\gamma_{1}^{d_{1}}\gamma_{2}^{d_{2}}\ll \frac{(N-p)!}{N!}$. Because there are $N$ diagonal elements, the necessary condition is $\gamma_{1}^{d_{1}}\gamma_{2}^{d_{2}}\ll N^{-p-1}$. It follows automatically that multiples of $p$ also yield vanishing parts.\\

Next we consider the off-diagonal determinant contributions. Every term in the determinant is a product of diagonal and off-diagonal factors. These off-diagonal factors are cycles
\begin{align*}
&M_{k_{1}k_{2}}M_{k_{2}k_{3}}\cdots M_{k_{l-1}k_{l}} \quad\leftrightarrow\quad (k_{1}k_{2}\ldots k_{l})\quad,
\end{align*}
to which we assign a length
\begin{equation*}
|k_{1}-k_{2}|+|k_{2}-k_{3}|+\ldots+|k_{l}-k_{1}|=2\mathpzc{L}\quad.
\end{equation*}
This shows that every such cycle corresponds to a composition of $2\mathpzc{L}$. Since we are neglecting signs, it is not difficult to see that there are $2^{l}$ possible sign combinations. Assuming that $k_{1}$ is the smallest element, there are no more than $l2^{l-2}$ cycles with length $2\mathpzc{L}$ and $l$ elements. The number of compositions of $2\mathpzc{L}$ is thus smaller than\footnote{The number of compositions of $n$ with $m$ terms, which are all greater than or equal to $1$, is $\binom{n}{m-1}$.}
\begin{equation*}
N\sum_{t=1}^{\mathpzc{L}}\binom{2\mathpzc{L}-1}{t-1}2^{t-2}\leq \frac{1}{2}N9^{\mathpzc{L}}\quad.
\end{equation*}
If all the off-diagonal matrix elements satisfy either
\begin{align}
&M_{k_{1}k_{2}}\leq \zeta^{|k_{1}-k_{2}|}N^{-k_{1}-\frac{1}{2}}\quad\text{ or }
\quad M_{k_{1}k_{2}}\leq \zeta^{|k_{1}-k_{2}|}N^{-k_{2}-\frac{1}{2}}\qquad,\label{e:scex4e0}
\end{align}
for some $\zeta\rightarrow0$, then
\begin{align*}
&\hspace{-8mm}M_{k_{1}k_{2}}M_{k_{2}k_{3}}\cdots M_{k_{l}k_{1}} \leq \zeta^{2\mathpzc{L}} N^{-1-\sum_{j=1}^{l}k_{j}}\leq \zeta^{2\mathpzc{L}} N^{-1}\cdot\Big(\prod_{j=1}^{l}M_{k_{j}k_{j}}\Big)\quad.
\end{align*}
Summing over all $\mathpzc{L}$ shows that these corrections are small.\\

The parameters also need to satisfy 
\begin{equation*}
\gamma_{j} \ll \frac{N^{-n_{j}}}{\max_{l}|y_{l}|^{n_{j}-m_{j}}}
\end{equation*}
from Lemma~\ref{l:scex3}. Combining this with (\ref{e:scex4e0}) yields
\begin{align*}
&\gamma_{1}\leq \frac{N^{-\frac{1}{2}-n_{1}}}{1+\max_{l}|y_{l}|^{n_{1}-m_{1}}}\qquad\text{and}\qquad\gamma_{2}\leq N^{-\frac{1}{2}-m_{2}}\quad.
\end{align*}
The claim now holds, when these conditions are satisfied for every pair $(j,\tilde{j})$.
\end{proof}

\begin{exm}
It is interesting to test the last four lemmas~\ref{l:scex1}-\ref{l:scex4} numerically. In Figure~\ref{f:scex3} a numerical example with $N=5$ and  $x_{k}=y_{k}=1+\sqrt{k/N}$ is given for
\begin{align}
&Q=\log\varepsilon^{-\binom{N}{2}}\det_{1\leq k,l\leq N}\big(\exp[\varepsilon x_{k}y_{l}]\big)\qquad;\nonumber\\
&R=\log\varepsilon^{-\binom{N}{2}}\det_{1\leq k,l\leq N}\big(\exp[\varepsilon x_{k}(y_{l}+N^{-3}y^{3})]\big)\qquad;\nonumber\\
&S=\log\varepsilon^{-\binom{N}{2}}\det_{1\leq k,l\leq N}\big(\exp[\varepsilon x_{k}y_{l}+N^{-4.5}(\varepsilon x_{k}y_{l})^{3}]\big)\qquad;\nonumber\\
&T=\log\varepsilon^{-\binom{N}{2}}\det_{1\leq k,l\leq N}\big(\exp[\varepsilon x_{k}y_{l}+N^{-3.5}(\varepsilon x_{k})^{2}(y_{l})^{3}\nonumber\\
&\qquad\qquad\qquad\qquad\qquad+N^{-2.5}(\varepsilon x_{k})^{3}(y_{l})^{2}\,]\;\big)\qquad\text{and}\nonumber\\
&U=\log\Delta(x_{1},\ldots x_{N})\big[\prod_{t=0}^{N-1}\frac{1}{t!}\big]\Delta(y_{1},\ldots y_{N})\qquad.\label{e:scex4e2}
\end{align}
For larger $N$ such examples suffer from numerical instability.

\begin{figure}[!ht]
\includegraphics[width=0.99\textwidth]{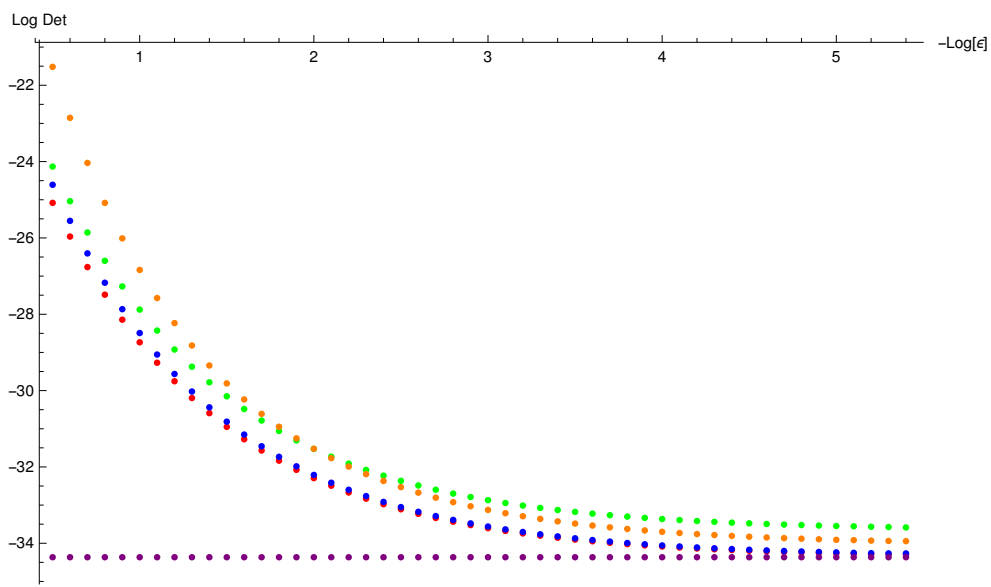}
\caption{The logarithms of the determinants for $N=5$ from (\ref{e:scex4e2}) with $-\log(\varepsilon)$ on the $x$-axis with $Q$ (red), $R$ (green), $S$ (blue), $T$ (orange) and $U$ (purple).\label{f:scex3}}
\end{figure}
\end{exm}

\chapter{The polytope of symmetric stochastic matrices\label{sec:polytope}}

The cubic Kontsevich model has the pleasant feature that the matrix integral factorises through diagonalisation. Diagonalising the matrices and applying the Harish-Chandra-Itzykson-Zuber integral to the unitary part yields a product of eigenvalue integrals, only intertwined through Vandermonde determinants. For general matrix models this is not the case. A nontrivial denominator will appear, which frustrates straightforward integration. In the Grosse-Wulkenhaar model the way to factorise the eigenvalue integrals in the partition function is through the volume of the polytope of symmetric stochastic matrices. More precisely, this is done by integrating against its diagonal subpolytopes. With the exception of Paragraph~\ref{sec:Evcp}, this chapter can be found in~\cite{jins1}.\\

The asymptotic volume of the polytope of symmetric stochastic matrices can be determined by asymptotic enumeration techniques as in the case of the Birkhoff polytope. These methods can be extended to polytopes of symmetric stochastic matrices with given diagonal, if this diagonal varies not too wildly. To this end, the asymptotic number of symmetric matrices with natural entries, zero diagonal and varying row sums is determined.

\section{Introduction}
Convex polytopes arise naturally in various places in mathematics. A fundamental problem is the polytope's volume. Some results are known for low-dimensional setups~\cite{gruber}, polytopes with only a few vertices, or highly symmetric cases~\cite{mckay3, canfield1}. This approach belongs to the latter category.\\
\begin{dfnt}
A convex polytope $\mathpzc{P}$ is the convex hull of a finite set \\$S_{\mathpzc{P}}=\{v_{j}\in\mathbb{R}^{n}\}$ of vertices.
\end{dfnt}
Stochastic matrices are square matrices with nonnegative entries, such that every row of the matrix sums to one. The symmetric stochastic $N\times N$-matrices are an example of a convex polytope. It will be denoted by $\mathcal{P}_{N}$. Its vertices are given by the symmetric permutation matrices. There are $\sum_{j=0}^{N/2}\binom{N}{2j}(2j-1)!!$ such matrices. It follows directly from the Birkhof-Von Neumann theorem that all symmetric stochastic matrices are of this form. A basis for this space is given by
\begin{equation*}
\{I_{N}\}\cup\{B^{(jk)}|1\leq j<k\leq N\}\quad,
\end{equation*}
where $I_{N}$ is the $N\times N$ identity matrix and the matrix elements of $B^{(jk)}$ are given by
\begin{equation*}
B^{(jk)}_{lm}=\left\{\begin{array}{ll}B^{(jk)}_{lm}=1&\text{, if }\{l,m\}=\{j,k\}\quad;\\B^{(jk)}_{lm}=1&\text{, if }j\neq l=m\neq k\quad;\\B^{(jk)}_{lm}=0&\text{, otherwise }\end{array}\right.\quad.
\end{equation*}
All these vertices are linearly independent and it follows that the polytope is\\
\mbox{$\binom{N}{2}$-dimensional}. 
\begin{dfnt}
A convex subpolytope $\mathpzc{P}'$ of a convex polytope $\mathpzc{P}$ is the convex hull of a finite set $\{v'_{j}\in \mathpzc{P}\}$ of elements in $\mathpzc{P}$.
\end{dfnt}
Slicing a polytope yields a surface of section, which is itself a convex space and, hence, a polytope. Determining its vertices is in general very difficult.\\
Spaces of symmetric stochastic matrices with several diagonal entries fixed are examples of such slice subpolytopes of $\mathcal{P}_{N}$, provided that these entries lie between zero and one. The slice subpolytope of $\mathcal{P}_{N}$, obtained by fixing all diagonal entries $h_{j}\in[0,1]$, will be called the diagonal subpolytope $P_{N}(h_{1},\ldots,h_{N})$ here. This is a polytope of dimension $N(N-3)/2$. These polytopes form the main subject of this chapter.\\

To keep the notation light, vectors of $N$ elements are usually written by a bold symbol. The diagonal subpolytope with entries $h_{1},\ldots,h_{N}$ will thus be written by $P_{N}(\vec{h})$.\\

The main results are the following two theorems.
\begin{thrm}\label{thrm:p1}
Let $V_{N}(\vec{t};\lambda)$ be the number of symmetric $N\times N$-matrices with an empty diagonal and entries in the natural numbers such that $t_{j}$ is the $j$-th row sum. Denote the total entry sum by $x=\sum_{j=1}^{N}t_{j}$ and let $\lambda$ be the average matrix entry 
\begin{equation*}
\lambda=\frac{x}{N(N-1)}>\frac{C}{\log N}\quad.
\end{equation*}
If for some $\omega\in(0,\frac{1}{4})$ the limit
\begin{equation*}
\lim_{N\rightarrow\infty}\frac{t_{j}-\lambda(N-1)}{\lambda N^{\frac{1}{2}+\omega}}=0\qquad\text{for all }j=1,\ldots,N\quad,
\end{equation*}
then the number of such matrices is asymptotically ($N\rightarrow\infty$) given by
\begin{align*}
&\hspace{-8mm}V_{N}(\vec{t};\lambda)=\frac{\sqrt{2}(1+\lambda)^{\binom{N}{2}}}{(2\pi\lambda(\lambda+1)N)^{\frac{N}{2}}}\big(1+\frac{1}{\lambda}\big)^{\frac{x}{2}}\exp[\frac{14\lambda^{2}+14\lambda-1}{12\lambda(\lambda+1)}]\\
&\times\exp[\frac{-1}{2\lambda(\lambda+1)N}\sum_{m}(t_{m}-\lambda(N-1))^{2}]\\
&\times\exp[\frac{-1}{\lambda(\lambda+1)N^{2}}\sum_{m}(t_{m}-\lambda(N-1))^{2}]\\
&\times\exp[\frac{2\lambda+1}{6\lambda^{2}(\lambda+1)^{2}N^{2}}\sum_{m}(t_{m}-\lambda(N-1))^{3}]\\
&\times\exp[-\frac{3\lambda^{2}+3\lambda+1}{12\lambda^{3}(\lambda+1)^{3}N^{3}}\sum_{m}(t_{m}-\lambda(N-1))^{4}]\\
&\times\exp[\frac{1}{4\lambda^{2}(\lambda\!+\!1)^{2}N^{4}}\big(\!\sum_{m}(t_{m}\!-\!\lambda(N\!-\!1))^{2}\big)^{2}]\!\times\!\Big(\!1\!+\!\mathcal{O}(N^{-\!\frac{1}{2}\!+\!6\omega})\!\Big)\;.
\end{align*}
\end{thrm}

\begin{thrm}\label{thrm:p2}
Let $\vec{h}=h_{1},\ldots,h_{N}$ with $h_{j}\in[0,1]$ and $\chi=\sum_{j=1}^{N}h_{j}$. If
\begin{equation*}
\lim_{N\rightarrow\infty}N^{\frac{1}{2}-\omega}\frac{N-1}{N-\chi}\cdot \big|h_{j}-\frac{\chi}{N}\big|=0\quad\text{for all }j=1,\ldots,N\qquad,
\end{equation*}
and for some $\omega\in(\frac{\log\log N}{2\log N},\frac{1}{4})$, then the asymptotic volume ($N\rightarrow\infty$) of the polytope of symmetric stochastic $N\times N$-matrices with diagonal $(h_{1},\ldots,h_{N})$ is given by
\begin{align*}
&\hspace{-8mm}\vol(P_{N}(\vec{h}))=\sqrt{2}e^{\frac{7}{6}}\Big(\frac{e(N-\chi)}{N(N-1)}\Big)^{\binom{N}{2}}\Big(\frac{N(N-1)^{2}}{2\pi(N-\chi)^{2}}\Big)^{\frac{N}{2}}\\
&\times\exp[-\frac{N(N-1)^{2}}{2(N-\chi)^{2}}\sum_{j}(h_{j}-\frac{\chi}{N})^{2}]\nonumber\\
&\times\exp[-\frac{(N-1)^{2}}{(N-\chi)^{2}}\sum_{j}(h_{j}-\frac{\chi}{N})^{2}]\exp[-\frac{N(N-1)^{3}}{3(N-\chi)^{3}}\sum_{j}(h_{j}-\frac{\chi}{N})^{3}]\nonumber\\
&\times\exp[-\frac{N(N-1)^{4}}{4(N-\chi)^{4}}\sum_{j}(h_{j}-\frac{\chi}{N})^{4}]\exp[\frac{(N-1)^{4}}{4(N-\chi)^{4}}\big(\sum_{j}(h_{j}-\frac{\chi}{N})^{2}\big)^{2}]\\
&\times\big(1+\mathcal{O}(N^{-\frac{1}{2}+6\omega})\big)\quad.
\end{align*}
\end{thrm}

The outline of this chapter is as follows. In Paragraph~\ref{sec:cp} the volume problem is formulated as a counting problem and subsequently as a contour integral. Some exact volume computations for small dimensions are performed in Paragraph~\ref{sec:Evcp}. Under the assumption of a restricted region this is subsequently integrated in Paragraph~\ref{sec:icp}. Paragraph~\ref{sec:rir} is dedicated to a fundamental lemma to actually restrict the integration region. The volume of the diagonal subpolytopes is extracted from the counting result in Paragraph~\ref{sec:pv}.

\section{Counting problem\label{sec:cp}}
The volume of a polytope $\mathpzc{P}$ in $\mathbb{R}^{n}$ with basis $\{\mathcal{B}_{j}\in\mathbb{R}^{n}|1\leq j\leq d\}$ is obtained by
\begin{equation*}
\int_{[0,1]^{d}}\ud\vec{u}\,\mathbf{1}_{\mathpzc{P}}(\sum_{j=1}^{d}u_{j}\mathcal{B}_{j})\quad,
\end{equation*}
where $\mathbf{1}_{\mathpzc{P}}$ is the indicator function for the polytope $\mathpzc{P}$. If the polytope is put on a lattice $(a\mathbb{Z})^{n}$ with lattice parameter $a\in(0,1)$, an approximation of this volume is obtained by counting the lattice sites inside the polytope and multiplying this by the volume $a^{n}$ of a single cell. This approximation becomes better as the lattice parameter shrinks. In the limit this yields
\begin{equation}
\vol(\mathpzc{P})=\lim_{a\rightarrow 0} a^{n}\,\,|\{\mathpzc{P}\cap (a\mathbb{Z})^{n}\}|\quad.\label{e:ehrhart}
\end{equation}
This approach is formalised by the Ehrhart polynomial~\cite{stanley1}, which counts the number of lattice sites of $\mathbb{Z}^{n}$ in a dilated polytope. A dilation of a polytope $\mathpzc{P}$ by a factor $a^{-1}>1$ yields the polytope $a^{-1}\mathpzc{P}$, which is the convex hull of the dilated vertices $S_{a^{-1}\mathpzc{P}}=\{a^{-1}v|v\in S_{\mathpzc{P}}\}$. That the obtained volume is the same, follows from the observation
\begin{equation*}
|\{a^{-1}\mathpzc{P}\cap \mathbb{Z}^{n}\}|=|\{\mathpzc{P}\cap (a\mathbb{Z})^{n}\}|\quad.
\end{equation*}

The volume integral of the diagonal subpolytope $P_{N}(\vec{h})$ is
\begin{equation*}
\vol(P_{N}(\vec{h}))=\Big\{\!\!\prod_{1\leq k<l\leq N}\int_{0}^{1}\ud u_{kl}\Big\}\quad\mathbf{1}_{P_{N}(\vec{h})}\big(I_{N}+\!\!\!\!\sum_{1\leq k<l\leq N}\!\!u_{kl}(B^{(kl)}-I_{N})\big)\quad.
\end{equation*}
To see that this integral covers the polytope, it suffices to see that the any symmetric stochastic matrix $A=(a_{kl})$ is decomposed in basis vectors as
\begin{equation*}
A=\big(a_{kl}\big)=I_{N}+\sum_{1\leq k<l\leq N}a_{kl}(B^{(kl)}-I_{N})\quad.
\end{equation*}

The next step is to introduce a lattice $(a\mathbb{Z})^{\binom{N}{2}}$ and count the sites inside the polytope. Each such site is a symmetric stochastic matrix with $h_{1},\ldots,h_{N}$ on the diagonal.\\
Since the volume depends continuously on the extremal points, it can be assumed without loss of generality that all $h_{j}$ are rational. This implies that a dilation factor $a^{-1}$ exists, such that all $a^{-1}(1-h_{j})=t_{j}\in\mathbb{N}$ and that the matrices that solve
\begin{equation}
\left(\begin{array}{cccc}0&b_{12}&\cdots&b_{1N}\\b_{12}&0&\cdots&b_{2N}\\\vdots&\vdots&\ddots&\vdots\\b_{1N}&b_{2N}&\cdots&0\end{array}\right)\left(\begin{array}{c}1\\1\\\vdots\\1\end{array}\right)=\left(\begin{array}{c}t_{1}\\t_{2}\\\vdots\\t_{N}\end{array}\right)\label{e:natmateq}
\end{equation}
with $t_{j},b_{jk}\in\mathbb{N}$ are to be counted. This yields a number $V_{N}(\vec{t})$. The polytope volume is then given by
\begin{equation*}
\vol(P_{N}(\vec{h}))=\lim_{a\rightarrow0}\,a^{\frac{N(N-3)}{2}}V_{N}(\frac{1-h_{1}}{a},\ldots,\frac{1-h_{N}}{a})\quad,
\end{equation*}
where 
\begin{align}
&\hspace{-8mm}V_{N}(\vec{t})=\oint_{\mathcal{C}}\frac{\ud w_{1}}{2\pi i w_{1}^{1+t_{1}}}\ldots \oint_{\mathcal{C}}\frac{\ud w_{N}}{2\pi i w_{N}^{1+t_{N}}}\,\prod_{1\leq k<l\leq N}\frac{1}{1-w_{k}w_{l}}\label{e:V2}\quad.
\end{align}
To see this, let the possible values $m$ for the matrix element $b_{jk}$ be given by the generating function
\begin{equation*}
\frac{1}{1-w_{j}w_{k}}=\sum_{m=0}^{\infty}(w_{j}w_{k})^{m}\quad.
\end{equation*}
Applying this to all matrix entries shows that $V_{N}(\vec{t})$ is given by the coefficient of the term $w_{1}^{t_{1}}w_{2}^{t_{2}}\ldots w_{N}^{t_{N}}$ in $\prod_{1\leq j<k\leq N}\frac{1}{1-w_{j}w_{k}}$. Formulating this in derivatives yields
\begin{align*}
&\hspace{-8mm}V_{N}(\vec{t})=\frac{1}{t_{1}!}\frac{\ud}{\ud w_{1}}\Big|_{w_{1}=0}^{t_{1}}\ldots \frac{1}{t_{N}!}\frac{\ud}{\ud w_{N}}\Big|_{w_{N}=0}^{t_{N}}\;\prod_{1\leq k<l\leq N}\frac{1}{1-w_{k}w_{l}}\quad.
\end{align*}
By Cauchy's integral formula the number of matrices (\ref{e:V2}) follows from this. The contour $\mathcal{C}$ encircles the origin once in the positive direction, but not the pole at $w_{k}w_{l}=1$.\\

The next step is to parametrise this contour explicitly and find a way to compute the integral for $N\rightarrow\infty$. This must be done in such a way that a combinatorial treatments is avoided. A convenient choice is
\begin{equation}
w_{j}=\sqrt{\frac{\lambda_{j}}{\lambda_{j}+1}}e^{i\varphi_{j}}\qquad\text{, with }\lambda_{j}\in\mathbb{R}_{+}\text{ and }\varphi_{j}\in[-\pi,\pi)\quad.\label{e:conpar}
\end{equation}
Later a specific value for $\lambda_{j}$ will be chosen.\\
The counting problem has now been turned into an integral over the $N$-dimensional torus
\begin{align}
&\hspace{-8mm}V_{N}(\vec{t})=\Big(\prod_{j=1}^{N}(1+\frac{1}{\lambda_{j}})^{\frac{t_{j}}{2}}\Big)(2\pi)^{-N}\int_{\mathbb{T}^{N}}\!\!\!\ud\vec{\varphi}\,e^{-i\sum_{j=1}^{N}\varphi_{j}t_{j}}\nonumber\\
&\times\Big\{\prod_{1\leq k<l\leq N}\!\!\frac{\sqrt{(1+\lambda_{k})(1+\lambda_{l})}}{\sqrt{(1+\lambda_{k})(1+\lambda_{l})}-\sqrt{\lambda_{k}\lambda_{l}}}\nonumber\\
&\times\big(1-\frac{\sqrt{\lambda_{k}\lambda_{l}}}{\sqrt{(1+\lambda_{k})(1+\lambda_{l})}-\sqrt{\lambda_{k}\lambda_{l}}} (e^{i(\varphi_{k}+\varphi_{l})}-1)\big)^{-1}\Big\}\quad,\label{e:V3}
\end{align}
where we have written $\ud\vec{\varphi}$ for $\ud\varphi_{1}\ldots \ud\varphi_{N}$.\\

The notations
\begin{equation*}
x=\sum_{j}t_{j}=\sum_{j=1}^{N}t_{j}\qquad\text{and}\qquad \sum_{k<l}(\varphi_{k}+\varphi_{l})=\sum_{1\leq k<l\leq N}(\varphi_{k}+\varphi_{l})
\end{equation*}
are used, when no doubt about $N$ can exist. When no summation bounds are mentioned, these will always be $1$ and $N$. The notation $a\ll b$ indicates that $a<b$ and $a/b\rightarrow0$.\\

The main tool for these integrals will be the stationary phase method, also called the saddle-point method. In the form used in this chapter, the exponential of a function $f$ is integrated around its maximum $\tilde{x}$, so that
\begin{align}
&\hspace{-8mm}\lim_{\Lambda\rightarrow\infty}\int\!\ud x\,e^{\Lambda f(x)}=\lim_{\Lambda\rightarrow\infty}\exp[\Lambda f(\tilde{x})]\int\!\ud x\,\exp[\frac{\Lambda f^{(2)}(\tilde{x})}{2}(x\!-\!\tilde{x})^{2}]\nonumber\\
&\times\exp[\frac{\Lambda f^{(3)}(\tilde{x})}{6}(x\!-\!\tilde{x})^{3}+\frac{\Lambda f^{(4)}(\tilde{x})}{24}(x\!-\!\tilde{x})^{4}]\nonumber\\
&=\exp[\Lambda f(\tilde{x})]\sqrt{\frac{-2\pi}{\Lambda f^{(2)}(\tilde{x})}}\nonumber\\
&\times\Big(1+\frac{15}{16}\frac{2(f^{(3)}(\tilde{x}))^{2}}{9\Lambda(-f^{(2)}(\tilde{x}))^{3}}+\frac{3}{4}\frac{f^{(4)}(\tilde{x})}{6\Lambda (f^{(2)}(\tilde{x}))^{2}}+\mathcal{O}(\Lambda^{-2})\Big)\quad.\label{e:spm}
\end{align}

Many counting problems can be computed asymptotically by the saddle-point method~\cite{mckay1,canfield2}. Often it is assumed that all $t_{j}$ are equal, but we show that it suffices to demand that they do not deviate too much from this symmetric case.

\section{Exact volume computations for low-dimensional polytopes\label{sec:Evcp}}

The computation of the volumes of the polytope of symmetric stochastic matrices is based on the counting of symmetric matrices with positive integer entries. For small $N$ it is possible to compute this volume exactly. To this end it is needed to compute
\begin{equation}
V_{N}(\vec{t})=\left(\prod_{i=1}^{N}\left[\frac{1}{t_{i}!}\frac{\ud^{t_{i}}}{\ud w_{i}^{t_{i}}}\right]_{w_{i}=0}\right)\,\Big(\prod_{k<l}\frac{1}{1-w_{k}w_{l}}\Big)\label{e:volP}\quad.
\end{equation}
As the lattice parameter $a\rightarrow 0$, the $t_{j}=h_{j}/a$ becomes larger and the volume approximates the continuous volume
\begin{equation*}
\vol\big(P_{N}(\vec{h})\big)=\lim_{a\rightarrow0}a^{\frac{N}{2}(N-3)}V_{N}(\vec{t})\quad.
\end{equation*}

The polytope for $N=3$ is the simplest with expected volume $1$. Writing \\$\partial_{m}=\left.\frac{\ud}{\ud w_{m}}\right)_{w_{m}=0}$ it is easily checked that
\begin{equation*}
\partial_{k}^{t_{k}}(1-w_{k}w_{l})^{-1}=t_{k}!w_{l}^{t_{k}}\quad,
\end{equation*}
so that (\ref{e:volP}) becomes
\begin{align}
&\hspace{-8mm}V_{3}(\vec{t})=\frac{\partial_{1}^{t_{1}}\partial_{2}^{t_{2}}}{(t_{1}!)(t_{2}!)}\sum_{c_{3}=0}^{t_{3}}w_{2}^{c_{3}}w_{1}^{t_{3}-c_{3}}(1-w_{1}w_{2})^{-1}\nonumber\\
&=\sum_{c_{3}=0}^{t_{3}}\frac{\partial_{1}^{t_{1}}}{t_{1}!}w_{1}^{t_{2}+t_{3}-2c_{3}}\vartheta(\frac{t_{2}+t_{3}-t_{1}}{2}\in\mathbb{N}_{0})\vartheta(c_{3}\leq t_{3})\vartheta(c_{3}\leq t_{2})\label{e:v31}\\
&=\vartheta(z_{1}\in\mathbb{N}_{0})\vartheta(z_{2}\in\mathbb{N}_{0})\vartheta(z_{3}\in\mathbb{N}_{0})\quad,\label{e:v32}
\end{align}
where the notation
\begin{equation*}
z_{i_{1}\ldots i_{n}}=(\frac{1}{2}\sum_{j=1}^{N}t_{j})-\sum_{k=1}^{n}t_{i_{k}}\quad
\end{equation*}
was used. The first step function in (\ref{e:v31}) arises, because $t_{1}$ must be equal to $t_{2}+t_{3}-2c_{3}$, which is only possible if $-t_{1}+t_{2}+t_{3}$ is even. This implies that $c_{3}=(-t_{1}+t_{2}+t_{3})/2$, so that the final two step functions of (\ref{e:v31}) turn into the final two step functions of (\ref{e:v32}). The volume $1$ is retrieved, if the $t$'s are chosen suitably.\\

The strategy for $N=4$ is to distribute the differentials with respect to $w_{4}$ over the functions $(1-w_{k}w_{4})^{-1}$ and then use the above result for $V_{3}(\vec{t})$. It yields
\begin{align}
&\hspace{-8mm}V_{4}(\vec{t})=a^{2}\frac{\partial_{1}^{t_{1}}}{t_{1}!}\frac{\partial_{2}^{t_{2}}}{t_{2}!}\frac{\partial_{3}^{t_{3}}}{t_{3}!}\sum_{t_{4}=c_{4,1}+c_{4,2}+c_{4,3}} w_{1}^{c_{4,1}}w_{2}^{c_{4,2}}w_{3}^{c_{4,3}}\prod_{1\leq k<l\leq 3}(1-w_{k}w_{l})^{-1}\nonumber\\
&=a^{2}\!\!\!\!\!\sum_{\substack{t_{4}=\sum_{i=1}^{3}c_{4,i} \\ 0\leq c_{4,i} \leq t_{i}}} \frac{\partial_{1}^{t_{1}-c_{4,1}}}{(t_{1}-c_{4,1})!}\frac{\partial_{2}^{t_{2}-c_{4,2}}}{(t_{2}-c_{4,2})!}\frac{\partial_{3}^{t_{3}-c_{4,3}}}{(t_{3}-c_{4,3})!}\prod_{1\leq k<l\leq 3}(1-w_{k}w_{l})^{-1}\nonumber\\
&=a^{2}\!\!\!\!\!\sum_{\substack{t_{4}=\sum_{i=1}^{3}c_{4,i} \\ 0\leq c_{4,i} \leq t_{i}}} \vartheta(c_{4,1}+(-t_{1}+t_{2}+t_{3}-t_{4})/2\in\mathbb{N}_{0})\nonumber\\
&\times\vartheta(c_{4,2}+(t_{1}-t_{2}+t_{3}-t_{4})/2\in\mathbb{N}_{0})\nonumber\\
&\times\vartheta(c_{4,3}+(t_{1}+t_{2}-t_{3}-t_{4})/2\in\mathbb{N}_{0})\label{e:v41}\quad.
\end{align}

The volume of the polytope is symmetric under exchange of the variables. Although manifest in (\ref{e:volP}), due to the explicit order of applying the differentials in the calculation of $V_{4}(\vec{t})$, this is not straightforward anymore. It must still be there.\\

Assuming for the moment that $t_{4}$ is the smallest of $t$'s and that $z_{12}$ and $z_{13}$ are negative, while $z_{14}$ is positive. Choosing $a$ sufficiently small, these are all even, so that the step functions in (\ref{e:v41}) are satisfied. Counting the number of configurations $c_{4,1}+c_{4,2}+c_{4,3}=t_{4}$ yields the answer $(t_{4}+1)(t_{4}+2)/2$.\\
If $z_{14}$ were also negative, $c_{4,1}$ had to be larger than or equal to $-z_{14}$, leaving $t_{4}+z_{14}+1=z_{1}+1$ possibilities. In that case, there remain for $c_{4,2}$ $z_{1}+1$ possible values, so that in total there are $(z_{1}+1)^{2}/2$ configurations.\\
Because the volume of the polytope is symmetric under permutation of the arguments, this fully determines the volume
\begin{align}
&\hspace{-8mm}\vol\big(P_{4}(\vec{h})\big)=\frac{h_{1}^{2}}{2}\,\vartheta(s_{12})\vartheta(s_{13})\vartheta(s_{14})+\frac{s_{1}^{2}}{2}\,\vartheta(-s_{12})\vartheta(-s_{13})\vartheta(-s_{14})\nonumber\\
&+\frac{h_{2}^{2}}{2}\,\vartheta(s_{12})\vartheta(-s_{13})\vartheta(-s_{14})+\frac{s_{2}^{2}}{2}\,\vartheta(-s_{12})\vartheta(s_{13})\vartheta(s_{14})\nonumber\\
&+\frac{h_{3}^{2}}{2}\,\vartheta(-s_{12})\vartheta(s_{13})\vartheta(-s_{14})+\frac{s_{3}^{2}}{2}\,\vartheta(s_{12})\vartheta(-s_{13})\vartheta(s_{14})\nonumber\\
&+\frac{h_{4}^{2}}{2}\,\vartheta(-s_{12})\vartheta(-s_{13})\vartheta(s_{14})+\frac{s_{4}^{2}}{2}\,\vartheta(s_{12})\vartheta(s_{13})\vartheta(-s_{14})\quad,\label{e:v42}
\end{align}
where the scaled parameters
\begin{equation*}
s_{i_{1}\ldots i_{n}}=az_{i_{1}\ldots i_{n}}
\end{equation*}
are used.

\section{Integrating the central part\label{sec:icp}}
The integrals in (\ref{e:V3}) are too difficult to compute in full generality. A useful approximation can be obtained from the observation that the integrand
\begin{equation}
\big|\frac{1}{1-\mu(e^{iy}-1)}\big|^{2}=\frac{1}{1-2\mu(\mu+1)(\cos(y)-1)}\qquad\text{for }y\in(-2\pi,2\pi)\label{e:intfac}
\end{equation}
is concentrated in a neighbourhood of the origin and the antipode $y=\pm2\pi$, where it takes the value $1$. This is plotted in Figure~\ref{f:estplot}. For small $y$ and $\mu y$ the absolute value of the integrand factor can be written as
\begin{equation}
\big|\frac{1}{1-\mu(e^{iy}-1)}\big|=\sqrt{\frac{1}{1+\mu(\mu+1)y^{2}}}\big(1+\mathcal{O}(y^{4})\big)\,\quad.\label{e:fracest}
\end{equation}

\begin{figure}[!hb]\centering
\includegraphics[width=0.7\textwidth]{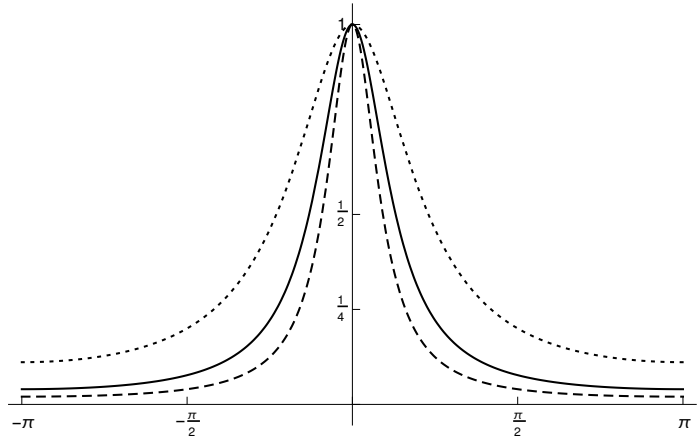}
\caption{The absolute value squared of the integrand factor (\ref{e:intfac}) for $\mu=1,2$ and $3$ in dotted, continuous and dashed lines respectively.\label{f:estplot}}
\end{figure}

It is concentrated in a small region around the origin and the antipode. The form of the region is assumed to be $[\delta_{N},\delta_{N}]^{N}$ with 
\begin{equation*}
\delta_{N}=\frac{N^{-\alpha}\zeta_{N}}{\min_{j}\{\lambda_{j}\}}\quad,
\end{equation*}
where $\alpha\in(0,1/2)$ and $\zeta_{N}$ tends slowly to infinity. In the remainder of this paragraph the integral inside this box will be computed.\\
To this end, a lower bound is introduced. Below this threshold we do not strive for accuracy. The aim is thus to find the asymptotic number $V_{N}(\vec{t})$ for configurations $\vec{t}$, such that this number is larger than the Lower bound.
\begin{dfnt}\emph{Lower bound}\label{dfnt:lb}\\
For $N$, $\alpha\in(0,1/2)$, $t_{j}\in\mathbb{N}$ and $\lambda_{j}\in\mathbb{R}_{+}^{N}$ for $j=1,\ldots,N$ we define the Lower bound by
\begin{align*}
&\hspace{-8mm}\mathcal{E}_{\alpha} = (2\pi\lambda(\lambda\!+\!1)N)^{-\frac{N}{2}}\Big(\prod_{j=1}^{N}(1\!+\!\frac{1}{\lambda_{j}})^{\frac{t_{j}}{2}}\Big)\Big(\prod_{k<l}\frac{\sqrt{(1\!+\!\lambda_{k})(1\!+\!\lambda_{l})}}{\sqrt{(1\!+\!\lambda_{k})(1\!+\!\lambda_{l})}\!-\!\sqrt{\lambda_{k}\lambda_{l}}}\Big)\\
&\times\exp[\frac{14\lambda^{2}+14\lambda-1}{12\lambda(\lambda+1)}]\exp[-N^{1-2\alpha}]\quad,
\end{align*}
where $\lambda=N^{-1}\sum_{j}\lambda_{j}$.
\end{dfnt}

The integral in $[-\delta_{N},\delta_{N}]^{N}$ can now be cast into a simpler form, where the size $\delta_{N}$ of this box can be used as an expansion parameter. The expansion used is
\begin{equation}
\frac{1}{1-\mu(\exp[iy]-1)}=\exp[\sum_{j=1}^{k}A_{j}(iy)^{j}]+\mathcal{O}(y^{k+1}(1+\mu)^{k+1})\label{e:fracexpa}\quad.
\end{equation}
The coefficients $A_{j}(\mu)$ (or $A_{j}$ if the argument is clear) are polynomials in $\mu$ of degree $j$. They are obtained as the polylogarithms
\begin{equation*}
A_{n}(\mu)=\frac{(-1)^{n}}{n!}\; \Li_{1-n}(1+\frac{1}{\mu})\quad.
\end{equation*}
The first four coefficients are
\begin{align}
&A_{1}=\mu\quad;\quad A_{2}=\frac{\mu}{2}(\mu+1)\quad;\quad A_{3}=\frac{\mu}{6}(\mu+1)(2\mu+1)\quad\nonumber\\
&\text{and}\quad A_{4}=\frac{\mu}{24}(\mu+1)(6\mu^{2}+6\mu+1)\quad.\label{e:coef}
\end{align}
The value of the parameter $\mu$ in the above formules can be approximated. Assuming that $\varepsilon_{k}$ is small compared to $\lambda$ and writing $\varepsilon=\max_{k}\varepsilon_{k}$, this is
\begin{align*}
&\hspace{-8mm}\frac{\sqrt{(\lambda+\varepsilon_{k})(\lambda+\varepsilon_{l})}}{\sqrt{(\lambda+\varepsilon_{k}+1)(\lambda+\varepsilon_{l}+1)}-\sqrt{(\lambda+\varepsilon_{k})(\lambda+\varepsilon_{l})}}\approx\lambda+\frac{\varepsilon_{k}+\varepsilon_{l}}{2}\\
&-\frac{2\lambda+1}{8\lambda(\lambda+1)}(\varepsilon_{k}-\varepsilon_{l})^{2}+\frac{2\lambda^{2}+2\lambda+1}{16\lambda^{2}(\lambda+1)^{2}}(\varepsilon_{k}^{3}-\varepsilon_{k}^{2}\varepsilon_{l}-\varepsilon_{k}\varepsilon_{l}^{2}+\varepsilon_{l}^{3})\\
&+\mathcal{O}(\frac{\varepsilon^{4}}{\lambda^{3}})\quad.
\end{align*}

Applying this in combination with (\ref{e:fracexpa}) produces the combinations
\begin{align}
&\hspace{-8mm}\sum_{k<l}(\varphi_{k}+\varphi_{l})\cdot(\frac{\sqrt{\lambda_{k}\lambda_{l}}}{\sqrt{(1+\lambda_{k})(1+\lambda_{l})}-\sqrt{\lambda_{k}\lambda_{l}}})\nonumber\\
&=\sum_{j=1}^{N}\varphi_{j}\big[\frac{N-2}{2}\lambda_{j}+\frac{N}{2}\lambda-B_{1}\big(N\varepsilon_{j}^{2}+\sum_{m}\varepsilon_{m}^{2}\big)\nonumber\\
&+C_{1}\big(N\varepsilon_{j}^{3}-\varepsilon_{j}\sum_{m}\varepsilon_{m}^{2}+\sum_{m}\varepsilon_{m}^{3}\big)\big]\times(1+\mathcal{O}(N\frac{\varepsilon^{4}}{\lambda^{4}}))\qquad;\nonumber\\
&\hspace{-8mm}\sum_{k<l}(\varphi_{k}+\varphi_{l})^{2}\cdot A_{2}(\frac{\sqrt{\lambda_{k}\lambda_{l}}}{\sqrt{(1+\lambda_{k})(1+\lambda_{l})}-\sqrt{\lambda_{k}\lambda_{l}}})\nonumber\\
&=\big[\sum_{j=1}^{N}\varphi_{j}^{2}\Big((N-2)A_{2}+\varepsilon_{j}B_{2}(N-4)-(N-4)C_{2}\varepsilon_{j}^{2}-C_{2}\sum_{m}\varepsilon_{m}^{2}\Big)\nonumber\\
&+\!\sum_{j=1}^{N}\varphi_{j}\Big(\!A_{2}\!\sum_{m}\!\varphi_{m}\!+\!2B_{2}\varepsilon_{j}\!\sum_{m}\!\varphi_{m}\!-\!2C_{2}\varepsilon_{j}^{2}\!\sum_{m}\!\varphi_{m}\!+\!D_{2}\varepsilon_{j}\!\sum_{m}\!\varepsilon_{m}\varphi_{m}\!\Big)\big]\nonumber\\
&\times(1+\mathcal{O}(N\frac{\varepsilon^{3}}{\lambda^{3}}))\qquad;\nonumber\\
&\hspace{-8mm}\sum_{k<l}(\varphi_{k}+\varphi_{l})^{3}\cdot A_{3}(\frac{\sqrt{\lambda_{k}\lambda_{l}}}{\sqrt{(1+\lambda_{k})(1+\lambda_{l})}-\sqrt{\lambda_{k}\lambda_{l}}})\nonumber\\
&=\big[\sum_{j}\varphi_{j}^{3}A_{3}(N-4)+3A_{3}\sum_{j}\varphi_{j}^{2}\sum_{m}\varphi_{m}\big]\times(1+\mathcal{O}(\frac{\varepsilon}{\lambda}))\qquad\text{and}\nonumber\\
&\hspace{-8mm}\sum_{k<l}(\varphi_{k}+\varphi_{l})^{4}\cdot A_{4}(\frac{\sqrt{\lambda_{k}\lambda_{l}}}{\sqrt{(1+\lambda_{k})(1+\lambda_{l})}-\sqrt{\lambda_{k}\lambda_{l}}})\nonumber\\
&=\big[\sum_{j}\varphi_{j}^{4}A_{4}(N-8)+4A_{4}\sum_{j}\varphi_{j}^{3}\sum_{m}\varphi_{m}+3A_{4}(\sum_{j}\varphi_{j}^{2})^{2}\big]\nonumber\\
&\times(1+\mathcal{O}(\frac{\varepsilon}{\lambda}))\quad.\label{e:comb}
\end{align}

Here we used the additional combinations
\begin{align}
&B_{1} = \frac{2\lambda+1}{8\lambda(\lambda+1)}\quad;\quad C_{1}=\frac{2\lambda^{2}+2\lambda+1}{16\lambda^{2}(\lambda+1)^{2}}\quad;\quad B_{2}=\frac{2\lambda+1}{4}\nonumber\\
&C_{2}=\frac{2\lambda^{2}+2\lambda+1}{16\lambda(\lambda+1)}\quad;\quad D_{2}=\frac{6\lambda^{2}+6\lambda+1}{8\lambda(\lambda+1)}\label{e:coef2}
\end{align}
to simplify the notation.\\

The simplest way to compute this integral is to ensure that the linear part of the exponent is small.
Splitting $\lambda_{j}=\lambda+\varepsilon_{j}$ and choosing the value 
\begin{align*}
&\hspace{-8mm}\varepsilon_{j}=\frac{2}{N-2}\big(t_{j}-\lambda(N-1)\big)
\end{align*}
is done therefore. Combined with the assumption that 
\begin{equation*}
x=\sum_{j}t_{j}=\lambda N(N-1)\quad,
\end{equation*}
this implies that $\sum_{m}\varepsilon_{m}=0$. Assuming furthermore that 
\begin{equation*}
|t_{j}-\lambda(N-1)|\ll \lambda N^{\frac{1}{2}+\omega}\quad,
\end{equation*}
the error terms $|\varepsilon/\lambda|\ll N^{-\frac{1}{2}+\omega}$ follow.\\

The first step now is to focus on the integral inside the box $[-\delta_{N},\delta_{N}]^{N}$, simplify and calculate this.
\begin{rmk}
The estimates in Lemma~\ref{l:funris} cause the integral (\ref{e:V3}) to depend nontrivially on $\lambda$. For that reason $\lambda$ is explicitly mentioned as an argument.
\end{rmk}
\begin{lemma}\label{l:diff}
Assume that $K,N\in\mathbb{N}$, $\omega,\alpha\in\mathbb{R}_{+}$ are chosen such that \\$\omega\in(0,\frac{\log(K\alpha-2) + \log\log N}{4\log N})$, $\alpha\in(0,\frac{1}{4}-\omega)$ and $K>2/\alpha+1$. Define
\begin{equation*}
\delta_{N}=\frac{N^{-\alpha}\zeta_{N}}{\min \{\lambda_{j}\}}\quad,
\end{equation*}
so that $\zeta_{N}\rightarrow\infty$ and $N^{-\delta}\zeta_{N}\rightarrow0$ for any $\delta>0$, when $N\rightarrow\infty$. If $x=\sum_{j}t_{j}$, the average matrix entry $\lambda=\frac{x}{N(N-1)}$ and
\begin{equation*}
\lim_{N\rightarrow\infty}\frac{t_{j}-\lambda(N-1)}{\lambda N^{\frac{1}{2}+\omega}}=0\quad\text{for }j=1,\ldots,N\qquad,
\end{equation*}
then the integral
\begin{align*}
&\hspace{-8mm}V_{N}(\vec{t})=\Big(\prod_{j=1}^{N}(1+\frac{1}{\lambda_{j}})^{\frac{t_{j}}{2}}\Big)(2\pi)^{-N}\int_{[-\delta_{N},\delta_{N}]^{N}}\!\!\!\ud\vec{\varphi}\,e^{-i\sum_{j=1}^{N}\varphi_{j}t_{j}}\\
&\times\Big\{\prod_{1\leq k<l\leq N}\!\!\frac{\sqrt{(1+\lambda_{k})(1+\lambda_{l})}}{\sqrt{(1+\lambda_{k})(1+\lambda_{l})}-\sqrt{\lambda_{k}\lambda_{l}}}\\
&\times\big(1-\frac{\sqrt{\lambda_{k}\lambda_{l}}}{\sqrt{(1+\lambda_{k})(1+\lambda_{l})}-\sqrt{\lambda_{k}\lambda_{l}}} (e^{i(\varphi_{k}+\varphi_{l})}-1)\big)^{-1}\Big\}
\end{align*}
is given by
\begin{align*}
&\hspace{-8mm}V_{N}(\vec{t};\lambda)=\frac{2}{(2\pi)^{N}}\Big(\prod_{j=1}^{N}(1+\frac{1}{\lambda_{j}})^{\frac{t_{j}}{2}}\Big)\cdot\Big(\prod_{1\leq k<l\leq N}\!\!\frac{\sqrt{(1+\lambda_{k})(1+\lambda_{l})}}{\sqrt{(1+\lambda_{k})(1+\lambda_{l})}-\sqrt{\lambda_{k}\lambda_{l}}}\Big)\\
&\times\int_{[-\delta_{N},\delta_{N}]^{N}}\!\!\!\!\!\!\!\!\!\!\!\!\!\ud\vec{\varphi}\,\exp[-i\!\sum_{j}\!\varphi_{j}t_{j}]\\
&\times\exp[\sum_{n=1}^{K-1}\!i^{n}\!\sum_{k<l}\!A_{n}\big(\frac{\sqrt{\lambda{k}\lambda_{l}}}{\sqrt{(1\!+\!\lambda{k})(1\!+\!\lambda_{l})}-\sqrt{\lambda{k}\lambda_{l}}}\big)\cdot\big(\varphi_{k}\!+\!\varphi_{l}\big)^{n}]+\mathcal{D}\,,
\end{align*}
up to a difference $\mathcal{D}$ that satisfies
\begin{align*}
&\hspace{-8mm}|\mathcal{D}|\leq \mathcal{O}(N^{2-K\alpha})\frac{\sqrt{2}(1+\lambda)^{\binom{N}{2}}}{(2\pi\lambda(\lambda+1)N)^{\frac{N}{2}}}\big(1+\frac{1}{\lambda}\big)^{\frac{x}{2}}\exp[\frac{10\lambda^{2}+10\lambda+1}{4\lambda(\lambda+1)}]\\
&\exp[\frac{-1}{2\lambda(\lambda+1)N}\sum_{m}(t_{m}-\lambda(N-1))^{2}]\\
&\times\exp[\frac{3}{4\lambda^{2}(\lambda+1)^{2}N^{2}}\sum_{m}(t_{m}-\lambda(N-1))^{2}]\\
&\times\exp[\frac{2\lambda+1}{6\lambda^{2}(\lambda+1)^{2}N^{2}}\sum_{m}(t_{m}-\lambda(N-1))^{3}]\\
&\times\exp[\frac{6\lambda^{2}+6\lambda+1}{24\lambda^{3}(\lambda+1)^{3}N^{3}}\sum_{m}(t_{m}-\lambda(N-1))^{4}]\\
&\times\exp[\frac{6\lambda^{2}+6\lambda+1}{8\lambda^{3}(\lambda+1)^{3}N^{4}}\big(\sum_{m}(t_{m}-\lambda(N-1))^{2}\big)^{2}]\quad.
\end{align*}
\end{lemma}
\begin{proof}
To the fraction 
\begin{equation*}
\Big(1-\frac{\sqrt{\lambda_{k}\lambda_{l}}}{\sqrt{(1+\lambda_{k})(1+\lambda_{l})}-\sqrt{\lambda_{k}\lambda_{l}}} (e^{i(\varphi_{k}+\varphi_{l})}-1)\Big)^{-1}
\end{equation*}
in the integral (\ref{e:V3}) the expansion (\ref{e:fracexpa}) in combination with (\ref{e:coef}) and (\ref{e:coef2}) is applied. To prove that contributions in (\ref{e:fracexpa}) of $K$-th order or higher are irrelevant, we put these in the exponential $\exp[h(x)]$. To estimate their contribution, the estimate
\begin{equation*}
|\int\ud x\,e^{f(x)}(e^{h(x)}-1)|\leq \mathcal{O}(\sup_{x} |e^{h(x)}-1|)\cdot \int \ud x\, |e^{f(x)}|
\end{equation*}
is applied to the integral. Taking the absolute value of the integrand sets the imaginary parts of the exponential to zero. In terms of (\ref{e:coef}) and (\ref{e:coef2}) this means that $A_{3}$, $B_{1}$ and $C_{1}$ are set to zero. This integral is calculated in Lemma~\ref{l:labval}. Taking this result and setting these coefficients to zero completes the proof.
\end{proof}

\begin{lemma}\label{l:labval}
Assume that $K,N\in\mathbb{N}$, $\omega,\alpha\in\mathbb{R}_{+}$ are chosen such that \\$\omega\in(0,\frac{\log(K\alpha-2) + \log\log N}{4\log N})$, $\alpha\in(0,\frac{1}{4}-\omega)$ and $K>2/\alpha+1$. Define
\begin{equation*}
\delta_{N}=\frac{N^{-\alpha}\zeta_{N}}{\min \{\lambda_{j}\}}\quad,
\end{equation*}
so that $\zeta_{N}\rightarrow\infty$ and $N^{-\delta}\zeta_{N}\rightarrow0$ for any $\delta>0$, when $N\rightarrow\infty$. If $x=\sum_{j}t_{j}$, the average matrix entry $\lambda=\frac{x}{N(N-1)}>\frac{C}{\log N}$ and
\begin{equation}
\lim_{N\rightarrow\infty}\frac{t_{j}-\lambda(N-1)}{\lambda N^{\frac{1}{2}+\omega}}=0\quad\text{for }j=1,\ldots,N\qquad,
\end{equation}
then the integral
\begin{align*}
&\hspace{-8mm}V_{N}(\vec{t};\lambda)=\frac{2}{(2\pi)^{N}}\Big(\prod_{j=1}^{N}(1+\frac{1}{\lambda_{j}})^{\frac{t_{j}}{2}}\Big)\cdot\Big(\prod_{1\leq k<l\leq N}\!\!\frac{\sqrt{(1\!+\!\lambda_{k})(1\!+\!\lambda_{l})}}{\sqrt{(1\!+\!\lambda_{k})(1\!+\!\lambda_{l})}\!-\!\sqrt{\lambda_{k}\lambda_{l}}}\Big)\\
&\times\int_{[-\delta_{N},\delta_{N}]^{N}}\!\!\!\!\!\!\!\!\!\!\!\!\!\ud\vec{\varphi}\,\exp[-i\!\sum_{j}\!\varphi_{j}t_{j}]\\
&\times\exp[\sum_{n=1}^{K-1}\!i^{n}\!\sum_{k<l}\!A_{n}\big(\frac{\sqrt{\lambda{k}\lambda_{l}}}{\sqrt{(1\!+\!\lambda{k})(1\!+\!\lambda_{l})}-\sqrt{\lambda{k}\lambda_{l}}}\big)\cdot\big(\varphi_{k}\!+\!\varphi_{l}\big)^{n}]
\end{align*}
is asymptotically ($N\rightarrow\infty$) given by
\begin{align*}
&\hspace{-8mm}V_{N}(\vec{t};\lambda)=\frac{\sqrt{2}}{(2\pi\lambda(\lambda\!+\!1)N)^{\frac{N}{2}}}\big[\!\prod_{n}(1\!+\!\frac{1}{\lambda_{n}}\!)^{\frac{t_{n}}{2}}\!\big]\big[\!\prod_{k<l}\!\frac{\sqrt{(1\!+\!\lambda_{k})(1\!+\!\lambda_{l})}}{\sqrt{(1\!+\!\lambda_{k})(1\!+\!\lambda_{l})}\!-\!\sqrt{\lambda_{k}\lambda_{l}}}\big]\\
&\times\exp[\frac{14\lambda^{2}\!+\!14\lambda\!-\!1}{12\lambda(\lambda+1)}]\exp[\frac{\sum_{m}\varepsilon_{m}^{2}}{16\lambda^{2}(\lambda\!+\!1)^{2}}]\exp[-\frac{(2\lambda\!+\!1)^{2}(\sum_{m}\varepsilon_{m}^{2})^{2}}{128\lambda^{3}(\lambda\!+\!1)^{3}}]\\
&\times\exp[-\frac{(2\lambda+1)^{2}N}{128\lambda^{3}(\lambda+1)^{3}}\sum_{m}\varepsilon_{m}^{4}]\\
&\times\Big(1+\mathcal{O}(N^{-\frac{1}{2}+6\omega}+N^{2+\frac{1}{3C}-K\alpha}\exp[N^{4\omega}])\Big)\\
&=\frac{\sqrt{2}(1+\lambda)^{\binom{N}{2}}}{(2\pi\lambda(\lambda+1)N)^{\frac{N}{2}}}\big(1+\frac{1}{\lambda}\big)^{\frac{x}{2}}\exp[\frac{14\lambda^{2}+14\lambda-1}{12\lambda(\lambda+1)}]\\
&\times\exp[\frac{-1}{2\lambda(\lambda+1)N}\sum_{m}(t_{m}-\lambda(N-1))^{2}]\\
&\times\exp[\frac{-1}{\lambda(\lambda+1)N^{2}}\sum_{m}(t_{m}-\lambda(N-1))^{2}]\\
&\times\exp[\frac{2\lambda+1}{6\lambda^{2}(\lambda+1)^{2}N^{2}}\sum_{m}(t_{m}-\lambda(N-1))^{3}]\\
&\times\exp[-\frac{3\lambda^{2}+3\lambda+1}{12\lambda^{3}(\lambda+1)^{3}N^{3}}\sum_{m}(t_{m}-\lambda(N-1))^{4}]\\
&\times\exp[\frac{1}{4\lambda^{2}(\lambda+1)^{2}N^{4}}\big(\sum_{m}(t_{m}-\lambda(N-1))^{2}\big)^{2}]\\
&\times\Big(1+\mathcal{O}(N^{-\frac{1}{2}+6\omega}+N^{2+\frac{1}{3C}-K\alpha}\exp[N^{4\omega}])\Big)\quad.
\end{align*}
This is much larger than the Lower bound from Definition~\ref{dfnt:lb}
\begin{equation*}
\frac{V_{N}(\vec{t};\lambda)}{\mathcal{E}_{\alpha}}\rightarrow\infty\quad.
\end{equation*}
\end{lemma}
\begin{proof}
Define $\varepsilon_{j}=\frac{2}{N-2}\big(t_{j}-\lambda(N-1)\big)$ and assume that $|\varepsilon_{j}|\leq \lambda N^{-\frac{1}{2}+\omega}$ with $0<\omega<1/14$. It follows that $\sum_{j}\varepsilon_{j}=0$.\\

To the integral $V_{N}(\vec{t};\lambda)$ the expansion (\ref{e:fracexpa}) for $k=4$ in combination with (\ref{e:coef}) and (\ref{e:coef2}) is applied. It will follow automatically that the higher orders ($K>5$) in this expansion will yield asymptotically irrelevant factors. This expansion produces the combinations (\ref{e:comb}). 
Introducing $\delta$-functions for $S_{1}=\sum_{m}\varphi_{m}$, \\$S_{2}=\sum_{m}\varphi_{m}^{2}$, $T_{3}=\sum_{m}\varepsilon_{m}\varphi_{m}$ and $T_{4}=\sum_{m}\varepsilon_{m}^{2}\varphi_{m}$ through their Fourier representation yields the integral
\begin{align}
&\hspace{-8mm}V_{N}(\vec{t};\lambda)=\frac{2}{(2\pi)^{N}}\big[\prod_{n}(1\!+\!\frac{1}{\lambda_{n}})^{\frac{t_{n}}{2}}\big]\big[\prod_{k<l}\!\frac{\sqrt{(1+\lambda_{k})(1+\lambda_{l})}}{\sqrt{(1\!+\!\lambda_{k})(1\!+\!\lambda_{l})}\!-\!\sqrt{\lambda_{k}\lambda_{l}}}\big]\!\!\int\!\!\ud\tau_{1}\!\!\int\!\!\ud S_{1}\nonumber\\
&\times\int\!\!\ud T_{3}\!\!\int\!\!\ud\tau_{3}\!\!\int\!\!\ud T_{4}\!\!\int\!\!\ud \tau_{4}\!\!\int\!\!\ud S_{2}\!\!\int\!\!\ud\tau_{2}\,\exp[2\pi i(\tau_{1}S_{1}\!+\!\tau_{2}S_{2}\!+\!\tau_{3}T_{3}\!+\!\tau_{4}T_{4})\nonumber\\
&-A_{2}S_{1}^{2}-2B_{2}S_{1}T_{3}+2C_{2}S_{1}T_{4}-2D_{2}T_{3}^{2}+3A_{4}S_{2}^{2}]\nonumber\\
&\times\Big\{\prod_{j}\int_{-\delta_{N}}^{\delta_{N}}\!\!\!\!\!\ud\varphi_{j}\,\exp\big[i\varphi_{j}\big(-B_{1}N\varepsilon_{j}^{2}-B_{1}\sum_{m}\varepsilon_{m}^{2}-2\pi\tau_{1}-3A_{3}S_{2}\nonumber\\
&+C_{1}N\varepsilon_{j}^{3}-C_{1}\varepsilon_{j}\sum_{m}\varepsilon_{m}^{2}+C_{1}\sum_{m}\varepsilon_{m}^{3}-2\pi\tau_{3}\varepsilon_{j}\big)\big]\nonumber\\
&\times \exp\!\big[\!-\!\varphi_{j}^{2}\big(A_{2}(N\!-\!2)\!+\!B_{2}(N\!-\!4)\varepsilon_{j}\!-\!(N\!-\!4)C_{2}\varepsilon_{j}^{2}\!-\!C_{2}\!\sum_{m}\!\varepsilon_{m}^{2}\!+\!2\pi i\tau_{2}\big)\big]\nonumber\\
&\times \exp\big[-i\varphi_{j}^{3}\big(A_{3}(N-4)+4iA_{4}S_{1}\big)\big]\nonumber\\
&\times \exp\big[\varphi_{j}^{4}\big(A_{4}(N\!-\!8)\big)\big]\,\Big\}\label{e:labval1}\quad.
\end{align}
To ensure that that overall error consists of asymptotically irrelevant factors only, the $\varphi_{j}$-integral must be computed up to $\mathcal{O}(N^{-1})$. Dividing the integration parameter $\varphi_{j}$ by $\sqrt{A_{2}(N-2)}$ shows that the $\varphi_{j}$-integral is of the form
\begin{align}
&\hspace{-8mm}\frac{1}{\sqrt{A_{2}(N-2)}}\int_{-\delta_{N}\sqrt{A_{2}(N-2)}}^{\delta_{N}\sqrt{A_{2}(N-2)}}\!\!\!\!\!\!\ud\varphi\,\exp[\frac{i\varphi Q_{1}}{\sqrt{A_{2}(N-2)}}-\varphi^{2}Q_{2}]\nonumber\\
&\times\exp[-\frac{i\varphi^{3}Q_{3}}{(A_{2}(N-2))^{3/2}}+\frac{\varphi^{4}Q_{4}}{(A_{2}(N-2))^{2}}]\nonumber\\
&=\sqrt{\frac{\pi}{A_{2}(N-2)}}\big[Q_{2}+\frac{3iQ_{3}\tilde{\varphi}}{(A_{2}(N-2))^{3/2}}\big]^{-\frac{1}{2}}\nonumber\\
&\times\exp[\frac{iQ_{1}\tilde{\varphi}}{\sqrt{A_{2}(N-2)}}-Q_{2}\tilde{\varphi}^{2}-\frac{iQ_{3}\tilde{\varphi}^{3}}{(A_{2}(N-2))^{3/2}}+\frac{Q_{4}\tilde{\varphi}^{4}}{A_{2}^{2}(N-2)^{2}}]\nonumber\\
&\times\big\{1-\frac{15Q_{3}^{2}}{16(A_{2}(N-2))^{3}(Q_{2}+\frac{3iQ_{3}\tilde{\varphi}}{(A_{2}(N-2))^{3/2}})^{3}}\nonumber\\
&+\frac{3Q_{4}}{4A_{2}^{2}(N-2)^{2}(Q_{2}+\frac{3iQ_{3}\tilde{\varphi}}{(A_{2}(N-2))^{3/2}})^{2}}\big\}\quad,\label{e:labval2}
\end{align}
which is calculated by the (\ref{e:spm}) around the maximum $\tilde{\varphi}$ of the integrand. Observing that $Q_{1}=\mathcal{O}(N^{2\omega})$, $Q_{2}=\mathcal{O}(1)$ and $Q_{3,4}=\mathcal{O}(N)$, shows that
\begin{equation*}
\tilde{\varphi}=\frac{iQ_{1}}{2Q_{2}\sqrt{A_{2}(N-2)}}+\mathcal{O}(N^{-\frac{3}{2}+4\omega})=\mathcal{O}(N^{-\frac{1}{2}+2\omega})
\end{equation*}
is sufficient for the desired accuracy. This implies that
\begin{align*}
&\hspace{-8mm}\exp[\frac{iQ_{1}\tilde{\varphi}}{\sqrt{A_{2}(N-2)}}-Q_{2}\tilde{\varphi}^{2}-\frac{iQ_{3}\tilde{\varphi}^{3}}{(A_{2}(N-2))^{3/2}}]\\
&=\exp\big[-\frac{Q_{1}^{2}}{4A_{2}(N-2)}\big]\times\big(1+\mathcal{O}(N^{-2+6\omega})\big)\quad.
\end{align*}
The terms in square and curly brackets are then rewritten using
\begin{equation*}
\frac{1}{\sqrt{1+y}}\approx e^{-\frac{y}{2}+\frac{y^{2}}{4}}\qquad\text{and}\qquad 1+z\approx \exp[z]
\end{equation*}
respectively. Using the same order of factors as in (\ref{e:labval2}), the result of the $\varphi_{j}$-integral is
\begin{align*}
&\hspace{-8mm}\sqrt{\frac{\pi}{A_{2}(N-2)}}\exp\Big[-\frac{B_{2}(N-4)\varepsilon_{j}}{2A_{2}(N-2)}+\frac{C_{2}(N-4)\varepsilon_{j}^{2}}{2A_{2}(N-2)}+\frac{C_{2}\sum_{m}\varepsilon_{m}^{2}}{2A_{2}(N-2)}\\
&-\frac{i\pi\tau_{2}}{A_{2}(N-2)}-\frac{3A_{3}B_{1}N(N-4)\varepsilon_{j}^{2}}{4A_{2}^{2}(N-2)^{2}}-\frac{3A_{3}B_{1}(N-4)\sum_{m}\varepsilon_{m}^{2}}{4A_{2}^{2}(N-2)^{2}}\\
&-\frac{3\pi\tau_{1}A_{3}(N-4)}{2A_{2}^{2}(N-2)^{2}}-\frac{9A_{3}^{2}S_{2}(N-4)}{4A_{2}^{2}(N-2)^{2}}\Big]\\
&\times\exp\big[\frac{B_{2}^{2}(N-4)^{2}\varepsilon_{j}^{2}}{4A_{2}^{2}(N-2)^{2}}\big]\exp\Big[-\frac{B_{1}^{2}N^{2}\varepsilon_{j}^{4}}{4A_{2}(N-2)}-\frac{B_{1}^{2}N\varepsilon_{j}^{2}\sum_{m}\varepsilon_{m}^{2}}{2A_{2}(N-2)}\\
&-\frac{\pi\tau_{1}B_{1}N\varepsilon_{j}^{2}}{A_{2}(N-2)}-\frac{3A_{3}B_{1}S_{2}N\varepsilon_{j}^{2}}{2A_{2}(N-2)}-\frac{B_{1}^{2}(\sum_{m}\varepsilon_{m}^{2})^{2}}{4A_{2}(N-2)}-\frac{\pi\tau_{1}B_{1}\sum_{m}\varepsilon_{m}^{2}}{A_{2}(N-2)}\\
&-\frac{3A_{3}B_{1}S_{2}\sum_{m}\varepsilon_{m}^{2}}{2A_{2}(N-2)}-\frac{\pi^{2}\tau_{1}^{2}}{A_{2}(N-2)}-\frac{3\pi\tau_{1}A_{3}S_{2}}{A_{2}(N-2)}-\frac{9A_{3}^{2}S_{2}^{2}}{4A_{2}(N-2)}\Big]\\
&\times\exp\big[-\frac{15A_{3}^{2}(N-4)^{2}}{16A_{2}^{3}(N-2)^{3}}\big]\exp\big[\frac{3A_{4}(N-8)}{4A_{2}^{2}(N-2)^{2}}\big]\;.
\end{align*}
Integrating now $\tau_{2}$ yields a delta function that assigns the value
\begin{equation*}
S_{2}=\frac{N}{2A_{2}(N-2)}\quad.
\end{equation*}
Doing the same for $\tau_{3}$ and $\tau_{4}$ yields $T_{3}=0$ and $T_{4}=0$. The $S_{1}$-integral is
\begin{equation*}
\int \ud S_{1}\,\exp[2\pi i\tau_{1}S_{1}-A_{2}S_{1}^{2}]=\sqrt{\frac{\pi}{A_{2}}}\exp[-\frac{\pi^{2}\tau_{1}^{2}}{A_{2}}]\quad.
\end{equation*}
and the final integral
\begin{align*}
&\hspace{-8mm}\int \ud\tau_{1}\,\exp\big[-\frac{2\pi^{2}(N-1)\tau_{1}^{2}}{A_{2}(N-2)}-\frac{3\pi\tau_{1}A_{3}N}{A_{2}^{2}(N-2)}-\frac{2\pi\tau_{1}B_{1}N\sum_{m}\varepsilon_{m}^{2}}{A_{2}(N-2)}\big]\\
&=\sqrt{\frac{A_{2}(N-2)}{2\pi (N-1)}}\exp\big[\frac{9A_{3}^{2}N^{2}}{8A_{2}^{3}(N-1)(N-2)}+\frac{B_{1}^{2}N^{2}(\sum_{m}\varepsilon_{m}^{2})^{2}}{2A_{2}(N-1)(N-2)}\big]\\
&\times\exp\big[\frac{3A_{3}B_{1}N^{2}\sum_{m}\varepsilon_{m}^{2}}{2A_{2}^{2}(N-1)(N-2)}\big]\quad.
\end{align*}
Putting this all together yields
\begin{align}
&\hspace{-8mm}V_{N}(\vec{t};\lambda)=\frac{\sqrt{2}}{(2\pi\lambda(\lambda\!+\!1)N)^{\frac{N}{2}}}\big[\!\prod_{n}(1\!+\!\frac{1}{\lambda_{n}})^{\frac{t_{n}}{2}}\big]\big[\!\prod_{k<l}\!\frac{\sqrt{(1+\lambda_{k})(1+\lambda_{l})}}{\sqrt{(1\!+\!\lambda_{k})(1\!+\!\lambda_{l})}\!-\!\sqrt{\lambda_{k}\lambda_{l}}}\big]\nonumber\\
&\times\exp[\frac{14\lambda^{2}\!+\!14\lambda\!-\!1}{12\lambda(\lambda+1)}]\exp[\frac{\sum_{m}\varepsilon_{m}^{2}}{16\lambda^{2}(\lambda\!+\!1)^{2}}]\exp[-\frac{(2\lambda+1)^{2}}{128\lambda^{3}(\lambda\!+\!1)^{3}}(\sum_{m}\varepsilon_{m}^{2})^{2}]\nonumber\\
&\times\exp[-\frac{(2\lambda+1)^{2}N}{128\lambda^{3}(\lambda+1)^{3}}\sum_{m}\varepsilon_{m}^{4}]\quad.\label{e:labval3}
\end{align}
Comparing (\ref{e:labval3}) to the Lower bound $\mathcal{E}_{\alpha}$, it is immediately clear that $V_{n}(\vec{t};\lambda)$ is much larger. Expand the products in square brackets around $\lambda$,
\begin{align*}
&\hspace{-8mm}\Big[\prod_{j}(1+\frac{1}{\lambda_{j}})^{\frac{\lambda(N-1)+(N-2)\varepsilon_{j}/2}{2}}\Big]\Big[\prod_{k<l}\frac{\sqrt{(1+\lambda_{k})(1+\lambda_{l})}}{\sqrt{(1+\lambda_{k})(1+\lambda_{l})}-\sqrt{\lambda_{k}\lambda_{l}}}\Big]\\
&=(1\!+\!\frac{1}{\lambda})^{\frac{x}{2}}(1\!+\!\lambda)^{\binom{N}{2}}\exp[\frac{2\lambda(N\!-\!1)\!+\!(N\!-\!2)\varepsilon_{j}}{4}\log\!\big(\frac{1\!+\!\lambda\!+\!\varepsilon_{j}}{1+\lambda}\frac{\lambda}{\lambda\!+\!\varepsilon_{j}}\big)]\\
&\times\exp[-\!\sum_{k<l}\log\big(1\!+\!\lambda\!-\!\sqrt{(1\!+\!\lambda-\frac{1+\lambda}{1\!+\!\lambda\!+\!\varepsilon_{k}})(1+\lambda-\frac{1+\lambda}{1\!+\!\lambda\!+\!\varepsilon_{l}})}\big)]\\
&=(1+\frac{1}{\lambda})^{\frac{x}{2}}(1+\lambda)^{\binom{N}{2}}\\
&\times\exp\Big[(\sum_{m}\varepsilon_{m}^{2})\cdot\big[N\frac{\lambda^{2}}{4\lambda^{2}(\lambda+1)^{2}}+\frac{\lambda}{4\lambda^{2}(\lambda+1)^{2}}\\
&-(N-1)\frac{3\lambda^{2}+\lambda}{8\lambda^{2}(\lambda+1)^{2}}-\frac{1}{8\lambda(\lambda+1)}\big]\\
&+(\sum_{m}\varepsilon_{m}^{3})\cdot\big[-N\frac{6\lambda^{3}+3\lambda^{2}+\lambda}{24\lambda^{3}(\lambda+1)^{3}}+N\frac{14\lambda^{3}+9\lambda^{2}+3\lambda}{48\lambda^{3}(\lambda+1)^{3}}\big]\\
&+(\sum_{m}\varepsilon_{m}^{4})\cdot\big[N\frac{6\lambda^{4}+6\lambda^{3}+4\lambda^{2}+\lambda}{24\lambda^{4}(\lambda+1)^{4}}-N\frac{30\lambda^{4}+28\lambda^{3}+19\lambda^{2}+5\lambda}{128\lambda^{4}(\lambda+1)^{4}}\big]\\
&+(\sum_{m}\varepsilon_{m}^{2})^{2}\cdot\big[\frac{6\lambda^{2}+6\lambda+1}{128\lambda^{3}(\lambda+1)^{3}}\big]\times\big(1+\mathcal{O}(N^{-\frac{1}{2}+5\omega})\big)\quad,
\end{align*}
yields combined with (\ref{e:labval3}) the desired result. To determine the error from the difference $\mathcal{D}$ from Lemma~\ref{l:diff}, we divide it by $V_{N}(\vec{t};\lambda)$.\\
Assuming that $|t_{j}-\lambda(N-1)|=\lambda N^{\frac{1}{2}+\omega}$ takes maximal values, it follows that the relative difference is at most
\begin{align*}
&\hspace{-8mm}\mathcal{O}(N^{2-K\alpha})\exp[\frac{4\lambda^{2}+4\lambda+1}{3\lambda(\lambda+1)}]\exp[\frac{4\lambda^{2}+4\lambda+3}{4(\lambda+1)^{2}}N^{2\omega}]\\
&\times\exp[\frac{(2\lambda+1)^{2} \lambda N^{4\omega}}{8(\lambda+1)^{3}}]\exp[\frac{(2\lambda+1)^{2} \lambda N^{4\omega}}{8(\lambda+1)^{3}}]\quad.
\end{align*}
Only the first exponential can become large, if $\lambda$ is small. Assuming that $\lambda > C/\log(N)$, this factor adds an error $N^{\frac{1}{3C}}$. To keep this relative error small, it is furthermore necessary that $\exp[N^{4\omega}]\ll N^{K\alpha-2}$. Solving this yields
\begin{equation*}
0<\omega<\frac{\log(K\alpha-2)+\log\big(\log(N)\big)}{4\log(N)}\quad.
\end{equation*}
\end{proof}

Choosing the value of $\lambda$ may seem arbitrary at first. It is not. Comparing (\ref{e:labval3}) to the Lower bound $\mathcal{E}_{\frac{1}{2}-r}$ for some small $r>0$, the outcome is only much larger, if
\begin{equation*}
\sum_{m}\varepsilon_{m}^{4}\ll N^{2r}\qquad\text{and}\qquad \sum_{m}\varepsilon_{m}^{2}\ll N^{r}\quad.
\end{equation*}
It follows that $\lambda N(N-1)=x$ in the limit. In~\cite{mckay1} the number of matrices $V_{N}(\vec{t};\lambda)$ has been calculated for the case that all $t_{j}$ are equal. They require $\lambda$ to be the average matrix entry for infinitely large matrices. Because Lemma~\ref{l:labval} covers this case too, the same value for $\lambda$ had to be expected.\\

Methods to treat such multi-dimensional combinatorical Gaussian integrals in more generality have been discussed in~\cite{mckay2}.

\section{Reduction of the integration region\label{sec:rir}}

In the previous paragraph the result of the integral (\ref{e:V3}) in a small box around the origin was obtained. Knowing this makes it much easier to compare the contribution inside and outside of this box. This is the main aim of Lemma~\ref{l:funris}.

\begin{lemma}\label{l:est}
For $a\in[0,1]$ and $n\in\mathbb{N}$ the estimates
\begin{equation*}
\exp[na\log (2)]\leq (1+a)^{n}\leq \exp[na]
\end{equation*}
hold.
\end{lemma}
\begin{proof}
The right-hand side follows from
\begin{equation*}
(1+a)^{n}=\sum_{j=0}^{n}a^{j}\binom{n}{j}=\sum_{j=0}^{n}\frac{(na)^{j}}{j!}\frac{n!}{n^{j}(n-j)!}\leq \sum_{j=0}^{n}\frac{(na)^{j}}{j!}\leq \exp[na]\quad.
\end{equation*}
For the left-hand side it suffices to show that $\log(1+a)\geq a\log(2)$. Because equality holds at one and zero, this follows from the concavity of the logarithm.
\end{proof}

\begin{lemma}\label{l:funris}
For any $\omega\in(0,\frac{1}{4})$ and $\alpha\in(0,\frac{1}{4}-\omega)$, define
\begin{equation*}
\delta_{N}=\frac{N^{-\alpha}\zeta_{N}}{\min \{\lambda_{j}\}}\quad,
\end{equation*}
such that $\zeta_{N}\rightarrow\infty$ and $N^{-\delta}\zeta_{N}\rightarrow0$ for any $\delta>0$. Assuming that $x=\sum_{j}t_{j}=\lambda N(N-1)$, $|t_{j}-\lambda(N-1)|\ll \lambda N^{\frac{1}{2}+\omega}$ and $\lambda>C/log(N)$, the integral
\begin{align*}
&\hspace{-8mm}V_{N}(\vec{t})=\Big(\prod_{j=1}^{N}(1+\frac{1}{\lambda_{j}})^{\frac{t_{j}}{2}}\Big)(2\pi)^{-N}\int_{\mathbb{T}^{N}}\!\!\!\ud\vec{\varphi}\,e^{-i\sum_{j=1}^{N}\varphi_{j}t_{j}}\\
&\times\Big\{\prod_{1\leq k<l\leq N}\!\!\frac{\sqrt{(1+\lambda_{k})(1+\lambda_{l})}}{\sqrt{(1+\lambda_{k})(1+\lambda_{l})}-\sqrt{\lambda_{k}\lambda_{l}}}\\
&\times\big(1-\frac{\sqrt{\lambda_{k}\lambda_{l}}}{\sqrt{(1+\lambda_{k})(1+\lambda_{l})}-\sqrt{\lambda_{k}\lambda_{l}}} (e^{i(\varphi_{k}+\varphi_{l})}-1)\big)^{-1}\Big\}
\end{align*}
can be restricted to
\begin{align*}
&\hspace{-8mm}V_{N}(\vec{t};\lambda)=\frac{2}{(2\pi)^{N}}\Big(\prod_{j=1}^{N}(1+\frac{1}{\lambda_{j}})^{\frac{t_{j}}{2}}\Big)\int_{[-\delta_{N},\delta_{N}]^{N}}\!\!\!\!\!\!\!\!\!\!\!\ud\vec{\varphi}\,\exp[-i\sum_{j}\varphi_{j}(t_{j}-\lambda(N-1))]\\
&\times\Big\{\prod_{1\leq k<l\leq N}\!\!\frac{\sqrt{(1+\lambda_{k})(1+\lambda_{l})}}{\sqrt{(1+\lambda_{k})(1+\lambda_{l})}-\sqrt{\lambda_{k}\lambda_{l}}}\\
&\times\big(1-\frac{\sqrt{\lambda_{k}\lambda_{l}}}{\sqrt{(1+\lambda_{k})(1+\lambda_{l})}-\sqrt{\lambda_{k}\lambda_{l}}} (e^{i(\varphi_{k}+\varphi_{l})}-1)\big)^{-1}\Big\}\\
&\times\big(1+\mathcal{O}(N^{\frac{3}{2}}\exp[-N^{1-2\alpha}\zeta_{N}^{2}])\big)\quad.
\end{align*}
\end{lemma}

\begin{proof}
The idea of the proof is to consider the integrand in a small box $[-\delta_{N},\delta_{N}]^{N}$ and see what happens to it if some of the angles $\vec{\varphi}$ lie outside of it.\\
Because $x$ is even, it follows that the integrand takes the same value at $\vec{\varphi}$ and $\vec{\varphi}+\vec{\pi}=(\varphi_{1}+\pi,\ldots,\varphi_{N}+\pi)$. This means that only half of the space has to be considered and the result must be multiplied by $2$.\\

This estimate follows directly from application of (\ref{e:fracexpa}-\ref{e:comb}) to the integrand and a computation like the one in the proof of Lemma~\ref{l:labval}. Writing
\begin{equation*}
\mu_{kl}=\frac{\sqrt{\lambda_{k}\lambda_{l}}}{\sqrt{(1+\lambda_{k})(1+\lambda_{l})}-\sqrt{\lambda_{k}\lambda_{l}}}\qquad\text{and}\qquad \varepsilon_{j}=\lambda_{j}-\lambda
\end{equation*}
with $|\varepsilon_{j}|\ll \lambda N^{-\frac{1}{2}+\omega}$ this yields
\begin{align}
&\hspace{-8mm}\Big|\int_{[-\delta_{N},\delta_{N}]^{N}}\!\!\!\!\!\!\!\!\ud\vec{\varphi}\,\prod_{1\leq k<l\leq N}\frac{1}{1-\mu_{kl}(\exp[i(\varphi_{k}+\varphi_{l})]-1)}\Big|\nonumber\\
&\leq \int_{[-\delta_{N}/2,\delta_{N}/2]^{N}}\!\!\!\!\!\!\!\!\ud\vec{\varphi}\,\big|\exp[\sum_{m=1}i^{m}\sum_{k<l}A_{m}(\mu_{kl})(\varphi_{k}+\varphi_{l})^{m}]\big|\nonumber\\
&\leq\sqrt{2}\big(\frac{2\pi}{\lambda(\lambda+1)N}\big)^{\frac{N}{2}}\exp[\frac{10\lambda^{2}+10\lambda+1}{4\lambda(\lambda+1)}]\exp[N^{\frac{1}{2}+2\omega}]\quad.\label{e:cl1}
\end{align}
The final exponent $\exp[N^{\frac{1}{2}+2\omega}]$ here comes from the estimate 
\begin{equation*}
\mu_{kl}\geq \lambda(1-N^{-\frac{1}{2}+\omega})\quad.
\end{equation*}

Now we argue case by case why other configurations of the angles $\varphi_{j}$ are asymptotically suppressed.\\
\emph{Case 1.} All but finitely many angles lie in the box $[-\delta_{N},\delta_{N}]^{N}$. A finite number of $m$ angles lies outside of it. We label these angles $\{\varphi_{1},\ldots,\varphi_{m}\}$. The maximum of the integrand
\begin{align*}
&\hspace{-8mm}f:(\varphi_{m+1},\ldots,\varphi_{N})\mapsto \prod_{1\leq k<l\leq N}\frac{1}{1-\mu_{kl}(\exp[i(\varphi_{k}+\varphi_{l})]-1)}
\end{align*}
in absolute value is given by the equations
\begin{equation*}
0=\partial_{\varphi_{j}}|f|=\sum_{k\neq j}\frac{\sin(\varphi_{j}+\varphi_{k})}{1\!-\!2\mu_{kl}(\mu_{kl}\!+\!1)(\cos(\varphi_{j}\!+\!\varphi_{k})\!-\!1)}\quad\text{ for }j=m\!+\!1,\ldots,N\,.
\end{equation*}
It is clear that the maximum is found for $\tilde{\varphi}=\varphi_{m+1}=\ldots=\varphi_{N}$. The first order solution to this is then
\begin{equation*}
\tilde{\varphi}=\frac{-1}{2(N-m-1)}\sum_{k=1}^{m}\frac{\sin(\varphi_{k})}{1+2\mu_{kj}(\mu_{kj}+1)(1-\cos\varphi_{k})}\quad.
\end{equation*}
This shows that the maximum will lie in the box $[-\delta_{N}/2,\delta_{N}/2]^{N}$. This implies that $|\varphi_{j}-\varphi_{k}|>\delta_{N}/2$, when $1\leq j\leq m$ and $m+1\leq k\leq N$. Applying the estimate (\ref{e:fracest}) to pairs of such angles and afterwards (\ref{e:cl1}) to the remaining $N-m$ angles in the box $[-\delta_{N},\delta_{N}]^{N-m}$ gives us an upper bound of
\begin{align*}
&\hspace{-8mm}\frac{2\sqrt{2}}{\big(2\pi\lambda(\lambda\!+\!1)(N\!-\!m)\big)^{\frac{N\!-\!m}{2}}}\Big(\!\prod_{j}(1\!+\!\frac{1}{\lambda_{j}})^{\frac{t_{j}}{2}}\!\Big)\!\cdot\!\Big(\!\prod_{k<l}\!\frac{\sqrt{(1+\lambda_{k})(1+\lambda_{l})}}{\sqrt{(1\!+\!\lambda_{k})(1\!+\!\lambda_{l})}\!-\!\sqrt{\lambda_{k}\lambda_{l}}}\Big)\\
&\times\binom{N}{m}\exp[\frac{30\lambda^{2}+30\lambda+3}{12\lambda(\lambda+1)}]\exp[N^{\frac{1}{2}+2\omega}](1+\frac{\lambda(\lambda+1)\delta_{N}^{2}}{4})^{-\frac{Nm}{2}}
\end{align*}
on the part of the integral in the small box $[-\delta_{N},\delta_{N}]^{N}$. There are $\binom{N}{m}$ ways to select the $m$ angles. Applying Lemma~\ref{l:est} to the final factor and comparing the result with the Lower bound, shows that this may be neglected if
\begin{align*}
&\hspace{-8mm}2\sqrt{2}e^{\frac{16\lambda^{2}+16\lambda+4}{12\lambda(\lambda+1)}}e^{m/2}N^{\frac{3m}{2}}(2\pi\lambda(\lambda+1))^{\frac{m}{2}}\exp[N^{1-2\alpha}+N^{\frac{1}{2}+2\omega}]\\
&\times \exp[-\frac{Nm\log(2)}{8}\lambda(\lambda+1)\delta_{N}^{2}]\rightarrow0\quad.
\end{align*}
The condition $0<\alpha<\frac{1}{4}-\omega$ and the sequence $\zeta_{N}\rightarrow\infty$ guarantee this. In fact, the same argument works for all $m$ such that $m/N\rightarrow0$.\\

\emph{Case 2.} If the number $m=\rho N$ of angles outside the integration box $[-\delta_{N},\delta_{N}]^{N}$ increases faster, another estimate is needed, because the maximum $\tilde{\varphi}$ may lie outside of $[-\delta_{N}/2,\delta_{N}/2]$. It is clear that $0<\rho<1$ in the limit.\\
Estimate the location $\varphi_{j}=\tilde{\varphi}$ of the maximum is much trickier now. Regardless of its precise location, we will take the maximum value as the estimate for the integrand in the entire integration box. The smaller box $[-\delta_{N}/2,\delta_{N}/2]^{N}$ is considered once more. We distinguish two options.\\
\emph{-Case 2a.} The maximum lies in $[-\delta_{N}/2,\delta_{N}/2]^{N}$, thus $\tilde{\varphi}\in[-\delta_{N}/2,\delta_{N}/2]$.\\
Applying the estimate (\ref{e:fracest}) to this yields an upper bound
\begin{align*}
&\hspace{-8mm}\binom{N}{\rho N}(2\delta_{N})^{N(1-\rho)}(2\pi)^{\rho N}\Big(\prod_{j}(1+\frac{1}{\lambda_{j}})^{\frac{t_{j}}{2}}\Big)\\
&\times\Big(\prod_{k<l}\frac{\sqrt{(1+\lambda_{k})(1+\lambda_{l})}}{\sqrt{(1+\lambda_{k})(1+\lambda_{l})}-\sqrt{\lambda_{k}\lambda_{l}}}\Big)(1+\frac{1}{4}\lambda(\lambda+1)\delta_{N}^{2})^{-\frac{N^{2}\rho(1-\rho)}{4}}\quad.
\end{align*}
Applying Lemma~\ref{l:est} to the last factor and dividing this by $\mathcal{E}_{\alpha}$ shows that
\begin{align*}
&\hspace{-8mm}\big(2^{\frac{1}{\rho}}\pi \delta^{\frac{1-\rho}{\rho}}(2\pi\lambda(\lambda\!+\!1)N)^{\frac{2}{\rho}}N\\
&\times\exp[\frac{N^{-2\alpha}}{\rho}+\frac{N^{-\frac{1}{2}+2\omega}(1\!-\!\rho)}{\rho}-\frac{\lambda(\lambda\!+\!1)\delta_{N}^{2}N(1\!-\!\rho)\log(2)}{16}]\big)^{\rho N}\rightarrow0
\end{align*}
is a sufficient and satisfied condition.\\
\emph{-Case 2b.} The maximum lies not in $[-\delta_{N}/2,\delta_{N}/2]^{N}$. This is the same as $\delta_{N}/2<|\tilde{\varphi}|\leq\delta_{N}$.\\
Applying (\ref{e:fracest}) only to the angles $\varphi_{\rho N+1},\ldots,\varphi_{\rho N}$ in the integration box gives an upper bound
\begin{align*}
&\hspace{-8mm}\binom{N}{\rho N}(2\delta_{N})^{N(1-\rho)}(2\pi)^{\rho N}\Big(\prod_{j}(1+\frac{1}{\lambda_{j}})^{\frac{t_{j}}{2}}\Big)\\
&\times\Big(\prod_{k<l}\frac{\sqrt{(1+\lambda_{k})(1+\lambda_{l})}}{\sqrt{(1+\lambda_{k})(1+\lambda_{l})}-\sqrt{\lambda_{k}\lambda_{l}}}\Big)(1+\frac{1}{4}\lambda(\lambda+1)\delta_{N}^{2})^{-\frac{N^{2}(1-\rho)^{2}}{4}}\quad.
\end{align*}
The same steps as in Case 2a. will do.\\
This shows that the integration can be restricted to the box $[-\delta_{N},\delta_{N}]^{N}$. The error terms follow from \emph{Case 1.}, since convergence there is much slower.
\end{proof}

Lemma~\ref{l:funris} shows that for every $\alpha\in(0,1/4-\omega)$ and $N\in\mathbb{N}$ there is a box that contains most of the integral's mass. As $N$ increases, this box shrinks and the approximation becomes better. The parameter $\alpha$ determines how fast this box shrinks. Smaller values of $\alpha$ lower the Lower bound and, hence, increase the number of configurations within reach at the price of more intricate integrals and less accuracy.\\

The observation that $\zeta_{N}=\log(N)$ and $K\geq \alpha^{-1}(\frac{3}{2}+\frac{1}{3C}-6\omega)$ satisfies all the demands proves Theorem~\ref{thrm:p1}.\\
An idea of the accuracy of these formulas can be obtained from Table~\ref{t:t1} and~\ref{t:t2}, where the reference values
\begin{equation}
y_{k}=\sum_{j=1}^{N}(t_{j}-\lambda(N-1))^{k}\qquad\text{for }k\geq 2\label{e:refval}
\end{equation}
are defined to compare configurations to the reference values $2^{-k}\lambda^{k}N^{1+\frac{k}{2}}$ for $k\geq 2$.

\begin{table}[!hb]\begin{footnotesize}
\hspace{5mm}\begin{tabular}{ccccccc||c||ccc||c||c}
$t_{1}$ & $t_{2}$ & $t_{3}$ & $t_{4}$ & $t_{5}$ & $t_{6}$ & $t_{7}$ & $\#$ & $y_{2}$ & $y_{3}$ & $y_{4}$ & $V_{N}(\vec{t;\lambda})$ & ratio \\\hline
8&8&8&8&8&8&8&5.42E7&0&0&0&5.03E7&0.928\\
7&8&8&8&8&8&9&5.07E7&2&0&2&4.74E7&0.935\\
7&7&8&8&8&9&9&4.75E7&4&0&4&4.47E7&0.941\\
7&7&7&8&9&9&9&4.45E7&6&0&6&4.21E7&0.947\\
6&8&8&8&8&8&10&4.15E7&8&0&32&3.96E7&0.955\\
6&7&8&8&9&9&9&4.13E7&8&-6&20&3.94E7&0.953\\
7&7&7&8&8&9&10&4.18E7&8&6&20&4.00E7&0.956\\
5&8&8&8&9&9&9&3.53E7&12&-24&84&3.40E7&0.964\\
7&7&7&8&8&8&11&3.71E7&12&24&84&3.62E7&0.976\\
5&7&8&8&9&9&10&3.12E7&16&-18&100&3.05E7&0.977\\
6&7&7&7&9&10&10&3.23E7&16&6&52&3.16E7&0.977\\
7&7&7&7&8&8&12&2.91E7&20&60&260&2.96E7&1.017\\
5&5&5&9&10&11&11&1.08E7&50&-18&422&1.11E7&1.031\\
5&7&7&7&7&9&14&1.17E7&50&186&1382&1.34E7&1.143\\
4&6&7&7&8&10&14&7.92E6&62&150&1586&8.94E6&1.128
\end{tabular}\end{footnotesize}
\caption{\small The number ($\#$) of symmetric $7\times 7$-matrices with zero diagonal and natural entries summing to $x=56$ such that the $j$-th row sums to $t_{j}$ and the asymptotic estimates for this number by $V_{N}(\vec{t};\lambda)$ from Lemma~\ref{l:labval} with $\lambda=x/(N(N-1))$ the average matrix entry. The parameters $y_{2}$, $y_{3}$ and $y_{4}$ are defined in (\ref{e:refval}) and their reference values are are $22$, $38$ and $68$ respectively. The convergence condition is $|t_{j}-\lambda(N-1)|\leq1.3$. The notation $1.0E6=1.0\times 10^{6}$ is used here.\label{t:t1}}
\end{table}
 
\begin{table}[!hb]\begin{footnotesize}
\hspace{25mm}\begin{tabular}{ccc|c|c|c}
$N$ & $t$ & $\lambda$ & $\#$ & $V_{N}(t;\lambda)$ & ratio\\\hline
6&6&1.20&3.69E4&3.34E4&0.906\\
7&8&1.33&5.42E7&5.03E7&0.928\\
8&9&1.29&1.10E11&1.04E11&0.938\\
9&10&1.25&8.46E14&8.00E14&0.946\\
10&11&1.22&2.45E19&2.34E19&0.952\\
11&12&1.20&2.71E24&2.60E24&0.957\\
12&13&1.18&1.14E30&1.10E30&0.961\\
13&14&1.17&1.86E36&1.79E36&0.965\\
14&15&1.15&1.16E43&1.12E43&0.968\\
15&14&1.00&6.36E46&6.18E46&0.971\\
16&12&0.80&6.32E47&6.15E47&0.974\\
17&12&0.75&9.55E52&9.32E52&0.976\\
18&12&0.71&2.02E58&1.97E58&0.978
\end{tabular}\end{footnotesize}
\caption{\small The number ($\#$) of symmetric $N\times N$-matrices~\cite{mckay4} with zero diagonal and natural entries, such that each row sums to $t$ and the asymptotic estimates for these numbers by $V_{N}(t;\lambda)$ Lemma~\ref{l:labval}, where $\lambda=x/(N(N-1))$ the average matrix entry. The notation $1.0E6=1.0\times 10^{6}$ is used here.\label{t:t2}}
\end{table}

\begin{figure}[!hbt]
\begin{subfigure}[t]{0.5\linewidth}
\includegraphics[width=0.9\textwidth]{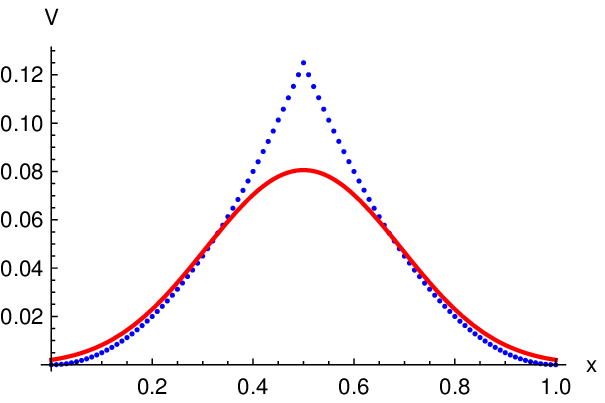}
\caption{N=4\label{f:f1a}}
\end{subfigure}%
\begin{subfigure}[t]{0.5\linewidth}
\includegraphics[width=0.9\textwidth]{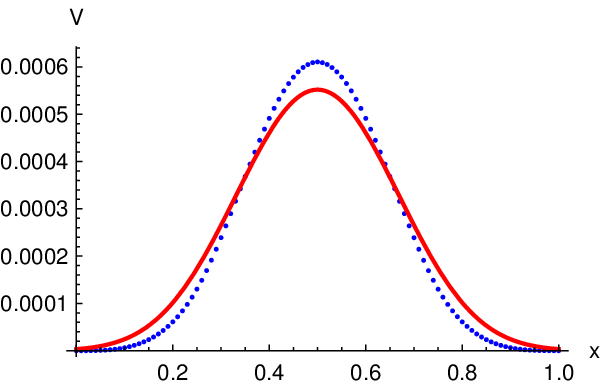}
\subcaption{N=5\label{f:f1b}}
\end{subfigure}\\
\begin{subfigure}[t]{0.5\linewidth}
\includegraphics[width=0.9\textwidth]{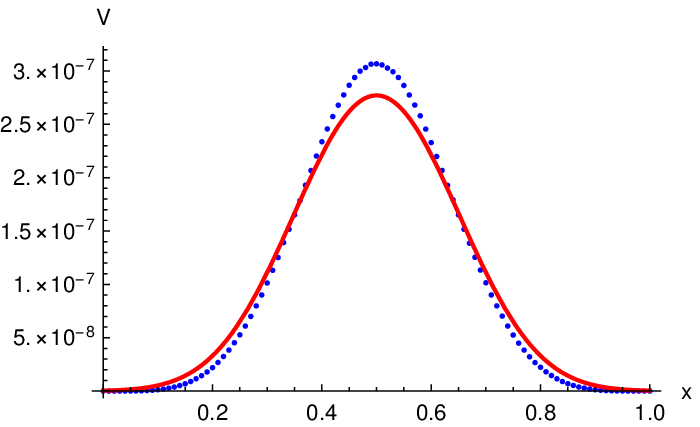}
\caption{N=6\label{f:f1c}}
\end{subfigure}%
\begin{subfigure}[t]{0.5\linewidth}
\includegraphics[width=0.9\textwidth]{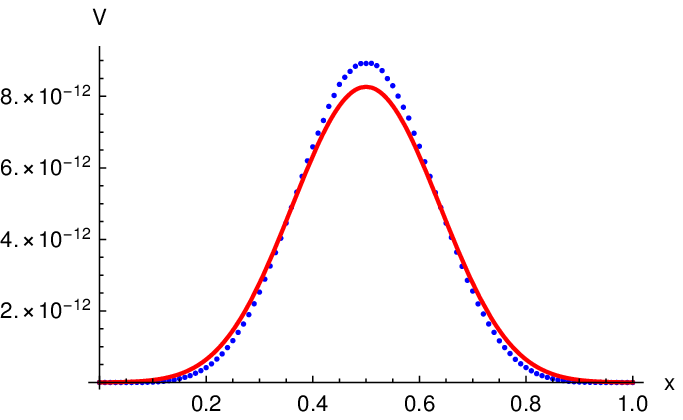}
\subcaption{N=7\label{f:f1d}}
\end{subfigure}\\
\begin{subfigure}[t]{0.5\linewidth}
\includegraphics[width=0.9\textwidth]{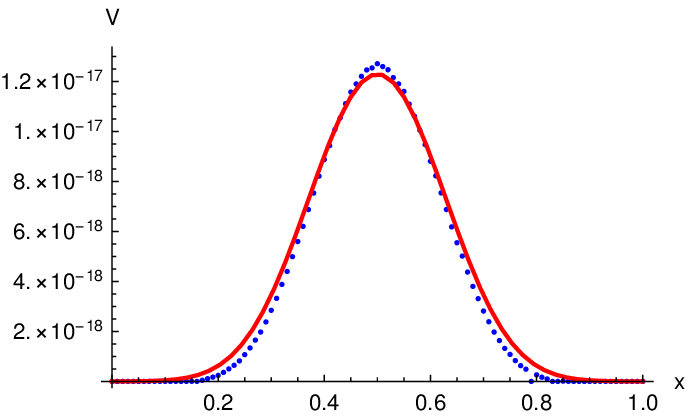}
\caption{N=8\label{f:f1e}}
\end{subfigure}%
\begin{subfigure}[t]{0.5\linewidth}
\includegraphics[width=0.9\textwidth]{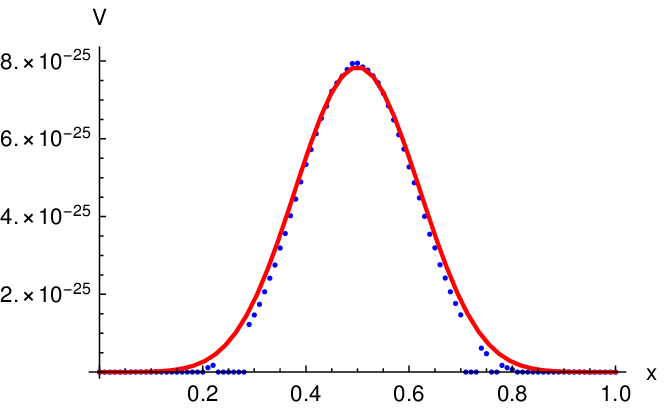}
\subcaption{N=9\label{f:f1f}}
\end{subfigure}
\caption{\small The volume result (\ref{e:pv1}) (red) and the volume of $P_{N}(0.5,\ldots,0.5,x,1-x)$ (blue) for $N=4,5,6,7,8,9$. The latter were determined by a numerical integration algorithm for convex multidimensional step functions on the basis of Monte Carlo integration.\label{f:f1}}
\end{figure}

\begin{figure}[!hbt]
\begin{subfigure}[t]{0.5\linewidth}
\includegraphics[width=0.9\textwidth]{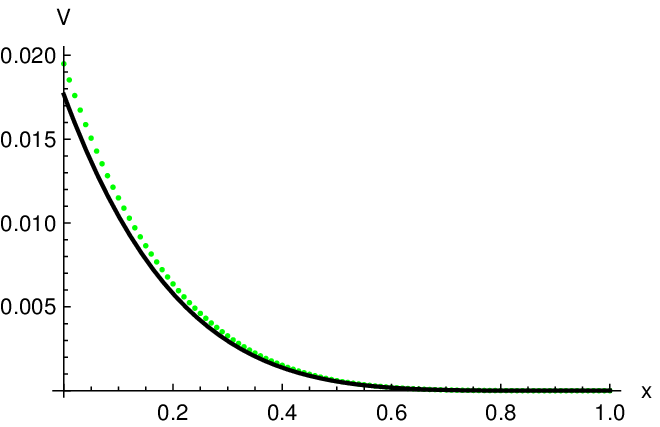}
\caption{$\vol\big(P_{5}(x,x,x,x,x)\big)$\label{f:f2a}}
\end{subfigure}%
\begin{subfigure}[t]{0.5\linewidth}
\includegraphics[width=0.9\textwidth]{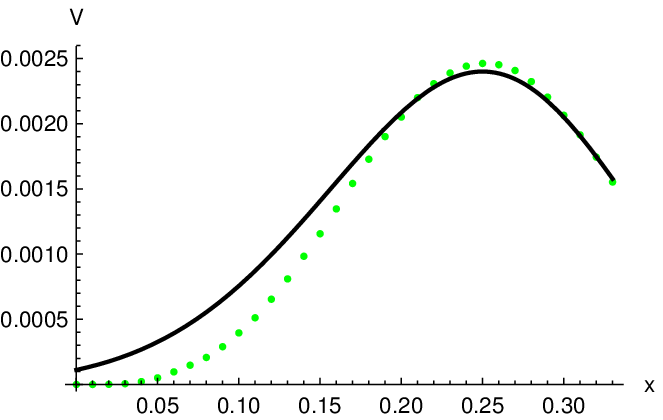}
\subcaption{$\vol\big(P_{5}(0.5,x,x,x,1-3x)\big)$\label{f:f2b}}
\end{subfigure}
\caption{\small The volume result (\ref{e:pv1}) (black) and the volume for two functions (green) for $N=5$. The latter were determined by straightforward Monte Carlo integration.\label{f:f2}}
\end{figure}

\section{Polytopes\label{sec:pv}}
In the previous paragraphs the asymptotic counting of symmetric matrices with zero diagonal and entries in the natural numbers was discussed. This allows us to return to the polytopes. The first step is to count the total number of symmetric matrices with zero diagonal and integer entries summing up to $x$ to see which fraction of such matrices are covered by Theorem~\ref{thrm:p1}.\\
This is easily done by a line of $\binom{N}{2}+\frac{x}{2}$ elements, for example unit elements $1$, and $\binom{N}{2}-1$ semicolons. Putting the semicolons between the elements, such that the line begins and ends with a unit element and no semicolons stand next to each other, creates such a matrix. The number of elements before the first semicolon minus one is the first matrix element $b_{12}$. The number of elements minus one between the first and second semicolon yields the second matrix element $b_{13}$. In this way, we obtain the $\binom{N}{2}$ elements of the upper triangular matrix. There are $\binom{N}{2}-1+\frac{x}{2}$ positions to put $\binom{N}{2}-1$ semicolons and thus
\begin{equation}
\binom{\binom{N}{2}-1+\frac{x}{2}}{\binom{N}{2}-1}\approx \frac{1}{N}\sqrt{\frac{1}{\pi\lambda(\lambda+1)}}(1+\lambda)^{\binom{N}{2}}(1+\frac{1}{\lambda})^{\frac{x}{2}}\big(1+\mathcal{O}(N^{-1})\big)\label{e:nosm}
\end{equation}
such matrices, where we have used Stirling's approximation and the average matrix entry condition $\lambda=x/(N(N-1))$ for the approximation.\\

The next step is to estimate the number of matrices within reach of Theorem~\ref{thrm:p1}. Using only the leading order, the number of covered matrices is given by
\begin{align}
&\hspace{-8mm}\int \ud \vec{t}\;V_{N}(\vec{t};\lambda)\,\delta(\lambda N(N-1)-\sum_{j}t_{j})\nonumber\\
&=\frac{(1\!+\!\lambda)^{\binom{N}{2}}(1\!+\!\frac{1}{\lambda})^{\frac{x}{2}}}{\pi^{\frac{N}{2}}\sqrt{\lambda (\lambda\!+\!1)N}}\exp[\frac{14\lambda^{2}\!+\!14\lambda\!-\!1}{12\lambda(\lambda+1)}]\!\!\int\!\!\ud S\!\! \int\!\!\ud\sigma\!\!\int\!\!\ud\tau\,\exp[2\pi i\sigma S\!+\!S^{2}]\nonumber\\
&\times\Big\{\!\prod_{j=1}^{N}\!\int_{-N^{\omega}}^{N^{\omega}}\!\!\!\!\!\!\ud y_{j}\,\exp[2\pi i\tau y_{j}]\exp[-y_{j}^{2}(1\!+\!\frac{2}{N}\!-\!\frac{2\pi i\sigma}{N})]\exp[\frac{\sqrt{2}(2\lambda\!+\!1)y_{j}^{3}}{3\sqrt{\lambda(\lambda\!+\!1)N}}]\nonumber\\
&\times\exp[-\frac{3\lambda^{2}+3\lambda+1}{3\lambda(\lambda+1)N}y_{j}^{4}]\Big\}\nonumber\\
&=\frac{(1\!+\!\lambda)^{\binom{N}{2}}(1\!+\!\frac{1}{\lambda})^{\frac{x}{2}}}{N\sqrt{\pi\lambda (\lambda\!+\!1)}}\exp[\frac{14\lambda^{2}\!+\!14\lambda\!-\!1}{12\lambda(\lambda+1)}]\!\int\!\!\ud S\!\! \int\!\!\ud\sigma\!\!\int\!\!\ud\tau\,\exp[2\pi i\sigma S\!+\!S^{2}]\nonumber\\
&\times\exp[-\pi^{2}\tau^{2}(1-\frac{2}{N}+\frac{2i\pi \sigma }{N})]\exp[-1+\pi i \sigma+\frac{\pi i\tau(2\lambda+1)}{\sqrt{2\lambda(\lambda+1)}}]\nonumber\\
&\times\exp[\frac{5(2\lambda+1)^{2}}{24\lambda (\lambda+1)}]\exp[-\frac{3\lambda^{2}+3\lambda+1}{4\lambda(\lambda+1)}]\times\big(1-\mathcal{O}(\frac{\exp[-N^{2\omega}]}{N^{2\omega}})\big)^{N}\nonumber\\
&=\frac{(1+\lambda)^{\binom{N}{2}}(1+\frac{1}{\lambda})^{\frac{x}{2}}}{N\sqrt{\pi\lambda (\lambda+1)}}\exp[-\frac{1}{4\lambda(\lambda+1)}]\nonumber\\
&\times\big(1-\mathcal{O}(\frac{\exp[-N^{2\omega}]}{N^{2\omega}})\big)^{N}\times\big(1+\mathcal{O}(N^{-1})\big)\quad.\label{e:matcov}
\end{align}

A fraction $\exp[-\frac{1}{4\lambda(\lambda)}]$ of the matrices is covered, provided that $\omega$ is large enough. A sufficient condition is that 
\begin{equation}
\omega \geq \frac{\log \log N}{2\log N}\quad.
\end{equation}
Combining this with the condition 
\begin{equation*}
\omega \leq \frac{\log (K\alpha-2) + \log(\log N)}{4\log (N)}
\end{equation*}
shows that $K\geq \log(N)/\alpha+2$  is necessary to satisfy both demands. However, such large values of $K$ remain without consequences, because higher values of $K$ only influence the the error term in Lemma~\ref{l:labval}.\\
As $\lambda\rightarrow\infty$, the fraction of covered matrices tends to one and the volume of the diagonal subpolytopes of symmetric stochastic matrices can be determined by (\ref{e:ehrhart}). In terms of the variables
\begin{equation*}
t_{j}=\frac{1-h_{j}}{a}\quad\text{and}\quad \chi=\sum_{j}h_{j}
\end{equation*}
the volume of the diagonal subpolytope is calculated by
\begin{align}
&\hspace{-8mm}\vol(P_{N}(\vec{h}))=\lim_{a\rightarrow0}a^{\frac{N(N-3)}{2}}V_{N}(\frac{\vec{1}-\vec{h}}{a};\frac{N-\chi}{aN(N-1)})\nonumber\\
&=\sqrt{2}e^{\frac{7}{6}}\Big(\frac{e(N-\chi)}{N(N-1)}\Big)^{\binom{N}{2}}\Big(\frac{N(N-1)^{2}}{2\pi(N-\chi)^{2}}\Big)^{\frac{N}{2}}\nonumber\\
&\times\exp[-\frac{N(N-1)^{2}}{2(N-\chi)^{2}}\sum_{j}(h_{j}-\frac{\chi}{N})^{2}]\exp[-\frac{(N-1)^{2}}{(N-\chi)^{2}}\sum_{j}(h_{j}-\frac{\chi}{N})^{2}]\nonumber\\
&\times\exp[-\frac{N(N-1)^{3}}{3(N-\chi)^{3}}\sum_{j}(h_{j}-\frac{\chi}{N})^{3}]\exp[-\frac{N(N-1)^{4}}{4(N-\chi)^{4}}\sum_{j}(h_{j}-\frac{\chi}{N})^{4}]\nonumber\\
&\times\exp[\frac{(N-1)^{4}}{4(N-\chi)^{4}}\big(\sum_{j}(h_{j}-\frac{\chi}{N})^{2}\big)^{2}]\quad.\label{e:pv1}
\end{align}

The convergence criterion becomes
\begin{equation*}
\frac{|t_{j}-\lambda (N-1)|}{\lambda N^{\frac{1}{2}+\omega}}=\frac{N^{\frac{1}{2}-\omega}(N-1)}{N-\chi}|h_{j}-\frac{\chi}{N}|\rightarrow0\quad.
\end{equation*}
This is the same as
\begin{equation*}
\sum_{j}|h_{j}-\frac{\chi}{N}|^{k}\ll (\frac{N-\chi}{N-1})^{k}N^{1-\frac{k}{2}+k\omega}\quad\text{for all }k\geq2\quad.
\end{equation*}

This means that we only have accuracy in a small neighbourhood around $\vec{\chi/N}$. However, the calculation (\ref{e:matcov}) shows that this corresponds to almost all matrices asymptotically, so that outside of this region the polytopes will have very small volumes. There, not all relevant factors are known, but missing factors will be small compared to the dominant factor. This means that for diagonals that satisfy
\begin{equation*}
\lim_{N\rightarrow\infty}\frac{(N-1)^{2}\sum_{j}(h_{j}-\frac{\chi}{N})^{2}}{(N-\chi)^{2}\log(N)}=0
\end{equation*}
qualitatively reasonable results are expected.\\

Since we are calculating a $\binom{N}{2}$-dimensional volume with only one length scale, it follows that no correction can become large in this limit. It inherits the relative error from Theorem~\ref{thrm:p1}. This proves Theorem~\ref{thrm:p2}. Examples of this formula at work are given in Figure~\ref{f:f1} and~\ref{f:f2}. 

\chapter{The vacuum sector for weak coupling\label{sec:OP_for_QMM}}\normalsize
Having discussed all techniques needed for the evaluation of the partition function, it is time to practise them. This will be done for weak, but strictly positive coupling. This assumption provides some context to the computations. Furthermore, the partition function will be restricted to the vacuum sector. As will be discussed later, this is not much of a restriction, but mainly a technical simplification.\\

The matrix expansion of the Moyal product (\ref{e:m1})
\begin{equation*}
\left(a \star_{\Theta} b\right)(x)=\int\frac{\ud^{d}k\,\ud^{d}y}{(2\pi)^{d}}a(x+\frac{1}{2}\Theta k)b(x+y)e^{ik\cdot y}
\end{equation*}
allows an interpretation of a noncommutative $\varphi^{4}$-model as a matrix model in $d$ even dimensions. The specific quantum field theory is then given by the action (\ref{e:GW})
\begin{equation}
S[\varphi]=\int_{\mathbb{R}^{4}}\ud^{4}x\,\varphi(x)\frac{1}{2}\big(-\Delta+\mu^{2}+|2\Theta^{-1}x|^{2}\big)\varphi(x)+\frac{g}{4}\varphi^{\star 4}(x)\label{e:ncqft1}
\end{equation}
with partition function
\begin{equation*}
\mathscr{Z}=\int \ud \varphi \,e^{-S[\varphi]}\quad,
\end{equation*}
where the integral is over some suitable class of functions corresponding precisely to the Hermitean matrices. The associated matrix expansion is given by
\begin{equation}
\mathscr{Z}=\int\ud X\,e^{\Tr\Big(-E X^{2}-g X^{4}+JX\Big)}\label{e:qmm1}\quad,
\end{equation}
where we assume naturally that $g>0$. This procedure is discussed briefly in Paragraph~\ref{sec:moyalp}. Such an expansion provides both an \textsc{ir} and \textsc{uv}-regulariser and avoids the complications of \textsc{ir}/\textsc{uv} mixing~\cite{langmann1}.\\

The matrix $J$ in (\ref{e:qmm1}) lies around zero and represents a source, whereas $E$ corresponds to the two-point function and is an unbounded self-adjoint matrix with compact resolvent. It may be assumed diagonal. Any further assumptions on its entries will be made if the computations require it.\\
Our aim is to compute the free energy density of this model. This means that we would like to compute the partition function $\mathscr{Z}$ 
\begin{equation*}
F=\lim_{N\rightarrow \infty}\frac{1}{\mathpzc{V}(N)}\log(\mathscr{Z})\quad,
\end{equation*}
where $\mathpzc{V}(N)$ is the volume of the discretised momentum space for $N$ degrees of freedom. To write down the system meaningfully a regulariser $N$ is introduced, where $N$ is the number of degrees of freedom included.\\

The partition function (\ref{e:qmm1}) in this case is regularised by the size $N$ of the matrices. This implies that $E$ has a maximal eigenvalue, so that $N$ may be interpreted as a kind of momentum cut-off. To make it a true cut-off it needs to be coupled to a dimensionful scale, so that it becomes a meaningful interpretation as renormalisation tool. The finite volume is then the volume of the accessible part of the momentum space.

\section{The vacuum sector of the quartic matrix model\label{sec:VOP_for_QMM}}\normalsize
Various aspects of the theory can be seen without external fields and thus without source. This simplifies the technical challenges of the model considerably. For example, he Harish-Chandra-Itzykson-Zuber integral from Theorem~\ref{thrm:hciz} can be applied to eliminate the matrix integral. Setting $J=0$ in (\ref{e:qmm1}) means that
\begin{align}
&\hspace{-8mm}\mathscr{Z}[J=0]=\int\ud X\,e^{-\Tr(EX^{2}+gX^{4})}\label{e:qmm2}\\
&=\mathpzc{U}(-1)^{\binom{N}{2}}\big[\!\prod_{k=0}^{N-1}k!\big]\!\int_{-\infty}^{\infty}\!\!\!\!\!\!\!\ud\vec{\lambda}\,\Delta^{2}(\lambda_{1},\ldots,\lambda_{N})e^{-\Tr(g\Lambda^{4})}\!\int_{U(N)}\!\!\!\!\!\!\!\!\ud U\,e^{-\Tr(\Lambda^{2}UEU^{*})}\;,\nonumber
\end{align}
where $\Lambda$ is the diagonal matrix $\diag(\lambda_{1},\ldots,\lambda_{N})$ of the eigenvalues of the matrix $X$ and 
\begin{equation*}
\mathpzc{U}=\frac{\pi^{\binom{N}{2}}}{\prod_{m=0}^{N}m!}
\end{equation*}
is the constant associated with the coordinate transformation from (\ref{e:Umeastra1}) and (\ref{e:Umeastra2}). 

\begin{exm}\label{exm:constant1}
The various constants, such as $\mathpzc{U}$, are difficult to trace. An example is a convenient tool to check them. It can be checked directly that
\begin{align*}
&\hspace{-8mm}\mathpzc{Z}=\frac{\pi^{\frac{9}{2}}}{8}=\big[\!\!\prod_{1\leq k\leq l\leq 3}\int_{-\infty}^{\infty}\!\!\!\!\ud M_{ij}^{(r)}\big]\cdot\big[\!\!\prod_{1\leq k< l\leq 3}\int_{-\infty}^{\infty}\!\!\!\!\ud M_{ij}^{(i)}\big]\;\\
&\times\exp[-\!\!\sum_{j=1}^{3}\!(M_{jj}^{(r)})^{2}-2\!\!\!\sum_{1\leq k<l\leq r}\!\!\!(M_{ij}^{(r)})^{2}+(M_{ij}^{(i)})^{2}]
\end{align*}
and after diagonalisation
\begin{align*}
&\hspace{-8mm}\mathpzc{Z}=\frac{\pi^{\frac{9}{2}}}{8}=\frac{\pi^{3}}{\prod_{n=0}^{3}n!}\big[\prod_{m=1}^{3}\int_{-\infty}^{\infty}\!\!\!\!\ud \lambda_{m}\big]\;\big[\prod_{1\leq k<l\leq 3}(\lambda_{l}-\lambda_{k})^{2}\big]\exp[-\sum_{m=1}^{3}\lambda_{m}^{2}]\quad.
\end{align*}
as well. By Monte Carlo integration techniques it is also checked that (\ref{e:ZJ0b}) is correct. The poles $\lambda_{m}=-\lambda_{n}$ form a small problem in this check. To overcome this a symmetric integration prescription can be used. This confirms that our starting point is correct.
\end{exm}

By the Harish-Chandra-Itzykson-Zuber integral from Theorem~\ref{thrm:hciz} this equals
\begin{align}
&\hspace{-8mm}\mathscr{Z}[0]=\mathpzc{U}(-1)^{\binom{N}{2}}[\prod_{k=0}^{N-1}k!]\int\ud\vec{\lambda}\,\frac{\Delta^{2}(\lambda_{1},\ldots,\lambda_{N})}{\Delta(e_{1},\ldots,e_{N})\Delta(\lambda_{1}^{2},\ldots,\lambda_{N}^{2})}\nonumber\\
&\times e^{-\sum_{j=1}^{N}g\lambda_{j}^{4}}\det_{k,l}\big(e^{-e_{k}\lambda_{l}^{2}}\big)\label{e:ZJ0}\\
&=\mathpzc{U}\frac{(-1)^{\binom{N}{2}}[\prod_{k=0}^{N-1}k!]g^{-\frac{N^{2}}{4}}}{\Delta(\sqrt{\mu_{1}},\ldots,\sqrt{\mu_{N}})}\int\ud\vec{\lambda}\,\frac{\Delta(\lambda_{1},\ldots,\lambda_{N})}{\prod_{m<n}(\lambda_{m}+\lambda_{n})}e^{-\sum_{j=1}^{N}\lambda_{j}^{4}}\nonumber\\
&\times\det_{k,l}\big(e^{-\sqrt{\mu}_{k}\lambda_{l}^{2}}\big)\label{e:ZJ0b}\\
&=\!\mathpzc{U}\frac{(\!-1)^{\binom{N}{2}}[\prod_{k=0}^{N-1}k!](N!)}{\Delta(\sqrt{\mu_{1}},\ldots,\sqrt{\mu_{N}})}\!\!\int\!\!\ud\vec{\lambda}\,\frac{\Delta(\lambda_{1},\ldots,\lambda_{N})}{\prod_{m<n}(\lambda_{m}\!+\!\lambda_{n})}e^{-\sum_{j}(g\lambda_{j}^{4}+e_{j}\lambda_{j}^{2})}\,,\label{e:ZJ0c}
\end{align}
where 
\begin{equation*}
\Delta(\lambda_{1}^{2},\ldots,\lambda_{N}^{2})=\prod_{1\leq k<l\leq N}(\lambda_{l}^{2}-\lambda_{k}^{2})
\end{equation*}
and
\begin{equation}
\mu_{k}=\frac{e_{k}^{2}}{g}\quad.\label{e:newvar}
\end{equation}
The dynamic variables $\mu_{k}$ are especially useful when discussing the case of strong coupling, when the quartic term in the exponential is dominant. This will be used in Chapter~\ref{sec:Str_coup}.\\

The symmetrised partition function (\ref{e:ZJ0c}) was introduced in Paragraph~\ref{sec:symHCIZ} and is practical for the weak coupling regime.\\

As in previous sections the notation $\sum_{j}$ for $\sum_{j=1}^{N}$ and $\sum_{k<l}$ for $\sum_{k=1}^{N}\sum_{l=k+1}^{N}$ will be used, when no confusion about the indices can exist. The same generalisation  is used in other contexts.

\subsection{Generalisation from the vacuum sector\label{sec:source}}
It was argued that the restriction to the vacuum sector results in an easier, but still interesting regime of the Grosse-Wulkenhaar model. However, other sectors of the quartic matrix model are of interest too. Knowing how to compute the partition function for the vacuum sector, an extension to the full theory presents itself.\\

For this the shifted matrix action $\exp[-\Tr(X^{4}+(E+J)\cdot(X^{2}+X\kappa))]$ may be considered. The kinetic eigenvalues $e_{j}$ are then replaced by the eigenvalues of the Hermitean matrices $E+J$. The consequence of the shift $X^{2}+X\kappa$ is that a linear exponential factor $\exp[-\kappa\sum_{j}\mathpzc{e}_{j}\lambda_{j}]$ must be added to (\ref{e:ZJ0c}).\\

This is of the form of the vacuum sector and can in principle be treated by the same methods. The essential observation is then that the formal parameter $\kappa$ is small. The terms in the asymptotic expansion of the shifted partition function in $\kappa$ and the entries $J_{kl}$, where the powers of $\kappa$ and $J$ are identical, compose the full partition function. Successfully computing the vacuum sector of the partition function demonstrates then that the partition function of the full theory can be computed.\\

This forms an extra argument to study the vacuum sector of the partition function of the Grosse-Wulkenhaar model.

\subsection{Vanishing coupling\label{sec:limsec}}
In some limiting cases the partition functions are much easier to determine. These results may be used to check  more advanced results for consistency later. For now, the case of vanishing coupling is of interest. Vanishing coupling means that $g=0$ in (\ref{e:qmm2}). The measure then is given by
\begin{equation*}
\ud X=\Big\{\prod_{k<l}\int_{-\infty}^{\infty}\ud X_{kl}^{(r)}\,\int_{-\infty}^{\infty}\ud X_{kl}^{(i)}\Big\}\times\Big\{\prod_{k=1}^{N}\int_{-\infty}^{\infty}\ud X_{kk}^{(r)}\Big\}\quad,
\end{equation*}
so that the resulting matrix
\begin{equation*}
X=\sum_{k=1}^{N}X_{kk}^{(r)}+\sum_{1\leq k<l\leq N}X_{kl}^{(r)}+i X_{kl}^{(i)}+\sum_{1\leq n<m\leq N}X_{mn}^{(r)}-i X_{mn}^{(i)}
\end{equation*}
is Hermitean. In these components it follows that the trace of the square of such matrices is given by
\begin{equation*}
\Tr EX^{2} = \sum_{k=1}^{N}E_{kk}\big((X_{kk}^{(r)})^{2}+\sum_{l=1}^{k-1}(X_{lk}^{(r)})^{2}+(X_{lk}^{(i)})^{2}+\sum_{l=k+1}^{N}(X_{kl}^{(r)})^{2}+(X_{kl}^{(i)})^{2}\big)\,.
\end{equation*}
This implies that the free partition function is given by
\begin{equation}
\mathscr{Z}[0]=\Big\{\prod_{k=1}^{N}\sqrt{\frac{\pi}{e_{k}}}\,\Big\}\cdot\Big\{\prod_{1\leq k<l\leq N}\frac{\pi}{e_{k}+e_{l}}\Big\}\quad.\label{e:zerocoup}
\end{equation}

\section{Schwinger splitting\label{sec:Ss}}

There is an obvious obstacle towards the integration of the partition function (\ref{e:ZJ0c}) for large $N$. It is the intertwinement of the eigenvalue integrals through the denominator. Performing the integral over $\lambda_{1}$ would change the form of the integrand. The only hope to perform all these integrals is to reformulate the denominator in such a way that the partition functions factorises.\\

A well-known first step is the Schwinger trick. For real $\zeta$ one rewrites
\begin{equation}
\frac{1}{\lambda_{k}+\lambda_{l}}=i\int_{0}^{\zeta}\ud u_{kl}\,e^{-iu_{kl}(\lambda_{k}+\lambda_{l})}+\frac{e^{-i\zeta(\lambda_{k}+\lambda_{l})}}{\lambda_{k}+\lambda_{l}}=\text{\textsf{S}}+\text{\textsf{R}}\label{e:ST}
\end{equation}
respectively. Applying this trick to the denominators in (\ref{e:ZJ0c}) yields
\begin{align}
&\hspace{-8mm}\prod_{k<l}\frac{1}{\lambda_{k}+\lambda_{l}}=\prod_{k<l}\left(i\int_{0}^{\zeta}\ud u_{kl}\,e^{-iu_{kl}(\lambda_{k}+\lambda_{l})}+\frac{e^{-i\zeta(\lambda_{k}+\lambda_{l})}}{\lambda_{k}+\lambda_{l}}\right)\label{e:ST2}\quad.
\end{align}

In (\ref{e:ST}) the denominator $(\lambda_{k}+\lambda_{l})^{-1}$ is rewritten in a part \textsf{S} and a rest term \textsf{R} respectively. The idea is now to choose a regime where all terms with \textsf{R} vanish. Schematically, the right-hand side of (\ref{e:ST2}) should
\begin{equation*}
(\textsf{S}+\textsf{R})^{(N-1)N/2}\rightarrow\textsf{S}^{(N-1)N/2}\quad,\qquad N\rightarrow\infty\quad.
\end{equation*}
This only happens, if \textsf{R} is strongly suppressed. Integrating the rest against the Gaussian function shows that this is certainly the case, if $\zeta\rightarrow\infty$. This argument can be modified to account for the weak coupling.

\begin{exm}\label{exm:constant2}

The steps in Example~\ref{exm:constant1} can be repeated numerically with a diagonal matrix $E$ with kinetic eigenvalues $\{1.0,1.1,1.2\}$ and coupling $g=0$
\begin{equation*}
\mathpzc{Z} = \int_{\mathpzc{H}_{3}} \ud M\, \exp[-\Tr EM^{2}]\approx14.64\quad,
\end{equation*}
which is close to the exact result $\mathpzc{Z}=14.142$, see (\ref{e:zerocoup}). This is a well-behaved $9$-dimensional integral.\\
This is the case $N=3$ of Hermitean $3\times 3$-matrices. The $3$-dimensional numerical integral after diagonalisation yields $\mathpzc{Z}\approx14.18$.\\
The next step would then be the application of the Harish-Chandra-Itzykson-Zuber integral. The reciprocals of a sum of two integration parameters $(\lambda_{k}+\lambda_{l})^{-1}$ can be written using the Schwinger trick and the polytope volume as the integral
\begin{equation*}
\mathpzc{Z} = \frac{(-i\pi)^{3}}{2}\int_{\mathbb{R}^{3}}\ud\vec{\lambda}\int_{\mathbb{R}_{+}^{3}}\ud\vec{u}\,\exp[-\sum_{j}e_{j}\lambda_{j}^{2}+iu_{j}\lambda_{j}]\,V_{N}(\vec{u})\cdot\big[\prod_{1\leq k<l\leq 3}\frac{\lambda_{l}-\lambda_{k}}{e_{l}-e_{k}}\big]\quad.
\end{equation*}
In the case $N=3$ or $N=4$ the polytope volume, (\ref{e:v32}) and (\ref{e:v42}) respectively, is easily computed exactly. This means that this step can be tested too. To avoid complex numbers, the average of the integal above and its symmetric ($\vec{\lambda} \leftrightarrow -\vec{\lambda}$) is computed. This yields
\begin{align}
&\hspace{-8mm}\mathpzc{Z} = \frac{\pi^{\binom{N}{2}}}{2}\int_{\mathbb{R}^{3}}\ud\vec{\lambda}\int_{\mathbb{R}_{+}^{3}}\ud\vec{u}\,V_{N}(\vec{u})\,\exp[-\sum_{j}e_{j}\lambda_{j}^{2}]\cdot\big[\prod_{1\leq k<l\leq 3}\frac{\lambda_{l}-\lambda_{k}}{e_{l}-e_{k}}\big]\\[2mm]
&\times\left\{\begin{array}{ll}
i^{\binom{N}{2}}\cos(\sum_{j}u_{j}\lambda_{j}) &\text{, if }(-1)^{\binom{N}{2}}=1\\
i^{1+\binom{N}{2}}\sin(\sum_{j}u_{j}\lambda_{j}) &\text{, if }(-1)^{\binom{N}{2}}=-1
\end{array}\right.\qquad.\label{e:constant2_e1}
\end{align}
However, even in the case $N=3$ it is not easy to find a stable result. The simple Monte Carlo methods used above yields wildly varying results. Even straightforward numerical integration methods with equal spacing depend highly on the lattice parameter chosen, which is explained by the oscillating part of the integral. It is a result of the interplay between the periodic sites and the periodic trigonometric functions. Especially for large $\lambda_{j}$'s, the (co)sine is rapidly oscillating as a function of the $u_{j}$'s. The lattice parameter $a$ needed to integrate this reliably is very small. Choosing the lattice parameter too big, produces instable results.\\

The eigenvalues $\lambda_{j}\in[-3.0,3.0]$ and $u_{j}\in[0.0,3.0]$ yields the results from Table~\ref{t:constant2} for various lattice parameters $a$. Although not decisive, this check slowly moves in the right direction as $a$ becomes smaller. Furthermore, it makes clear why similar checks for larger $N$ are not feasible. The asymptotic volume formula would have to be used and it is only reasonably accurate for $N\geq 7$.
\begin{table}[!h]\small
\begin{tabular}{c|ccc|c}
$a$ & $\mathpzc{Z}$ &\hspace{15mm}& $a$ &$\mathpzc{Z}$\\\hline
$0.200$ & $08.68$ && $0.050$ & $15.12$\\
$0.150$ & $20.85$ && $0.048$ & $10.83$\\
$0.100$ & $06.93$ && $0.046$ & $14.83$\\
$0.090$ & $15.08$ && $0.044$ & $11.52$\\
$0.080$ & $11.17$ && $0.042$ & $15.01$\\
$0.070$ & $16.13$ && $0.040$ & $11.39$\\
$0.060$ & $10.63$ &&         &
\end{tabular}\normalsize
\caption{Numerical results for (\ref{e:constant2_e1}) with $N=3$ on $[-3,3]^{3}\times[0,3]^{3}$ for various lattice parameters $a$.\label{t:constant2}}
\end{table}

\end{exm}

\section{The partition function of the vacuum sector}

In the previous paragraph we have shown that if $\zeta$ goes to infinity, all relevant contributions to the partition function are contained in
\begin{align}
&\hspace{-8mm}\mathscr{Z}[0]=\mathpzc{U}\frac{(-i)^{\binom{N}{2}}[\prod_{k=0}^{N-1}k!]}{\Delta(e_{1},\ldots,e_{N})}\int\ud\vec{\lambda}\,\Delta(\lambda_{1},\ldots,\lambda_{N})e^{-g\sum_{j=1}^{N}\lambda_{j}^{4}}\nonumber\\
&\times\det_{m,n}\big(e^{-e_{m}\lambda_{n}^{2}}\big)\cdot\Big(\prod_{k<l}\int_{0}^{\zeta}\ud u_{kl}\,e^{-iu_{kl}(\lambda_{k}+\lambda_{l})}\Big)\quad.\label{e:zld1}
\end{align}
In (\ref{e:zld1}) the dependence on the $u_{kl}$'s is in fact a dependence on 
\begin{equation}
u_{k}=\sum_{j=1}^{k}u_{jk}+\sum_{j=k+1}^{N}u_{kj}\quad.\label{e:uj}
\end{equation}

Substituting the integration variables yields
\begin{equation*}
\prod_{k<l}\int_{0}^{\zeta}\ud u_{kl}\, f(\vec{u})=\frac{1}{2}\Big(\prod_{m=1}^{N}\int_{0}^{(N-1)\zeta}\ud u_{m}\Big)\, f(\vec{u}) V_{N}(\vec{u})\quad.
\end{equation*}
To find out what the function $V_{N}(\vec{u})$ is, the equation (\ref{e:uj}) is written explicitly as a matrix equation
\begin{equation*}
\left(\begin{array}{cccc}0&u_{12}&\ldots&u_{1N}\\u_{12}&0&\ldots&u_{2N}\\\vdots&\vdots&\ddots&\vdots\\u_{1N}&u_{2N}&\ldots&0\end{array}\right)\left(\begin{array}{c}1\\1\\\vdots\\1\end{array}\right)=\left(\begin{array}{c}u_{1}\\u_{2}\\\vdots\\u_{N}\end{array}\right)\quad.
\end{equation*}
Choosing for now the formulation with all $u_{j}\in(0,1)$, it follows that $h_{j}=1-u_{j}$ lies between $0$ and $1$ and that to such a matrix equation a unique symmetric stochastic matrix 
\begin{equation*}
\left(\begin{array}{cccc}h_{1}&u_{12}&\ldots&u_{1N}\\u_{12}&h_{2}&\ldots&u_{2N}\\\vdots&\vdots&\ddots&\vdots\\u_{1N}&u_{2N}&\ldots&h_{N}\end{array}\right)
\end{equation*}
corresponds. A matrix is stochastic, when all its entries are nonnegative and every row sums to $1$.\\
This implies that the function $V_{N}(\vec{u})$ with all $u_{j}\in[0,1]$ is the volume of the space of symmetric stochastic matrices with diagonal entries $\{1-u_{1},\ldots,1-u_{N}\}$. It is straightforward to check that this space is convex, so that this space is a $\frac{N(N-3)}{2}$-dimensional polytope.\\

The relevant parameters in the integrals in (\ref{e:zld1}) over the $u_{kl}$'s are the sums $u_{j}$. For $N\geq3$ a change of integration variables can be made. Instead of the integration variables $u_{12},\ldots,u_{1N},u_{23}$, the sums $u_{1},\ldots,u_{N}$ are used. Because these sums cover all variables twice, the absolute determinant of the Jacobian of this transformation is $2$. It is not dificult to see that this holds for any $N$. An example of this is the $N=4$-Jacobian in Table~\ref{t:jac4}. For this substitution of variables the determinant is easily calculated. However, because these vectors are not perpendicular, we can not integrate them independently.\\
A more symmetric coordinate transformation, for example\\$\{u_{12},u_{23},u_{34},\ldots,u_{(N-1)N},u_{1N}\}\rightarrow\{u_{1},\ldots u_{N}\}$ has Jacobian $0$.\\
\footnotesize
\begin{table}[!h]
\hspace*{3cm}\footnotesize{\begin{tabular}{l|cccccc}
 & $u_{12}$ & $u_{13}$&$u_{23}$ &$u_{14}$ &$u_{24}$ &$u_{34}$ \\\hline
$u_{1}$  &$1 $&$1 $&$0 $&$1 $&$0 $&$0 $\\
$u_{2}$  &$1 $&$0 $&$1 $&$0 $&$1 $&$0 $\\
$u_{3}$  &$0 $&$1 $&$1 $&$0 $&$0 $&$1 $\\
$u_{4}$  &$0 $&$0 $&$0 $&$1 $&$1 $&$1 $\\
$u_{24}$ &$0 $&$0 $&$0 $&$0 $&$1 $&$0 $\\
$u_{34}$ &$0 $&$0 $&$0 $&$0 $&$0 $&$1 $
\end{tabular}}
\caption{The Jacobian corresponding to (\ref{e:uj}) for $N=4$. \label{t:jac4}}
\end{table}
\vspace{5mm}\\\normalsize

These steps demonstrate that the partition function (\ref{e:zld1}) can be rewritten as
\begin{align}
&\hspace{-8mm}\mathscr{Z}[0]=\frac{\mathpzc{U}(-i)^{\binom{N}{2}}[\prod_{k=0}^{N-1}k!]}{2\Delta(e_{1},\ldots,e_{N})}\int_{0}^{(N-1)\zeta}\ud\vec{u}\,V_{N}(\vec{u})\nonumber\\
&\times\int\ud\vec{\lambda}\,\Delta(\lambda_{1},\ldots,\lambda_{N})e^{-g\sum_{j}(\lambda_{j}^{4}+iu_{j}\lambda_{j})}\det_{m,n}\big(e^{-e_{m}\lambda_{n}^{2}}\big)\label{e:zld2}\\
&=\frac{\mathpzc{U}(-i)^{\binom{N}{2}}((N-1)\zeta)^{N}}{2\Delta(e_{1},\ldots,e_{N})}\int_{0}^{1}\ud\vec{u}\,V_{N}((N-1)\zeta\vec{u})\nonumber\\
&\times\int\ud\vec{\lambda}\,\Delta(\lambda_{1},\ldots,\lambda_{N})e^{-g\sum_{j}(\lambda_{j}^{4}+i(N-1)\zeta u_{j}\lambda_{j})}\det_{m,n}\big(e^{-e_{m}\lambda_{n}^{2}}\big)\nonumber\\
&=\frac{\mathpzc{U}(-i(N-1)\zeta)^{\binom{N}{2}}}{2\Delta(e_{1},\ldots,e_{N})}\int_{0}^{1}\ud\vec{u}\,V_{N}(\vec{u})\nonumber\\
&\times\int\ud\vec{\lambda}\,\Delta(\lambda_{1},\ldots,\lambda_{N})e^{-g\sum_{j}(\lambda_{j}^{4}+i(N-1)\zeta u_{j}\lambda_{j})}\det_{m,n}\big(e^{-e_{m}\lambda_{n}^{2}}\big)\nonumber\quad,
\end{align}
where we have used that the volume of the polytope obeys the scaling law
\begin{equation*}
V_{N}(\vec{u})=M^{\frac{-N(N-3)}{2}}V_{N}(M\vec{u})\qquad\text{for}\quad M\in\mathbb{R}_{+}\quad.
\end{equation*}

In Chapter~\ref{sec:polytope} it is shown that the polytope volume is given by
\begin{align*}
&\hspace{-8mm}\vol(P_{N}(\vec{h}))=\sqrt{2}e^{\frac{7}{6}}\big(\frac{e(N-\chi)}{N(N-1)}\big)^{\binom{N}{2}}\big(\frac{N(N-1)^{2}}{2\pi(N-\chi)^{2}}\big)^{\frac{N}{2}}\\
&\times\exp[-\frac{(N\!-\!1)^{2}}{2(N\!-\!\chi)^{2}}(N\!+\!2)\sum_{j}(h_{j}\!-\!\frac{\chi}{N})^{2}]\exp[-\frac{N(N\!-\!1)^{3}}{3(N\!-\!\chi)^{3}}\sum_{j}(h_{j}\!-\!\frac{\chi}{N})^{3}]\\
&\times\exp[-\frac{N(N-1)^{4}}{4(N-\chi)^{4}}\sum_{j}(h_{j}-\frac{\chi}{N})^{4}]\exp[\frac{(N-1)^{4}}{4(N-\chi)^{4}}\big(\sum_{j}(h_{j}-\frac{\chi}{N})^{2}\big)^{2}]\quad,
\end{align*}
provided that for all $j=1,\ldots,N$
\begin{equation*}
\lim_{N\rightarrow\infty}N^{\frac{1}{4}}\frac{N-1}{N-\chi}\cdot\big|h_{j}-\frac{\chi}{N}\big|=0\quad
\end{equation*}
where $\chi=\sum_{j}h_{j}$. Substituting $h_{j}=1-u_{j}$ and $\chi=N-S$ with $S=\sum_{j}u_{j}$ taking values in $[0,N]$ now yields
\begin{align}
&\hspace{-8mm}V_{N}(\vec{u})=\sqrt{2}e^{\frac{7}{6}}\big(\frac{eS}{N(N\!-\!1)}\big)^{\binom{N}{2}}\big(\frac{N(N\!-\!1)^{2}}{2\pi S^{2}}\big)^{\frac{N}{2}}\exp[-\frac{(N\!-\!1)^{2}}{2S^{2}}(N\!+\!2)\sum_{j}(u_{j}\!-\!\frac{S}{N})^{2}]\nonumber\\
&\times\exp[\frac{N(N-1)^{3}}{3S^{3}}\sum_{j}(u_{j}-\frac{S}{N})^{3}]\exp[-\frac{N(N-1)^{4}}{4S^{4}}\sum_{j}(u_{j}-\frac{S}{N})^{4}]\nonumber\\
&\times\exp[\frac{(N-1)^{4}}{4S^{4}}\big(\sum_{j}(u_{j}-\frac{S}{N})^{2}\big)^{2}]\quad,\label{e:polvol}
\end{align}
provided that
\begin{equation}
\lim_{N\rightarrow\infty}N^{\frac{1}{4}}\frac{N-1}{S}\cdot\big|u_{j}-\frac{S}{N}\big|=0\quad.\label{e:polvolcond}
\end{equation}
This condition implies that
\begin{equation}
\lim_{N\rightarrow\infty}\sum_{j=1}^{N}N^{-1+\frac{k}{4}}\big(\frac{N-1}{S}\big)^{k}\cdot\big|u_{j}-\frac{S}{N}\big|^{k}=0\qquad\forall k\geq2\quad.\label{e:polvolcond_k}
\end{equation}
Asymptotically, these conditions cover almost all volume of the polytope of symmetric stochastic matrices. The lion's share of the volume is located at small $\chi$, or large $S$.\\

Including the Fourier representation of a delta-function $\delta(S-\sum_{j}u_{j})$ in the partition function (\ref{e:zld2}) shows that
\begin{align}
&\hspace{-8mm}\mathscr{Z}[0]=\mathpzc{U}\frac{\sqrt{2}e^{\frac{7}{6}}(-i)^{\binom{N}{2}}[\prod_{k=0}^{N-1}k!]}{2\Delta(e_{1},\ldots,e_{N})}\!\int_{-\infty}^{\infty}\!\!\!\!\!\!\ud\kappa\!\int_{0}^{N(N-1)\zeta}\!\!\!\!\!\!\ud S\!\int_{0}^{(N-1)\zeta}\!\!\!\!\!\!\ud\vec{u}\,\big(\frac{eS}{N(N\!-\!1)}\big)^{\binom{N}{2}}\nonumber\\
&\times\big(\frac{N(N\!-\!1)^{2}}{2\pi S^{2}}\big)^{\frac{N}{2}}\exp[2\pi i\kappa(S-\sum_{j}u_{j})]\nonumber\\
&\times\exp[-\frac{(N-1)^{2}}{2S^{2}}(N+2)\sum_{j}(u_{j}-\frac{S}{N})^{2}]\exp[\frac{N(N-1)^{3}}{3S^{3}}\sum_{j}(u_{j}-\frac{S}{N})^{3}]\nonumber\\
&\times\exp[-\frac{N(N-1)^{4}}{4S^{4}}\sum_{j}(u_{j}-\frac{S}{N})^{4}]\exp[\frac{(N-1)^{4}}{4S^{4}}\big(\sum_{j}(u_{j}-\frac{S}{N})^{2}\big)^{2}]\nonumber\\
&\times\int_{-\infty}^{\infty}\!\!\!\!\ud\vec{\lambda}\,\Delta(\lambda_{1},\ldots,\lambda_{N})e^{-\sum_{j}(g\lambda_{j}^{4}+iu_{j}\lambda_{j})}\det_{m,n}\big(e^{-e_{m}\lambda_{n}^{2}}\big)\label{e:zld5}\quad.
\end{align}

\section{No coupling\label{sec:noco}}

So far, the factorisation of the partition function is little more than a nice idea. The polytope volume as an integration measure has only been demonstrated for $N=3$ in a small numerical check. To test whether this evaluation method has any chance of succeeding we return to the free theory. The partition function for the free model (\ref{e:zerocoup}) is rewritten using the polytope volume. Comparing the outcome to the starting point will provide us with some information on this matter.\\

Additionally, the various steps up to this point have made the expressions only more complicated. A way to find out how these expressions should be treated is through a test calculation, where the outcome is known in advance. Introducing the polytope volume in (\ref{e:zerocoup}) yields
\begin{align}
&\hspace{-8mm}\lim_{g\rightarrow0}\mathscr{Z}[0]=\big\{\prod_{k=1}^{N}\sqrt{\frac{\pi}{e_{k}}}\big\}\prod_{k<l}\frac{\pi}{e_{k}+e_{l}}\label{e:tc1}\\
&=\pi^{\binom{N}{2}}\big\{\prod_{k=1}^{N}\sqrt{\frac{\pi}{e_{k}}}\big\}\prod_{k<l}\int_{0}^{\zeta}\!\!\ud u_{kl}\,e^{-u_{kl}(e_{k}+e_{l})}\nonumber\\
&=\frac{1}{2}\pi^{\binom{N}{2}}\big\{\prod_{k=1}^{N}\sqrt{\frac{\pi}{e_{k}}}\big\}\int_{0}^{(N-1)\zeta}\!\!\!\!\!\!\ud^{N}\vec{u}\,V_{N}(\vec{u})e^{-\sum_{m}u_{m}e_{m}}\nonumber\\
&=\frac{1}{2}\pi^{\binom{N}{2}}\big\{\prod_{k=1}^{N}\sqrt{\frac{\pi}{e_{k}}}\big\}\int_{0}^{(N-1)\zeta}\!\!\!\!\!\!\!\!\!\ud\vec{u}\,e^{-\sum_{m}u_{m}e_{m}}\int_{-\infty}^{\infty}\!\!\!\!\!\!\ud\kappa\int_{-\infty}^{\infty}\!\!\!\!\!\!\ud q\int_{0}^{N(N-1)\zeta}\!\!\!\!\!\!\!\!\!\!\!\!\ud S\int_{0}^{\infty}\!\!\!\!\!\!\ud Q\nonumber\\
&\times\sqrt{2}e^{\frac{7}{6}}\big(\frac{eS}{N(N\!-\!1)}\big)^{\binom{N}{2}}\big(\frac{N(N\!-\!1)^{2}}{2\pi S^{2}}\big)^{\frac{N}{2}}\exp[2\pi i\kappa(S\!-\!\sum_{j}\!u_{j})]\nonumber\\
&\times\exp[-\frac{(N\!-\!1)^{2}}{2S^{2}}(N\!+\!2)\sum_{j}(u_{j}\!-\!\frac{S}{N})^{2}]\nonumber\\
&\times\exp[\frac{N(N-1)^{3}}{3S^{3}}\sum_{j}(u_{j}-\frac{S}{N})^{3}]\exp[-\frac{N(N-1)^{4}}{4S^{4}}\sum_{j}(u_{j}-\frac{S}{N})^{4}]\nonumber\\
&\times\exp[\frac{(N\!-\!1)^{4}}{4S^{4}}Q^{2}]\exp[2\pi iq(Q\!-\!\sum_{j}\!(u_{j}\!-\!\frac{S}{N})^{2})]\label{e:tc2}\\
&=\frac{1}{2}\pi^{\binom{N}{2}}\big\{\!\prod_{k=1}^{N}\!\sqrt{\frac{\pi}{e_{k}}}\big\}(N\!-\!1)\!\!\int_{-\infty}^{\infty}\!\!\!\!\!\!\!\ud\kappa\!\int_{-\infty}^{\infty}\!\!\!\!\!\!\!\!\ud q\!\int_{0}^{N\zeta}\!\!\!\!\!\!\!\!\ud S\!\int_{0}^{\sqrt{N}}\!\!\!\!\!\!\!\!\ud Q\!\int_{-N^\frac{1}{4}}^{N^\frac{1}{4}}\!\!\ud\vec{x}\,\frac{\sqrt{2N}e^{7/6}}{(2\pi)^{N/2}S}\nonumber\\
&\times\exp[-2\pi i\kappa\!\sum_{j}\!x_{j}-\!\sum_{j}\!e_{j}S(\frac{x_{j}}{\sqrt{N}}\!+\!\frac{N\!-\!1}{N})]\big(\frac{eS}{N}\big)^{\binom{N}{2}}\exp[-\frac{N\!+\!2}{2N}\!\sum_{j}\!x_{j}^{2}]\nonumber\\
&\times\exp[\frac{\sum_{j}x_{j}^{3}}{3\sqrt{N}}-\frac{\sum_{j}x_{j}^{4}}{4N}+\frac{Q^{2}}{4}]\exp[2\pi iq(Q-\sum_{j}\frac{x_{j}^{2}}{N})]\quad,\label{e:tc3}
\end{align}
where in the last step $u_{j}\rightarrow x_{j}+S/N$, $S\rightarrow S(N-1)$, $x_{j}\rightarrow x_{j}S/\sqrt{N}$, $Q\rightarrow QS^{2}$, $q\rightarrow qS^{-2}$ and $\kappa\rightarrow \kappa \sqrt{N}/S$ respectively. The application range of the polytope volume formula (\ref{e:polvolcond}), $|u_{j}-S/N|\ll SN^{-\frac{5}{4}}$ in (\ref{e:tc2}), implies integration boundaries for the integration parameters $x_{j}$ in (\ref{e:tc3}).\\

It is not straightforward to see what the most convenient integration order in (\ref{e:tc3}) is. There are several starting points conceivable. What can be seen is that the integral over $S$ can be performed directly. This yields a Gamma function and a fraction depending on $\sum_{j}x_{j}e_{j}$ to the power $\small{\binom{N}{2}}$. For the integral over $x_{j}$ it becomes necessary to make some assumptions on the kinetic model parameters $e_{j}$. Assuming that
\begin{equation}
e_{j}=\xi(1+\tilde{\varepsilon}_{j})\qquad\text{, where }\xi=\frac{1}{N}\sum_{j=1}^{N}e_{j}\label{e:kinpar}
\end{equation}
with $|\tilde{\varepsilon}_{j}|\ll N^{-\frac{1}{4}}$ allows us to approximate the $x_{j}$-dependence in the fraction $\sum(j)e_{j}(\frac{N-1}{N}+\frac{x_{j}}{\sqrt{N}})$ with exponentials, which can be integrate by the stationary phase method. This is also used for the integral over $\kappa$. The remaining integrals are Fourier transforms of the Dirac delta and are therefore straightforward to integrate. All details of this computation can be found in Appendix~\ref{sec:no_coup_comp}.\\

Application of Stirling's formula (\ref{e:gammastir}) to the Gamma funcions in (\ref{e:tc6c}) yields the partition function of the free theory computed via the polytope volume
\begin{align}
&\hspace{-8mm}\lim_{g\rightarrow0}\mathscr{Z}[0]=\big(\frac{\pi}{2\xi}\big)^{\binom{N}{2}}\big\{\!\prod_{k=1}^{N}\!\sqrt{\frac{\pi}{e_{k}}}\big\}\exp[\frac{N\!-\!2}{8}\sum_{j}\!\tilde{\varepsilon}_{j}^{2}\!-\!\frac{N\!-\!6}{24}\sum_{j}\!\tilde{\varepsilon}_{j}^{3}\!+\!\frac{N}{64}\sum_{j}\!\tilde{\varepsilon}_{j}^{4}]\nonumber\\
&\times\exp[\frac{3}{64}(\sum_{j}\tilde{\varepsilon}_{j}^{2})^{2}-\frac{1}{16}(\sum_{j}\tilde{\varepsilon}_{j}^{2})(\sum_{j}\tilde{\varepsilon}_{j}^{3})+\frac{7}{128}(\sum_{j}\tilde{\varepsilon}_{j}^{2})(\sum_{j}\tilde{\varepsilon}_{j}^{4})]\nonumber\\
&\times\exp[\frac{3}{128}(\sum_{j}\tilde{\varepsilon}_{j}^{3})^{2}-\frac{5}{128}(\sum_{j}\tilde{\varepsilon}_{j}^{3})(\sum_{j}\tilde{\varepsilon}_{j}^{4})+\frac{1}{16N}(\sum_{j}\tilde{\varepsilon}_{j}^{2})^{3}]\nonumber\\
&\exp[-\frac{11}{128N}(\sum_{j}\tilde{\varepsilon}_{j}^{2})^{2}(\sum_{j}\tilde{\varepsilon}_{j}^{3})]\label{e:tc6d}\quad.
\end{align}

This is to be compared to
\begin{align}
&\hspace{-8mm}\lim_{g\rightarrow0}\mathscr{Z}[0]=\big\{\prod_{k=1}^{N}\sqrt{\frac{\pi}{e_{k}}}\big\}\prod_{k<l}\frac{\pi}{e_{k}+e_{l}}=\big\{\prod_{k=1}^{N}\sqrt{\frac{\pi}{e_{k}}}\big\}\big(\frac{\pi}{2\xi}\big)^{\binom{N}{2}}\nonumber\\
&\times\exp\big[\!-\!\sum_{k<l}\!\big\{\!\big(\frac{\tilde{\varepsilon}_{k}\!+\!\tilde{\varepsilon}_{l}}{2}\big)-\frac{1}{2}\big(\frac{\tilde{\varepsilon}_{k}\!+\!\tilde{\varepsilon}_{l}}{2}\big)^{2}+\frac{1}{3}\big(\frac{\tilde{\varepsilon}_{k}\!+\!\tilde{\varepsilon}_{l}}{2}\big)^{3}-\frac{1}{4}\big(\frac{\tilde{\varepsilon}_{k}\!+\!\tilde{\varepsilon}_{l}}{2}\big)^{4}\big\}\big]\nonumber\\
&=\big\{\prod_{k=1}^{N}\sqrt{\frac{\pi}{e_{k}}}\big\}\big(\frac{\pi}{2\xi}\big)^{\binom{N}{2}}\exp[\frac{N-2}{8}\sum_{j}\tilde{\varepsilon}_{j}^{2}-\frac{N-4}{24}\sum_{j}\tilde{\varepsilon}_{j}^{3}]\nonumber\\
&\times\exp[\frac{N\!-\!8}{64}\sum_{j}\!\tilde{\varepsilon}_{j}^{4}\!+\!\frac{3}{64}(\sum_{j}\!\tilde{\varepsilon}_{j}^{2})^{2}\!-\!\frac{N\!-\!16}{160}\sum_{j}\!\tilde{\varepsilon}_{j}^{5}\!-\!\frac{1}{16}(\sum_{j}\tilde{\varepsilon}_{j}^{2})(\sum_{j}\tilde{\varepsilon}_{j}^{3})]\nonumber\\
&\times\exp[\frac{N-32}{384}\sum_{j}\tilde{\varepsilon}_{j}^{6}+\frac{5}{128}(\sum_{j}\tilde{\varepsilon}_{j}^{2})(\sum_{j}\tilde{\varepsilon}_{j}^{4})+\frac{5}{96}(\sum_{j}\tilde{\varepsilon}_{j}^{3})^{2}]\quad.\label{e:tc7}
\end{align}
This is the same, provided that $\tilde{\varepsilon}_{j} \ll N^{-\frac{1}{3}}$. It does not seem possible to extend this any further.\\

The above computation has not yielded new insight into the free theory's partition function. However, it does show that the method of factorisation and integration against the polytope volume functions. Besides that, it indicates that the formulas used are correct. And finally, it provides us with some experience performing such calculations.\\

The polytope volume calculation turned all parameters into symmetric sums, whereas the direct computation in Paragraph~\ref{sec:limsec} is given in pairs of eigenvalues. The difference stems from the asymptotic formulation of the polytope volume, where the matrix structure has disappeared. This leads to a trade off between structural integrity and computability.

\section{Weak coupling\label{sec:weco}}

Inspired by the success and insights of the previous paragraph we may try to repeat the calculation for weak coupling $g\rightarrow0$. The computation without coupling in Paragraph~\ref{sec:noco} shows that the factorisation procedure with the asymptotic polytope volume alters the partition function structure. The strictly positive coupling means that the eigenvalue integrals must be performed after factorisation. To complicate the analysis not further, it is assumed that the coupling is small, so that several terms may be ignored. It should be pointed out explicitly that this is nowhere essential and it is only done to keep these lengthy calculations as succinct as possible.\\

There is one way to find out what the consequences of the factorisation procedure for the partition function will be. The starting point is (\ref{e:zld2}) with (\ref{e:polvol}), where the eigenvalue integrals are symmetrised as described in Paragraph~\ref{sec:symHCIZ}, so that
\begin{align}
&\hspace{-8mm}\mathscr{Z}[0]=\mathpzc{U}\frac{(-i)^{\binom{N}{2}}[\prod_{k=0}^{N-1}k!](N!)}{2\Delta(e_{1},\ldots,e_{N})}\int_{0}^{(N-1)\zeta}\!\!\!\!\ud \vec{u}\int \ud \vec{\lambda}\,\Delta(\lambda_{1},\ldots,\lambda_{N})V_{N}(\vec{u})\nonumber\\
&\times\exp[\sum_{j}(-g\lambda_{j}^{4}-e_{j}\lambda_{j}^{2}-iu_{j}\lambda_{j})]\nonumber\\
&=\frac{\sqrt{2N}e^{\frac{7}{6}}(N-1)}{2\Delta(e_{1},\ldots,e_{N})}\big(\frac{1}{2\pi}\big)^{\frac{N}{2}}\!\!\!\int_{-\infty}^{\infty}\!\!\!\!\!\!\!\ud\kappa\!\int_{0}^{N\zeta}\frac{\ud S}{S}\!\int_{-\infty}^{\infty}\!\!\!\!\!\!\!\ud\vec{x}\!\int_{-\infty}^{\infty}\!\!\!\!\!\!\!\ud q\!\int_{0}^{\infty}\!\!\!\!\!\!\!\ud Q\!\,\big(\frac{-\pi ieS}{N}\big)^{\binom{N}{2}}\nonumber\\
&\times\int_{-\infty}^{\infty}\!\!\!\!\!\!\!\ud\vec{\lambda}\,\Delta(\lambda_{1},\ldots,\lambda_{N})\,\exp[\sum_{j}\big(-e_{j}\lambda_{j}^{2}-g\lambda_{j}^{4}-iS\frac{x_{j}}{\sqrt{N}}\lambda_{j}-iS\frac{N-1}{N}\lambda_{j}\big)]\nonumber\\
&\times \exp[2\pi i q (Q-\sum_{j}\frac{x_{j}^{2}}{N})]\exp[-2\pi i\kappa \sum_{j}x_{j}]\exp[-\frac{N+2}{2N}\sum_{j}x_{j}^{2}]\nonumber\\
&\times\exp[\frac{1}{3\sqrt{N}}\sum_{j}x_{j}^{3}]\exp[\frac{-1}{4N}\sum_{j}x_{j}^{4}]\exp[\frac{Q^{2}}{4}]\label{e:wc0}
\end{align}
is obtained. The same transformations as in Paragraph~\ref{sec:noco} were used here, $u_{j}=x_{j}+S/N$, $S\rightarrow S(N-1)$, $x_{j}\rightarrow x_{j}S/\sqrt{N}$, $Q\rightarrow QS^{2}$, $q\rightarrow qS^{-2}$ and $\kappa\rightarrow \kappa\sqrt{N}/S$. 
The integration order of the extra integrals is borrowed from Paragraph~\ref{sec:noco}. This means that the integral over $S$ is performed first. 
Using that for $A\in\mathbb{R}$, unequal to zero
\begin{equation*}
\int_{0}^{\infty}\ud S\,S^{n-1}e^{-iAS}=\big(\frac{-i}{A}\big)^{n}\Gamma(n)
\end{equation*}
this becomes
\begin{align}
&\hspace{-8mm}\mathscr{Z}[0]=\frac{\sqrt{2N}e^{\frac{7}{6}}(N-1)}{2\Delta(e_{1},\ldots,e_{N})}\big(\frac{1}{2\pi}\big)^{\frac{N}{2}}\!\!\!\int_{-\infty}^{\infty}\!\!\!\!\!\!\!\ud\kappa\!\int_{-\infty}^{\infty}\!\!\!\!\!\!\!\ud\vec{x}\!\int_{-\infty}^{\infty}\!\!\!\!\!\!\!\ud q\!\int_{0}^{\infty}\!\!\!\!\!\!\!\ud Q\!\int_{-\infty}^{\infty}\!\!\!\!\!\!\!\ud \mu\!\int_{-\infty}^{\infty}\!\!\!\!\!\!\!\ud X\!\int_{-\infty}^{\infty}\!\!\!\!\!\!\!\ud \nu\!\int_{-\infty}^{\infty}\!\!\!\!\!\!\!\ud \Lambda\,\nonumber\\
&\times\big(\frac{-\pi e}{N}\big)^{\binom{N}{2}}\int_{-\infty}^{\infty}\!\!\!\!\!\!\!\ud\vec{\lambda}\,\Delta(\lambda_{1},\ldots,\lambda_{N})\,\exp[\sum_{j}\big(-e_{j}\lambda_{j}^{2}-g\lambda_{j}^{4}\big)]\nonumber\\
&\times \exp[2\pi i q (Q-\sum_{j}\frac{x_{j}^{2}}{N})]\exp[-2\pi i\kappa \sum_{j}x_{j}]\exp[-\frac{N+2}{2N}\sum_{j}x_{j}^{2}]\nonumber\\
&\times\exp[\frac{1}{3\sqrt{N}}\sum_{j}x_{j}^{3}]\exp[\frac{-1}{4N}\sum_{j}x_{j}^{4}]\exp[\frac{Q^{2}}{4}]\frac{\Gamma(\binom{N}{2})}{(\Lambda +X)^{\binom{N}{2}}}\nonumber\\
&\times\exp[2\pi i\mu(X-\frac{1}{\sqrt{N}}\sum_{j}x_{j}\lambda_{j})+2\pi i\nu(\Lambda-\frac{N-1}{N}\sum_{j}\lambda_{j})]\label{e:wc16}\quad.
\end{align}
Integrating over $x_{j}$ using the stationary phase method from Lemma~\ref{l:aspm} and scaling $\kappa\rightarrow \kappa/\sqrt{N}$ yields
\begin{align}
&\hspace{-8mm}\mathscr{Z}[0]=\frac{\sqrt{2}e^{\frac{7}{6}}(N-1)\Gamma(\binom{N}{2})}{2\Delta(e_{1},\ldots,e_{N})}\int_{-\infty}^{\infty}\!\!\!\!\!\!\!\ud\kappa\!\int_{-\infty}^{\infty}\!\!\!\!\!\!\!\ud q\!\int_{0}^{\infty}\!\!\!\!\!\!\!\ud Q\!\int_{-\infty}^{\infty}\!\!\!\!\!\!\!\ud \mu\!\int_{-\infty}^{\infty}\!\!\!\!\!\!\!\ud X\!\int_{-\infty}^{\infty}\!\!\!\!\!\!\!\ud \nu\!\int_{-\infty}^{\infty}\!\!\!\!\!\!\!\ud \Lambda\,\nonumber\\
&\times\big(\frac{-\pi e}{N(\Lambda +X)}\big)^{\binom{N}{2}}\int_{-\infty}^{\infty}\!\!\!\!\!\!\!\ud\vec{\lambda}\,\Delta(\lambda_{1},\ldots,\lambda_{N})\,\exp[\sum_{j}\big(-e_{j}\lambda_{j}^{2}-g\lambda_{j}^{4}\big)]\nonumber\\
&\times \exp[2\pi i q Q+\frac{Q^{2}}{4}+2\pi i\mu X+2\pi i\nu(\Lambda-\frac{N\!-\!1}{N}\sum_{j}\lambda_{j})]\exp[\frac{5}{6}-\frac{3}{4}]\nonumber\\
&\times(1+\frac{2}{N}+\frac{4\pi iq}{N})^{-\frac{N}{2}}\exp[\sum_{j}\frac{[-2\pi i(\kappa+\mu\lambda_{j})]^{2}}{2N\{1+\frac{2}{N}+\frac{4\pi iq}{N}\}}]\label{e:wc17b}\\
&\times\exp[\sum_{j}\frac{[\ldots]^{3}}{3N^{2}\{\ldots\}^{3}}]\exp[\sum_{j}\frac{[\ldots]}{N\{\ldots\}^{2}}]\exp[-\sum_{j}\frac{[\ldots]^{4}}{4N^{3}\{\ldots\}^{4}}]\nonumber\\
&\times\exp[\sum_{j}\!\frac{[\ldots]^{4}}{2N^{3}\{\ldots\}^{5}}]\exp[\sum_{j}\!\frac{-3[\ldots]^{2}}{2N^{2}\{\ldots\}^{3}}]\exp[\sum_{j}\!\frac{2[\ldots]^{2}}{N^{2}\{\ldots\}^{4}}]\label{e:wc17}\quad.
\end{align}
The expressions inside the brackets $[\ldots]$ and $\{\ldots\}$ are the same as the expressions in the same brackets in the exponential on line (\ref{e:wc17b}). Selecting leading terms and using that $\sum_{j}\lambda_{j}=\frac{N}{N-1}\Lambda$ gives
\begin{align*}
&\hspace{-8mm}\int_{-\infty}^{\infty}\ud\kappa\,\exp[-2\pi i\kappa\big(1-\frac{4}{N}-\frac{8\pi iq}{N}-2\pi i\frac{\mu}{N\!-\!1}\Lambda(1-\frac{2}{N}-\frac{4\pi iq}{N})\big)]\\
&\times\exp[-2\pi^{2}\kappa^{2}(1-\frac{2}{N}-\frac{4\pi iq}{N}-\frac{4\pi i\mu\Lambda}{N(N\!-\!1)})]\exp[\frac{8\pi^{3}i\kappa^{3}}{3N}]\\
&=\frac{1}{\sqrt{2\pi}}\exp[-\frac{1}{2}\frac{\big(1-\frac{4}{N}-\frac{8\pi iq}{N}-2\pi i\frac{\mu}{N\!-\!1}\Lambda(1-\frac{2}{N}-\frac{4\pi iq}{N})\big)^{2}}{(1-\frac{2}{N}-\frac{4\pi iq}{N}-\frac{4\pi i\mu\Lambda}{N(N\!-\!1)})}]\\
&\times\exp[\frac{1}{N}\frac{\big(1-\frac{4}{N}-\frac{8\pi iq}{N}-2\pi i\frac{\mu}{N\!-\!1}\Lambda(1-\frac{2}{N}-\frac{4\pi iq}{N})\big)}{(1-\frac{2}{N}-\frac{4\pi iq}{N}-\frac{4\pi i\mu\Lambda}{N(N\!-\!1)})^{2}}]\\
&\times\exp[-\frac{1}{3N}\frac{\big(1-\frac{4}{N}-\frac{8\pi iq}{N}-2\pi i\frac{\mu}{N\!-\!1}\Lambda(1-\frac{2}{N}-\frac{4\pi iq}{N})\big)^{3}}{(1-\frac{2}{N}-\frac{4\pi iq}{N}-\frac{4\pi i\mu\Lambda}{N(N\!-\!1)})^{3}}]\\
&=\frac{1}{\sqrt{2\pi}}\exp[-\frac{1}{2}+\frac{2\pi i\mu\Lambda}{N-1}(1-\frac{4}{N})+\frac{2\pi^{2}\mu^{2}\Lambda^{2}}{(N-1)^{2}}]\quad.
\end{align*}
Repeating these steps for $\mu$ results in
\begin{align*}
&\hspace{-8mm}\int_{-\infty}^{\infty}\ud\mu\,\exp[2\pi i\mu\big(X+\frac{\Lambda}{N-1}(1-\frac{4}{N}-\frac{8\pi iq}{N})-\frac{\Lambda}{N-1}(1-\frac{4}{N}-\frac{8\pi iq}{N})\big)]\\
&\times\exp[2\pi^{2}\mu^{2}\big(\frac{\Lambda^{2}}{(N-1)^{2}}-\frac{4\Lambda^{2}}{N(N-1)^{2}}-\frac{1}{N}\sum_{j}\lambda^{2}\big)]=\delta(X)\quad.
\end{align*}
The integrand $\lambda_{j}^{N}\exp[-e_{j}\lambda_{j}^{2}]$ is maximal for $\tilde{\lambda}_{j}^{2}=N/(2e_{j})$, so that we may assume that the quadratic term is small in the integral over $\mu$. This yields then $\delta(X)$. This makes the integral over $X$ trivial. Putting things together yields
\begin{align}
&\hspace{-8mm}\mathscr{Z}[0]=\frac{e^{-\frac{1}{4}}(N-1)\Gamma(\binom{N}{2})}{2\sqrt{\pi}\Delta(e_{1},\ldots,e_{N})}\int_{-\infty}^{\infty}\!\!\!\!\!\!\!\ud q\!\int_{0}^{\infty}\!\!\!\!\!\!\!\ud Q\!\int_{-\infty}^{\infty}\!\!\!\!\!\!\!\ud \nu\!\int_{-\infty}^{\infty}\!\!\!\!\!\!\!\ud \Lambda\,\nonumber\\
&\times\big(\frac{-\pi e}{N\Lambda}\big)^{\binom{N}{2}}\int_{-\infty}^{\infty}\!\!\!\!\!\!\!\ud\vec{\lambda}\,\Delta(\lambda_{1},\ldots,\lambda_{N})\,\exp[\sum_{j}\big(-e_{j}\lambda_{j}^{2}-g\lambda_{j}^{4}\big)]\nonumber\\
&\times \exp[2\pi i q (Q-1)+\frac{Q^{2}}{4}]\exp[2\pi i\nu(\Lambda-\frac{N-1}{N}\sum_{j}\lambda_{j})]\label{e:wc18}\quad
\end{align}
which shows that the integral over $q$ and $Q$ are straightforward too. Now is the moment to write the determinant
\begin{equation}
\Delta(\lambda_{1},\ldots,\lambda_{N})=\big(i\varepsilon\big)^{\binom{N}{2}}\sum_{\sigma}\sgn(\sigma)\exp[\sum_{j}i\varepsilon\sigma(j)\lambda_{j}]\label{e:detexp}
\end{equation}
and integrate over $\lambda_{j}$
\begin{align*}
&\hspace{-8mm}\int_{-\infty}^{\infty}\ud\lambda_{j}\,\exp[-g\lambda_{j}^{4}-e_{j}\lambda_{j}^{2}-\lambda_{j}(\frac{2\pi i\nu(N\!-\!1)}{N}-i\varepsilon\sigma(j))]\\
&=\sqrt{\frac{\pi}{e_{j}}}\exp[-\frac{3g}{4e_{j}^{2}}-\frac{1}{4e_{j}}(\frac{2\pi \nu(N\!-\!1)}{N}-\varepsilon\sigma(j))^{2}]\\
&\times\exp[-\frac{g}{16e_{j}^{4}}(\frac{2\pi \nu(N\!-\!1)}{N}-\varepsilon\sigma(j))^{4}-\frac{3g}{4e_{j}^{3}}(\frac{2\pi \nu(N\!-\!1)}{N}-\varepsilon\sigma(j))^{2}]\\
&\approx\sqrt{\frac{\pi}{e_{j}}}\exp[-\frac{3g}{4e_{j}^{2}}-\frac{\pi^{2}\nu^{2}(N\!-\!1)^{2}}{N^{2}e_{j}}+\frac{\pi\nu(N\!-\!1)\varepsilon\sigma(j)}{Ne_{j}}]\quad.
\end{align*}
Here it is used that the coupling is small, so that the additional terms in the exponential may be ignored. Technically, it is not necessary to do this. A consequence would be that the matrix of derivatives from Paragraph~\ref{sec:cmd} becomes nondiagonal. However, since all the off-diagonal elements will depend on the coupling $g$, the coupling may be assumed to be so small, that the approximation results from Lemmas~\ref{l:scex1} -~\ref{l:scex4} may apply. This would yield a fundamental bound on the coupling, such that this method is still applicable.\\

Using the same integration strategy once more gives
\begin{align*}
&\hspace{-8mm}\int_{-\infty}^{\infty}\ud\nu\,\exp[2\pi i\nu(\Lambda-\sum_{j}\frac{i(N\!-\!1)\varepsilon\sigma(j)}{2Ne_{j}})-\frac{\pi^{2}\nu^{2}(N\!-\!1)^{2}}{N^{2}}(\sum_{m}e_{m}^{-1})]\\
&=\frac{N}{N\!-\!1}\sqrt{\frac{1}{\pi(\sum_{m}e_{m}^{-1})}}\\
&\times\exp[-\frac{N^{2}}{(N\!-\!1)^{2}(\sum_{m}e_{m}^{-1})}(\Lambda-\sum_{j}\frac{i(N\!-\!1)\varepsilon\sigma(j)}{2Ne_{j}})^{2}]
\end{align*}
and leads to the final integral
\begin{align*}
&\hspace{-8mm}\mathcal{Q}=\int_{-\infty}^{\infty}\ud\Lambda\,(i\varepsilon\Lambda)^{-\binom{N}{2}}\exp[-\frac{N^{2}\Lambda^{2}}{(N\!-\!1)^{2}(\sum_{m}e_{m}^{-1})}+\sum_{j}\frac{i\varepsilon N\Lambda\sigma(j)}{(N\!-\!1)(\sum_{m}e_{m}^{-1})e_{j}}\\
&+\frac{1}{\sum_{m}e_{m}^{-1}}(\sum_{j}\frac{\varepsilon\sigma(j)}{2e_{j}})^{2}]\\
&=\int_{-\infty}^{\infty}\!\!\!\!\!\!\ud\Lambda\,(i\varepsilon\Lambda)^{-\binom{N}{2}}\exp[-\frac{N^{2}\Lambda^{2}}{(N\!-\!1)^{2}(\sum_{m}e_{m}^{-1})}\!+\!\sum_{j}\frac{i\varepsilon N\Lambda\sigma(j)}{(N\!-\!1)(\sum_{m}e_{m}^{-1})e_{j}}]\\
&\times\exp[-\frac{(N\!-\!1)^{2}(\sum_{m}e_{m}^{-1})}{4N^{2}\Lambda^{2}}\big(\sum_{j}\frac{i\varepsilon N\Lambda\sigma(j)}{(N\!-\!1)(\sum_{m}e_{m}^{-1})e_{j}}\big)^{2}]\quad.
\end{align*}
In this formulation one may recognise a matrix determinant. It is explained in Chapter~\ref{sec:Det_tech} that the first nonvanishing coefficient is that of $\varepsilon^{\binom{N}{2}}$. Applying the power series expansion
\begin{equation*}
\partial_{x=0}^{\binom{N}{2}}e^{x+ax^{2}}=\partial_{x=0}^{\binom{N}{2}}\sum_{m=0}^{\infty}\frac{x^{m}}{m!}\sum_{l=0}^{\lfloor m/2\rfloor}\frac{(m!)\,a^{l}}{(l!)\cdot(m-2l)!}
\end{equation*}
with (\ref{e:detexp}) in the opposite direction yields a Vandermonde-determinant of the reciprocals of the dynamic model parameters. The remaining terms take the form
\begin{align}
&\hspace{-8mm}\mathcal{Q}=
\frac{(N\!-\!1)\sqrt{\sum_{m}e_{m}^{-1}}}{N}\big(\frac{N}{(N\!-\!1)(\sum_{m}e_{m}^{-1})}\big)^{\binom{N}{2}}\Delta(\frac{1}{e_{1}},\ldots,\frac{1}{e_{N}})\nonumber\\
&\times\int\ud \Lambda\,\exp[-\Lambda^{2}]\cdot\sum_{l=0}^{\lfloor\frac{N(N\!-\!1)}{4}\rfloor}\frac{\binom{N}{2}!}{(l!)\cdot(\binom{N}{2}-2l)!}\big(\frac{-1}{4\Lambda^{2}}\big)^{l}\label{e:wc19}
\end{align}

The formulation in (\ref{e:wc19}) is ambiguous, because it is not clear how the integral over $\Lambda$ should be performed for $l>0$. As a real integral a single term is divergent. Divergencies with $\Lambda=0$ are not present in (\ref{e:ZJ0b}), but may be the asymptotic generalisation of the termwise divergencies that occur there when $\lambda_{m}=-\lambda_{n}$. There it was possible to overcome this by symmetric (numerical) integration, but such an option does not appear here.\\

An alternative treatment presents itself. Extending the upper bound on the summation to infinity shows that this is a Meijer $G$-function, which can be integrated against any other such function. Most common functions can be represented as a Meijer $G$-function. For simplicity this can be done using the Gamma function termwise
\begin{equation*}
\int_{-\infty}^{\infty}\ud x\,x^{-2n}e^{-x^{2}}=\frac{2}{2}\int_{0}^{\infty}\ud t\,e^{-t}t^{-n-\frac{1}{2}}=\Gamma(\frac{1}{2}-n)=\frac{(-2)^{n}\sqrt{\pi}}{(2n-1)!!},
\end{equation*}
which would lead to
\begin{align*}
&\hspace{-8mm}\mathcal{Q}'=\sqrt{\frac{\pi}{\sum_{m}e_{m}^{-1}}}\big(\frac{N}{(N\!-\!1)(\sum_{m}\!e_{m}^{-1})}\big)^{\binom{N}{2}-1}\Delta(\frac{1}{e_{1}},\ldots,\frac{1}{e_{N}})\!\!\sum_{l=0}^{\lfloor\frac{N(N\!-\!1)}{4}\rfloor}\!\!\binom{\binom{N}{2}}{2l}\\
&=\frac{\sqrt{\pi}}{2}\frac{(N\!-\!1)\sqrt{\sum_{m}e_{m}^{-1}}}{N}\big(\frac{2N}{(N\!-\!1)(\sum_{m}e_{m}^{-1})}\big)^{\binom{N}{2}}\Delta(\frac{1}{e_{1}},\ldots,\frac{1}{e_{N}})\quad.
\end{align*}
The combinatorial identity used at the last equality follows from
\begin{equation*}
\sum_{l=0}^{\lfloor\frac{N(N\!-\!1)}{4}\rfloor}\binom{\binom{N}{2}}{2l}=\frac{1}{2}\sum_{l=0}^{\binom{N}{2}}\binom{\binom{N}{2}}{l} (1^{l}+(-1)^{l})=\frac{1}{2}\big((1+1)^{\binom{N}{2}}+(1-1)^{\binom{N}{2}}\big)\quad.
\end{equation*}
The result of this alternative is a multiplication by a factor $2^{\binom{N}{2}-1}$. The factor is easily overlooked and there is no compelling argument to remove it. Certainly, it is not due to the formulation as an asymptotic expansion in $\varepsilon$. Alternative formulations or integration orders also yield this factor.\\

Nonetheless, it should not be there. This is suspiciously similar to the practice of regulating and subtracting divergent parts to obtain the desired part. Including only the $l=0$-term in (\ref{e:wc19}) sets
\begin{align*}
&\mathcal{Q}_{0}=\sqrt{\pi}\frac{(N\!-\!1)\sqrt{\sum_{m}e_{m}^{-1}}}{N}\big(\frac{N}{(N\!-\!1)(\sum_{m}e_{m}^{-1})}\big)^{\binom{N}{2}}\Delta(\frac{1}{e_{1}},\ldots,\frac{1}{e_{N}})\nonumber\quad.
\end{align*}

All integrals are now computed. The Gamma function is approximated by Stirling's formula 
\begin{equation}
\Gamma\Big(\binom{N}{2}\Big)=2\sqrt{\frac{\pi}{ N(N-1)}}\Big(\frac{N(N\!-\!1)}{2e}\Big)^{\binom{N}{2}}\times\big(1+\mathcal{O}(N^{-2})\big)\quad\label{e:gammastir}
\end{equation}
and the Vandermonde determinant of the inverses can be written as
\begin{equation*}
\Delta(\frac{1}{e_{1}},\ldots,\frac{1}{e_{N}})=\prod_{k<l}\frac{1}{e_{l}}-\frac{1}{e_{k}}=(-1)^{\binom{N}{2}}\Delta(e_{1},\ldots,e_{N})\prod_{m=1}^{N}e_{m}^{1-N}\quad.
\end{equation*}
Putting this together yields the partition function for small but strictly positive coupling
\begin{align}
&\hspace{-8mm}\mathscr{Z}[0]=\sqrt{\frac{N\!-\!1}{N}}\big[\!\prod_{m=1}^{N}\!\sqrt{\frac{\pi}{e_{m}}}e_{m}^{1-N}\big]\big(\frac{\pi N}{2(\sum_{m}e_{m}^{-1})}\big)^{\binom{N}{2}}\exp[-\sum_{m}\frac{3g}{4e_{m}^{2}}]\label{e:wc20}\;.
\end{align}
To compare this to the result without coupling (\ref{e:tc7}) the parameter convention
\begin{equation*}
\xi=\frac{1}{N}\sum_{j}e_{j}\qquad\text{and}\qquad e_{j}=\xi(1+\tilde{\varepsilon}_{j})
\end{equation*}
for small $\tilde{\varepsilon}_{j}\ll 1$ is used again. This yields
\begin{align}
&\hspace{-8mm}\mathscr{Z}[0]=\big[\!\prod_{m=1}^{N}\!\sqrt{\frac{\pi}{e_{m}}}\big]\big(\frac{\pi}{2\xi}\big)^{\binom{N}{2}}\exp[-\sum_{m}\frac{3g}{4e_{m}^{2}}]\exp[-\binom{N}{2}\log\big(\frac{1}{N}\sum_{m}\frac{1}{1+\tilde{\varepsilon}_{m}}\big)]\nonumber\\
&\times\exp[-(N\!-\!1)\sum_{m}\log(1+\tilde{\varepsilon}_{m})]\nonumber\\
&=\big[\!\prod_{m=1}^{N}\!\sqrt{\frac{\pi}{e_{m}}}\big]\big(\frac{\pi}{2\xi}\big)^{\binom{N}{2}}\exp[-\sum_{m}\frac{3g}{4e_{m}^{2}}]\nonumber\\
&\times\exp[-\binom{N}{2}\cdot\log(1\!+\!\frac{1}{N}\sum_{m}\tilde{\varepsilon}_{m}^{2}\!-\!\frac{1}{N}\sum_{m}\tilde{\varepsilon}_{m}^{3}\!+\!\frac{1}{N}\sum_{m}\tilde{\varepsilon}_{m}^{4}\!-\!\frac{1}{N}\sum_{m}\tilde{\varepsilon}_{m}^{5}\!+\!\frac{1}{N}\sum_{m}\tilde{\varepsilon}_{m}^{6})]\nonumber\\
&\times\exp[-(N\!-\!1)\cdot\sum_{m}(-\frac{1}{2}\tilde{\varepsilon}_{m}^{2}+\frac{1}{3}\tilde{\varepsilon}_{m}^{3}-\frac{1}{4}\tilde{\varepsilon}_{m}^{4}+\frac{1}{5}\tilde{\varepsilon}_{m}^{5}-\frac{1}{6}\tilde{\varepsilon}_{m}^{6})]\nonumber\\
&=\big[\!\prod_{m=1}^{N}\!\sqrt{\frac{\pi}{e_{m}}}\big]\big(\frac{\pi}{2\xi}\big)^{\binom{N}{2}}\exp[-\sum_{m}\frac{3g}{4e_{m}^{2}}]\exp[\frac{N\!-\!1}{6}\sum_{m}\tilde{\varepsilon}_{m}^{3}-\frac{N\!-\!1}{4}\sum_{m}\tilde{\varepsilon}_{m}^{4}]\nonumber\\
&\times\exp[\frac{3(N\!-\!1)}{10}\sum_{m}\tilde{\varepsilon}_{m}^{5}-\frac{N\!-\!1}{3}\sum_{m}\tilde{\varepsilon}_{m}^{6}]\nonumber\\
&\times\exp[\frac{N\!-\!1}{4N}(\sum_{m}\tilde{\varepsilon}_{m}^{2})^{2}-\frac{1}{2}(\sum_{m}\tilde{\varepsilon}_{m}^{2})(\sum_{n}\tilde{\varepsilon}_{n}^{3})+\frac{1}{2}(\sum_{m}\tilde{\varepsilon}_{m}^{2})(\sum_{n}\tilde{\varepsilon}_{n}^{4})]\nonumber\\
&\times\exp[\frac{1}{4}(\sum_{m}\tilde{\varepsilon}_{m}^{3})^{2}-\frac{1}{6N}(\sum_{m}\tilde{\varepsilon}_{m}^{2})^{3}]\quad.\label{e:wc21}
\end{align}

There is no neighbourhood of parameters such that this is equal to (\ref{e:tc6c}) as one would naively expect. Any divergence from the symmetric situation $e_{m}=\xi$ for all $m=1,\ldots,N$ modifies the partition function significantly. Although the integration against the polytope volume allows an evaluation of the partition function that is nonpertubative in the weak coupling, the dynamic parameters are fixed to the symmetric case. Also this is reminiscent of perturbative quantum field theory. To make connection to the free theory, some model parameters must be fixed to their trivial values.\\

The matrix structure in the regulated partition function has been removed by the polytope volume to factorise the computation. This was described as a trade-off between structural integrity and computability. Processing this new structure instead of the (reformulated) matrices may cause artificial structures to appear. However, the numerical proximity demands that a limit case must exist, in which the original (numerical) value is retrieved. The difference between the obtained partition function and the limit case is either vanishing or diverging. The latter corresponds to the subtraction of divergent terms to obtain the desired partition function. This practice is common in perturbative quantum field theory, where it is performed on the level of Feynman diagrams. This observation supports the suggestion that the experimentally successful models of particle physics are theoretically treated in the wrong framework.

\chapter{Strong coupling\label{sec:Str_coup}}

In Chapter~\ref{sec:OP_for_QMM} the partition function for small nonzero coupling was determined. We were led to conclusion that the partition function for nonzero coupling could only be determined for equal dynamic eigenvalues. There is no reason to assume that for stronger coupling this is the case too.\\
The partition function for strong coupling would be very interesting to have. Since the kinetic eigenvalues encode the manifold on which the \textsc{qft} is defined, the freedom to choose those, although not unrestricted, allows the possibility to study quantum field theory on a curved space. It seems possible in favourable cases to reconstruct a Lagrangian density endowed with the Moyal product on a curved manifold. Because the Moyal product in this formulation yields a nonvanishing commutator between the fields and the coordinate functions. This may be interpreted as a toy model being both a quantum field theory and a quantum gravitational theory. Clearly, not much can be said about the global gravitational and causal structures. However, this is more an outlook of the possible applications than a concrete road map. And it is not within reach of the current project.\\

Nonetheless, an overview can be given. First, a direct method will be given, along the lines of the weak coupling method. For comparison, the difficulties of an alternative method relying heavily on the Pearcey integral are sketched. 

\section*{Vanishing kinematics}
The limit of zero coupling, corresponding to the free theory, made the determinantion of the partition function for weak coupling much easier. Repeating those steps for $E=0$ in (\ref{e:qmm2}) yields
\begin{align}
&\hspace{-8mm}\mathscr{Z}[0]=\mathpzc{U}g^{-\frac{N^{2}}{4}}\int\ud\vec{\lambda}\,\left|\begin{array}{ccc}1&\ldots&\lambda_{1}^{N-1}\\\vdots&\ddots&\vdots\\1&\ldots&\lambda_{N}^{N-1}\end{array}\right|^{2}\,\exp[-\sum_{j=1}^{N}\lambda_{j}^{4}]\nonumber\\
&=\mathpzc{U}g^{-\frac{N^{2}}{4}}(N!)\prod_{t=0}^{N-1}h_{t}\quad,\label{e:vk1}
\end{align}
where $h_{t}=\langle P_{t},P_{t}\rangle$ is the length of the monic orthogonal polynomials from Theorem~\ref{thrm:op} for the weight $w(\lambda)=\exp[-\lambda^{4}]$. Accurate approximations schemes for this are discussed in Paragraph~\ref{sec:Dai}. It is standard that $h_{m}=R_{m}h_{m-1}$, see also (\ref{e:hrec}). Together with the very narrow band of values 
\begin{equation*}
\sqrt{\frac{m}{12}}<R_{m}<\sqrt{\frac{m}{12}}\exp[(2m)^{-2}]\quad,
\end{equation*}
see (\ref{e:Rapp}), it is obvious that $P_{0}(\lambda)=1$, so that $h_{0}=\frac{1}{2}\Gamma(\frac{1}{4})$.

\section{Weak kinematics}
To explain the setup, it is best to recall the starting point (\ref{e:zld5})
\begin{align}
&\hspace{-8mm}\mathscr{Z}[0]=\mathpzc{U}\frac{\sqrt{2}e^{\frac{7}{6}}(-i)^{\binom{N}{2}}[\prod_{k=0}^{N-1}k!]}{2g^{N^{2}/4}\Delta(\sqrt{\mu_{1}},\ldots,\sqrt{\mu_{N}})}\!\int_{-\infty}^{\infty}\!\!\!\!\!\!\ud\kappa\!\int_{0}^{N(N-1)\zeta}\!\!\!\!\!\!\ud S\!\int_{0}^{(N-1)\zeta}\!\!\!\!\!\!\ud\vec{u}\,\big(\frac{eS}{N(N\!-\!1)}\big)^{\binom{N}{2}}\nonumber\\
&\times\big(\frac{N(N\!-\!1)^{2}}{2\pi S^{2}}\big)^{\frac{N}{2}}\exp[2\pi i\kappa(S-\sum_{j}u_{j})]\nonumber\\
&\times\exp[-\frac{(N-1)^{2}}{2S^{2}}(N+2)\sum_{j}(u_{j}-\frac{S}{N})^{2}]\exp[\frac{N(N-1)^{3}}{3S^{3}}\sum_{j}(u_{j}-\frac{S}{N})^{3}]\nonumber\\
&\times\exp[-\frac{N(N-1)^{4}}{4S^{4}}\sum_{j}(u_{j}-\frac{S}{N})^{4}]\exp[\frac{(N-1)^{4}}{4S^{4}}\big(\sum_{j}(u_{j}-\frac{S}{N})^{2}\big)^{2}]\nonumber\\
&\times\int_{-\infty}^{\infty}\!\!\!\!\ud\vec{\lambda}\,\Delta(\lambda_{1},\ldots,\lambda_{N})e^{-\sum_{j=1}^{N}(\lambda_{j}^{4}+iu_{j}\lambda_{j})}\det_{m,n}\big(e^{-\sqrt{\mu_{m}}\lambda_{n}^{2}}\big)\label{e:zldsc1}\quad.
\end{align}
Weak kinematics mean that the the model parameters $\mu_{k}$ are small compared to $1$. Setting them equal to zero yields the case of vanishing kinematics.\\

This formulation shows the same two difficulties as in the case of weak coupling. The volume of the diagonal subpolytope of symmetric stochastic matrices from Chapter~\ref{sec:polytope}, which factorises that eigenvalue integrals, appears again. The other obstacle is the determinant remaining in the partition function. Similar techniques from Chapter~\ref{sec:Det_tech} may be applied again. The main differences come from the integration methods applicable.

\section{Matrix factorisation}
One approach to the partition function for strong coupling is to write the entire partition function as a matrix product, so that the determinant computation reduces to the determinant computation of the individual matrices. To demonstrate the idea we perform a small test calculation first. The eigenvalue integral structure of (\ref{e:zldsc1}) is given by.
\begin{align*}
&\hspace{-8mm}\mathcal{P}=\det_{k,l}\int_{-\infty}^{\infty}\ud\lambda\,\lambda^{k-1}\exp[iA\lambda-B_{l}\lambda^{2}-D\lambda^{4}]\\
&=\varepsilon^{-\binom{N}{2}}\det_{k,l}\int_{-\infty}^{\infty}\ud\lambda\,\exp[(iA+k\varepsilon)\lambda-B_{l}\lambda^{2}-D\lambda^{4}]\quad.
\end{align*}
Scaling $\lambda\rightarrow \lambda D^{-1/4}$ shows that the we are interested in the rotated Pearcey function
\begin{equation*}
P(\alpha,\beta)=\int_{-\infty}^{\infty}\ud\lambda\,\exp[-\alpha\lambda-\beta\lambda^{2}-\lambda^{4}]\quad.
\end{equation*}
The notation used before connects via 
\begin{equation*}
\mathcal{P}(A,B,D)=\int_{-\infty}^{\infty}\ud\lambda\,e^{iA\lambda-B\lambda^{2}-D\lambda^{4}}=D^{-\frac{1}{4}}P(\frac{A}{D^{\frac{1}{4}}},\frac{B}{\sqrt{D}})\quad.
\end{equation*}
The needed boundary conditions are given by
\begin{equation*}
P(0,0)=\frac{1}{2}\Gamma(\frac{1}{4})\quad;\quad \partial_{\alpha=0}P(\alpha,0)=0\quad;\quad \partial_{\alpha=0}^{2}P(\alpha,0)=-\frac{1}{2}\Gamma(\frac{3}{4})\quad.
\end{equation*}
The Pearcey function $P(\alpha,0)$ satisfies 
\begin{equation*}
0=(\frac{\alpha}{4}-\partial_{\alpha}^{3})P(\alpha,0)\quad.
\end{equation*}
Solutions are given by hypergeometric functions $\,{}_{p}F_{q}(a_{1},\ldots,a_{p};b_{1},\ldots,b_{q};z)$, which in general satisfy
\begin{equation*}
0=\Big[-z\big(\prod_{n=1}^{p}(z\partial_{z}+a_{n})\big)+z\partial_{z}\big(\prod_{n=1}^{q}(z\partial_{z}+b_{n}-1)\big)\Big]\,{}_{p}F_{q}(a_{1},\ldots,a_{p};b_{1},\ldots,b_{q};z)\quad.
\end{equation*}
For $w=\,{}_{0}F_{2}$ and $z=s^{m}$ this implies
\begin{align}
&\hspace{-8mm}0=\big[-z+z\partial_{z}(z\partial_{z}+b_{1}-1)(z\partial_{z}+b_{2}-1)\big]w\nonumber\\
&=\big[-s^{m}+\frac{s\partial_{s}}{m}(\frac{s\partial_{s}}{m}+b_{1}-1)(\frac{s\partial_{s}}{m}+b_{2}-1)\big]w\nonumber\\
&=\big[-s^{m}+\frac{s^{3}\partial_{s}^{3}}{m^{3}}+(b_{1}+b_{2}-2+\frac{3}{m})\frac{s^{2}\partial_{s}^{2}}{m^{2}}\nonumber\\
&+(b_{1}-1+\frac{1}{m})(b_{2}-1+\frac{1}{m})\frac{s\partial_{s}}{m}\big]w\label{e:ghgm}\quad.
\end{align}
With $m=4$, $b_{1}=1/2$ and $b_{2}=3/4$ this becomes
\begin{equation*}
0=\big[64s-\partial_{s}^{3}\big]w\quad,
\end{equation*}
so that $\,{}_{0}F_{2}(\frac{1}{2},\frac{3}{4};\frac{\alpha^{4}}{256})$ is a solution. Continue with (\ref{e:ghgm}) and $s^{3}\partial_{s}^{3}=s^{2}\partial_{s}^{3}s-3s^{2}\partial_{s}^{2}$ shows that
\begin{align*}
&\hspace{-8mm}0=-s^{m-1}\cdot sw(s)+\big[\frac{s^{2}\partial_{s}^{3}s}{m^{3}}+(b_{1}+b_{2}-2+\frac{3}{m}-\frac{3}{m})\frac{s^{2}\partial_{s}^{2}}{m^{2}}\\
&+(b_{1}-1+\frac{1}{m})(b_{2}-1+\frac{1}{m})\frac{s\partial_{s}}{m}\big]w(s)\quad,
\end{align*}
so that $\frac{\alpha}{4}\,{}_{0}F_{2}(\frac{3}{4},\frac{5}{4};\frac{\alpha^{4}}{256})$ is a second solution. Once more with $s^{3}\partial_{s}^{3}=s\partial_{s}^{3}s^{2}-6s^{2}\partial_{s}^{2}-6s\partial_{s}$ yields a third solution $(\alpha^{2}/16)\,{}_{0}F_{2}(\frac{5}{4},\frac{3}{2};\frac{\alpha^{4}}{256})$.\\
From this we obtain that 
\begin{equation*}
\mathcal{P}(A,0,D)=\frac{\Gamma(\frac{1}{4})}{2D^{\frac{1}{4}}}\,{}_{0}F_{2}(\frac{1}{2},\frac{3}{4};\frac{A^{4}}{256D})-\frac{\Gamma(\frac{3}{4})}{4D^{\frac{3}{4}}}A^{2}\,{}_{0}F_{2}(\frac{5}{4},\frac{3}{2};\frac{A^{4}}{256D})\quad.
\end{equation*}
Under the requirement that $B$ is real and strictly positive the heat kernel can be used to construct $\mathcal{P}(A,B,D)$ from this, because
\begin{equation*}
0=(\partial_{B}-\partial_{A}^{2})\mathcal{P}(A,B,D)\quad.
\end{equation*}
Choosing initial and boundary conditions yields a solution formulation, such as
\begin{equation}
\mathcal{P}(A,B,D)=\int_{-\infty}^{\infty}\frac{\ud z}{\sqrt{4\pi B}}\exp[-\frac{(z-A)^{2}}{4B}]\mathcal{P}(z,0,D)\quad.\label{e:sol1}
\end{equation}
An alternative possibility is to choose spatial and temporal condition and construct a solution for $(A,B)\in\mathbb{R}_{+}\times\mathbb{R}_{+}$.\\
Writing the hypergeometric functions as power series and integrating term-wise using
\begin{equation*}
\int_{-\infty}^{\infty}\ud x\,x^{n}\exp[-\frac{(x-y)^{2}}{w}]=(-i\sqrt{\frac{w}{2}})^{n}\sqrt{\pi w}H_{n}(iy\sqrt{\frac{2}{w}})
\end{equation*}
turns (\ref{e:sol1}) into
\begin{align}
&\hspace{-8mm}\mathcal{P}(A,B,D)=\sum_{n=0}^{\infty}\Big\{\frac{\Gamma(\frac{1}{4})}{2D^{\frac{1}{4}}}\frac{(-i\sqrt{2B})^{4n}H_{4n}(\frac{iA}{\sqrt{2B}})}{(\frac{1}{2})_{n}\cdot(\frac{3}{4})_{n}\cdot(n!)\cdot(256D)^{n}}\nonumber\\
&-\frac{\Gamma(\frac{3}{4})}{4D^{\frac{3}{4}}}\frac{(-i\sqrt{2B})^{4n+2}H_{4n+2}(\frac{iA}{\sqrt{2B}})}{(\frac{3}{2})_{n}\cdot(\frac{5}{4})_{n}\cdot(n!)\cdot(256D)^{n}}\Big\}\nonumber\\
&=\sum_{n=0}^{\infty}\Big\{\frac{\Gamma(\frac{1}{4})\Gamma(4n+1)}{2D^{\frac{1}{4}}(256D)^{n}}\sum_{m=0}^{2n}\frac{(-B)^{m}(iA)^{4n-2m}}{(\frac{1}{2})_{n}\cdot(\frac{3}{4})_{n}\cdot(n!)\cdot(m!)\cdot(4n-2m)!}\nonumber\\
&-\frac{\Gamma(\frac{3}{4})\Gamma(4n+3)}{4D^{\frac{3}{4}}(256D)^{n}}\sum_{m=0}^{2n+1}\frac{(-B)^{m}(iA)^{4n+2-2m}}{(\frac{3}{2})_{n}\cdot(\frac{5}{4})_{n}\cdot(n!)\cdot(m!)\cdot(4n+2-2m)!}\Big\}\nonumber\\
&=\sum_{k,l=0}^{\infty}\frac{(-A^{2})^{k}}{(2k)!}\frac{(-B)^{l}}{l!}\sum_{n=0}^{\infty}\delta_{4n,2k+2l}\frac{\Gamma(4n+1)\Gamma(\frac{1}{4})\Gamma(\frac{1}{2})\Gamma(\frac{3}{4})}{2D^{\frac{1}{4}}(256D)^{n}\Gamma(n+\frac{1}{2})\Gamma(n+\frac{3}{4})\Gamma(n+1)}\nonumber\\
&+\delta_{4n+2,2k+2l}\frac{\Gamma(4n+3)\Gamma(\frac{3}{4})\Gamma(\frac{3}{2})\Gamma(\frac{5}{4})}{4D^{\frac{3}{4}}(256D)^{n}\Gamma(n+\frac{3}{2})\Gamma(n+\frac{5}{4})\Gamma(n+1)}\nonumber\\
&=\sum_{k,l=0}^{\infty}\frac{(A^{2})^{k}}{(2k)!}\frac{(B)^{l}}{l!}(-1)^{k+l}\frac{\Gamma(2k+2l+1)\Gamma(\frac{1}{4})\Gamma(\frac{1}{2})\Gamma(\frac{3}{4})}{2D^{\frac{1}{4}}(256D)^{\frac{k+l}{2}}\Gamma(\frac{k+l+1}{2})\Gamma(\frac{k+l+1}{2}+\frac{1}{4})\Gamma(\frac{k+l+2}{2})}\nonumber\\
&=\sum_{k,l=0}^{\infty}\frac{(A^{2})^{k}}{(2k)!}\frac{(B)^{l}}{l!}(-1)^{k+l}\frac{\Gamma(\frac{k+l}{2}+\frac{1}{4})}{2D^{\frac{1}{4}}D^{\frac{k+l}{2}}}\quad,\label{e:sol2}
\end{align}
where we have used that Hermite polynomials are given by
\begin{equation*}
H_{n}(x)=(n!)\cdot\sum_{m=0}^{\lfloor\frac{n}{2}\rfloor}(-\frac{1}{2})^{m}\frac{x^{n-2m}}{(m!)\cdot (n-2m)!}\quad,
\end{equation*}
the Pochhammer symbol $(a)_{n}=\Gamma(n+a)/\Gamma(a)$ and the multiplication formula
\begin{equation*}
\prod_{k=0}^{m-1}\Gamma(z+\frac{k}{m})=(2\pi)^{\frac{m-1}{2}}m^{\frac{1}{2}-mz}\Gamma(mz)\qquad\forall z\notin \mathbb{Z}\quad.
\end{equation*}
for Gamma functions. An alternative method to get here is to integrate the power series for $\exp[iA\lambda-B\lambda^{2}]$ term by term against the weight function $\exp[-D\lambda^{4}]$. These methods also yield
\begin{align}
&\hspace{-8mm}\mathcal{P}_{k}(A_{t},B_{t},D)=\int_{-\infty}^{\infty}\ud\lambda\,\lambda^{k}e^{iA_{t}\lambda-B_{t}\lambda^{2}-D\lambda^{4}}\nonumber\\
&=\sum_{m,n=0}^{\infty}\frac{(iA_{t})^{m}}{m!}\frac{(-B_{t})^{n}}{n!}\frac{\Gamma(\frac{k+m+2n+1}{4})}{D^{\frac{k+m+2n+1}{4}}}\,\vartheta(\frac{k+m}{2}\in\mathbb{N}_{0})\nonumber\\
&\approx\frac{1}{2}\sum_{n=0}^{N-1}\sum_{r=\lceil k/2\rceil}^{\lceil k/2\rceil+N-1}\frac{(-B_{t})^{n}}{n!}\frac{(iA_{t})^{2r-k}}{(2r-k)!}D^{-\frac{2n+2r+1}{4}}\Gamma(\frac{2n\!+\!2r\!+\!1}{4})\,,\label{e:sol3}
\end{align}
which is a matrix formulation that allows the computation of a determinant. However, convergence of these determinants is not trivial. First we check, if the expression converges. Applying Stirling's approximation to (\ref{e:sol3}) shows that the first omitted term is of order
\begin{equation*}
\big(\frac{Be}{N}\big)^{N}\big(\frac{Ae}{2N}\big)^{2N}\big(\frac{5N}{4De}\big)^{\frac{5N}{4}}\quad,
\end{equation*}
meaning it is not automatically small. However, for small $B$ or large $D$ this converges. Proving that this is also the case for the determinants yields an approximation for small dynamical eigenvalues.\\

Including a cubic term in the exponent brings us back to the case of interest.
The subleading exponentials may be expanded as power series. Truncating these series at the $N$-th term yields an approximation in terms of matrices
\begin{align}
&\hspace{-8mm}\det_{0\leq k,l\leq N-1}\int_{-\infty}^{\infty}\ud\lambda\,\lambda^{k}\exp[iA\lambda-B_{l}\lambda^{2}+iC\lambda^{3}-D\lambda^{4}]\nonumber\\
&\approx\det_{0\leq k,l\leq N-1}\Big(\frac{(-B_{l})^{n}}{n!}\Big)\bullet\Big(\frac{\Gamma(\frac{2n+2q+1}{4})}{2D^{\frac{2n+2q+1}{4}}}\Big)\nonumber\\
&\bullet\Big(\frac{(iA)^{2q-r}}{(2q-r)!}\vartheta(2q-r\in\mathbb{N}_{0})\Big)\bullet\Big(\frac{(iC)^{\frac{r-k}{3}}}{(\frac{r-k}{3})!}\vartheta(\frac{r-k}{3}\in\mathbb{N}_{0})\Big)\label{e:sol4}\\
&=\det \mathpzc{B}\bullet\mathpzc{G}\bullet\mathpzc{M}\bullet\mathpzc{C}\label{e:sol5}\\
&=\big\{\!\prod_{t=0}^{N-1}\!\frac{(-1)^{t}}{t!}\big\}\Delta(B_{1},\ldots,B_{N})\cdot\frac{\tilde{\gamma}}{2^{N}}D^{-\binom{N}{2}-\frac{N}{4}}\cdot (iA)^{\binom{N}{2}}\big\{\!\prod_{t=2}^{N-1}\!\frac{1}{(2t\!-\!1)!!}\big\}\nonumber\quad.
\end{align}
To check whether this approximation converges, the determinants and matrix inverses must be computed. To show that the determinants of two matrices, $Q$ and $A$, related through $Q=A+\tilde{A}$, converge, it must be shown that $\det(1+A^{-1}\tilde{A})\rightarrow1$. This is done most easily by Hadamard's inequality~\cite{stewart}. This is the same as what was done in Lemma~\ref{l:detcom}.
\begin{lemma}\{\emph{Hadamard's inequality}\}\\
A $n\times n$-matrix $M$, consisting of $n$ linearly independent columns $v_{j}$, satisfies the inequality
\begin{equation*}
\det(M)\leq \prod_{j=1}^{n}\Vert v_{j} \Vert\quad,
\end{equation*}
where equality holds if and only if the vectors are orthogonal.
\end{lemma}

\subsection{Determinants and inverses\label{sec:Dai}}
The partition function (\ref{e:sol4}) is factorised in matrices. The determinant reduces to the product of the determinants of the factor matrices. To this end the matrices $\mathpzc{G}$ and $\mathpzc{M}$ must be analysed. The determinant of $\mathpzc{G}$ is estimated using orthogonal polynomials, by
\begin{align}
&\hspace{-8mm}\det_{0\leq k,l\leq N-1}\!\!\Gamma(\frac{2k\!+\!2l\!+\!1}{4})=\!\tilde{\gamma}\!=\det_{k,l}\int_{0}^{\infty}\!\!\!\!\!\!\!\ud t\,e^{-t}\,t^{\frac{2k+2l-3}{4}}=\det_{k,l}\Big(2\!\!\int_{-\infty}^{\infty}\!\!\!\!\!\!\!\ud \lambda\,e^{-\lambda^{4}}\lambda^{2k\!+\!2l}\Big)\nonumber\\
&=\Big\{\prod_{m=0}^{N-1}2\int_{-\infty}^{\infty}\ud \lambda_{m}\,\lambda_{m}^{2m}e^{-\lambda_{m}^{4}}\Big\}\cdot\det_{k,l}\big(\lambda_{k}^{2l}\big)\nonumber\\
&=\Big\{\!\prod_{m=0}^{N-1}2\!\int_{-\infty}^{\infty}\!\!\!\ud \lambda_{m}\,\lambda_{m}^{2m}e^{-\lambda_{m}^{4}}\Big\}\cdot\det_{k,l}\big(P_{2l}(\lambda_{k})\big)=2^{N}\prod_{m=0}^{N-1}h_{2m}\;,\label{e:sol5}
\end{align}
where $P_{2l}$ are the monic orthogonal polynomials of degree $2l$ with respect to the weight function $w(\lambda)=\exp[-\lambda^{4}]$ on $\mathbb{R}$ from Theorem~\ref{thrm:op}. Because this is a symmetric weight function, the even polynomials consist only of terms of even degree and satisfy the recursion relation
\begin{equation*}
\lambda P_{m}=P_{m+1} + R_{m}P_{m-1}\quad.
\end{equation*}
The length of these polynomials
\begin{equation*}
h_{m}=\langle P_{m},P_{m}\rangle_{w}=\int \ud \lambda\,w(\lambda)P_{m}(\lambda)^{2}
\end{equation*}
satisfies then
\begin{align}
&\hspace{-8mm}h_{m}=\langle P_{m}, \lambda P_{m-1}\rangle_{w}=\langle \lambda P_{m}, P_{m-1}\rangle_{w}=\langle P_{m+1}+R_{m}P_{m-1}, P_{m-1}\rangle_{w}\nonumber\\
&=R_{m}h_{m-1}\label{e:hrec}
\end{align}
and the difference equation
\begin{align*}
&\hspace{-8mm}mh_{m-1}=\langle P_{m}', P_{m-1}\rangle_{w}=-\langle P_{m}, P_{m-1}'\rangle_{w}+4\langle \lambda^{2}P_{m}, \lambda P_{m-1}\rangle_{w}\\
&=4\langle P_{m}+R_{m-1}P_{m-2}, P_{m+2}+(R_{m+1}+R_{m})P_{m}+R_{m}R_{m-1}P_{m-2}\rangle_{w}\\
&=4h_{m}(R_{m}+R_{m+1})+4R_{m}R_{m-1}^{2}h_{m-2}\\
&=4h_{m-1}(R_{m+1}R_{m}+R_{m}^{2}+R_{m}R_{m-1})\quad.
\end{align*}
Combined with $h_{0}=\frac{1}{2}\Gamma(1/4)$ and $R_{1}=\Gamma(3/4)/\Gamma(1/4)$ this defines the determinant (\ref{e:sol5}). Because this determinant does not contain parameters of the quantum field theory, its asymptotics are only needed for the analysis to see whether (\ref{e:sol3}) converges or not. To this end, the asymptotic approximation
\begin{equation*}
R_{m}=\sqrt{\frac{m}{12}}\exp[\frac{1}{24m^{2}}-\frac{5}{384m^{4}}]\times\big(1+\mathcal{O}(m^{-6})\big)
\end{equation*}
suffices.\\

For the inverse $\mathpzc{G}^{-1}$ we construct a triangular matrix $U$. Its entries are the coefficients of the orthogonal polynomials $P_{2n}(\lambda)=\sum_{m=0}^{N-1}U_{nm}\lambda^{2m}$. From $(U\mathpzc{G}U^{T})_{kl}=h_{2k}\delta_{kl}$ it follows that
\begin{equation}
\big(\mathpzc{G}^{-1}\big)_{kl}=\sum_{m=0}^{N-1}U^{T}_{km}\frac{1}{h_{2m}}U_{ml}=\sum_{m=0}^{N-1}U_{mk}\frac{1}{h_{2m}}U_{ml}\quad.\label{e:GInv}
\end{equation}
It is not difficult to see that these coefficients for $m\geq k$ are given by $U_{km}=\delta_{km}$ and for $m< k$ by the recursion
\begin{align}
&\hspace{-8mm}U_{mk}=U_{(m-1)(k-1)}-(R_{2m-1}+R_{2m-2})U_{(m-1)k}\nonumber\\
&-R_{2m-2}R_{2m-3}U_{(m-2)k}\quad.\label{e:recrel1}
\end{align}
The first observation is that for $m\geq 1$ the coefficients are clamped by the asymptotic expansion
\begin{equation}
\sqrt{\frac{m}{12}}<R_{m}<\sqrt{\frac{m}{12}}\exp[(2m)^{-2}]\quad.\label{e:Rapp}
\end{equation}
\begin{table}[!hb]
\hspace*{3cm}\footnotesize{\begin{tabular}{l||l|l|l}
$m$& $\sqrt{m/12}$&$R_{m}$&$\sqrt{m/12}\exp[(2m)^{-2}]$\\\hline
$1$ & $0.2887$ & $0.3380$ & $0.3707$\\
$2$ & $0.4082$ & $0.4017$ & $0.4346$\\
$3$ & $0.5000$ & $0.5051$ & $0.5141$\\
$4$ & $0.5774$ & $0.5781$ & $0.5864$\\
$5$ & $0.6455$ & $0.6468$ & $0.6520$\\
$6$ & $0.7071$ & $0.7079$ & $0.7120$\\
$7$ & $0.7638$ & $0.7644$ & $0.7677$\\
$8$ & $0.8165$ & $0.8170$ & $0.8197$\\
$9$ & $0.8660$ & $0.8665$ & $0.8687$\\
$10$ & $0.9129$ & $0.9132$ & $0.9152$
\end{tabular}}
\caption{The first few recursion coefficients $R_{m}$ for orthogonal polynomials and their approximation from (\ref{e:Rapp}) for the weight function $w(z)=\exp[-z^{4}]$.\label{t:Rapp}}
\end{table}

Using $\frac{1}{4}\sum_{m=1}^{\infty}m^{-2}=\frac{\pi^{2}}{24}$ (\ref{e:hrec}) immediately implies that 
\begin{equation*}
\frac{1}{2}\Gamma(\frac{1}{4})\sqrt{\frac{m!}{12^{m}}}\leq h_{m}\leq \frac{1}{2}\Gamma(\frac{1}{4})\sqrt{\frac{m!}{12^{m}}}e^{\frac{\pi^{2}}{24}}\qquad\text{, for }m\geq 0\quad.
\end{equation*}

There are two sides from which the entries can be approximated. Both will be demonstrated. All diagonal matrix entries are equal to one, so that the diagonal band below it satisfies
\begin{align*}
&\hspace{-8mm}-U_{m(m-1)}=-U_{(m-1)(m-2)}+(R_{2m-1}+R_{2m-2})=\sum_{j=1}^{m}R_{2(j-\frac{1}{2})}+R_{2(j-\frac{1}{2})-1}\\
&\leq e^{\frac{1}{4}}\sum_{j=1}^{m}2\sqrt{\frac{2(j-1/2)}{12}}\leq e^{\frac{1}{4}}C_{2}\int_{0}^{m}\ud x\,\sqrt{\frac{2x}{3}}=e^{\frac{1}{4}}C_{2}\big(\frac{2m}{3}\big)^{\frac{3}{2}}\quad,
\end{align*}
where
\begin{equation*}
C_{2}=\max_{n\in\mathbb{N}}\sqrt{n}\big[\int_{n-\frac{1}{2}}^{n+\frac{1}{2}}\ud x\,\sqrt{x}\big]^{-1}\leq 1.02\quad.
\end{equation*}
The recursion for the next band $U_{m(m-2)}$ involves both diagonal bands above. However, the diagonal, multiplied by $(m-1)/6$, is small in comparison to $U_{m(m-1)}$ multiplied by $\sqrt{2(m-\frac{1}{2})/3}$. Furthermore, it is of opposite sign. This shows that the sign of the entry $U_{mk}$ is given by $(-1)^{m+k}$. Restricting ourselves to the absolute value of the matrix entries, we may ignore it. Repeating the same calculation as above with the estimate
\begin{equation*}
\max_{p\in\mathbb{N}}\max_{n\in\mathbb{N}}n^{\frac{p+1}{2}}\big[\int_{n-\frac{1}{2}}^{n+\frac{1}{2}}\ud x\,x^{\frac{p+1}{2}}\big]^{-1}\leq 1
\end{equation*}
yields
\begin{align*}
&\hspace{-8mm}U_{m(m-2)}\leq\sum_{j=2}^{m}(R_{2j-1}+R_{2j-2})U_{(j-1)(j-2)}\\
&\leq C_{2}e^{\frac{1}{4}+\frac{1}{16}}\sum_{j=2}^{m}(\frac{2(j-1)}{3})^{\frac{3}{2}}\sqrt{\frac{2(j-\frac{1}{2})}{3}}\\
&\leq C_{2}e^{\frac{\pi^{2}}{24}}\frac{4}{9}\int_{1}^{m}\ud x\,x^{2}\leq C_{2}e^{\frac{\pi^{2}}{24}}\frac{4}{27}m^{3}\quad.
\end{align*}
Repeating the steps it is not difficult to see that
\begin{equation}
|U_{mk}|\leq \big(\frac{2}{3}\big)^{\frac{m-k}{2}}C_{2}\exp[\frac{\pi^{2}}{24}]m^{\frac{3(m-k)}{2}}\qquad\text{, for }k\leq m\quad.\label{e:Uest}
\end{equation}
Although this works, it is desirable to do a little better. It is straightforward to check that $U_{10}=-R_{1}, U_{20}=R_{3}R_{1}$ and that
\begin{equation*}
U_{m0}=\prod_{j=1}^{m}(-R_{2j-1})\qquad\text{and}\qquad |U_{m0}|\leq e^{\frac{\pi^{2}}{32}}\sqrt{\frac{(2m-1)!!}{12^m}}\quad,
\end{equation*}
where it was used that $\sum_{j=1}^{\infty}(2j-1)^{-2}=\frac{\pi^{2}}{8}$. This gives much tighter bounds than (\ref{e:Uest}). For the second column, i.e. $k=1$, the first few elements are given by
\begin{align*}
&U_{11}=1\quad;\quad U_{21}=-(R_{1}+R_{2}+R_{3})\quad;\quad U_{31}=R_{5}R_{3}+R_{5}R_{2}\\
&+R_{5}R_{1}+R_{4}R_{2}+R_{4}R_{1}+R_{3}R_{1}\quad;\quad U_{41}=-R_{7}R_{5}(R_{3}+R_{2}+R_{1})\\
&-R_{7}R_{4}(R_{2}+R_{1})-R_{7}R_{3}R_{1}-R_{6}R_{4}(R_{2}+R_{1})-R_{6}R_{3}R_{1}-R_{5}R_{3}R_{1}\quad.
\end{align*}
Removing $U_{(m-1)k}$ from (\ref{e:recrel1}) by applying the recursion relation once more, it is clear that all terms from $U_{(m-2)k}$ in (\ref{e:recrel1}) cancel terms from $U_{(m-1)k}$. This shows that
\begin{equation*}
U_{mk}=(-1)^{m+k}\sum_{\stackrel{1\leq q_{1}<\ldots<q_{m-k}\leq 2m-1}{q_{j}-q_{j-1}\geq2}}\big[\prod_{j=1}^{m-k}R_{q_{j}}\big]\quad.
\end{equation*}
All terms in $U_{mk}$ are products of $m-1$ different $R_{j}$'s with the largest involved being $R_{2m-1}$ and the indices of various factors in a term differ by at least $2$. The number of such combinations is given by $\binom{m+k}{m-k}=\binom{m+k-2}{m-k}+2\binom{m+k-1}{m-k-1}-\binom{m+k-2}{m-k-2}$. In combination with the fact that $k<l$ implies that $R_{k}<R_{l}$ is sufficient to produce an upper bound. Another estimate is thus given by
\begin{equation}
|U_{mk}|\leq \binom{m+k}{m-k}e^{\frac{\pi^{2}}{32}}12^{\frac{k-m}{2}}\sqrt{\frac{(2m-1)!!}{(2k-1)!!}}\quad.\label{e:Uest2}
\end{equation}
A quick numerical check, Tables~\ref{t:U} and~\ref{t:Uapp}, shows that this is much better than (\ref{e:Uest}).\\

\begin{table}[!hb]
\hspace*{0cm}\footnotesize{\begin{tabular}{l|ccccccccccc}
$m\backslash k$& $0$&$1$&$2$& $3$&$4$&$5$& $6$&$7$&$8$& $9$&$10$\\\hline
$0$&$1.0$&$0.0$&$0.0$&$0.0$&$0.0$&$0.0$&$0.0$&$0.0$&$0.0$&$0.0$&$0.0$\\
$1$&$0.34$&$1.0$&$0.0$&$0.0$&$0.0$&$0.0$&$0.0$&$0.0$&$0.0$&$0.0$&$0.0$\\
$2$&$0.17$&$1.24$&$1.0$&$0.0$&$0.0$&$0.0$&$0.0$&$0.0$&$0.0$&$0.0$&$0.0$\\
$3$&$0.11$&$1.40$&$2.47$&$1.0$&$0.0$&$0.0$&$0.0$&$0.0$&$0.0$&$0.0$&$0.0$\\
$4$&$0.08$&$1.60$&$4.58$&$3.94$&$1.0$&$0.0$&$0.0$&$0.0$&$0.0$&$0.0$&$0.0$\\
$5$&$0.07$&$1.91$&$7.77$&$10.6$&$5.63$&$1.0$&$0.0$&$0.0$&$0.0$&$0.0$&$0.0$\\
$6$&$0.07$&$2.38$&$12.8$&$24.5$&$20.3$&$7.50$&$1.0$&$0.0$&$0.0$&$0.0$&$0.0$\\
$7$&$0.07$&$3.10$&$21.1$&$52.7$&$60.6$&$34.7$&$9.54$&$1.0$&$0.0$&$0.0$&$0.0$\\
$8$&$0.08$&$4.20$&$35.1$&$109.$&$163.$&$128.$&$54.5$&$11.7$&$1.0$&$0.0$&$0.0$\\
$9$&$0.10$&$5.94$&$59.3$&$224.$&$413.$&$419.$&$243.$&$80.7$&$14.1$&$1.0$&$0.0$\\
$10$&$0.12$&$8.72$&$102.$&$455.$&$1012.$&$1268.$&$946.$&$427.$&$114.$&$16.6$&$1.0$
\end{tabular}}
\caption{The absolute value of the first coefficients $U_{mk}$ for orthogonal polynomials for the weight function $w(z)=\exp[-z^{4}]$.\label{t:U}}
\end{table}
\begin{table}[!hb]
\hspace*{0cm}\footnotesize{\begin{tabular}{l|ccccccccccc}
$m\backslash k$& $0$&$1$&$2$& $3$&$4$&$5$& $6$&$7$&$8$& $9$&$10$\\\hline
$0$&$1.0$&$0.0$&$0.0$&$0.0$&$0.0$&$0.0$&$0.0$&$0.0$&$0.0$&$0.0$&$0.0$\\
$1$&$0.39$&$1.0$&$0.0$&$0.0$&$0.0$&$0.0$&$0.0$&$0.0$&$0.0$&$0.0$&$0.0$\\
$2$&$0.20$&$2.04$&$1.0$&$0.0$&$0.0$&$0.0$&$0.0$&$0.0$&$0.0$&$0.0$&$0.0$\\
$3$&$0.13$&$2.64$&$4.39$&$1.0$&$0.0$&$0.0$&$0.0$&$0.0$&$0.0$&$0.0$&$0.0$\\
$4$&$0.10$&$3.36$&$10.1$&$7.28$&$1.0$&$0.0$&$0.0$&$0.0$&$0.0$&$0.0$&$0.0$\\
$5$&$0.08$&$4.36$&$20.3$&$25.2$&$10.6$&$1.0$&$0.0$&$0.0$&$0.0$&$0.0$&$0.0$\\
$6$&$0.08$&$5.84$&$39.0$&$72.4$&$50.8$&$14.3$&$1.0$&$0.0$&$0.0$&$0.0$&$0.0$\\
$7$&$0.08$&$8.11$&$73.0$&$188.$&$194.$&$89.5$&$18.4$&$1.0$&$0.0$&$0.0$&$0.0$\\
$8$&$0.09$&$11.7$&$136.$&$463.$&$650.$&$434.$&$144.$&$22.8$&$1.0$&$0.0$&$0.0$\\
$9$&$0.11$&$17.3$&$254.$&$1103.$&$2012.$&$1807.$&$858.$&$217.$&$27.5$&$1.0$&$0.0$\\
$10$&$0.14$&$26.7$&$480.$&$2578.$&$5907.$&$6821.$&$4318.$&$1550.$&$312.$&$32.5$&$1.0$
\end{tabular}}
\caption{The approximation (\ref{e:Uest2}) of the absolute value of the first coefficients $U_{mk}$ for orthogonal polynomials for the weight function $w(z)=\exp[-z^{4}]$.\label{t:Uapp}}
\end{table}

There is another matrix that has to be inverted. The inverse of $\mathpzc{M}$ follows from a reformulation of its entries by
\begin{equation*}
\mathpzc{M}_{nm}=\Big(\binom{2n}{m}\Big)=\sum_{j=0}^{m}(2n)^{j}\frac{s_{1}(m,j)}{m!}=\Big(n^{j}\Big)\bullet\Big(2^{j}\frac{s_{1}(m,j)}{m!}\Big)\quad,
\end{equation*}
where $s_{1}(n,k)$ is the Stirling number of the first kind. This is in fact a matrix equation, so that the inverse is given by the inverse Vandermonde-matrix (\ref{e:invVdm}) multiplied by the Stirling numbers of the second kind, which are the matrix inverse of the Stirling numbers of the first kind
\begin{equation*}
\delta_{kl}=\sum_{m\geq0}s_{1}(k,m)\cdot S_{2}(m,l)=\sum_{m\geq0}S_{2}(k,m)\cdot s_{1}(m,l)\quad.
\end{equation*}
This yields
\begin{equation*}
\delta_{km}=\mathpzc{M}^{-1}_{kn}\mathpzc{M}_{nm}=\Big((k!)\frac{S_{2}(t,k)}{2^{t}}\Big)\bullet\Big(\tilde{v}_{tn}\Big)\bullet\Big(n^{j}\Big)\bullet\Big(2^{j}\frac{s_{1}(m,j)}{m!}\Big)\quad,
\end{equation*}
where 
\begin{equation}
\tilde{v}_{tn}=\frac{(-1)^{t}}{(t!)\cdot(N-1-t)!}\sum_{\stackrel{0\leq m_{0}<\cdots<m_{N-1-n}\leq N-1}{m_{0},\ldots,m_{N-1-n}\neq t}}(-1)^{n}m_{0}\cdots m_{N-1-n}\quad.
\end{equation}
Its determinant 
\begin{equation*}
\det_{0\leq k,l\leq N-1}\mathpzc{M}_{kl}=\prod_{t=1}^{N-1}\frac{1}{(2t-1)!!}
\end{equation*}
is computed in Lemma~\ref{l:sfacdet2}.\\

The inverse of $\mathpzc{B}$ from (\ref{e:sol4}) follows from the inverse Vandermonde-matrix (\ref{e:invVdm}) and is given by
\begin{equation*}
\big(\mathpzc{B}^{-1}\big)_{nl}=(n!)\frac{1}{\prod_{\stackrel{t=0}{t\neq l}}^{N-1}(B_{t}-B_{l})}\sum_{\stackrel{0\leq m_{0}<\cdots<m_{N-1-n}\leq N-1}{m_{0},\ldots,m_{N-1-n}\neq l}}(-1)^{n}B_{m_{0}}\cdots B_{m_{N-1-n}}\quad.
\end{equation*}

It follows from
\begin{align*}
&\hspace{-8mm}\sum_{k=0}^{N-1}\frac{(iC)^{\frac{r-k}{3}}}{(\frac{r-k}{3})!}\cdot\frac{(-iC)^{\frac{k-s}{3}}}{(\frac{k-s}{3})!}\cdot\vartheta(\frac{k-s}{3}\in\mathbb{N}_{0})\cdot\vartheta(\frac{r-k}{3}\in\mathbb{N}_{0})\\
&=\frac{(iC)^{\frac{r-s}{3}}}{(\frac{r-s}{3})!}\vartheta(\frac{r-s}{3}\in\mathbb{N}_{0})\sum_{m=0}^{\frac{r-s}{3}}\binom{\frac{r-s}{3}}{m}(-1)^{m}=\delta_{rs}
\end{align*}
that the inverse of $\mathpzc{C}$ is given by
\begin{equation*}
\big(\mathpzc{C}^{-1}\big)_{ks}=\frac{(-iC)^{\frac{k-s}{3}}}{(\frac{k-s}{3})!}\vartheta(\frac{k-s}{3}\in\mathbb{N}_{0})\quad.
\end{equation*}

Putting the things together now, shows that the determinant decomposition in (\ref{e:sol4}) converges, if the determinant of the matrix with entries
\begin{align}
&\hspace{-8mm}\mathpzc{R}_{jk}=\delta_{jk}+\Big((iA)^{l}\delta_{jl}\Big)_{jl}\bullet\Big(2^{-w}(l!)S_{2}(w,l)\Big)_{lw}\nonumber\\
&\bullet\Big(\frac{\delta_{wt}}{(iA)^{2t}}\Big)_{wt}\bullet\Big(\sum_{\stackrel{0\leq g_{0}<\ldots<g_{N-1-m\leq N-1}}{g_{i}\neq t}}\!\!\!\!\!(-1)^{t+m}\frac{g_{0}\cdots g_{N-1-m}}{(t!)\cdot((N-1-t)!)}\Big)_{tm}\nonumber\\
&\bullet\Big(2D^{\frac{2m+1}{4}}\delta_{mz}\Big)_{mz}\bullet\Big(\sum_{r=0}^{N-1}\frac{U_{rz}U_{rv}}{h_{2r}}\Big)_{zv}\bullet\Big(D^{\frac{n}{2}}\delta_{vn}\Big)_{vn}\nonumber\\
&\bullet\Big(\frac{(n!)\cdot(-1)^{n}}{\prod_{\stackrel{t=0}{t\neq p}}^{N-1}(B_{t}-B_{p})}\sum_{\stackrel{0\leq m_{0}<\cdots<m_{N-1-n}\leq N-1}{m_{0},\ldots,m_{N-1-n}\neq p}}B_{m_{0}}\cdots B_{m_{N-1-n}}\Big)_{np}\nonumber\\
&\bullet\Big(\frac{(-B_{p})^{a}}{a!}\Big)_{pa}\bullet\Big(\frac{\Gamma(\frac{2a+2b+1}{4})}{2D^{\frac{2a+2b+1}{4}}}\Big)_{ab}\nonumber\\
&\bullet\Big(\frac{(iA)^{2b-c}}{(2b-c)!}\vartheta(2b-c\in\mathbb{N}_{0})\Big)_{bc}\bullet\Big(\frac{(iC)^{\frac{c-q}{3}}}{(\frac{c-q}{3})!}\vartheta(\frac{c-q}{3}\in\mathbb{N}_{0})\Big)_{cq}\nonumber\\
&\bullet\Big(\frac{(-iC)^{\frac{q-k}{3}}}{(\frac{q-k}{3})!}\vartheta(\frac{q-k}{3}\in\mathbb{N}_{0})\Big)_{qk}\quad,\label{e:CM1}
\end{align}
converges to one. Some possible estimates on this are given below.\\

We refer to the matrices by their indices. It should be noticed that the matrix indices $a,b$ and $c$ run from $N$ to infinity. All other indices run from $0$ to $N-1$. There are two sources that can make the entries small. One is a very large $D$ and the other is $B_{p}^{a}/(a!)$. It reflects our idea of a strong coupling that the quartic term must be dominant or that the other parameters must be small. Since there are no model parameters left in the quartic term, it must be the second. The question to answer is thus: ``How small must the $B_{m}$'s be to guarantee convergence?''\\

The Stirling numbers of the second kind in the $lw$-matrix are bounded by $(l!)S_{2}(w,l)\leq l^{w}\vartheta(w\geq l)$. The entries in the $tm$-matrix are given by
\begin{equation*}
\frac{(-1)^{t}}{(t!)\cdot(N-1-t)!}\frac{\partial_{\varepsilon=0}^{m}}{m!}\prod_{\stackrel{g=0}{g\neq t}}^{N-1}(g-\varepsilon)
\end{equation*}
and can be estimated in absolute value by $\binom{N-1}{t}$.\\
It follows from (\ref{e:GInv}) that the $zv$-matrix composed with the diagonal $mz$-matrix and the diagonal $vn$-matrix is the inverse of $\mathpzc{G}$ from (\ref{e:sol4}). Its entries can be bound using (\ref{e:Uest2}). It is straightforward to see that $(2r-1)!!\leq\sqrt{(2r)!}$ and it follows from Stirling's approximation that
\begin{equation*}
\sum_{r=0}^{N-1}\binom{r+m}{r-m}\binom{r+n}{r-n}\leq 2^{4(N-1)}\quad.
\end{equation*}
For the composition of the $np$-matrix with the $pa$-matrix the formulas (\ref{e:VdMinvest})~-~(\ref{e:hgfest2})\footnote{The indices there are $j=a-N+1$ and $n=l-1$.} are very useful. Some formula of this kind is also necessary, because it prevents that the corrections diverge, when two parameters approximate each other.\\

The combination of the final two matrices yields
\begin{equation*}
\Big((-1)^{\lfloor\frac{k}{3}\rfloor}\frac{(iC)^{\frac{c-k}{3}}}{\big(\frac{c-k}{3}\big)!}\vartheta(\frac{c-k}{3}\in\mathbb{N}_{0})\Big)_{ck}\quad,
\end{equation*}
which can be combined with the $bc$-matrix
\begin{equation*}
\frac{A^{2b}}{(2b)!}\big(\sum_{c'=0}^{\frac{2b}{3}}\frac{C^{c'}}{c'!}A^{-3c'}\binom{2b}{3c'}\big)\leq\frac{(A+C^\frac{1}{3})^{2b}}{(2b)!}\quad.
\end{equation*}

Using the estimates described above we can write the correction matrix as
\begin{align}
&\hspace{-8mm}|\mathpzc{R}_{jk}|\leq\delta_{jk}+\Big(A^{j}(\frac{j}{2A^{2}})^{t}\vartheta(t-j\in\mathbb{N}_{0})\Big)_{jt}\bullet\Big(\frac{\partial_{\varepsilon=0}^{m}}{m!}\frac{\prod_{\stackrel{g=0}{g\neq t}}^{N-1}(g-\varepsilon)}{(t!)\cdot((N-1-t)!)}\Big)_{tm}\nonumber\\
&\bullet\Big(\frac{2^{4N-2}e^{\frac{\pi^{2}}{16}}}{\Gamma(1/4)}D^{\frac{2m+2n+1}{4}}\frac{12^{\frac{m+n}{2}}}{\sqrt{((2m-1)!!)\cdot((2n-1)!!)}}\Big)_{mn}\nonumber\\
&\bullet\Big(\frac{n!}{a!}\frac{2^{2N+a-n}}{\sqrt{\pi N}}\cdot(\max_{0\leq t\leq N-1}|B_{t}|)^{a-n}\Big)_{na}\bullet\Big(\frac{\Gamma(\frac{2a+2b+1}{4})}{2D^{\frac{2a+2b+1}{4}}}\Big)_{ab}\nonumber\\
&\bullet\Big(\frac{(A+C^\frac{1}{3})^{2b}}{(2b)!}\Big)_{bk}\quad.\label{e:CM2}
\end{align}
For brevity we write $\max_{l}|B_{l}|=\mathcal{B}$. Furthermore, $A=\mathcal{O}(N^{1/2})$, $C=\mathcal{O}(N^{-1/2})$ and $D=\mathcal{O}(N^{-1})$. Following (\ref{e:zldsc1}) it seems unlikely that $\mathcal{B}<\frac{1}{2}$. This must be assumed though, since this method will fail otherwise.\\

It is straightforward to see that for $k,l\geq0$
\begin{equation*}
(k!)(l!)\leq(k+l)!\leq (k!)(l!)2^{k+l}\qquad;\qquad \Gamma(2+k+l)\leq 2^{k+l+1}\Gamma(k+\frac{3}{2})\Gamma(l+\frac{3}{2})\quad.
\end{equation*}
Hoping that factors $2^{N}$ are subleading, this is a useful approximation. The next step is to shift the summations over $b=b'+N$ and $a=a'+N$. These run now from $0$ to infinity. For some contant $K_{1}>0$ we can approximate
\begin{align*}
&\hspace{-8mm}\Gamma(\frac{2a+2b+1}{4})=\Gamma(N+\frac{2a'+2b'+1}{4})\leq K_{1}\,2^{N+\frac{a'+b'}{2}}\,\Gamma(1+N)\,\Gamma(1+\frac{a'+b'}{2})\\
&\leq K_{1}^{2}\,2^{N+a'+b'}\,\Gamma(1+N)\,\Gamma(1+\frac{a'}{2})\,\Gamma(1+\frac{b'}{2})\quad.
\end{align*}

The sum over $b'$ takes the form
\begin{equation*}
\frac{A^{2N}}{(2N)!}\sum_{b'=0}^{\infty}\frac{\Gamma(1+\frac{b'}{2})}{(2b')!}\big(\frac{A^{2}}{\sqrt{D}}\big)^{b'}\leq K_{2} \frac{A^{2N}}{(2N)!}\exp[2\frac{A^{\frac{4}{3}}}{D^{\frac{1}{3}}}]\quad.
\end{equation*}
It follows from a simple estimate and the geometric series that terms with $b'>\frac{e A^{\frac{4}{3}}}{D^{\frac{1}{3}}}$ yield at most a constant $2$. Ignoring these terms, the asymptotic behaviour above is found.\\

For $a$ the assumption $\mathcal{B}\ll\sqrt{D}$ implies that an appropriate upper bound is given by
\begin{equation*}
\frac{(2\mathcal{B})^{N}}{N!}\sum_{a'=0}^{\infty}(\frac{2\mathcal{B}}{\sqrt{D}})^{a'}\frac{\Gamma(1+\frac{a'}{2})}{a'!}\leq K_{3}\frac{(2\mathcal{B})^{N}}{N!}\exp[\frac{2\sqrt{2e}\mathcal{B}^{2}}{D}]
\end{equation*}
for some constant $K_{3}$. From Stirling's approximation it follows that the terms with $a'\gg\frac{2e\mathcal{B}^{2}}{D}$ may be neglected. Using this value for the Gamma function produces the exponential.\\

For the sum over $n$, we use that modula a constant
\begin{equation*}
\frac{n!}{\sqrt{(2n-1)!!}}\leq C_{1}\big(\frac{n}{e}\big)^{\frac{n}{2}}\sqrt{n}2^{-\frac{n}{2}}\quad,
\end{equation*}
so that we obtain an upper bound for the sum
\begin{align}
&\hspace{-8mm}\sum_{n=0}^{N-1}\frac{n!}{\sqrt{(2n-1)!!}}\big(\frac{3D}{\mathcal{B}^{2}}\big)^{\frac{n}{2}}\nonumber\\
&\leq C_{1}\sqrt{N}\sum_{n=0}^{N-1}\big(\frac{n3D}{2e\mathcal{B}^{2}}\big)^{\frac{n}{2}}\nonumber\\
&\leq C_{1}\sqrt{N}\left\{
\begin{array}{ll}
\frac{1}{1-\sqrt{\frac{3DN}{e\mathcal{B}^{2}}}} & \text{, if }\frac{\sqrt{ND}}{\mathcal{B}}\ll1 \\ 
\big(\frac{3DN}{e\mathcal{B}^{2}}\big)^{\frac{N-1}{2}}\frac{1}{1-\sqrt{\frac{e\mathcal{B}^{2}}{3DN}}} & \text{, if }\frac{\sqrt{DN}}{\mathcal{B}}\gg1
\end{array}\right.\quad.\label{e:nub}
\end{align}
An upper bound for $\partial_{\varepsilon=0}^{m}$ is given by $m=0$, so that the sum over $m$ can be estimated by $\exp[\sqrt{12D}]$.

From the definition (\ref{e:zldsc1}) of $A$ it follows that $j\leq A^{2}$, implying
\begin{equation*}
A^{j}\sum_{t=j}^{N-1}\binom{N-1}{t}(\frac{j}{2A^{2}})^{t}\leq (\frac{3}{2})^{N-1}A^{j}\quad.
\end{equation*}
Now we can finally put things together, resulting in
\begin{align}
&\hspace{-8mm}|\mathpzc{R}_{jk}|\leq\delta_{jk}+\mathcal{K}N^{\frac{N-1}{2}}A^{j}192^{N}\frac{A^{2N}}{(2N)!}\frac{\mathcal{B}^{2}}{D^{1+\frac{N}{2}}}\exp[2\frac{A^{\frac{4}{3}}}{D^{\frac{1}{3}}}]\exp[\frac{2\sqrt{2e}\mathcal{B}^{2}}{D}]\quad,\label{e:CM3}
\end{align}
where the constants that are independent of $A, \mathcal{B}, D$ or $N$ are put in a constant $\mathcal{K}$ and we have assumed that $\mathcal{B}\ll\sqrt{DN}$. Assuming that all $|\mathpzc{R}_{jk}|\ll N^{-1}$ the determinant can be computed this way.\\

This is the point to stop. As mentioned at the beginning of this chapter, the emphasis will be on the methods and not on the specific computational steps. The case of weak coupling makes it clear that there is a long way from a working idea to an exact result. Nonetheless, most known problems of the computation are, at least theoretically, resolved. Walking the steps to a partition function with (\ref{e:vk1}) as a limit case is `all' there is left to do. It must be reminded, though, that the asymptotic polytope volume has been used too for the partition function for strong coupling. It is not evident that this will lead to the same unfortunate situation as in the case of weak coupling, but the possibility remains.

\section{Integration via the Pearcey function}
In Paragraph~\ref{sec:PI} we discussed the Pearcey function and some methods to integrate or approximate it. In the current context it is given by
\begin{equation}
\mathpzc{P}(a,b)=\int_{-\infty}^{\infty}\ud \lambda\, e^{-f(a,b;\lambda)}\qquad,\text{where}\quad f(a,b;\lambda)=\lambda^{4}+b\lambda^{2}+ia\lambda\quad\label{e:realPearcey}
\end{equation}
with $a,b\in\mathbb{R}$. The central tools are paths in $\mathbb{C}$ of constant $\Im[f(a,b;\lambda)]$ that connect $-\infty$ and $+\infty$ possibly via $\pm i\infty$. The real parts of this are paths of steepest descent.\\
A helpful observation is that one path of constant imaginary part is given by the imaginary axis, i.e $\Re[\lambda]=0$. Thus at least one saddle point must lie on this line. Furthermore, it is not difficult to see that $f(a,b;\lambda)=\overline{f}(a,b;-\overline{\lambda})$, so that the lines of constant imaginary part are symmetric in the imaginary axis. This implies that either one or two paths are needed.\\
The polynomial's real coefficients imply that on the imaginary axis either one or three saddle points will lie. Viewed as a real function on a (real) line, this line must be passed at a local maximum, because $-f$ has a local minimum there. Because these points are saddle points, this means that this is the heighest point of the real part along the curve of constant imaginary part. If there is one saddle points on the imaginary axis, this is a local minimum and the local maxima lie at $\pm i\infty$. This means that two integration curves are needed. If there are three saddle points, the middle one will correspond to a local maximum and only one curve is needed. In Paragraph~\ref{sec:PIrp} more details on this function are given.\\

These saddle points are given by
\begin{equation*}
0=\lambda^{3}+\frac{b}{2}\lambda+ia\qquad\text{or}\qquad 0=-x^{3}+\frac{b}{2}x+a\quad,
\end{equation*}
where $\lambda = ix$ is chosen to move to real arguments. Solutions to this are given by
\begin{equation}
x_{m}=\frac{R}{3}e^{2\pi i\frac{m}{3}}+\frac{b}{2R}e^{-2\pi i\frac{m}{3}}\qquad\text{with }m=0,1,2\quad\label{e:rPxm}
\end{equation}
and where
\begin{equation*}
R=\big[\frac{27 a}{8}\pm \sqrt{(\frac{27a}{8})^{2}-(\frac{3b}{2})^{3}}\big]^{\frac{1}{3}}\quad,
\end{equation*}
where the $+$-sign will be used here. The strategy is to approximate the function around the saddle points
\begin{align*}
&\hspace{-8mm}f(a,b;ix_{m}+z_{R}+iz_{I})\approx \big\{-\frac{b}{2}x_{m}^{2}-\frac{3a}{4}x_{m}\big\}\\
&+\frac{z_{R}^{2}-z_{I}^{2}+2iz_{R}z_{I}}{2}\Big(-2B-\frac{4R^{2}}{3}e^{\frac{4\pi im}{3}}-\frac{3b^{2}}{R^{2}}e^{-\frac{4\pi im}{3}}\Big)\\
&=\big\{-\frac{b}{2}x_{m}^{2}-\frac{3a}{4}x_{m}\big\}+\frac{z_{R}^{2}-z_{I}^{2}+2iz_{R}z_{I}}{2}\Big(X+iY\Big)\quad.
\end{align*}
To keep the imaginary part constant, we solve
\begin{equation*}
z_{I}=-z_{R}\frac{X}{Y}\big(-1+\sqrt{1+\frac{Y^{2}}{X^{2}}}\big)\quad.
\end{equation*}
This shows that the real part of the second derivative in terms of the real part $z_{R}$ is
\begin{equation*}
\frac{X}{2}(z_{R}^{2}-z_{I}^{2})-z_{R}z_{I}Y=z_{R}^{2}X\big((1+\frac{X^{2}}{Y^{2}})(-1+\sqrt{1+\frac{Y^{2}}{X^{2}}})\big)\quad.
\end{equation*}

\subsection*{One contour: $8b^{3}>27a^{2}$}
Naturally, this implies that $b>0$. The ugly part of the solution can be parametrised by
\begin{align*}
&\hspace{-8mm}R=|R|e^{i\varphi}\qquad\text{, where}\quad|R|=\sqrt{\frac{3b}{2}}\\
&\text{and}\quad \varphi=\frac{1}{3}\arccos(\frac{27a/8}{(3b/2)^{3/2}})=\frac{1}{3}\arcsin(\frac{\sqrt{(\frac{27a}{8})^{2}-(\frac{3b}{2})^{3}}}{(3b/2)^{3/2}})\quad.
\end{align*}
The zeroes $x_{m}$ from (\ref{e:rPxm}) are given by
\begin{equation*}
x_{m}=\frac{1}{3}\sqrt{\frac{3b}{2}}e^{\frac{2\pi im}{3}+i\varphi}+\frac{b}{2\sqrt{\frac{3b}{2}}}e^{-\frac{2\pi im}{3}-i\varphi}\quad
\end{equation*}
and thus real-valued. It depends on the value of $\varphi$ which value of $m$ yields the middle saddle point. The symmetry prescribes that the contour is parallel to the real axis at the saddle point, which makes integration simple.

\subsection*{Two contours: $8b^{3}<27a^{2}$}
In this case
\begin{equation*}
R=|R|=\big[\frac{27a}{8}+\sqrt{(\frac{27a}{8})^{2}-(\frac{3b}{2})^{3}}\big]^{\frac{1}{3}}\quad,
\end{equation*}
so that $m=0$ corresponds to a real root and $m=1,2$ give the complex conjugated roots, which are needed for the integration.

\subsection*{Two contours: $8b^{3}=27a^{2}$}
A special case is given by the case of zero discriminant. In this case only one contour is needed, but it is not continuously differentiable, as can be seen in Figure~\ref{f:rPp}. 
\begin{figure}[!htb]
\hspace{1.5cm}\includegraphics[width=0.6\textwidth]{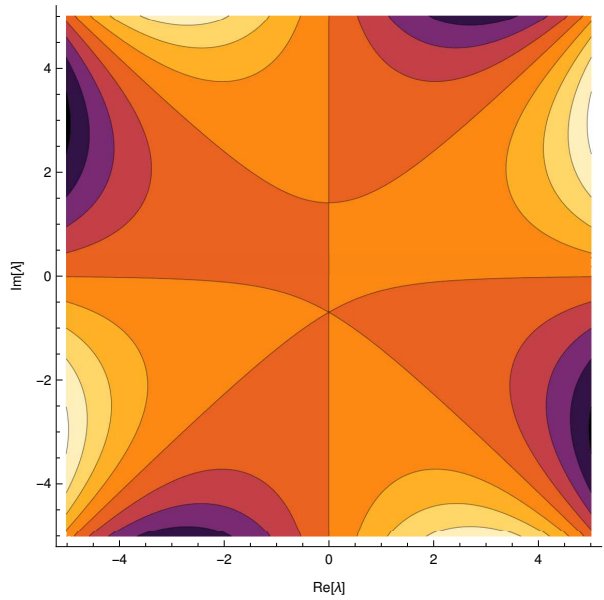}
\caption{The imaginary part of $f(2\sqrt{2},3,\lambda)$ for $\lambda\in\mathbb{C}$ and some contours of constant phase.\label{f:rPp}}
\end{figure}
Therefore, we split the integration curve into two half contours, starting/ending at the imaginary axis.

\begin{exm}
Picking $a=-24$ and $b=14$ yields zeroes at $2i,-3i$ and $i$. The approximation of the integral is given by
\begin{equation}
\mathpzc{I}_{P}=\sqrt{\frac{\pi}{b-6x_{2}^{2}}}\exp[\frac{b}{2}x_{2}^{2}+\frac{3a}{4}x_{2}]\quad,\label{e:Paxexm}
\end{equation}
where the root is given by
\begin{equation*}
x_{2}=\sqrt{\frac{b}{6}}\cos\big(\frac{4\pi}{3}+\frac{1}{3}\arccos(\frac{27a}{(6b)^{3/2}})\big)\quad.
\end{equation*}
To next step is to include polynomial factors by acting on the integral by $(i\partial_{a})^{k}$ and compare this to the actual integrals
\begin{equation}
\mathpzc{I}=\int_{-\infty}^{\infty}\ud\lambda \,\lambda^{k}\exp[-f(a,b;\lambda)]\quad.\label{e:Pexm}
\end{equation}
Some numerical results for this are shown in Table~\ref{t:Pexm}.
\begin{table}[!h]
\hspace*{3cm}\footnotesize{\begin{tabular}{l||l|l|r}
$k$& $\mathpzc{I}$&$\mathpzc{I}_{P}$& ratio\\\hline
$0$& $\phantom{-}1.01E-5$& $\phantom{-}1.04E-5$& $1.03$\\
$1$& $\phantom{-}9.75E-6i$& $\phantom{-}1.02E-5i$& $1.05$\\
$2$& $-8.83E-6$& $-9.36E-6$& $1.06$\\
$3$& $-7.45E-6i$& $-7.98E-6i$& $1.07$\\
$4$& $\phantom{-}5.80E-6$& $\phantom{-}6.26E-6$& $1.08$\\
$5$& $\phantom{-}4.10E-6i$& $\phantom{-}4.44E-6i$& $1.08$\\
$6$& $-2.54E-6$& $-2.76E-6$& $1.08$\\
$7$& $-1.30E-6i$& $-1.40E-6i$& $1.08$\\
$8$& $\phantom{-}4.58E-7$& $\phantom{-}4.84E-7$& $1.06$
\end{tabular}}
\caption{The integral (\ref{e:Pexm}) for some values of $k$ and $a=-24$, $b=14$ and the approximation (\ref{e:Paxexm}) based on the saddle points of $f$. Higher values of $k$ generally yields poorer results.\label{t:Pexm}}
\end{table}
\end{exm}

It is not surprising that this does not yield excellent results for higher $k$. We approximate the Pearcey function around one of it saddle points and try to reconstruct a modified version of this function on the complex plane from its derivatives.\\

All in all this is not a simple method. Additional complications are to be expected, since the partition function is not factorised here and that no routine for the determinant computation is given in this context. The obvious advantage of this method is the light numerical implementation. The steepest descent computation for the Pearcey function is computationally efficient, as discussed in Paragraph~\ref{sec:PI}.

\chapter*{Outlook \& Conlusions\label{sec:Conc}}
\addcontentsline{toc}{chapter}{Outlook \& Conlusions}

The main goal of this work is the free energy density for the vacuum sector to supplement the knowledge from the Schwinger-Dyson equations of the Grosse-Wulkenhaar model. The asymptotic volume of the diagonal subpolytope of symmetric stochastic matrices was determined to this end. Integrating against this volume factorizes the partition function, so that direct evaluation is possible.\\

Compared with the Schwinger-Dyson equations the dimension of the Euclidean spacetime plays a minor role in partition function methods. Additionally, there is no need to require specific kinetic eigenvalues. Instead of that, conditions on the differences between these eigenvalues are introduced. This allows an extension of this method to curved spacetime. Furthermore, no requirements on the amount of noncommutativity are made, so that this method can be used to study the $\Phi_{4}^{4}$-model on the Moyal plane for finite and vanishing noncommutativity. \\

However, the asymptotic polytope volume alters the structure of the partition function. In particular, it removes the matrix structure and replaces it by symmetric sums. A consequence of this is the appearance of fictitious divergencies in the regime of weak coupling. Although it was explicitly tried to avoid the common practice in perturbative \textsc{qft} of regulating and simply removing divergencies, it emerged from the asymptotic factorisation procedure in the end. This suggests that perturbative quantum field theory should be studied in another structure to avoid the divergency calculus. Another consequence is the restriction of the kinetic eigenvalues to their symmetric value.\\

The situation may be different for the strong coupling partition function factorised with the polytope volume. Therefore, some modifications of this method to evaluate the partition function in this regime have been discussed as well.\\

The factorisation using polytope volumes forms a method to treat a \textsc{qft} partition function. The advantages of this method compared to Schwinger-Dyson equations are twofold. On one side, the renormalisation of the free energy density for the vacuum sector provides additional insight into this model. On the other side, the explicit nature of all steps yields new ways to treat similar models. The price for this is an asymptotic structure at the heart of the model, which may alter its structure in abstruse ways.

\appendix
\chapter{Integration against the polytope volume for the free theory\label{sec:no_coup_comp}}

It is mainly lengthy bookkeeping that leads from (\ref{e:tc3}) to (\ref{e:tc6d}). The steps are elaborated in this appendix.\\

The integral over $\kappa$ in (\ref{e:tc3}) yields a Dirac delta $\delta(\sum_{j}x_{j})$, which is the essential part for the integration of the symmetric case, $e_{j}=\xi$ for all $j$. That this and closely related cases can be integrated follows from the setup for which the volume of the polytope of symmetric stochastic matrices with almost equal diagonal entries can be determined.

We continue from (\ref{e:tc3}) with the integral over $S$. Because $|x_{j}|\ll N^{\frac{1}{4}}$
\begin{align}
&\hspace{-8mm}\int_{0}^{N\zeta}\frac{\ud S}{S}\,\big(\frac{eS}{N}\big)^{\binom{N}{2}}\exp\big[-S\sum_{j}\big(\frac{e_{j}x_{j}}{\sqrt{N}}+\frac{e_{j}N-1}{N}\big)\big]\nonumber\\
&=\big(\frac{e}{N}\big)^{\binom{N}{2}}\,\Gamma\big(\binom{N}{2}\big)\,\exp\big[-\binom{N}{2}\cdot\log\big(\frac{N-1}{N}\sum_{j}e_{j}\big)\big]\nonumber\\
&\times\exp\big[-\binom{N}{2}\cdot\log\big(1+\sum_{j}\frac{e_{j}x_{j}\sqrt{N}}{(N\!-\!1)\sum_{k}e_{k}}\big)\big]\nonumber\\
&=\big(\frac{e}{(N\!-\!1)(\sum_{j}e_{j})}\big)^{\binom{N}{2}}\,\Gamma\big(\binom{N}{2}\big)\,\exp\big[-\binom{N}{2}\cdot\sum_{j}\frac{e_{j}x_{j}\sqrt{N}}{(N\!-\!1)\sum_{k}e_{k}}\big]\nonumber\\
&\exp\big[\frac{N(N\!-\!1)}{4}\big(\!\sum_{j}\!\frac{e_{j}x_{j}\sqrt{N}}{(N\!-\!1)\sum_{k}\!e_{k}}\big)^{2}-\frac{N(N\!-\!1)}{6}\big(\!\sum_{j}\!\frac{e_{j}x_{j}\sqrt{N}}{(N\!-\!1)\sum_{k}\!e_{k}}\big)^{3}\big]\nonumber\\
&\times\exp\big[\frac{N(N\!-\!1)}{8}\cdot\big(\sum_{j}\frac{e_{j}x_{j}\sqrt{N}}{(N\!-\!1)\sum_{k}e_{k}}\big)^{4}\big]\quad.\label{e:tc5}
\end{align}
The problem here is that the term linear in $x_{j}$ becomes too large for the stationary phase method. If all $e_{j}$ were (almost) equal, a solution appears. In this case $e_{k}-\sum_{j}e_{j}/N$ is small. And the integral over $\kappa$ yields $\sum_{j}x_{j}=0$, so that the average is small. This method would work, if we chose all eigenvalues close together. Writing 
\begin{equation*}
e_{j}=\xi(1+\tilde{\varepsilon}_{j})\qquad\text{, where }\xi=\frac{1}{N}\sum_{j=1}^{N}e_{j}
\end{equation*}
and $\tilde{\varepsilon}_{j}\ll N^{-\frac{1}{4}}$ is necessary to integrate this using the stationary phase method. 
The value of $\xi$ does not matter at this stage, since it falls out. For now, it may be assumes that $\tilde{\varepsilon}_{j}$ is so small that this can be integrated.\\

Applying that $\sum_{j}x_{j}e_{j}=\xi\sum_{j}x_{j}\tilde{\varepsilon}_{j}$ and $\sum_{j}e_{j}=N\xi$ to the combination of (\ref{e:tc3}) and (\ref{e:tc5}) yields
\begin{align*}
&\hspace{-8mm}\lim_{g\rightarrow0}\mathscr{Z}[0]=\frac{1}{2}\pi^{\binom{N}{2}}\big\{\!\prod_{k=1}^{N}\!\sqrt{\frac{\pi}{e_{k}}}\big\}(N\!-\!1)\!\!\int_{-\infty}^{\infty}\!\!\!\!\!\!\!\ud\kappa\!\int_{-\infty}^{\infty}\!\!\!\!\!\!\!\!\ud q\!\int_{0}^{\sqrt{N}}\!\!\!\!\!\!\!\!\ud Q\!\int_{-N^\frac{1}{4}}^{N^\frac{1}{4}}\!\!\ud\vec{x}\,\frac{\sqrt{2N}e^{7/6}}{(2\pi)^{N/2}}\nonumber\\
&\times\exp[-2\pi i\kappa\sum_{j}\!x_{j}]\exp[-\frac{N\!+\!2}{2N}\sum_{j}\!x_{j}^{2}+\frac{\sum_{j}x_{j}^{3}}{3\sqrt{N}}-\frac{\sum_{j}x_{j}^{4}}{4N}]\nonumber\\
&\times\exp[2\pi iq(Q-\sum_{j}\frac{x_{j}^{2}}{N})+\frac{Q^{2}}{4}]\big(\frac{e}{(N\!-\!1)N\xi}\big)^{\binom{N}{2}}\cdot\Gamma\big(\binom{N}{2}\big)\,\\
&\times\exp\big[-\binom{N}{2}\cdot\sum_{j}\frac{\tilde{\varepsilon}_{j}x_{j}}{(N\!-\!1)\sqrt{N}}+\frac{N(N\!-\!1)}{4}\big(\sum_{j}\frac{\tilde{\varepsilon}_{j}x_{j}}{(N\!-\!1)\sqrt{N}}\big)^{2}\big]\nonumber\\
&\exp\big[-\frac{N(N\!-\!1)}{6}\big(\sum_{j}\frac{\tilde{\varepsilon}_{j}x_{j}}{(N\!-\!1)\sqrt{N}}\big)^{3}\big]+\frac{1}{4}\binom{N}{2}\big(\sum_{j}\frac{\tilde{\varepsilon}_{j}x_{j}}{(N\!-\!1)\sqrt{N}}\big)^{4}\big]\\
&=\frac{1}{2}\pi^{\binom{N}{2}}\big\{\!\prod_{k=1}^{N}\!\sqrt{\frac{\pi}{e_{k}}}\big\}(N\!-\!1)\!\!\int_{-\infty}^{\infty}\!\!\!\!\!\!\!\ud\omega\int_{-\infty}^{\infty}\!\!\!\!\!\!\!\ud R\int_{-\infty}^{\infty}\!\!\!\!\!\!\!\ud\kappa\!\int_{-\infty}^{\infty}\!\!\!\!\!\!\!\!\ud q\!\int_{0}^{\sqrt{N}}\!\!\!\!\!\!\!\!\ud Q\,\frac{\sqrt{2N}e^{7/6}}{(2\pi)^{N/2}}\nonumber\\
&\times\exp[2\pi iqQ+2\pi i\omega R+\frac{Q^{2}}{4}]\big(\frac{e}{(N\!-\!1)N\xi}\big)^{\binom{N}{2}}\cdot\Gamma\big(\binom{N}{2}\big)\,\\
&\times\int_{-N^\frac{1}{4}}^{N^\frac{1}{4}}\!\!\ud\vec{x}\,\exp[\sum_{j}x_{j}\cdot\big(-2\pi i(\kappa+\omega \tilde{\varepsilon}_{j})-\frac{\tilde{\varepsilon}_{j}\sqrt{N}}{2}\big)]\\
&\times\exp[-\sum_{j}x_{j}^{2}\cdot\big(\frac{1}{2}+\frac{1}{N}+\frac{2\pi iq}{N}\big)]\exp[\frac{\sum_{j}x_{j}^{3}}{3\sqrt{N}}-\frac{\sum_{j}x_{j}^{4}}{4N}]\\
&\times\exp[\frac{R^{2}}{4(N-1)}-\frac{R^{3}}{6\sqrt{N}(N-1)^{2}}+\frac{R^{4}}{8N(N-1)^{3}}]\quad.
\end{align*}
Here the Dirac delta $\delta(R-\sum_{j}x_{j}\tilde{\varepsilon}_{j})$ was written by its Fourier transform. As before, this can be integrated using the stationary phase method from Lemma~\ref{l:aspm}, provided that $\tilde{\varepsilon}_{j}$ is small enough. We obtain
\begin{align*}
&\hspace{-8mm}\lim_{g\rightarrow0}\mathscr{Z}[0]=\frac{\sqrt{2N}e^{7/6}(N\!-\!1)}{2}\pi^{\binom{N}{2}}\big\{\!\prod_{k=1}^{N}\!\sqrt{\frac{\pi}{e_{k}}}\big\}\!\!\int_{-\infty}^{\infty}\!\!\frac{\ud\kappa}{\sqrt{N}}\!\int_{-\infty}^{\infty}\!\!\!\!\!\!\!\!\ud \omega\int_{-\infty}^{\infty}\!\!\!\!\!\!\!\!\ud R\int_{-\infty}^{\infty}\!\!\!\!\!\!\!\!\ud q\int_{0}^{\sqrt{N}}\!\!\!\!\!\!\!\!\!\ud Q\\
&\times\exp[2\pi iqQ+2\pi i\omega R+\frac{Q^{2}}{4}]\big(\frac{e}{(N\!-\!1)N\xi}\big)^{\binom{N}{2}}\cdot\Gamma\big(\binom{N}{2}\big)\,e^{-1-2\pi iq}\\
&\times\exp[\frac{NR^{2}}{4(N-1)}-\frac{NR^{3}}{6(N-1)^{2}}+\frac{NR^{4}}{8(N-1)^{3}}]\\
&\times\exp\big[\sum_{j}\frac{-4\pi^{2}\frac{(\kappa+\omega\tilde{\varepsilon}_{j})^{2}}{N}+2\pi i(\kappa+\omega\tilde{\varepsilon}_{j})\tilde{\varepsilon}_{j}+\frac{N}{4}\tilde{\varepsilon}_{j}^{2}}{2(1+\frac{2}{N}+\frac{4\pi iq}{N})}\big]\\
&\times\exp\big[\sum_{j}\frac{8\pi^{3}i\frac{(\kappa\!+\!\omega\tilde{\varepsilon}_{j})^{3}}{N^{\frac{3}{2}}}\!+\!6\pi^{2}\frac{(\kappa\!+\!\omega\tilde{\varepsilon}_{j})^{2}}{\sqrt{N}}\tilde{\varepsilon}_{j}\!-\!\frac{3}{2}\pi i(\kappa\!+\!\omega\tilde{\varepsilon}_{j})\sqrt{N}\tilde{\varepsilon}_{j}^{2}\!-\!\frac{N^{\frac{3}{2}}}{8}\tilde{\varepsilon}_{j}^{3}}{3\sqrt{N}(1+\frac{2}{N}+\frac{4\pi iq}{N})^{3}}\big]\\
&\times\exp\big[\sum_{j}\frac{[-2\pi i\frac{(\kappa+\omega\tilde{\varepsilon}_{j})}{\sqrt{N}}-\frac{\sqrt{N}}{2}\tilde{\varepsilon}_{j}]^{4}}{-4N(1+\frac{2}{N}+\frac{4\pi iq}{N})^{4}}+\sum_{j}\frac{[-2\pi i\frac{(\kappa+\omega\tilde{\varepsilon}_{j})}{\sqrt{N}}-\frac{\sqrt{N}}{2}\tilde{\varepsilon}_{j}]^{4}}{2N(1+\frac{2}{N}+\frac{4\pi iq}{N})^{5}}\big]\\
&\times\exp\big[\sum_{j}\frac{-2\pi i\frac{(\kappa+\omega\tilde{\varepsilon}_{j})}{\sqrt{N}}-\frac{\sqrt{N}}{2}\tilde{\varepsilon}_{j}}{\sqrt{N}(1+\frac{2}{N}+\frac{4\pi iq}{N})^{2}}\big]\\
&\times\exp\big[\sum_{j}2\frac{-4\pi^{2} \frac{(\kappa+\omega\tilde{\varepsilon}_{j})^{2}}{N} +2\pi i(\kappa+\omega\tilde{\varepsilon}_{j}) \tilde{\varepsilon}_{j}+\frac{N\tilde{\varepsilon}_{j}^{2}}{4}}{N(1+\frac{2}{N}+\frac{4\pi iq}{N})^{4}}\big]\\
&\times\exp\big[\sum_{j}-\frac{3}{2}\frac{-4\pi^{2}\frac{(\kappa+\omega\tilde{\varepsilon}_{j})^{2}}{N}+2\pi i(\kappa+\omega\tilde{\varepsilon}_{j}) \tilde{\varepsilon}_{j}+\frac{N\tilde{\varepsilon}_{j}^{2}}{4}}{N(1+\frac{2}{N}+\frac{4\pi iq}{N})^{3}}\big]\\
&\times\exp\big[\sum_{j}\frac{-3}{4N(1+\frac{2}{N}+\frac{4\pi iq}{N})^{2}}\big]\exp\big[\sum_{j}\frac{5}{6N(1+\frac{2}{N}+\frac{4\pi iq}{N})^{3}}\big]\quad,
\end{align*}
where we have rescaled $\kappa\rightarrow\kappa/\sqrt{N}$, $\omega\rightarrow  \omega/\sqrt{N}$ and $R\rightarrow R\sqrt{N}$.\\
Collecting relevant terms yields
\begin{align}
&\hspace{-8mm}\lim_{g\rightarrow0}\mathscr{Z}[0]=\frac{N\!-\!1}{2}\big(\frac{e\pi}{(N\!-\!1)N\xi}\big)^{\binom{N}{2}}\big\{\!\prod_{k=1}^{N}\!\sqrt{\frac{\pi}{e_{k}}}\big\}\!\!\int_{-\infty}^{\infty}\!\!\!\!\!\!\!\!\ud \omega\!\int_{-\infty}^{\infty}\!\!\!\!\!\!\!\!\ud R\!\int_{-\infty}^{\infty}\!\!\!\!\!\!\!\!\ud q\!\int_{0}^{\sqrt{N}}\!\!\!\!\!\!\!\!\ud Q\,\sqrt{2}e^{1/4}\nonumber\\
&\times\Gamma\big(\binom{N}{2}\big)\,e^{2\pi iq(Q-1)+\frac{Q^{2}}{4}+2\pi i\omega R}\nonumber\\
&\times \exp[\frac{NR^{2}}{4(N-1)}-\frac{NR^{3}}{6(N-1)^{2}}+\frac{NR^{4}}{8(N-1)^{3}}]\nonumber\\
&\times \exp\big[\frac{N}{8}\sum_{j}\!\tilde{\varepsilon}_{j}^{2}-\frac{1}{4}\sum_{j}\!\tilde{\varepsilon}_{j}^{2}-\frac{\pi iq}{2}\sum_{j}\!\tilde{\varepsilon}_{j}^{2}+\frac{2\pi iq}{N^{2}}\sum_{j}\!\tilde{\varepsilon}_{j}^{2}-\frac{2\pi^{2}q^{2}}{N}\sum_{j}\!\tilde{\varepsilon}_{j}^{2}\nonumber\\
&-\frac{N}{24}\sum_{j}\tilde{\varepsilon}_{j}^{3}+\frac{1}{4}\sum_{j}\tilde{\varepsilon}_{j}^{3}+\frac{\pi iq}{2}\sum_{j}\tilde{\varepsilon}_{j}^{3}+\frac{N}{64}\sum_{j}\tilde{\varepsilon}_{j}^{4}-\frac{3}{16}\sum_{j}\tilde{\varepsilon}_{j}^{4}\big]\nonumber\\
&\times\exp[-\frac{3\pi iq}{8}\sum_{j}\tilde{\varepsilon}_{j}^{4}+\frac{1}{8}\sum_{j}\tilde{\varepsilon}_{j}^{2}-2\pi^{2}\omega^{2}(1-\frac{1}{N}-\frac{4\pi iq}{N})\sum_{j}\frac{\tilde{\varepsilon}_{j}^{2}}{N}]\nonumber\\
&\times\exp[\frac{8}{3}\pi^{3}i\omega^{3}\sum_{j}\frac{\tilde{\varepsilon}_{j}^{3}}{N^{2}}+4\pi^{4}\omega^{4}\sum_{j}\frac{\tilde{\varepsilon}_{j}^{4}}{N^{3}}\pi i\omega(1-\frac{1}{N}-\frac{4\pi iq}{N})\sum_{j}\tilde{\varepsilon}_{j}^{2}]\nonumber\\
&\times\exp[-\frac{1}{2}\pi i\omega\!\sum_{j}\tilde{\varepsilon}_{j}^{3}+2\pi^{2}\omega^{2}\!\sum_{j}\frac{\tilde{\varepsilon}_{j}^{3}}{N}+\frac{1}{4}\pi i\omega\sum_{j}\tilde{\varepsilon}_{j}^{4}-\frac{3}{2}\pi^{2} \omega^{2}\sum_{j}\frac{\tilde{\varepsilon}_{j}^{4}}{N}]\nonumber\\
&\times \exp[-4\pi^{3}i\omega^{3}\!\sum_{j}\!\frac{\tilde{\varepsilon}_{j}^{4}}{N^{2}}]\nonumber\\
&\times\int_{-\infty}^{\infty}\!\!\!\!\!\!\ud \kappa\,\exp[2\pi i\kappa\big(\!-1-\frac{1}{4}\!\sum_{j}\!\tilde{\varepsilon}_{j}^{2}+\frac{3\pi iq}{N}\!\sum_{j}\!\tilde{\varepsilon}_{j}^{2}\nonumber\\
&+\frac{1}{8}\sum_{j}\tilde{\varepsilon}_{j}^{3}-2\pi i\omega\sum_{j}\frac{\tilde{\varepsilon}_{j}^{2}}{N}+\frac{3}{2}\pi i\omega\sum_{j}\frac{\tilde{\varepsilon}_{j}^{2}}{N}-8\pi^{3}i\omega\sum_{j}\frac{\tilde{\varepsilon}_{j}^{3}}{N^{3}}\nonumber\\
&+4\pi^{2}\omega^{2}\sum_{j}\frac{\tilde{\varepsilon}_{j}^{3}}{N^{2}}-6\pi^{2}\omega^{2}\sum_{j}\frac{\tilde{\varepsilon}_{j}^{3}}{N^{2}}\big)]\times\exp[-2\pi^{2}\kappa^{2}\cdot\big(1+\frac{3}{4N}\sum_{j}\tilde{\varepsilon}_{j}^{2}\nonumber\\
&-\frac{1}{N}-\frac{4\pi iq}{N}+6\pi i\omega\!\sum_{j}\!\frac{\tilde{\varepsilon}_{j}^{2}}{N^{2}}-12\pi^{2}\omega^{2}\!\sum_{j}\!\frac{\tilde{\varepsilon}_{j}^{2}}{N^{3}}\big)]\nonumber\\
&\times\exp[\frac{8\pi^{3}i\kappa^{3}}{3N}+\frac{4\pi^{4}\kappa^{4}}{N^{2}}]\quad.\label{e:tc5b}
\end{align}

Although this is not precisely the setup of Lemma~\ref{l:aspm}, it is not difficult to see that this can be integrated in the same manner. Shift the integration parameter $\kappa$, so that the linear part is absorbed in the quadratic term. The cubic and quartic terms will now give rise to new linear, quadratic and cubic terms. The last five lines in (\ref{e:tc5b}) yield by the stationary phase method from Lemma~\ref{l:aspm}
\begin{align*}
&\hspace{-8mm}\frac{1}{\sqrt{2\pi}}\exp\!\big[\frac{-1}{2}\!-\!\frac{1}{32}(\sum_{j}\!\tilde{\varepsilon}_{j}^{2})^{2}\!\cdot\!\big(1\!-\!\frac{3}{4N}\!\sum_{j}\!\tilde{\varepsilon}_{j}^{2}\!-\!\frac{6\pi i\omega}{N^{2}}\!\sum_{j}\!\tilde{\varepsilon}_{j}^{2}\!+\!\frac{4\pi iq}{N}\!-\!\frac{6\pi iq}{N^{2}}\!\sum_{j}\!\tilde{\varepsilon}_{j}^{2}\big)\\
&-\frac{1}{128}(\sum_{j}\tilde{\varepsilon}_{j}^{3})^{2}\cdot\big(1+\frac{4\pi iq}{N}\big)-\frac{1}{4}\sum_{j}\tilde{\varepsilon}_{j}^{2}\cdot(1+\frac{4\pi iq}{N})+\frac{1}{8}\sum_{j}\tilde{\varepsilon}_{j}^{3}\\
&\!+\!\frac{\pi i\omega}{2N}\!\sum_{j}\!\tilde{\varepsilon}_{j}^{2}\!+\!\frac{3\pi iq}{N}\!\sum_{j}\!\tilde{\varepsilon}_{j}^{2}\!+\!\frac{1}{32}(\!\sum_{j}\!\tilde{\varepsilon}_{j}^{2})(\!\sum_{j}\!\tilde{\varepsilon}_{j}^{3})\cdot\big(1\!-\!\frac{3}{4N}\!\sum_{j}\!\tilde{\varepsilon}_{j}^{2}\!+\!\frac{4\pi iq}{N}\big)\\
&+\frac{\pi i\omega}{8N}(\sum_{j}\tilde{\varepsilon}_{j}^{2})^{2}\cdot\big(1-\frac{3}{4N}\sum_{j}\tilde{\varepsilon}_{j}^{2}\big)+\frac{3\pi iq}{4N}(\sum_{j}\tilde{\varepsilon}_{j}^{2})^{2}\cdot\big(1-\frac{3}{4}\sum_{j}\tilde{\varepsilon}_{j}^{2}\big)\\
&-\frac{\pi i\omega}{16N}(\sum_{j}\tilde{\varepsilon}_{j}^{2})(\sum_{j}\tilde{\varepsilon}_{j}^{3})-\frac{3\pi iq}{8N}(\sum_{j}\tilde{\varepsilon}_{j}^{2})(\sum_{j}\tilde{\varepsilon}_{j}^{3})\big]\\
&\times\exp[-\frac{1}{192N}(\sum_{j}\tilde{\varepsilon}_{j}^{2})^{3}-\frac{\pi i\omega}{32N^{2}}(\sum_{j}\tilde{\varepsilon}_{j}^{2})^{3}+\frac{1}{128N}(\sum_{j}\tilde{\varepsilon}_{j}^{2})^{2}(\sum_{j}\tilde{\varepsilon}_{j}^{3})\\
&+\frac{3\pi iq}{16N^{2}}(\sum_{j}\tilde{\varepsilon}_{j}^{2})^{3}-\frac{\pi iq}{16N^{2}}(\sum_{j}\tilde{\varepsilon}_{j}^{2})^{3}]\quad.
\end{align*}

The integral over $\omega$ can be performed now, resulting in the identification
\begin{align*}
&\hspace{-8mm}R=-\frac{1}{2}\sum_{j}\tilde{\varepsilon}_{j}^{2}+\frac{2\pi iq}{N}\sum_{j}\tilde{\varepsilon}_{j}^{2}+\frac{1}{4}\sum_{j}\tilde{\varepsilon}_{j}^{3}-\frac{1}{8}\sum_{j}\tilde{\varepsilon}_{j}^{4}-\frac{(\sum_{j}\tilde{\varepsilon}_{j}^{2})^{2}}{16N}\\
&+\frac{(\sum_{j}\tilde{\varepsilon}_{j}^{2})(\sum_{j}\tilde{\varepsilon}_{j}^{3})}{32N}\quad.
\end{align*}

This gives
\begin{align}
&\hspace{-8mm}\lim_{g\rightarrow0}\mathscr{Z}[0]=\frac{1}{2}\big(\frac{e\pi}{(N\!-\!1)N\xi}\big)^{\binom{N}{2}}\big\{\!\prod_{k=1}^{N}\!\sqrt{\frac{\pi}{e_{k}}}\big\}(N\!-\!1)\!\!\int_{-\infty}^{\infty}\!\!\!\!\!\!\!\!\ud q\!\int_{0}^{\sqrt{N}}\!\!\!\!\!\!\!\!\ud Q\,\frac{e^{1/4}}{\sqrt{\pi}}\nonumber\\
&\times\Gamma\big(\binom{N}{2}\big)\,e^{2\pi iq(Q-1)+\frac{Q^{2}}{4}}\nonumber\\
&\times \exp\Big[\frac{N}{4(N-1)}\big(\frac{1}{4}(\sum_{j}\tilde{\varepsilon}_{j}^{2})^{2}-\frac{1}{4}(\sum_{j}\tilde{\varepsilon}_{j}^{2})(\sum_{j}\tilde{\varepsilon}_{j}^{3})+\frac{1}{8}(\sum_{j}\tilde{\varepsilon}_{j}^{2})(\sum_{j}\tilde{\varepsilon}_{j}^{4})]\nonumber\\
&+\frac{1}{16N}(\sum_{j}\tilde{\varepsilon}_{j}^{2})^{3}-\frac{1}{32N}(\sum_{j}\tilde{\varepsilon}_{j}^{2})^{2}(\sum_{j}\tilde{\varepsilon}_{j}^{3})+\frac{1}{16}(\sum_{j}\tilde{\varepsilon}_{j}^{3})^{2}]\nonumber\\
&-\frac{1}{16}(\sum_{j}\tilde{\varepsilon}_{j}^{3})(\sum_{j}\tilde{\varepsilon}_{j}^{4})-\frac{1}{32N}(\sum_{j}\tilde{\varepsilon}_{j}^{2})^{2}(\sum_{j}\tilde{\varepsilon}_{j}^{3})-\frac{2\pi iq}{N}(\sum_{j}\tilde{\varepsilon}_{j}^{2})^{2}\nonumber\\
&-\frac{4\pi^{2} q^{2}}{N^{2}}(\sum_{j}\tilde{\varepsilon}_{j}^{2})^{2}+\frac{\pi iq}{N}(\sum_{j}\tilde{\varepsilon}_{j}^{2})(\sum_{j}\tilde{\varepsilon}_{j}^{3})\big)\Big]\nonumber\\
&\times\exp\Big[-\frac{N}{6(N-1)^{2}}\big(-\frac{1}{8}(\sum_{j}\tilde{\varepsilon}_{j}^{2})^{3}+\frac{3}{16}(\sum_{j}\tilde{\varepsilon}_{j}^{2})^{2}(\sum_{j}\tilde{\varepsilon}_{j}^{3})\big)\Big]\nonumber\\
&\times \exp\big[\frac{N}{8}\sum_{j}\tilde{\varepsilon}_{j}^{2}-\frac{1}{4}\sum_{j}\tilde{\varepsilon}_{j}^{2}-\frac{\pi iq}{2}\sum_{j}\tilde{\varepsilon}_{j}^{2}-\frac{N}{24}\sum_{j}\tilde{\varepsilon}_{j}^{3}+\frac{1}{4}\sum_{j}\tilde{\varepsilon}_{j}^{3}\nonumber\\
&+\frac{\pi iq}{2}\sum_{j}\tilde{\varepsilon}_{j}^{3}+\frac{N}{64}\sum_{j}\tilde{\varepsilon}_{j}^{4}-\frac{3}{16}\sum_{j}\tilde{\varepsilon}_{j}^{4}-\frac{3\pi iq}{8}\sum_{j}\tilde{\varepsilon}_{j}^{4}+\frac{1}{8}\sum_{j}\tilde{\varepsilon}_{j}^{2}\big]\nonumber\\
&\times \exp[-\frac{1}{2}-\frac{1}{32}(\sum_{j}\tilde{\varepsilon}_{j}^{2})^{2}+\frac{1}{32N}(\sum_{j}\tilde{\varepsilon}_{j}^{2})^{2}-\frac{1}{128}(\sum_{j}\tilde{\varepsilon}_{j}^{3})^{2}]\nonumber\\
&\times\exp[-\frac{1}{4}\sum_{j}\tilde{\varepsilon}_{j}^{2}+\frac{1}{8}\sum_{j}\tilde{\varepsilon}_{j}^{3}+\frac{1}{32}(\sum_{j}\tilde{\varepsilon}_{j}^{2})(\sum_{j}\tilde{\varepsilon}_{j}^{3})]\nonumber\\
&\times\exp[-\frac{1}{192N}(\sum_{j}\tilde{\varepsilon}_{j}^{2})^{3}+\frac{1}{128N}(\sum_{j}\tilde{\varepsilon}_{j}^{2})^{2}(\sum_{j}\tilde{\varepsilon}_{j}^{3})]\label{e:tc6}\quad.
\end{align}
The integral over $q$ yields another Dirac delta that assigns 
\begin{equation*}
Q=1+\frac{1}{4}\sum_{j}\tilde{\varepsilon}_{j}^{2}-\frac{1}{4}\sum_{j}\tilde{\varepsilon}_{j}^{3}+\frac{3}{16}\sum_{j}\tilde{\varepsilon}_{j}^{4}+\frac{1}{4N}(\sum_{j}\tilde{\varepsilon}_{j}^{2})^{2}-\frac{1}{8N}(\sum_{j}\tilde{\varepsilon}_{j}^{2})(\sum_{j}\tilde{\varepsilon}_{j}^{3})\;.
\end{equation*}
This yields
\begin{align}
&\hspace{-8mm}\lim_{g\rightarrow0}\mathscr{Z}[0]=\frac{1}{2}\big(\frac{e\pi}{(N\!-\!1)N\xi}\big)^{\binom{N}{2}}\big\{\!\prod_{k=1}^{N}\!\sqrt{\frac{\pi}{e_{k}}}\big\}(N\!-\!1)\frac{e^{1/4}}{\sqrt{\pi}}\,\Gamma\big(\binom{N}{2}\big)\,\nonumber\\
&\times\exp[\frac{1}{4}+\frac{1}{64}(\sum_{j}\tilde{\varepsilon}_{j}^{2})^{2}+\frac{1}{64}(\sum_{j}\tilde{\varepsilon}_{j}^{3})^{2}+\frac{1}{8}\sum_{j}\tilde{\varepsilon}_{j}^{2}-\frac{1}{8}\sum_{j}\tilde{\varepsilon}_{j}^{3}\nonumber\\
&-\frac{1}{32}(\sum_{j}\tilde{\varepsilon}_{j}^{2})(\sum_{j}\tilde{\varepsilon}_{j}^{3})-\frac{3}{128}(\sum_{j}\tilde{\varepsilon}_{j}^{3})(\sum_{j}\tilde{\varepsilon}_{j}^{4})-\frac{1}{32N}(\sum_{j}\tilde{\varepsilon}_{j}^{2})^{2}(\sum_{j}\tilde{\varepsilon}_{j}^{3})\nonumber\\
&+\frac{3}{128}(\sum_{j}\tilde{\varepsilon}_{j}^{2})(\sum_{j}\tilde{\varepsilon}_{j}^{4})+\frac{1}{32N}(\sum_{j}\tilde{\varepsilon}_{j}^{2})^{3}-\frac{1}{64N}(\sum_{j}\tilde{\varepsilon}_{j}^{2})^{2}(\sum_{j}\tilde{\varepsilon}_{j}^{3})]\nonumber\\
&\times \exp\Big[\frac{1}{16}(\sum_{j}\tilde{\varepsilon}_{j}^{2})^{2}-\frac{1}{16}(\sum_{j}\tilde{\varepsilon}_{j}^{2})(\sum_{j}\tilde{\varepsilon}_{j}^{3})+\frac{1}{32}(\sum_{j}\tilde{\varepsilon}_{j}^{2})(\sum_{j}\tilde{\varepsilon}_{j}^{4})]\nonumber\\
&+\frac{1}{64N}(\sum_{j}\tilde{\varepsilon}_{j}^{2})^{3}-\frac{1}{128N}(\sum_{j}\tilde{\varepsilon}_{j}^{2})^{2}(\sum_{j}\tilde{\varepsilon}_{j}^{3})+\frac{1}{64}(\sum_{j}\tilde{\varepsilon}_{j}^{3})^{2}]\nonumber\\
&-\frac{1}{64}(\sum_{j}\tilde{\varepsilon}_{j}^{3})(\sum_{j}\tilde{\varepsilon}_{j}^{4})-\frac{1}{128N}(\sum_{j}\tilde{\varepsilon}_{j}^{2})^{2}(\sum_{j}\tilde{\varepsilon}_{j}^{3})\Big]\nonumber\\
&\times\exp\Big[\frac{1}{48N}(\sum_{j}\tilde{\varepsilon}_{j}^{2})^{3}-\frac{1}{32N}(\sum_{j}\tilde{\varepsilon}_{j}^{2})^{2}(\sum_{j}\tilde{\varepsilon}_{j}^{3})\Big]\nonumber\\
&\times \exp\big[\frac{N}{8}\sum_{j}\tilde{\varepsilon}_{j}^{2}-\frac{1}{4}\sum_{j}\tilde{\varepsilon}_{j}^{2}-\frac{N}{24}\sum_{j}\tilde{\varepsilon}_{j}^{3}+\frac{1}{4}\sum_{j}\tilde{\varepsilon}_{j}^{3}\nonumber\\
&+\frac{N}{64}\sum_{j}\tilde{\varepsilon}_{j}^{4}-\frac{3}{16}\sum_{j}\tilde{\varepsilon}_{j}^{4}+\frac{1}{8}\sum_{j}\tilde{\varepsilon}_{j}^{2}\big]\nonumber\\
&\times \exp[-\frac{1}{2}-\frac{1}{32}(\sum_{j}\tilde{\varepsilon}_{j}^{2})^{2}+\frac{1}{32N}(\sum_{j}\tilde{\varepsilon}_{j}^{2})^{2}-\frac{1}{128}(\sum_{j}\tilde{\varepsilon}_{j}^{3})^{2}]\nonumber\\
&\times\exp[-\frac{1}{4}\sum_{j}\tilde{\varepsilon}_{j}^{2}+\frac{1}{8}\sum_{j}\tilde{\varepsilon}_{j}^{3}+\frac{1}{32}(\sum_{j}\tilde{\varepsilon}_{j}^{2})(\sum_{j}\tilde{\varepsilon}_{j}^{3})]\nonumber\\
&\times\exp[-\frac{1}{192N}(\sum_{j}\tilde{\varepsilon}_{j}^{2})^{3}+\frac{1}{128N}(\sum_{j}\tilde{\varepsilon}_{j}^{2})^{2}(\sum_{j}\tilde{\varepsilon}_{j}^{3})]\nonumber\\
&=\frac{1}{2}\big(\frac{e\pi}{(N\!-\!1)N\xi}\big)^{\binom{N}{2}}\big\{\!\prod_{k=1}^{N}\!\sqrt{\frac{\pi}{e_{k}}}\big\}(N\!-\!1)\frac{1}{\sqrt{\pi}}\,\Gamma\big(\binom{N}{2}\big)\,\nonumber\\
&\times\exp[\frac{N\!-\!2}{8}\sum_{j}\tilde{\varepsilon}_{j}^{2}-\frac{N\!-\!6}{24}\sum_{j}\tilde{\varepsilon}_{j}^{3}+\frac{N}{64}\sum_{j}\tilde{\varepsilon}_{j}^{4}+\frac{3}{64}(\sum_{j}\tilde{\varepsilon}_{j}^{2})^{2}]\nonumber\\
&\times\exp[-\frac{1}{16}(\sum_{j}\tilde{\varepsilon}_{j}^{2})(\sum_{j}\tilde{\varepsilon}_{j}^{3})\frac{7}{128}(\sum_{j}\tilde{\varepsilon}_{j}^{2})(\sum_{j}\tilde{\varepsilon}_{j}^{4})+\frac{3}{128}(\sum_{j}\tilde{\varepsilon}_{j}^{3})^{2}]\nonumber\\
&\times\exp[-\frac{5}{128}(\sum_{j}\tilde{\varepsilon}_{j}^{3})(\sum_{j}\tilde{\varepsilon}_{j}^{4})+\frac{1}{16N}(\sum_{j}\tilde{\varepsilon}_{j}^{2})^{3}]\nonumber\\
&\exp[-\frac{11}{128N}(\sum_{j}\tilde{\varepsilon}_{j}^{2})^{2}(\sum_{j}\tilde{\varepsilon}_{j}^{3})]\label{e:tc6c}\quad.
\end{align}

\bibliography{project.bib}{}
\bibliographystyle{unsrt}
\addcontentsline{toc}{chapter}{Bibliography}

\newpage
\begin{minipage}[t]{\textwidth}
\phantom{line\\}
\end{minipage}
\newpage
\chapter*{Lebenslauf}
\newpage
\chapter*{Acknowledgements}
\addcontentsline{toc}{chapter}{Acknowledgments}

First of all I would like to thank Raimar Wulkenhaar for his support, supervision and ideas for this project. Out of his enthusiasm for quantum field theory my fascination for exactly solvable \textsc{qft} and its developments has evolved.\\

Alex, Carlos, Jan and Romain have shaped the atmosphere in our working group. Besides teaching me much about quantum field theory, they have made my work much more pleasant. In addition, I want to express gratitude to Gabi. Her support made working here much easier.\\

Furthermore, I would like to thank friends and family for their support and all the distractions they provided me with.\\

The last words of this thesis can only be for Anne.








\end{document}